\documentclass[aps,pra,nofootinbib,superscriptaddress,10pt]{revtex4-2}

\newcommand{\lsp}{\hspace{0.1em}}

\usepackage{float}
\makeatletter
\let\newfloat\newfloat@ltx
\makeatother

\usepackage{minitoc}
\usepackage[toc,page,header]{appendix}
\usepackage{physics}
\usepackage{xfrac}
\usepackage{graphicx}
\usepackage{subfigure}
\usepackage{mathdots}
\usepackage{amsfonts,amssymb,amsmath}
\usepackage{amsthm}
\usepackage{mathrsfs}
\usepackage{mathtools}
\usepackage{dsfont}
\usepackage{longtable}
\usepackage{tabularx}
\usepackage{csquotes}
\MakeOuterQuote{"}
\usepackage{epstopdf}
\usepackage{tikz}
\usepackage{setspace}

\def\identity{\leavevmode\hbox{\small1\kern-3.8pt\normalsize1}}

\newtheorem{theorem}{Theorem}

\newtheorem{lemma}{Lemma}
\newtheorem{conjecture}{Conjecture}
\newtheorem{proposition}{Proposition}

\newtheorem{corollary}{Corollary}

\newcommand{\Cl}{\mathrm{Cl}}
\newcommand{\oCl}{\overline{\mathrm{Cl}}}
\newcommand{\Ob}{\mathfrak{O}}

\newcommand{\defe}{\mathrm{def}}
\newcommand{\even}{\mathrm{even}}
\newcommand{\odd}{\mathrm{odd}}
\newcommand{\sym}{\mathrm{sym}}
\newcommand{\Sym}{\mathrm{Sym}}
\newcommand{\iso}{\mathrm{iso}}
\newcommand{\LD}{\mathrm{LD}}
\newcommand{\RD}{\mathrm{RD}}
\newcommand{\Sp}{\mathrm{Sp}}
\newcommand{\ASp}{\mathrm{ASp}}

\newcommand{\imply}{\Rightarrow}
\newcommand{\ind}{\operatorname{ind}}
\newcommand{\mmod}{\!\mod}
\newcommand{\rk}{\operatorname{rank}}
\newcommand{\spa}{\operatorname{span}}
\newcommand{\id}{\mathrm{id}}

\newcommand{\supp}{\operatorname{supp}}

\newcommand{\Stab}{\mathrm{Stab}}

\newcommand{\totimes}{\tilde{\otimes}}

\newcommand{\sh}{\mathrm{sh}}
\newcommand{\qb}{\mathrm{QB}}
\newcommand{\QB}{\mathrm{QB}}

\newcommand{\ns}{\mathrm{ns}}
\newcommand{\orb}{\mathrm{orb}}

\newcommand{\can}{\mathrm{can}}

\newcommand{\bbE}{\mathbb{E}}
\newcommand{\bbF}{\mathbb{F}}
\newcommand{\bbI}{\mathbb{I}}
\newcommand{\bbN}{\mathbb{N}}
\newcommand{\bbR}{\mathbb{R}}
\newcommand{\bbT}{\mathbb{T}}
\newcommand{\bbZ}{\mathbb{Z}}

\newcommand{\bfg}{\mathbf{g}}

\newcommand{\bfp}{\mathbf{p}}
\newcommand{\bfq}{\mathbf{q}}

\newcommand{\bfu}{\mathbf{u}}
\newcommand{\bfv}{\mathbf{v}}

\newcommand{\bfx}{\mathbf{x}}
\newcommand{\bfy}{\mathbf{y}}
\newcommand{\bfz}{\mathbf{z}}

\newcommand{\bQ}{\bar{Q}}
\newcommand{\bPhi}{\bar{\Phi}}
\newcommand{\bcaQ}{\bar{\mathcal{Q}}}

\newcommand{\caC}{\mathcal{C}}
\newcommand{\caD}{\mathcal{D}}

\newcommand{\caH}{\mathcal{H}}
\newcommand{\caL}{\mathcal{L}}
\newcommand{\caM}{\mathcal{M}}
\newcommand{\caN}{\mathcal{N}}
\newcommand{\caO}{\mathcal{O}}

\newcommand{\caQ}{\mathcal{Q}}
\newcommand{\caR}{\mathcal{R}}
\newcommand{\caS}{\mathcal{S}}
\newcommand{\caT}{\mathcal{T}}

\newcommand{\caV}{\mathcal{V}}
\newcommand{\caW}{\mathcal{W}}

\newcommand{\hka}{\hat{\kappa}}

\newcommand{\rmd}{\mathrm{d}}
\newcommand{\rme}{\operatorname{e}}
\newcommand{\rmi}{\mathrm{i}}

\newcommand{\rmB}{\mathrm{B}}

\newcommand{\rmH}{\mathrm{H}}

\newcommand{\rmU}{\mathrm{U}}

\newcommand{\scrA}{\mathscr{A}}
\newcommand{\scrC}{\mathscr{C}}
\newcommand{\scrE}{\mathscr{E}}

\newcommand{\scrM}{\mathscr{M}}

\newcommand{\scrS}{\mathscr{S}}
\newcommand{\scrT}{\mathscr{T}}
\newcommand{\scrV}{\mathscr{V}}

\newcommand{\tf}{\tilde{f}}
\newcommand{\tbfx}{\tilde{\mathbf{x}}}
\newcommand{\tbfy}{\tilde{\mathbf{y}}}

\newcommand{\tx}{\tilde{x}}
\newcommand{\ty}{\tilde{y}}

\newcommand{\tg}{\tilde{g}}

\newcommand{\tR}{\tilde{R}}

\newcommand{\tka}{\tilde{\kappa}}
\newcommand{\tomega}{\tilde{\omega}}

\newcommand{\tmu}{\tilde{\mu}}
\newcommand{\tDelta}{{\tilde{\Delta}}}
\newcommand{\tPsi}{\tilde{\Psi}}

\newcommand{\tcaN}{\tilde{\mathcal{N}}}
\newcommand{\tcaT}{\tilde{\mathcal{T}}}

\newcommand{\tscrM}{\tilde{\mathscr{M}}}
\newcommand{\tscrP}{\tilde{\mathscr{P}}}

\renewcommand{\epsilon}{\varepsilon}

\newcommand{\vertiii}[1]{{\left\vert\kern-0.25ex\left\vert\kern-0.25ex\left\vert #1 
		\right\vert\kern-0.25ex\right\vert\kern-0.25ex\right\vert}}

\newcommand{\lref}[1]{Lemma~\ref{#1}}
\newcommand{\lsref}[1]{Lemmas~\ref{#1}}
\newcommand{\Lref}[1]{Lemma~\ref{#1}}
\newcommand{\Lsref}[1]{Lemmas~\ref{#1}}

\newcommand{\thref}[1]{Theorem~\ref{#1}}
\newcommand{\thsref}[1]{Theorems~\ref{#1}}
\newcommand{\Thref}[1]{Theorem~\ref{#1}}
\newcommand{\Thsref}[1]{Theorems~\ref{#1}}

\newcommand{\pref}[1]{Proposition~\ref{#1}}
\newcommand{\psref}[1]{Propositions~\ref{#1}}
\newcommand{\Pref}[1]{Proposition~\ref{#1}}

\newcommand{\coref}[1]{Corollary~\ref{#1}}
\newcommand{\cosref}[1]{Corollaries~\ref{#1}}

\newcommand{\cref}[1]{Conjecture~\ref{#1}}

\newcommand{\Cref}[1]{Conjecture~\ref{#1}}

\def\eqref#1{\textup{(\ref{#1})}}
\newcommand{\eref}[1]{Eq.~\textup{(\ref{#1})}}
\newcommand{\eqsref}[2]{Eqs.~(\ref{#1}) and (\ref{#2})}
\newcommand{\eqssref}[3]{Eqs.~(\ref{#1}),  (\ref{#2}), and (\ref{#3})}

\newcommand{\Eref}[1]{Equation~\textup{(\ref{#1})}}
\newcommand{\Eqsref}[2]{Equations~(\ref{#1}) and (\ref{#2})}

\newcommand{\tref}[1]{Table~\ref{#1}}

\newcommand{\tsref}[1]{Tables~\ref{#1}}

\newcommand{\sref}[1]{Sec.~\ref{#1}}

\newcommand{\ssref}[1]{Secs.~\ref{#1}}

\newcommand{\fref}[1]{Fig.~\ref{#1}}
\newcommand{\Fref}[1]{Figure~\ref{#1}}
\newcommand{\fsref}[1]{Figs.~\ref{#1}}
\newcommand{\Fsref}[1]{Figures~\ref{#1}}

\newcommand{\aref}[1]{Appendix~\ref{#1}}

\def\<{\langle}  
\def\>{\rangle}  

\newcommand{\rcite}[1]{Ref.~\cite{#1}}
\newcommand{\rscite}[1]{Refs.~\cite{#1}}

\begin{document}
	\title{Third moments of qudit Clifford orbits and 3-designs based on magic orbits}

\author{Huangjun Zhu}

\email{zhuhuangjun@fudan.edu.cn}

\author{Chengsi Mao}

\author{Changhao Yi}

\affiliation{State Key Laboratory of Surface Physics, Department of Physics, and Center for Field Theory and Particle Physics, Fudan University, Shanghai 200433, China}

\affiliation{Institute for Nanoelectronic Devices and Quantum Computing, Fudan University, Shanghai 200433, China}

\affiliation{Shanghai Research Center for Quantum Sciences, Shanghai 201315, China}

\date{\today}
	
\begin{abstract}
When the local dimension $d$ is an odd prime, the qudit Clifford group is only a 2-design, but not a 3-design, unlike the qubit counterpart. This distinction and its extension to Clifford orbits have profound implications for many applications in quantum information processing. In this work we systematically delve into  general qudit Clifford orbits with a focus on  the third moments and potential applications in shadow estimation. 
First,  we introduce the shadow norm to quantify the deviations of Clifford orbits from 3-designs and clarify its properties. Then, we show that  the third normalized frame potential and shadow norm are both $\caO(d)$ for any Clifford orbit, including the orbit of stabilizer states, although the operator norm of the third normalized moment operator may increase exponentially with the number $n$ of qudits when $d\neq 2\mmod 3$.  Moreover, we prove that the shadow norm of any magic orbit is upper bounded by the constant $15/2$, so a single magic gate can already eliminate the $\caO(d)$ overhead in qudit shadow estimation and bridge the gap between qudit systems and qubit systems. Furthermore, we propose simple recipes for constructing approximate and exact 3-designs (with respect to three figures of merit simultaneously) from  one or a few Clifford orbits. Notably, accurate approximate 3-designs 
can be constructed from only two Clifford orbits. For an infinite family of local dimensions, exact 3-designs can be constructed from  two or four Clifford orbits. In the course of study, we clarify the key properties of the commutant of the third Clifford tensor power and the underlying mathematical structures. 
\end{abstract}

\maketitle
\tableofcontents

\section{Introduction}
Random quantum states  are ubiquitous tools in quantum information processing. However, in many applications, it is unrealistic to produce Haar random quantum states and partial derandomization is indispensable. Fortunately, 
quite often it suffices to use finite ensembles of quantum states that can match Haar random states up to  the first $t$ moments for some positive integer $t$, which leads to the concept of  \emph{complex projective $t$-designs}   \cite{DelsGS77,Hogg82,Zaun11,ReneBSC04, Scot06,AmbaE07,ZhuKGG16}. Intuitively, a $t$-design is a (weighted) set of vectors on the complex unit sphere in a finite-dimensional Hilbert space $\caH$ that are  "evenly distributed", and a large value of $t$ means a high degree of uniformity. \emph{Unitary $t$-designs} are a natural generalization of complex projective $t$-designs to the unitary group \cite{Dank05the, GrosAE07,DankCEL09,RoyS09}. They can replace Haar random unitaries in numerous applications and are equally useful in various research areas.

In the past two decades, complex projective designs and unitary design have found increasing applications in quantum information processing,  including randomized benchmarking \cite{KnilLRB08,MageGE11}, quantum  tomography \cite{HayaHH05,Scot06,Scot08, ZhuE11,Zhu12the,Zhu14IOC}, quantum verification \cite{ZhuH19O,LiHZ19,LiHZ20,ZhuZ20}, quantum data hiding and cryptography \cite{DiViLT02,AmbaBW09,MattWW09},  decoupling \cite{AbeyDHW09,SzehDTR13}, circuit complexity \cite{HarrL09,BranH13,BranHH16L,HafeFKE22}, and scrambling \cite{HaydP07,RobeY17,LiuLZZ18L}. One of the most important applications that motivates the current study is   shadow estimation  \cite{HuanKP20}, which is a powerful tool in quantum characterization and verification \cite{EiseHWR20,CarrEKK21,KlieR21,YuSG22,MorrSGD22,ElbeFHK23}. Shadow estimation is particularly efficient for estimating the fidelity between the actual state and the target state when the underlying measurement ensemble is constructed from  a 3-design, in which case the sample complexity is independent of the system size. 

The Clifford group is one of the most important groups in quantum information processing \cite{BoltRW61I, BoltRW61II,Gott97the,BengZ17book}, with applications in numerous research areas, including quantum computation \cite{Gott97the,Gott99,GottC99,BravK05,NielC10book}, quantum error correction \cite{Gott97the,Terh15,NielC10book}, quantum simulation \cite{BravG16,BravBCC19,Hein21the,PashRKB22},
randomized benchmarking \cite{KnilLRB08,MageGE11,KlieR21,JafaWSS20},  shadow estimation \cite{HuanKP20,HelsW23,HuCY23,IppoLRK23,BertHHI24,LevyLC24,SchuHH24,KlieR21,ElbeFHK23}, quantum entanglement \cite{NestDM04,HeinEB04,NezaW20}, discrete Wigner functions \cite{CormGGP06,Gros06,Zhu16P,GrosNW21,RausOZF23}, mutually unbiased bases \cite{DurtEBZ10,Appl09P,Zhu15M}, symmetric informationally complete measurements \cite{Appl05,ScotG10,Zhu10,Zhu12the,FuchHS17}, and phase transition \cite{LuntSP21,SierT23,PiroLVN23}. For a qubit system (a system composed of one or more qubits), it is by now well known that the Clifford group forms a unitary 3-design \cite{Webb16,Zhu17MC,ZhuKGG16} (see also \rcite{GuraT05}), which means every Clifford orbit is a 3-design, including the orbit of stabilizer states in particular \cite{KuenG13}. Therefore, stabilizer measurements, which can be realized by random Clifford transformations followed by the computational basis measurement, are the ideal choice for shadow estimation. Moreover, designs with higher strengths can be constructed from Clifford circuits interleaved by a few non-Clifford gates \cite{HafeMHE23}. 

Unfortunately, when the local dimension $d$ is an odd prime,  the Clifford group is only a unitary 2-design \cite{DiViLT02,Webb16,Zhu17MC}, and the set of stabilizer states is only a 2-design too. Little is known about the efficiency of shadow estimation based on qudit stabilizer measurements, although a qudit system can also realize universal quantum computation \cite{Gott99}. Moreover, most 3-designs known for a qudit system are difficult to realize in practice because they lack simple structures. Can we overcome the gap between qudit systems and qubit systems in quantum information processing? This fundamental question is of interest not only to foundational studies, but also to many practical applications.

In this work we systematically delve into the third moments of general qudit Clifford orbits, where the local dimension $d$ is an odd prime. Inspired by shadow estimation \cite{HuanKP20}, we introduce a new figure of merit, the shadow norm, for quantifying the deviation of an ensemble of quantum states from a 3-design. Intuitively, the shadow norm of an ensemble is defined as the maximum shadow norm of a traceless observable normalized with respect to the Hilbert-Schmidt norm. This figure of merit has a clear operational meaning, in contrast with other common figures of merit, such as the third normalized frame potential and the operator norm of the third normalized moment operator. 

We  show that the third normalized frame potential and shadow norm are both $\caO(d)$ irrespective of the number $n$ of qudits and the Clifford orbit under consideration, although the operator norm of the third normalized moment operator may increase exponentially with $n$ when $d\neq 2\mmod 3$. Notably, this is the case for the orbit of stabilizer states. Therefore, the overhead of qudit stabilizer measurements in shadow estimation compared with qubit stabilizer measurements is only $\caO(d)$, which is independent of $n$. Surprisingly, an ensemble that is far from a 3-design with respect to one of the most common figures of merit can achieve a similar performance as a 3-design ensemble. 

Moreover,  by virtue of magic gates we can generate  magic orbits based on magic states, which can better approximate 3-designs and are thus more appealing to many practical applications. We prove that the shadow norm of any magic orbit is upper bounded by the constant $15/2$, so a single magic gate can already eliminate the $\caO(d)$ overhead in qudit shadow estimation and bridge the gap between qudit systems and qubit systems. Furthermore, we propose simple recipes for constructing approximate and exact 3-designs with respect to all three figures of merit mentioned above. Notably,
accurate approximate 3-designs can be constructed from  only two Clifford orbits. For an infinite family of local dimensions, exact 3-designs can be constructed from  two or four Clifford orbits. 

Our work shows that the distinction between the qudit Clifford group and qubit Clifford group is not so dramatic as might be expected. For important applications, such as shadow estimation, this distinction can be circumvented with only a few magic gates; see the companion paper \cite{MaoYZ24} for more discussions. This result may have  profound implications for  quantum information processing based on qudits.

\subsection{Outline of the main results and technical innovation}

In \sref{sec:design} we review basic results about complex projective $t$-designs and unitary $t$-designs. Motivated by the task of shadow estimation \cite{HuanKP20}, we then introduce the shadow norm  as a figure of merit for quantifying the deviation of an ensemble of quantum states  from a 3-design. Not surprisingly, the shadow norm attains its minimum if and only if (iff) the ensemble forms a 3-design. Furthermore, we  derive a lower bound for the shadow norm based on the third frame potential  and thereby endow the third frame potential with a simple operational interpretation (see \thref{thm:ShNormPhi3LB}). It turns out this lower bound is almost tight for the ensemble of stabilizer states (see \thref{thm:ShNormStab}). To derive this lower bound we establish a simple connection between the shadow norm and a modified moment operator and then generalize a typical method for proving the familiar  lower bound for the third frame potential. 

In \sref{sec:QuditStab} we first review basic concepts in the qudit stabilizer formalism, including the Heisenberg-Weyl group, Clifford Group, stabilizer codes, and stabilizer states. 
Then we review basic results about  the commutant of the $t$th Clifford tensor power as developed in  \rscite{NezaW20,GrosNW21} and introduce the key concepts in this framework, such as the stochastic orthogonal group $O_t(d)$, the set $\Sigma_{t,t}(d)$ of stochastic Lagrangian subspaces, and the operators  $R(\caT)$ for $\caT\in \Sigma_{t,t}(d)$, which span the commutant under certainty assumption. 

In \sref{sec:SLS} we clarify the basic properties of the stochastic orthogonal group $O_3(d)$ and the set $\Sigma(d):=\Sigma_{3,3}(d)$ of stochastic Lagrangian subspaces,  in preparation for studying the commutant of the third Clifford tensor power. Furthermore, we introduce an alternative method for constructing stochastic Lagrangian subspaces and cubic characters for distinguishing three types of stochastic Lagrangian subspaces. Our study unravels many interesting properties of stochastic Lagrangian subspaces that are ignored or unexpected in the literature \cite{NezaW20,GrosNW21}. These results are instrumental to studying the third moments of Clifford orbits, including  magic orbits in particular. 

In \sref{sec:Commutant} we clarify the basic properties of the operators $R(\caT)$ for $\caT\in  \Sigma(d)$ and the commutant of the third Clifford tensor power. By virtue of the Gram matrix we  show that $\{R(\caT)\}_{\caT\in \Sigma(d)}$ spans the commutant of the third Clifford tensor power whenever $n\geq 1$ (see \thref{thm:spanR}). Then we construct 
a dual operator frame, which enables us to decompose any operator in the commutant in a simple way. Furthermore, we introduce shadow maps associated with  the operators $R(\caT)$ and clarify their basic properties  (see \lsref{lem:RTObStabProj}-\ref{lem:RTOdiag}) by virtue of the insight on  stochastic Lagrangian subspaces gained in the previous section. These results  are crucial to understanding the third moments of Clifford orbits and their applications in quantum information processing, such as shadow estimation. 

In \sref{sec:3rdmomentStab} we clarify the basic properties of the third moment of the ensemble of $n$-qudit stabilizer states, denoted by $\Stab(n,d)$ henceforth. Notably,  we determine the spectrum and operator norm of the third normalized moment operator  (see \pref{pro:Qspectra} and \tref{tab:bQnd3}). It turns out the operator norm is upper bounded by $(d+2)/3$ when $d=2\mmod 3$, but increases exponentially with $n$ when $d\neq 2\mmod 3$. The distinction between the two cases is tied to the different structures of stochastic Lagrangian subspaces as clarified in \sref{sec:SLS} and has a number theoretic origin.
Furthermore, we determine the shadow norms of stabilizer projectors with respect to  $\Stab(n,d)$ and the shadow norm of  $\Stab(n,d)$ itself (see \thsref{thm:StabShNormStab} and \ref{thm:ShNormStab}). It turns out the shadow norm of $\Stab(n,d)$ is upper bounded by $2d-1$, which is independent of $n$, although the deviation of $\Stab(n,d)$ from a 3-design grows exponentially with $n$ according to the operator norm of the third normalized moment operator when  $d\neq 2\mmod 3$. 

In \sref{sec:ThirdMomentGen} we clarify the basic properties of a general Clifford orbit. To this end, we  expand the third normalized moment operator in terms of the operator frame  $\{R(\caT)\}_{\caT\in \Sigma(d)}$ and a dual  frame and clarify the basic properties of the expansion coefficients. On this basis we derive  analytical formulas for the third normalized frame potential of the Clifford orbit
and the shadow norm of any stabilizer projector. We also derive nearly tight  bounds for the operator norm of the third normalized moment operator and the shadow norm of the Clifford orbit (see \thsref{thm:Phi3LUB}-\ref{thm:ShNormGen}). It turns out the third normalized frame potential and shadow norm are both upper bounded by $\caO(d)$. 
The operator norm is upper bounded by $\caO(d)$ when $d=2\mmod3$, but may grow exponentially with $n$ when $d\neq 2\mmod3$. The distinction between the two cases has the same origin as for the orbit of stabilizer states. In \sref{sec:Balance} we further clarify the properties of Clifford orbits based on balanced ensembles (see \thsref{thm:Phi3Balance}-\ref{thm:ShNormBalance}).

In \sref{sec:MagicState} we first review basic results about  qudit diagonal magic gates, referred to as  $T$ gates henceforth. Then we turn to qudit magic states, which can be generated by applying the Fourier gate and $T$ gates on the computational basis state. 
Both $T$ gates and magic states are tied to cubic polynomials on the finite field $\bbF_d$ (when $d=3$, the situation is slightly different). Furthermore, we introduce cubic characters to distinguish three types of $T$ gates and three types of single-qudit magic states in the case $d=1\mmod 3$. We also classify $n$-qudit magic states and highlight certain magic states that are of special interest. The existence of different types of $T$ gates and magic states unravels a rich story yet to be fully understood.

In \ssref{sec:MagicOrbitd23} and \ref{sec:MagicOrbitd1} we clarify the basic properties of magic orbits. 
Notably, we show that the third normalized frame potential, operator norm of the third normalized moment operator, and shadow norm all decrease exponentially with the number of $T$ gates (see \thsref{thm:Phi3Magicd23}-\ref{thm:ShNormMagicd1}). Therefore,  good approximate 3-designs   with respect to 
the third normalized frame potential and shadow norm can be constructed using  only a few $T$ gates. Actually, a single $T$ gate can already bridge the gap between qudit systems and qubit systems with respect to the third normalized frame potential and shadow norm. The situation for the operator norm is more complicated because it may grow exponentially with $n$ when $d\neq 2\mmod 3$. Nevertheless, good approximate 3-designs   with respect to all three figures of merit can be constructed  from only two Clifford orbits (see  \pref{pro:AccurateDesign} and \thsref{thm:Moment3MagicLUBd1}, \ref{thm:AccurateDesignMagic}; see \sref{sec:3designConstructd23} for the special case $d=3$). Furthermore, for an infinite family of local dimensions, exact 3-designs can be constructed from only two or four Clifford orbits (see \lref{lem:BalMagicExactd1}, \thref{thm:BalMagicExactd1}, and  \tref{tab:Exact3design}).

To establish the main results on magic orbits we first 
determine the expansion coefficients associated with 
magic states of a single qudit (see \lsref{lem:kappaMagicOned23} and \ref{lem:kappaMagicOned1}). Technically, a key  challenge is to count the number of solutions of a general cubic equation on each stochastic Lagrangian subspace, and one of our main contributions is to tackle this challenge after simplifying the counting problem. Our insight on  stochastic Lagrangian subspaces plays a crucial role in this endeavor. Incidentally, most sophisticated analysis is required to clarify the properties of the operator norm of the third normalized moment operator compared with the third normalized frame potential and shadow norm. Although this work was originally motivated by  shadow estimation, the scope and implications of this work have gone far beyond the original motivation.

In addition, results on Gauss sums and Jacobi sums are pretty useful to establishing our main results. Although our proofs only use basic facts on Gauss sums and Jacobi sums, 
the existence of an infinite family of local dimensions for which exact 3-designs can be constructed using two or four Clifford orbits is tied to a deep result on the phase distribution of cubic Gauss sums \cite{HeatP79}. This distribution problem was originally studied by Kummer in the 19th century and remains an active research topic in number theory until today~\cite{DunnR21}.

For the convenience of many practical applications to which  this work is relevant, we have tried to clarify not only the scaling behaviors of the figures of merit under consideration, but also the constants involved; constants without explicit values are avoided whenever it is possible. For the convenience of the readers, the main mathematical symbols used in this work are summarized in \tref{tab:Symbol}. 
To streamline the presentation, most technical proofs are relegated to the Appendix, which also includes some auxiliary results. In addition,   the Appendix includes a brief introduction on multiplicative characters,  Gauss sums, Jacobi sums, and cubic equations over finite fields.

\begin{center}
	\setlength{\tabcolsep}{2mm}
	\renewcommand{\arraystretch}{1.5}
	\begin{longtable}[tb]{c|c}
		\caption{\label{tab:Symbol}Main mathematical symbols used in this work. }\\
		\hline\hline
		Symbol & Meaning \\
		\hline
		\endfirsthead
		\hline\hline
		Symbol & Meaning \\
		\hline
		\endhead
		$\mathbf{1}_t=(1,1,\ldots, 1)^\top$ & the vector in $\bbF_d^t$ all of whose entries are equal to 1 \\
		
		$\mathds{1}=\mathds{1}_t$ & identity operator on $\bbF_d^t$\\

		$\|\cdot\|_\ell$,\; $\|\cdot\|= \|\cdot\|_\infty$ & Schatten $\ell$-norm and operator (spectral) norm \\
		
		$\Cl(n,d)$ & the $n$-qudit Clifford group\\

$\Cl(n,d)^{\totimes t}$ & $t$th diagonal tensor power of the Clifford group: $\{U^{\otimes t}\,|\, U\in \Cl(n,d)\}$\\
		
		$d$ & local dimension, which is an odd prime unless stated otherwise\\
		
		$D$ &total dimension, $D=d^n$ except in \sref{sec:design} and \aref{app:Shadow3design} \\
		
		$\Delta$ &   the diagonal stochastic Lagrangian subspace\\ 
		
		$\scrE$ & an ensemble of pure states\\
		
		$\|\scrE\|_\sh$  & shadow norm of the ensemble $\scrE$  \\
		
		$\eta_2(\cdot)$,\; $\eta_3(\cdot)$ & the quadratic character and a cubic character of $\bbF_d$  \\

		$\eta_3(f)$ & cubic character of the cubic polynomial $f$ (see Sec.~\ref{sec:MagicState}) \\

$\eta_3(\caT)$ & cubic character of the stochastic Lagrangian subspace $\caT$ [see \eref{eq:CharIndex}] \\

		$\bbF_d$,\; $\bbF_d^{\times}$ & the finite field  of $d$ elements and the subset of nonzero elements\\
		
		$\bfg_0$,\; $\bfg_1$ &   defect vectors defined in \eref{eq:TwoDefectSpaces} \\ 
		
		$g(m,a)$,\; $g(m)$,\; $G(\eta,a)$,\; $G(\eta)$ &Gauss sums defined in \eref{eq:GaussSumDef} \\
		
		$\tg(d)$ & $g(3,1)g(3,\nu)g(3,\nu^2)/d$\\
		
		$\gamma_{d,k}$,\; $\gamma_k$,\; $\gamma^\rmB_{d,k}$ & see \eqsref{eq:gammadk}{eq:kappatgammadk} \\
		
		$\caH_d$, \; $\caH$ &  local and total Hilbert spaces of dimensions $d$ and $D$, respectively  \\
		
		
		$\bbI$ & identity operator on  $\caH$ \\
		
		$\ind(\caT)$  & index of $\caT$ defined in \eref{eq:TindexDef}\\
		
		$J(\theta_1,\theta_2,\ldots, \theta_k)$ & Jacobi sum\\
		
		$\kappa(\Psi,\caT)$,\; $\tka(\Psi,\caT)$,\; $\hka(\Psi,\caT)$
		& see \eqsref{eq:kappapsiT}{eq:hkapsiT}\\

		$\kappa(\Psi,\scrT)$,\; $\kappa(\Psi,\scrT,j)$,\; $\hka(\Psi,\scrT)$ 
		& see  \eref{eq:kappascrT}\\
		
		$\kappa_{d,k}$ & see \eref{eq:kappatgammadk}\\

		$\kappa(f,\caT)$ &$\kappa(\psi_f,\caT)$\\
		
		$\scrM_{n,k}$,\; $\scrM_{n,k}^\can$,  $\scrM_{n,k}^\id$,\; $\scrM_{n,k}^\QB$,\; $\scrM_{n,k}^\rmB$ & five sets of magic states defined in \sref{sec:MagicState}\\

		$\mu_j=\mu_j(d)$ & see \eqsref{eq:muj}{eq:mujFormula}\\
		
		$n$ &  the number of qudits\\
		
		$\bbN$, \; $\bbN_0$ & the set of positive integers  and the set of  nonnegative integers  \\
		
		$\caN$ &  a generic defect subspace\\ 
		
		$\tcaN$,\; $\tcaN_0$\; $\tcaN_1$ &  special defect subspaces defined in \eqsref{eq:DefectSpaced3}{eq:TwoDefectSpaces} \\

		$N_c(f)=N(f=c)$ & the number of solutions to the equation $f=c$ \\
		
		$N_\alpha(f,\caT)$ & $|\{ (\bfx;\bfy)\in \caT \ | \  f(\bfy)-f(\bfx)=\alpha\}|$\\
		
		$\nu$ &  a primitive element in $\bbF_d$ \\
		
		$O_t(d)$ & stochastic orthogonal group\\
		
		$O_3^\even(d)$, \; $O_3^\odd(d)$ & the set of even (odd) stochastic isometries \\
		
		$\orb(\Psi)$ & Clifford orbit generated from the pure state $|\Psi\>$ \\
		
		$\orb(\scrE)$ & ensemble of pure states composed of Clifford orbits of states in $\scrE$\\
		
		$\Ob$,\; $\Ob_0$  & a linear operator on $\caH$ and its traceless part  \\

		$\|\Ob\|_\sh$  & shadow norm of $\Ob$ with respect to a given ensemble \\
		
		$\|\Ob\|_\Stab$,\;	$\|\Ob\|_{\orb(\Psi)}$,\; $\|\Ob\|_{\orb(\scrE)}$  & shadow norms of $\Ob$ with respect to $\Stab(n,d)$, $\orb(\Psi)$, and $\orb(\scrE)$ \\
		
		$\|\orb(\Psi)\|_\sh$,\; $\|\orb(\scrE)\|_\sh$  & shadow norms of $\orb(\Psi)$ and $\orb(\scrE)$, respectively \\

		$\omega=\omega_d$ &  $\rme^{ \frac{2\pi \rmi}{d}}$ \\
		
		$P_{[t]}$,\; $P_\sym$ &the projector  onto the $t$-partite symmetric subspace $\Sym_t(\caH)$,\; $P_{[3]}$ \\
		
		$\pi_{[t]}$ & $\tr P_{[t]}$\\
		
		$\tscrP_3(\bbF_d)$  & the set of cubic functions defined in \sref{sec:MagicState}\\
		
		$\Phi_t(\cdot)$, \; $\bar{\Phi}_t(\cdot)$ & the $t$th frame potential and  normalized frame potential\\

		$|\psi_f\>$ & magic state associated with the cubic function $f$\\
		
		$Q(n,d,t)$,\; $\bQ(n,d,t)$ & $t$th moment operator and normalized moment operator of $\Stab(n,d)$ \\
		
		$Q_t(\scrE)$,\; $\bQ_t(\scrE)$ & $t$th moment operator and normalized moment operator of $\scrE$ \\
		
		$Q(\cdot)=Q_3(\cdot)$,\; $\bQ(\cdot)=\bQ_3(\cdot)$
		& third moment operator and normalized moment operator\\

		$\caQ(\Ob)$ & shadow map defined in \eref{eq:ShadowMap}: $\tr_{BC}\bigl[Q\bigl(\mathbb{I} \otimes\Ob\otimes\Ob^{\dag}\bigr)\bigr]$ \\

		$\bcaQ_{n,d}(\cdot)$,\; 
		$\bcaQ_{\orb(\Psi)}(\cdot)$ & shadow maps associated with $\bQ(n,d,3)$ and $\bQ(\orb(\Psi))$\\
		
		$r(O)$, \; $R(O)$ & $r(\caT_O)$,\; $R(\caT_O)$\\
		
		$r(\caT)$,\; $R(\caT)=r(\caT)^{\otimes n}$ & operators defined in \eref{eq:rRT}\\

		$r(\scrT)$,\; $R(\scrT)$ & operators defined in \eref{eq:RscrT}\\

		$\tR(\caT)$ &dual frame operator defined in \eqsref{eq:DualBasis}{eq:DualBasisn1}\\

		$\caR_\caT(\Ob)$ & $\tr_{BC}\bigr[R(\caT)\bigl(\mathbb{I} \otimes\Ob\otimes\Ob^{\dag}\bigr)\bigr]$ \\
		
		$\Stab(n,d)$ &the set of $n$-qudit stabilizer states \\

		$S_3$ & the symmetric group generated by  $\zeta$ and $\tau_{12}$ in \eref{eq:zetatau} \\
		
		$\Sym_t(\caH)$ &  $t$-partite symmetric subspace in $\caH^{\otimes t}$ \\
		
		$\Sigma_{t,t}(d)$, \; $\Sigma(d)$ & the set of stochastic Lagrangian subspaces in $\bbF_d^{2t}$, \; $\Sigma_{3,3}(d)$	\\

		$T_f$ & $T$-gate associated with the cubic function $f$\\
		
		$\caT$ &  a generic stochastic Lagrangian subspace\\ 
		
		$\caT_\Delta$,\;  $\caT_\tDelta$ & $\caT\cap \Delta$,\; $\{\bfx\ | \ (\bfx;\bfx)\in \caT_\Delta\}$ \\
		
		$\caT_O$ & $\left\{ (O\bfx;\bfx)\, | \, \bfx \in \bbF_d^t \right\}$\\
		
		$\caT_\LD$,\; $\caT_\RD$ & left and right defect subspaces of $\caT$\\

		$\tcaT_i$,\; $\tcaT_{ij}$ for $i,j=0,1$ &  stochastic Lagrangian subspaces in $\scrT_\defe$ defined in \lref{lem:defectT} \\ 
		
		$\caT_\bfv$ & stochastic  Lagrangian subspace determined by  $\bfv$ as defined in \eref{eq:Tv}\\
		
		$\scrT_\sym$,\; $\scrT_\ns$,\;$\scrT_\iso$,\; $\scrT_\defe$,\; $\scrT_\even$,\; $\scrT_\odd$ & subsets of $\Sigma(d)$ defined in \eref{eq:scrTSIEO}\\

		$\scrT_0$,\; $\scrT_1$  & subsets of $\Sigma(d)$ defined in \eref{eq:scrT01}\\
		
		$\tau_{12}$,\; $\tau_{13}$,\; $\tau_{23}$,\; $\zeta$ & three transpositions and a cyclic permutation in $S_3$; see \eref{eq:zetatau} \\
		
		\hline\hline
	\end{longtable}
\end{center}

\section{\label{sec:design}Complex projective and unitary $t$-designs}

In this section we review basic results on (weighted complex projective) $t$-designs and unitary $t$-designs in preparation for the later study. In addition, motivated by shadow estimation \cite{HuanKP20},  we  introduce an alternative figure merit for quantifying the deviation of an ensemble from a 3-design and clarify its connection with the third frame potential. 

\subsection{Complex projective $t$-designs}
Let $\caH$ be a $D$-dimensional Hilbert, $\caL(\caH)$ the space of linear operators on $\caH$, $\caL_0(\caH)$ the space of traceless linear operators on $\caH$, $\caL^\rmH(\caH)$ the space of Hermitian operators on $\caH$, and $\rmU(\caH)$ the unitary group on $\caH$. Let $t$ be a positive integer. Denote by $P_{[t]}$  the projector  onto the $t$-partite symmetric subspace $\Sym_t(\caH)$ in $\caH^{\otimes t}$ and
\begin{equation}\label{eq:dimsym}
\pi_{[t]}=\tr(P_{[t]})=\begin{pmatrix}
D+t-1\\
t
\end{pmatrix}.
\end{equation}
Let $\scrE=\{|\psi_j\>, w_j \}_j$  be a weighted set of pure states in $\caH$. In this work we assume that all pure states are normalized, that is, $|\<\psi_j|\psi_j\>|^2=1$ for all $j$; the weights are nonnegative and normalized, that is, $w_j \geq 0$ for all $j$ and $\sum_j w_j=1$, so they form a probability distribution. Then  the weighted set $\scrE$ represents an ensemble of pure states. In addition, the weights may be omitted for a uniform distribution. The $t$th moment operator and normalized moment operator associated with the ensemble $\scrE$ are defined as
\begin{align}
Q_t(\scrE):=\sum_{j}w_j  \left(|\psi_j\> \<\psi_j|\right)^{\otimes t},\quad   \bQ_t(\scrE):=\pi_{[t]} Q_t(\scrE). 
\end{align}
By definition $Q_t(\scrE)$ and $\bQ_t(\scrE)$ are supported in $\Sym_t(\caH)$. 
The ensemble $\scrE$ is a (weighted complex projective) \emph{$t$-design} if $\bQ_t(\scrE)=P_{[t]}$ \cite{Zaun11,ReneBSC04,Scot06,AmbaE07,ZhuKGG16}. 
By definition a $t$-design is also a $t'$-design for $t'\leq t$.

The $t$th frame potential and normalized frame potential of the ensemble $\scrE$ are defined as
\begin{equation}\label{eq:PhitState}
\Phi_t(\scrE):=\sum_{j,k}w_jw_k |\langle\psi_j|\psi_k\rangle|^{2t}=\tr(Q_t^2),\quad \bar{\Phi}_t(\scrE)=\pi_{[t]}\Phi_t(\scrE)=\pi_{[t]}\tr(Q_t^2)=\frac{\tr(\bQ_t^2)}{\pi_{[t]}}. 
\end{equation}
It is well known that $\bar{\Phi}_t(\scrE)\geq1$, and the lower bound is saturated iff $\scrE$ is a $t$-design. The twirling superoperator $\bbT$ on $\caH^{\otimes t}$ is defined as follows,
\begin{align}
\bbT(M)=\int U^{\otimes t}M\bigl(U^{\otimes t}\bigr)^\dag \rmd U\quad \forall M\in \caL(\caH^{\otimes t}),
\end{align}
where the integration is over the normalized Haar measure on the unitary group $\rmU(\caH)$.
The operator $M$ is invariant under twirling if $\bbT(M)=M$.
Note that $\bbT(\bQ_t(\scrE))=P_{[t]}$ for any state ensemble $\scrE$.

\begin{proposition}\label{pro:tdesignEqui}
	The following four statements about the state ensemble $\scrE$ are equivalent:
	\begin{enumerate}
		\item $\scrE$ is a $t$-design.
		
		\item 	$\bQ_t(\scrE)=P_{[t]}$.
		
		\item $\bQ_t(\scrE)$ is invariant under twirling. 
		
		\item $\bar{\Phi}_t(\scrE)=1$. 
		
	\end{enumerate}
\end{proposition}

Next, we turn to approximate $t$-designs. 
The ensemble $\scrE$ is an $\epsilon$ approximate $t$-design if 
\begin{align}
\left\|\bQ_t(\scrE)-P_{[t]}\right\|\leq \epsilon. \label{eq:tdesignApprox}
\end{align}
Alternatively,  $\bar{\Phi}_t(\scrE)-1$ can also be used to quantify the deviation from an ideal $t$-design. Nevertheless, the criterion based on \eref{eq:tdesignApprox} is usually much more stringent because the norm $\|\bQ_t(\scrE)-P_{[t]}\|$ is much more
sensitive to low rank deviation in $\bQ_t(\scrE)$ compared with $\bar{\Phi}_t(\scrE)$. Notably, $\left\|\bQ_t(\scrE)-P_{[t]}\right\|$ might be very large even if $\bar{\Phi}_t(\scrE)-1$ is very small, which is the case for certain Clifford orbits, as we shall see in \ssref{sec:MagicOrbitd23} and \ref{sec:MagicOrbitd1}.

In this work, our main focus is the third moment. To simplify the notation, the third moment operator $Q_3(\scrE)$ and normalized moment operator $\bQ_3(\scrE)$ will be abbreviated as $Q(\scrE)$ and  $\bQ(\scrE)$, respectively, if there is no danger of confusion.

\subsection{The shadow norm and third frame potential}
Motivated by shadow estimation \cite{HuanKP20}, here we introduce an alternative figure merit for quantifying the deviation of an ensemble from a 3-design and clarify its connection with the third frame potential.

Let $Q$ be an operator acting on $\caH^{\otimes 3}$.  The \emph{shadow map} $\caQ$ associated with $Q$ is a map from $\caL(\caH)$ to $\caL(\caH)$ defined as follows,
\begin{align}\label{eq:ShadowMap}
\caQ(\Ob):=\tr_{BC}\bigl[Q \bigl(\bbI\otimes   \Ob\otimes \Ob^\dag\bigr)\bigr], \quad \Ob\in \caL(\caH). 
\end{align}
The squared \emph{shadow norm} of $\Ob$ with respect to $Q$, denoted by $\|\Ob\|^2_Q$,  is defined as the spectral norm of $\caQ(\Ob)$, that is,
\begin{align}\label{eq:ObShNormQ}
\|\Ob\|_Q^2:=\|\caQ(\Ob)\|.
\end{align}
The \emph{shadow norm} of $Q$ itself is defined as
\begin{align}\label{eq:ShadowNormDef}
\|Q\|_\sh&:=\max\{\left\|\caQ(\Ob)\right\|\ :\ \Ob\in \caL_0(\caH),\;\|\Ob\|_2=1\}, 
\end{align}
where the maximization is over all traceless operators in $\caL_0(\caH)$  that are normalized with respect to the Hilbert-Schmidt norm. The restriction on traceless operators is motivated by the fact that in shadow estimation, the variance in estimating the expectation value $\tr(\rho \Ob)$ with respect to the state $\rho$ is determined by the traceless part of $\Ob$, that is, $\Ob_0:=\Ob-\tr(\Ob)\bbI/D$.

We are mainly interested in the case in which $Q$ is permutation invariant, which means $Q$ commutes with all permutations among the three parties. Then it suffices to consider Hermitian operators
in  the above maximization, that is, 
\begin{align}\label{eq:ShadowNormDef2}
\|Q\|_\sh=\max\bigl\{\|\caQ(\Ob)\| \ : \ \Ob\in \caL_0(\caH),\; \Ob=\Ob^\dag,\; \|\Ob\|_2=1\bigr\}, 
\end{align}
given that $\caQ(\Ob_1+\rmi \Ob_2)= \caQ(\Ob_1)+\caQ(\Ob_2)$ whenever $\Ob_1$ and $\Ob_2$ are Hermitian operators on $\caH$.

Let $P_\sym=P_{[3]}$ be the projector onto the tripartite symmetric subspace in $\caH^{\otimes 3}$ and $\Ob\in \caL(\caH)$.
Then 
\begin{align}\label{eq:PsymOb}
6\tr_{BC}\bigl[P_\sym \bigl(\bbI\otimes \Ob\otimes \Ob^\dag\bigr)\bigr]&=\Ob\Ob^\dag+\Ob^\dag\Ob +\tr(\Ob^\dag\Ob)\bbI +|\tr\Ob|^2\bbI +\tr(\Ob)\Ob^\dag+\tr\bigl(\Ob^\dag\bigr)\Ob.  
\end{align}
If $\Ob$ is traceless, then 
\begin{equation}\label{eq:PsymObTraceless}
\begin{aligned}
6\tr_{BC}\bigl[P_\sym \bigl(\bbI\otimes \Ob\otimes \Ob^\dag\bigr)\bigr]&=\Ob\Ob^\dag+\Ob^\dag\Ob +\tr\bigl(\Ob^\dag\Ob\bigr)\bbI,\\  6\bigl\|\tr_{BC}\bigl[P_\sym \bigl(\bbI\otimes \Ob\otimes \Ob^\dag\bigr)\bigr]\bigr\|&=\|\Ob\Ob^\dag+\Ob^\dag\Ob\| +\|\Ob\|_2^2. 
\end{aligned}
\end{equation}
If  $\Ob$ is Hermitian and traceless, then 
\begin{align}
\|\Ob\|^2\leq \frac{D-1}{D}\|\Ob\|_2^2, \quad 6\|\tr_{BC}[P_\sym (\bbI\otimes   \Ob\otimes \Ob)]\|&=2\|\Ob\|^2+\|\Ob\|_2^2\leq \frac{3D-2}{D}\|\Ob\|_2^2, \label{eq:PsymShNorm}
\end{align}
and both inequalities are saturated when $\Ob=D|0\>\<0|-\bbI$. Based on this observation we can deduce that
\begin{align}\label{eq:ShadowNormSym}
\|P_\sym \|_\sh= \frac{3D-2}{6D}. 
\end{align}

Given any ensemble $\scrE=\{|\psi_j\>, w_j \}_j$  of pure states in $\caH$, the third moment operator $Q(\scrE)=Q_3(\scrE)$ and normalized moment operator $\bQ(\scrE)=\bQ_3(\scrE)$ 
are automatically permutation invariant. The squared \emph{shadow norm} of $\Ob\in \caL(\caH)$ with respect to $\scrE$ is defined as 
\begin{align}\label{eq:ObShNormDef}
\|\Ob\|^2_\scrE :=D(D+1)^2\|\Ob\|^2_{Q(\scrE)}=\frac{6(D+1)}{D+2}\|\Ob\|^2_{\bQ(\scrE)}=\frac{6(D+1)}{D+2}\left\|\tr_{BC}\bigl[\bQ(\scrE) \bigl(\bbI\otimes   \Ob\otimes \Ob^\dag\bigr)\bigr]\right\|,
\end{align}
 When the ensemble $\scrE$ is clear from the context, $\|\Ob\|^2_\scrE$ can also be written as $\|\Ob\|^2_\sh$. In shadow estimation based on the ensemble $\scrE$,  the variance in estimating the expectation value $\tr(\rho \Ob)$ of the observable $\Ob$  is upper bounded by  $\|\Ob\|^2_\scrE$ \cite{HuanKP20}, assuming that $\scrE$ is a 2-design, which is the case of interest here. This fact highlights the significance of the  squared shadow norm $\|\Ob\|^2_\scrE$. The shadow norm of the ensemble $\scrE$ is defined as
\begin{align}
\|\scrE\|_\sh:=D(D+1)^2 \|Q(\scrE)\|_\sh=\frac{6(D+1)}{D+2}\|\bQ(\scrE)\|_\sh. \label{eq:EnsembleShNormDef}
\end{align}
It turns out the shadow norm $\|\scrE\|_\sh$ is minimized iff the underlying ensemble forms a 3-design as shown in the following theorem and proved in \aref{app:Shadow3design}. 
\begin{theorem}\label{thm:ShNormPhi3LB}
	Suppose $\scrE=\{|\psi_j\>, w_j \}_j$ is an ensemble of pure states in $\caH$ and $\bQ=\bQ(\scrE)$. Then 
	\begin{gather}\label{eq:QbarShNormPhi3LB}
	\|\bQ\|_\sh\geq\frac{(D+1)(D+2)}{6(D-1)}(D^2\Phi_3-2D\Phi_2+\Phi_1)	\geq \frac{3D-2}{6D}= \|P_\sym\|_\sh,\\
	\|\scrE\|_\sh\geq\frac{(D+1)^2}{(D-1)}(D^2\Phi_3-2D\Phi_2+\Phi_1)	\geq \frac{(3D-2)(D+1)}{D(D+2)};  \label{eq:ShNormPhi3LB}	
	\end{gather}
	the final   bound for $\|\bQ\|_\sh$ (and similarly for $\|\scrE\|_\sh$)
	is saturated iff $\scrE$ is a 3-design. If   $\scrE$ is a 2-design, then 
	\begin{gather}\label{eq:QbarShNormPhi3LB2}
	\|\bQ\|_\sh\geq \frac{D}{D-1}\bar{\Phi}_3-\frac{(D+2)(3D-1)}{6D(D-1)}\geq \bar{\Phi}_3-\frac{3D+2}{6D},\\
	\|\scrE\|_\sh\geq \frac{6D(D+1)}{(D-1)(D+2)}\bar{\Phi}_3-\frac{(D+1)(3D-1)}{D(D-1)}> 6\bar{\Phi}_3-3-\frac{5}{D}.  \label{eq:ShNormPhi3LB2}
	\end{gather}
\end{theorem}
Here $\Phi_t$ and $\bar{\Phi}_t$ for $t=1,2,3$ are shorthands for $\Phi_t(\scrE)$ and $\bar{\Phi}_t(\scrE)$, respectively.

\subsection{Unitary $t$-designs}
Let $\{U_j, w_j\}_j$ be a weighted set of unitary operators on $\caH$, where  the weights are nonnegative and normalized as before. 
The weighted set $\{U_j, w_j\}_j$ is  a (weighted) unitary $t$-design if it satisfies \cite{DankCEL09, GrosAE07, RoyS09}
\begin{equation}
\sum_{j} w_j U_j^{\otimes t} M \bigl(U_j^{\otimes t}\bigr)^\dag = \int \rmd U U^{\otimes t} M\bigl(U^{\otimes t}\bigr)^\dag  \quad \forall M\in \caL(\caH^{\otimes t}),
\end{equation}
where the integration is over the normalized Haar measure on the unitary group $\rmU(\caH)$. Such a set can reproduce the first $t$ moments of Haar random unitaries. In analogy to \eref{eq:PhitState}, the $t$th frame potential of $\{U_j, w_j\}_j$ is defined as
\begin{equation}
\Phi_t(\{U_j, w_j\}_j):=\sum_{j,k}w_j w_k \bigl|\tr\bigl(U_jU_k^\dag\bigr)\bigr|^{2t}.
\end{equation}
It is known that 
\begin{equation}
\Phi_t(\{U_j, w_j\}_j)\geq \int \rmd U |\tr(U)|^{2t},
\end{equation}
and the lower bound is saturated iff  $\{U_j, w_j\}_j$ forms a unitary $t$-design \cite{DankCEL09, GrosAE07, RoyS09}. 
Here the integral is equal to the sum of squared multiplicities of irreducible components of the $t$th tensor representation of the unitary group. 
It is also equal to the number of permutations of $\{1,2,\ldots, t\}$ that have no increasing subsequence of length larger than $D$  \cite{DiacS94,Rain98}.
Explicit formulas for the integral are known in the following two cases \cite{Scot08},
\begin{equation}\label{eq:FramePmin}
 \int \rmd U |\tr(U)|^{2t}=\begin{cases}
\frac{(2t)!}{t! (t+1)!} &D=2,\\
t!& D\geq t.
\end{cases}
\end{equation}

\section{\label{sec:QuditStab}Qudit stabilizer formalism}\label{sec:stabform}
In preparation for the later study, in this section we review basic concepts and results on the qudit  stabilizer formalism. Here we assume that the Hilbert space $\caH$ is a tensor power of the form  $\caH=\caH_d^{\otimes n}$,  where the local dimension $d$ is a prime; accordingly, the total dimension reads $D=d^n$.  The computational basis of $\caH_d^{\otimes n}$ can be labeled by elements in $\bbF_d^n $, where $\bbF_d$ is the finite field composed of $d$ elements. The special case $d=2$ will be mentioned only briefly.

\subsection{Heisenberg-Weyl group and Clifford Group}

To start with we consider a $d$-dimensional Hilbert space $\caH_d$  with computational basis $\{|j\>\}_{j\in \bbF_d}$, where $d$ is an odd prime.  The phase operator $Z$ and cyclic-shift operator $X$ for a qudit are defined as follows,
\begin{equation}\label{eq:ZX}
Z|j\rangle = \omega^j|j\rangle, \quad X|j\rangle = |j+1\rangle,
\end{equation}
where  $\omega=\omega_d:=\rme^{\frac{2\pi \rmi}{d}}$ and the addition $j+1$ is modulo $d$.  The qudit \emph{Heisenberg-Weyl} (HW) group $\caW(d)$ is generated by $Z$ and $X$,
\begin{align}
\caW(d):=\<X,Z\>=\left\{\omega^j Z^kX^l | j,k,l\in \bbF_d\right\}. 
\end{align}
The elements in  $\caW(d)$ are called Weyl operators; up to phase factors they  can be labeled by vectors in the phase space $\bbF_d^2$. Let $u=(p,q)^\top\in \bbF_d^2$ and define
\begin{align}
W_u=W(p,q):=\omega^{-pq/2}Z^p X^q;  
\end{align}
Then $\caW(d)=\{\omega^jW_u\,|\, j\in \bbF_d, u\in \bbF_d^{2}\}$.

Next, we turn to an $n$-qudit system with Hilbert space $\caH=\caH_d^{\otimes n}$. The HW group $\caW(n,d)$ is the tensor product of $n$ copies of the single-qudit HW group $\caW(d)$;  up to phase factors Weyl operators in $\caW(n,d)$ can be labeled by vectors in the phase space $\bbF_d^{2n}$. Let $\bfu=(p_1, p_2, \ldots, p_n, q_1, q_2,\ldots, q_n)^\top\in \bbF_d^{2n}$ and define
\begin{align}
W_\bfu:=W(p_1, q_1)\otimes W(p_2, q_2)\otimes \cdots \otimes  W(p_n, q_n). 
\end{align}
Then $\caW(n,d)=\{\omega^jW_\bfu\,|\, j\in \bbF_d, \bfu\in \bbF_d^{2n}\}$. 
In addition, the Weyl operators  satisfy the following composition and commutation relations,
\begin{align}
W_\bfu W_\bfv=\omega^{[\bfu,\bfv]/2} W_{\bfu+\bfv},,\quad W_\bfu W_\bfv= \omega^{[\bfu,\bfv]} W_\bfv W_\bfu,\quad \bfu,\bfv\in \bbF_d^{2n},
\end{align}
where  $[\bfu,\bfv]$ denotes the symplectic product defined as follows,
\begin{equation}\label{symprod}
[\bfu,\bfv]:=\bfu^\top \Omega \bfv, \quad \Omega := \begin{pmatrix}
0_{n} & \mathds{1}_{n} \\
-\mathds{1}_{n} & 0_{n}
\end{pmatrix},
\end{equation}
where $\mathds{1}_n$  is the identity operator on $\bbF_d^n$. 
The symplectic group $\Sp(2n,d)$ on $\bbF_d^{2n}$ is the group of matrices that preserve the symplectic product, that is,
\begin{equation}
\Sp(2n,d):=\left\{S\in \bbF_d^{2n\times 2n} \ | \ S^\top \Omega S=\Omega\right\}. 
\end{equation}
By definition we have $[S\bfu,S\bfv]=[\bfu,\bfv]$ for all $S\in \Sp(2n,d)$ and $\bfu,\bfv\in \bbF_d^{2n}$.

The  \emph{Clifford group}  $\Cl(n,d)$ is the normalizer of the HW group $\caW(n,d)$, that is,
\begin{equation}
\Cl(n,d):=\left\{U \in \rmU\bigl(\caH_d^{\otimes n}\bigr)\, | \, U\caW(n,d) U^\dag = \caW(n,d) \right\}.
\end{equation}
By definition the commutation relations between Weyl operators are invariant under Clifford transformations. So every Clifford unitary induces a symplectic transformation on the phase space $\bbF_d^{2n}$. In this way we can establish a map from the Clifford group $\Cl(n,d)$ to the symplectic group $\Sp(2n,d)$, which turns out to be surjective. In the current study, overall phase factors of Clifford unitaries are irrelevant. The Clifford group modulo overall phase factors is also called the Clifford
group (also known as the projective Clifford
group) and is denoted by $\oCl(n,d)$ to distinguish it from $\Cl(n,d)$. When $d$ is an odd prime, the Clifford group  $\oCl(n,d)$ is isomorphic to the affine symplectic group
$\ASp(2n,d)=\Sp(2n,d) \ltimes \bbF_d^{2n}$ \cite{Hein21the}. This result is quite instructive to understanding the structure of the qudit Clifford group.

Incidentally, when $d=2$, the single-qudit HW group $\caW(d)$ reduces to the Pauli group, which is generated by the three Pauli operators. Alternatively, $\caW(d)$ is generated by the two Pauli operators $Z$ and $X$ as defined in \eref{eq:ZX} together with $\rmi \bbI$. The $n$-qudit HW group $\caW(n,d)$ is  the tensor product of $n$ copies of $\caW(d)$, and the  Clifford group  $\Cl(n,d)$ is the normalizer of $\caW(n,d)$ as before.

The second and third frame potential of the Clifford group $\Cl(n,d)$ read \cite{Zhu17MC}
\begin{equation}\label{eq:CliffordFP3}
\Phi_2(\Cl(n,d)) =2,\quad \Phi_3(\Cl(n,d)) = \begin{cases}
2d+1, \quad n=1,\\
2d+2, \quad n \geq 2, 
\end{cases}
\end{equation}
which are applicable whenever $d$ is a prime.
When $d$ is an odd prime, the Clifford group $\Cl(n,d)$ is a unitary 2-design, but not a 3-design.  When $d=2$, by contrast,  $\Cl(n,d)$ is not only a 2-design, but also a 3-design~\cite{Zhu17MC,Webb16}.

\subsection{Clifford orbits, stabilizer codes, and stabilizer states}

Given $|\Psi\>\in \caH_d^{\otimes n}$, denote by $\orb(\Psi)$ the Clifford orbit generated from $|\Psi\>$, that is, 
\begin{align}
\orb(\Psi):=\{U|\Psi\> \, |\, U\in \oCl(n,d) \}. 
\end{align}
Note that overall phase factors are irrelevant to the current study. Here $\orb(\Psi)$ is also regarded as a normalized ensemble of pure states: all states in the orbit have the same weight by default. Every Clifford orbit is a  2-design because the Clifford group is a  2-design; when $d=2$, every Clifford orbit is also  a 3-design because the Clifford group is a  3-design.

A stabilizer group $\caS$ is  an Abelian subgroup  of $\caW(n,d)$ that does not contain the operator $\omega \bbI$, where  $\omega=\rme^{\frac{2\pi \rmi}{d}}$. By definition the order of $\caS$ is a power of $d$ and is upper bounded by $d^n$. The \emph{stabilizer code} $\caC_\caS$ associated with  $\mathcal{S}$ is the common eigenspace of operators in $\mathcal{S}$ with eigenvalue 1 \cite{NielC10book,Gros06}. The corresponding eigenprojector is called a \emph{stabilizer projector} and can be expressed as follows,
\begin{equation}
P_{\mathcal{S}}=\frac{1}{|\mathcal{S}|} \sum_{U \in \mathcal{S}} U.
\end{equation}
Note that $\rk P_{\mathcal{S}}=\tr P_{\mathcal{S}}=d^n/|\caS|$.

If  the stabilizer group $\caS$ has the maximum order of $d^n$, then the stabilizer code $\caC_\caS$ is one-dimensional and can be represented by a normalized state, which is known as  a \emph{stabilizer state}. In addition, the common eigenbasis of $\caS$ is called a \emph{stabilizer basis}. Note that every basis state in the stabilizer basis is also a stabilizer state associated with some stabilizer group.  The set of all $n$-qudit stabilizer states is denoted by $\Stab(n,d)$ henceforth, which is also regarded as a normalized ensemble of pure states. It is known that all stabilizer codes of the same dimension form one orbit under the action of the Clifford group $\oCl(n,d)$. Notably, all stabilizer states in $\Stab(n,d)$ form one orbit \cite{Gros06}.  When $d=2$, the ensemble $\Stab(n,d)$ is a 3-design; when $d$ is an odd prime, by contrast,  $\Stab(n,d)$ is only a 2-design, but not a 3-design  according to \rcite{KuenG13}, which derived a recursion formula for the $t$th frame potential; see \eref{eq:Phi3Stab} in \sref{sec:3rdmomentStab} for the third frame potential.

\subsection{Commutant of the $t$th Clifford tensor power}
Schur-Weyl duality states that any operator on $\caH^{\otimes t}$ that commutes with $U^{\otimes t}$ for all unitaries $U$ on $\caH$ is necessarily a linear combination of permutation operators. Recently, Gross,  Nezami, and Walter (GNW) proved a theorem similar in flavor for the Clifford group \cite{GrosNW21}. Their work is 
the basis for understanding the third moments of
Clifford orbits. Here we recapitulate their main results that are relevant to the current study, assuming that the local dimension $d$ is an odd prime.

\subsubsection{\label{sec:SLSspanningSet}Stochastic Lagrangian subspaces and a spanning set of the commutant}

Let $\mathbf{1}_{2t} = (1,1,\dots,1)^\top$ be the column vector in $\bbF_d^{2t}$ whose entries are all equal to 1. A subspace $\caT \leq  \bbF_d^{2t}$ (the inequality means $\caT$ is a subspace of $\bbF_d^{2t}$) is a \emph{stochastic Lagrangian subspace} if it satisfies the following three conditions \cite{GrosNW21}:
\begin{enumerate}
	\item $\bfx \cdot \bfx - \bfy \cdot \bfy = 0 $ for any $(\bfx;\bfy) \in \caT$.
	\item $\caT$ has dimension $t$.
	\item $\mathbf{1}_{2t} \in \caT$.
\end{enumerate}
Here $\bfx$ and $\bfy$ are regarded as column vectors in  $\bbF_d^{t}$, and $(\bfx;\bfy)$ is regarded as a column vector in $\bbF_d^{2t}$ obtained by concatenating $\bfx$ and $\bfy$ vertically [in contrast with the notation $(\bfx,\bfy)$, which means horizontal concatenation]; all arithmetic are modulo $d$. Condition~1 means the subspace $\caT$ is totally isotropic with respect to the quadratic form $\mathfrak{q}$: $\bbF_d^{2t} \to \bbF_d$ defined as $\mathfrak{q}(\bfx,\bfy):=\bfx \cdot \bfx - \bfy \cdot \bfy$. 
Condition~2 requires $\caT$ to have the maximal possible dimension. A subspace satisfying the first two conditions
is called a  Lagrangian. Condition~3 means $\caT$ is stochastic.
Denote by $\Sigma_{t,t}(d)$ the set of stochastic Lagrangian subspaces in  $\bbF_d^{2t}$; its
cardinality reads \cite{GrosNW21}
\begin{equation}\label{eq:SigmattOrder}
|\Sigma_{t,t}(d)|=\prod_{k=0}^{t-2}(d^k+1).
\end{equation}

Given any subspace $\caT$ in  $\bbF_d^{2t}$, define the operators
\begin{align}\label{eq:rRT}
r(\caT):= \sum_{(\bfx;\bfy) \in \caT} |\bfx\> \<\bfy|, \quad 
R(\caT):= r(\caT)^{\otimes n},
\end{align}
which act on $\caH_d^{\otimes t}$ and $\bigl(\caH_d^{\otimes t}\bigr)^{\otimes n}$, respectively. Here $R(\caT)$ can also be regarded as an operator acting on $\bigl(\caH_d^{\otimes n}\bigr)^{\otimes t}$ given the isomorphism between $(\caH_d^{\otimes t})^{\otimes n}$ and  $\bigl(\caH_d^{\otimes n}\bigr)^{\otimes t}$.  This operator  is particularly important  when $\caT$ is a stochastic Lagrangian subspace, that is, $\caT\in \Sigma_{t,t}(d)$. In this case,
according to GNW \cite{GrosNW21}, $R(\caT)$  belongs to the commutant of
\begin{align}
\Cl(n,d)^{\totimes t}:=\{U^{\otimes t}\,|\, U\in \Cl(n,d)\},
\end{align}
 where the notation $\totimes$ means diagonal tensor power; moreover, 
the set  $\{R(\caT)\}_{\caT\in \Sigma_{t,t}(d)}$ spans the commutant of $\Cl(n,d)^{\totimes t}$ when  $n\geq t-1$. GNW also conjectured that the condition $n\geq t-1$ is not necessary. We shall prove their speculation for the special case of $t=3$ in \sref{sec:Spanning} by comparing the rank of the Gram matrix of $\{R(\caT)\}_{\caT\in \Sigma_{t,t}(d)}$  with the dimension of the commutant, which is equal to the $t$th frame potential of the Clifford group as reproduced in \eref{eq:CliffordFP3}. The definitions in \eref{eq:rRT} admit straightforward generalization to  any subset in $\Sigma_{t,t}(d)$, which is very convenient to later discussion,
\begin{align}\label{eq:RscrT}
r(\scrT):=\sum_{\caT\in \scrT} r(\caT),\quad  R(\scrT):=\sum_{\caT\in \scrT} R(\caT),\quad \scrT\subseteq \Sigma_{t,t}(d).
\end{align}

The \emph{diagonal stochastic Lagrangian subspace} in  $\bbF_d^{2t}$ is defined as \cite{GrosNW21}
\begin{equation}\label{eq:Delta}
\Delta :=\left\{(\bfx;\bfx) \ | \  \bfx \in \bbF_d^t \right\}. 
\end{equation}
It is of special interest because $r(\Delta)$ and $R(\Delta)$ are the identity operators on $\caH_d$ and $\caH_d^{\otimes n}$, respectively. In addition, it is useful to computing the trace of $R(\caT)$ for $\caT\in \Sigma_{t,t}(d)$. Define
\begin{align}\label{eq:TDeltaDef}
\caT_\Delta:=\caT\cap \Delta, \quad 	
\caT_{\tDelta}:=\{\bfx\ | \ (\bfx;\bfx)\in \caT_\Delta\};
\end{align}
then  $\mathbf{1}_{2t}\in \caT_\Delta$ and  $\mathbf{1}_t\in \caT_\tDelta$ by definition. In addition,
\begin{align}\label{eq:RTtr}
\tr R(\caT)&=[\tr r(\caT)]^n=d^{n\dim \caT_\Delta}=d^{n\dim \caT_{\tDelta}}\quad \forall \caT\in \Sigma_{t,t}(d). 
\end{align} 
This formula admits the following generalization,
\begin{align}\label{eq:RT1T2tr}
&\tr\bigl[R(\caT_1)^\dag R(\caT_2)\bigr]=\bigl\{\tr\bigl[r(\caT_1)^\dag r(\caT_2)\bigr]\bigr\}^n=d^{n\dim(\caT_1 \cap \caT_2)}\quad \forall \caT_1, \caT_2\in \Sigma_{t,t}(d). 
\end{align}

\subsubsection{\label{sec:Otd}Stochastic orthogonal group}

A  $t \times t$ matrix $O$ over $\bbF_d$ is a stochastic isometry if it satisfies the following two conditions:
\begin{enumerate}
	\item $O\bfx\cdot O\bfx=\bfx\cdot \bfx $ for any $\bfx \in \bbF_d^t$.
	\item $O \cdot \mathbf{1}_t = \mathbf{1}_t $, where $\mathbf{1}_t=(1,1,\ldots, 1)^\top$.
\end{enumerate}
The two conditions imply that $O^\top O=OO^\top=\mathds{1}$, where $\mathds{1}$ denotes the identity operator on $\bbF_d^t$.
The stochastic orthogonal group $O_t(d)$ is the group  of all stochastic isometries on $\bbF_d^t$ \cite{GrosNW21}. Note that $O_t(d)$ contains the group $S_t$ of permutation matrices on $\bbF_d^t$. 
The significance of this group is two-fold. First, given any $O \in O_t(d)$, the subspace 
\begin{equation}\label{eq:TOdef}
\caT_O:=\left\{ (O\bfx;\bfx)\, | \, \bfx \in \bbF_d^t \right\}
\end{equation}
is a stochastic Lagrangian subspace in $ \bbF_d^{2t}$, that is,  $\caT_O\in \Sigma_{t,t}(d)$. For simplicity of notation, we shall write 
\begin{equation}
r(O):=r(\caT_O), \quad
R(O):=R(\caT_O).
\end{equation}
Note that $r(O)$ and $R(O)$ are permutation matrices for any $O \in O_t(d)$.

Second, $O_t(d)$ is a natural symmetry group of $\Sigma_{t,t}(d)$ given the following left and right actions  \cite{GrosNW21},
\begin{equation}\label{eq:OTTO}
\begin{aligned}
O\caT:=\{(O\bfx;\bfy)\ | \ (\bfx;\bfy) \in \caT\},\quad
\caT O:=\bigl\{\bigl(\bfx; O^\top\bfy\bigr)\ | \ (\bfx;\bfy) \in \caT\bigr\}.
\end{aligned}
\end{equation}
By definition we have $R(O)R(\caT)R(O')=R(OTO')$ for $\caT\in \Sigma_{t,t}(d)$ and $O, O'\in O_t(d)$, so the actions are compatible with the composition of the operators $R(O)$, $R(\caT)$, and $R(O')$. These actions can be generalized to any subset $\scrT$ of $\Sigma_{t,t}(d)$ in a straightforward way: for example, $O\scrT=\{O\caT\,|\,\caT\in \scrT\}$. Based on the above actions, we can decompose $\Sigma_{t,t}(d)$ into a disjoint union of double cosets:
\begin{equation}\label{eq:doublecosets}
\Sigma_{t,t}(d) = O_t(d)\caT_1 O_t(d) \cup \cdots \cup O_t(d)\caT_k O_t(d),
\end{equation}
where $\caT_1, \caT_2,\dots, \caT_k\in \Sigma_{t,t}(d)$ are  representatives of the double cosets. Without loss of generality, we can choose $\caT_1=\Delta$, 
which is tied  to the special double coset $O_t(d)\Delta O_t(d)=\left\{\caT_O \ | \ O\in O_t(d)\right\}$.

When $t=2$,  $O_2(d)$ coincides with the symmetric group $S_2$, and $\Sigma_{t,t}(d)=\{\caT_O \ | \ O \in O_2(d)\}$ has only one double coset,
so all stochastic Lagrangian subspaces 
correspond to permutations. This result is consistent with the fact that the Clifford group with a prime local dimension is a unitary 2-design, and it is even a 3-design in the case of qubits \cite{Zhu17MC,Webb16}. 

\subsubsection{\label{sec:Defect}Defect subspaces}
To characterize general double cosets in the decomposition in \eref{eq:doublecosets}, we need to introduce some additional concepts. A nonzero vector $\bfx$ in  $\bbF_d^t$ is a \emph{defect vector} if $\bfx\cdot \bfx =\bfx\cdot \mathbf{1}_t=0$. A  subspace $\caN \leq \bbF_d^t$ is a \emph{defect subspace} if every nonzero vector in $\caN$ is a defect vector. 
The special defect subspace $\left\{\mathbf{0}_t\right\}$ associated with the zero vector  $\mathbf{0}_t$
in $\bbF_d^t$ is called trivial, while other defect subspaces are nontrivial.
By definition it is easy to verify that any two vectors $\bfx,\bfy$ in a defect subspace $\caN$ are orthogonal in the sense that  $\bfx\cdot \bfy=0$, given that  $d$ is an odd prime by assumption. 
Let  $\caN^{\perp}$ denote the orthogonal complement of $\caN$, that is, $\caN^{\perp}:=\left\{\bfy\in \bbF_d^t \,|\, \bfy\cdot \bfx = 0 \,\, \forall \bfx \in \caN\right\}$; then $\caN \leq \caN^{\perp}$. 
So the dimension of any defect subspace in $\bbF_d^t$ cannot surpass $t/2$. 
In analogy to \eref{eq:OTTO}, the left and right actions of $O\in O_t(d)$ on $\caN$ are defined as follows,
\begin{align}
O\caN := \{O\bfx\, | \, \bfx \in \caN \},\;\;
\caN O := \bigl\{O^\top \bfx \, | \, \bfx \in \caN \bigr\}.
\end{align}

The left and right defect subspaces of a stochastic Lagrangian subspace $\caT$ in $\Sigma_{t,t}(d)$ are defined as  \cite{GrosNW21}
\begin{align}
\caT_{\LD}:=\left\{ \bfx\,|\, (\bfx;\mathbf{0}_t) \in \caT\right\}, \quad \caT_{\RD}:=\left\{ \bfy \,|\, (\mathbf{0}_t;\bfy) \in \caT\right\},
\end{align}
which satisfy
\begin{align}
\dim \caT_{\LD}=\dim \caT_{\RD},\quad  \mathbf{1}_t\in \caT_{\LD}  \; \mbox{iff}\; \mathbf{1}_t\in \caT_{\RD}. 
\end{align}
Conversely, suppose $\caN,\caM \leq \bbF_d^t$ are two defect subspaces of the same dimension and $\mathbf{1}_t \in \caN$ iff $\mathbf{1}_t \in \caM$. Then  there  exists $O \in O_t(d)$ such that $O\caN=\caN O^\top=\caM$. In addition, there exists $\caT \in \Sigma_{t,t}(d)$ such that $\caT_{\LD}=\caN$ and $\caT_{\RD}=\caM$. Moreover, 
according to  Corollary 4.21 in \rcite{GrosNW21}, $\caT,\caT'\in \Sigma_{t,t}(d)$ belong to the same double coset, that is, $\caT' \in O_t(d)\caT O_t(d)$, iff the following conditions hold,
\begin{equation}\label{eq:SameCosetCon}
\dim \caT_{\LD}=\dim \caT'_{\LD}, \quad \mathbf{1}_t\in \caT_{\LD}  \; \mbox{iff}\; \mathbf{1}_t\in \caT'_{\LD}.
\end{equation}
Notably,  all stochastic Lagrangian subspaces whose defect subspaces are trivial belong to the same double coset, namely, $O_t(d)\Delta O_t(d)$.

\subsubsection{Connection with Calderbank-Shor-Steane codes}

Given any defect subspace $\caN$ in $\bbF_d^t$, a stochastic Lagrangian subspace $\caT$ with $\caT_{\LD}=\caT_{\RD}=\caN$ can be constructed as follows \cite{GrosNW21},
\begin{align}\label{eq:Tcss}
\caT = \left\{ (\bfx+\bfz;\bfx+\mathbf{w})\,|\, [\bfx] \in \caN^{\perp}/\caN, \bfz, \mathbf{w} \in \caN\right\}= \left\{(\bfx;\bfy)\,|\, \bfy \in \caN^{\perp}, \bfx\in [\bfy]\right\} =\left\{(\bfx;\bfy)\,|\, \bfx \in \caN^{\perp}, \bfy\in [\bfx]\right\},
\end{align}
where $[\bfx]$ denotes the image of $\bfx$ in the quotient space $\caN^{\perp}/\caN$. This stochastic Lagrangian subspace 
is of special interest because the operator $r(\caT)$ is tied to a Calderbank-Shor-Steane (CSS) code  \cite{CaldS96,Stea96}. Since  $\bfx \cdot \bfx =0 $ for all $\bfx \in \caN$ by definition,  we can construct a CSS stabilizer group as follows,
\begin{equation}
\mathrm{CSS}(\caN):=\left\{W_{(\bfp;\mathbf{0}_t)} W_{(\mathbf{0}_t;\bfq)} \,|\, \bfp,\bfq \in \caN\right\}.
\end{equation}
The corresponding stabilizer  projector reads
\begin{equation}\label{eq:CSSprojCosetS}
P_{\mathrm{CSS}(\caN)}=\sum_{[\bfx] \in \caN^\perp/\caN} |\caN,[\bfx]\> \<\caN,[\bfx]|,\quad |\caN,[\bfx]\>:= |\caN|^{-1/2} \sum_{\bfz \in \caN} |\bfx+\bfz\>,
\end{equation}
where $|\caN,[\bfx]\>$ is called a \emph{coset state}.  The relation between $r(\caT)$ and $P_{\mathrm{CSS}(\caN)}$ is clarified below \cite{GrosNW21}, 
\begin{equation}\label{eq:rTCSS}
r(\caT)=|\caN|P_{\mathrm{CSS}(\caN)}.
\end{equation}
In view of this relation,   subspaces of the form  in \eref{eq:Tcss} are called  stochastic Lagrangian subspaces of CSS type. 
If $\caN$ in \eref{eq:Tcss} is the trivial defect subspace, then $\caT=\Delta$, so  $\Delta$ is also of CSS type.

According to the criterion in \eref{eq:SameCosetCon},  any double coset in \eref{eq:doublecosets} contains a stochastic Lagrangian subspace of CSS type.  In addition, any stochastic Lagrangian subspace in $\bbF_d^{2t}$ can be expressed as $O\caT$, where $O\in O_t(d)$  and $\caT\in \Sigma_{t,t}(d)$ is a stochastic Lagrangian subspace of CSS type; the same conclusion holds if $O\caT$ is replaced by $\caT O$, although  usually $O\caT\neq \caT O$.

\subsection{Moments of stabilizer states}
The $t$th moment operator and normalized moment operator of stabilizer states are defined as  
\begin{align}\label{eq:MomentOtthDef}
Q(n,d,t)&: = Q_t(\Stab(n,d))= \mathbb{E}_{\Stab(n,d)}\bigl[(|S\> \<S|)^{\otimes t}\bigr], \quad \bQ(n,d,t):=\pi_{[t]}Q(n,d,t). 
\end{align}
Since the set of stabilizer states forms one orbit of the Clifford group,  $Q(n,d,t)$ and  $\bQ(n,d,t)$ commute with $U^{\otimes t}$ for all $U \in \Cl(n,d)$ and thus can be expressed 
as linear combinations of $R(\caT)$ for $\caT\in \Sigma_{t,t}(d)$ when $n\geq t-1$, where  $\Sigma_{t,t}(d)$ is the set of stochastic Lagrangian subspaces in $\bbF_d^{2t}$. It turns out the assumption $n\geq t-1$ is not necessary. 
According to  \rcite{GrosNW21}, $Q(n,d,t)$ and $\bQ(n,d,t)$ can be expressed as follows,
\begin{align}\label{eq:MomentOtth}
Q(n,d,t)&= \frac{1}{Z_{n,d,t}} \sum_{\caT \in \Sigma_{t,t}(d)} R(\caT),\quad \bQ(n,d,t)= \frac{\pi_{[t]}}{Z_{n,d,t}} \sum_{\caT \in \Sigma_{t,t}(d)} R(\caT),
\quad Z_{n,d,t} = D \prod_{k=0}^{t-2}(d^k+D). 
\end{align}
This result is applicable whenever $d$ is a prime, although we have not discussed the structure of $\Sigma_{t,t}(d)$ for the special case $d=2$. It is the basis for understanding the  operators $Q(n,d,t)$ and  $\bQ(n,d,t)$.

\section{\label{sec:SLS}Stochastic Lagrangian subspaces in  $\bbF_d^6$}
In preparation for studying the commutant of the third Clifford tensor power and the third moments of Clifford orbits here we clarify the basic properties of the stochastic orthogonal group $O_3(d)$ and the set $\Sigma(d):=\Sigma_{3,3}(d)$ of stochastic Lagrangian subspaces, assuming that $d$ is an odd prime.

\subsection{\label{sec:O3d}Stochastic orthogonal group $O_3(d)$}
When $d=3$, $O_3(d)$ coincides with the symmetric group $S_3$ generated by 
\begin{align}\label{eq:zetatau}
\zeta:=\begin{pmatrix}
0&0&1\\
1&0&0\\
0&1&0
\end{pmatrix},\quad \tau_{12}:=\begin{pmatrix}
0&1&0\\
1&0&0\\
0&0&1
\end{pmatrix}, 
\end{align}
which are tied to  the cyclic permutation  $1 \mapsto 2 \mapsto 3$ and  transposition $1 \leftrightarrow 2$, respectively.
Let $\tau_{13} =\zeta \tau_{12}$ and $\tau_{23}=\zeta^2 \tau_{12}$; 
then  $S_3=\{\mathds{1}, \zeta,\zeta^2, \tau_{12},\tau_{23},\tau_{13}\}$. 
In general, $O_3(d)$ contains $S_3$ as a subgroup.  Note that $\zeta$ and $\tau_{12}$ can be regarded as matrices over $\bbF_d$ for any prime $d$.

Denote by $\caD_h$ the dihedral group of order $2h$, which is the symmetry group of an $h$-sided regular polygon with $h\geq3$.  The following lemma is proved in \aref{appendix:O3Dh}. 
\begin{lemma}\label{lem:O3Dh}
	Suppose $d$ is an odd prime. Then $O_3(d)$ is isomorphic to $\caD_h$ and has order $2h$, where 
	\begin{equation}\label{eq:hvalue}
	h=
	\begin{cases}
	d+1 \quad & \text{if} \ d = 2 \mmod 3,\\
	d-1 \quad & \text{if} \ d = 1 \mmod 3,\\
	3 \quad & \text{if} \ d = 3.
	\end{cases}
	\end{equation}
\end{lemma}

Suppose $O\in O_3(d)$ and $a,b,c$ are the three entries in  a  column of $O$; then $a^2+b^2+c^2=a+b+c=1$.
Moreover, any row or column of $O$ is a permutation of $a,b,c$. 
If in addition $a,b,c$ are the three entries in the first column, then $O$ is equal to either $O_\even$ or $O_\odd$ (see \rcite{NezaW20} for the case  $d=2\mmod 3$), where
\begin{equation}\label{eq:Oevenodd}
O_\even=aI + b \zeta + c \zeta^{-1}=
\begin{pmatrix}
a & c & b \\
b & a & c \\
c & b & a 
\end{pmatrix},
\quad O_\odd=a\tau_{23} + b \tau_{12} + c \tau_{13}=
\begin{pmatrix}
a & b & c \\
b &c & a \\
c & a & b    
\end{pmatrix}. 
\end{equation}
The two stochastic isometries satisfy 
\begin{equation}\label{eq:Odet}
\det(O_\even)=-\det(O_\odd)=1, \quad O_\even\zeta=\zeta O_\even,
\quad
O_\odd^2 =\mathds{1},\quad  O_{\odd}\zeta O_{\odd}=\zeta^{-1}=\zeta^2. 
\end{equation}
Henceforth $O_\even$ ($O_\odd$) is  called an even (odd) stochastic isometry. Denote by $O_3^\even(d)$ ($O_3^\odd(d)$) the set of even (odd) stochastic isometries,
\begin{align}
O_3^\even(d):=\{O\in O_3(d)\ |\  \det(O)=1\},\quad
O_3^\odd(d):=\{O\in O_3(d)\ |\ \det(O)=-1\};
\end{align} 
then  $O_3(d)=O_3^\even(d)\sqcup O_3^\odd(d)$, where the symbol "$\sqcup$" denotes the union of disjoint sets. Moreover,  $O_3^\even(d)$ is the unique cyclic  subgroup of $O_3(d)$ of index 2. Additional results about the group $O_3(d)$ can be found in \aref{app:sec:O3dAux}.

\subsection{Nontrivial defect subspaces in   $\bbF_d^3$}
To determine the stochastic Lagrangian subspaces in $\Sigma(d)$, we first need to determine nontrivial defect subspaces in  $\bbF_d^3$, which are necessarily  one dimensional according to \sref{sec:Defect}. 
When $d=3$, there is only one  nontrivial defect subspace, namely, 
\begin{align}\label{eq:DefectSpaced3}
\tcaN:=\{(0;0;0),(1;1;1),(2;2;2)\}. 
\end{align} 
In the case $d = 1\mmod  3$, let $\xi$ be an order-3 element in $\bbF_d^\times:=\bbF_d\setminus \{0\}$ and define 
\begin{equation}\label{eq:TwoDefectSpaces}
\bfg_0 : =(1,\xi,\xi^2)^\top, \quad   \tcaN_0:=\spa(\bfg_0), \quad 
\bfg_1 : =(1,\xi^2,\xi)^\top,\quad \tcaN_1:=\spa(\bfg_1). 
\end{equation}
Then it is easy to verify that  
\begin{equation}\label{eq:g12Properties}
\zeta \bfg_0=\xi^2 \bfg_0, \quad \zeta \bfg_1=\xi \bfg_1, \quad 
\mathbf{1}_3\cdot \bfg_j=\bfg_j\cdot \bfg_j=0 \quad \forall j=0,1, 
\end{equation}
where $\zeta\in S_3$ is the cyclic permutation defined in \eref{eq:zetatau}.
So the three vectors $\mathbf{1}_3, \bfg_0, \bfg_1$ are linearly independent and form a basis of $\bbF_d^3$. In addition,  $\bfg_0, \bfg_1$ are defect vectors  and $\tcaN_0,\tcaN_1$ are two distinct nontrivial defect subspaces.
Actually, they are the only nontrivial defect subspaces as shown in \lref{lem:defect} below and proved in \aref{app:DefectProof}. By virtue of \eref{eq:g12Properties} we can further deduce that
\begin{align}\label{eq:NbotN}
\tcaN_j^{\perp} &=\spa(\{\bfg_i, \mathbf{1}_3\}), \quad
\tcaN_j^{\perp}/\tcaN_j = \{ [\bfx]\ |\ \bfx = x\cdot \mathbf{1}_3,\ x \in \bbF_d\} \quad   \forall j=0,1.
\end{align}

\begin{lemma}\label{lem:defect}
	Suppose $d$ is an odd prime. If $d=3$, then $\tcaN$ defined in \eref{eq:DefectSpaced3} is the only nontrivial defect subspace in 
	$\bbF_d^3$. If $d=2\mmod 3$, then there are no nontrivial defect subspaces in $\bbF_d^3$. If 
	$d = 1\mmod  3$, then  $\tcaN_0$ and $\tcaN_1$ defined in \eref{eq:TwoDefectSpaces} are the only two nontrivial defect subspaces in $\bbF_d^3$.   
\end{lemma}

Next, we clarify the action of $O_3(d)$ on nontrivial defect subspaces. The following lemma is also proved in \aref{app:DefectProof}. 
\begin{lemma}\label{lem:ONNO}
	Suppose $d$ is an odd prime. If $d=3$ and $\tcaN$ is the defect subspace defined in  \eref{eq:DefectSpaced3}, then	 
	\begin{align}\label{eq:ONNOd3}
	O \tcaN=\tcaN O=\tcaN  \quad \forall O\in O_3(d). 
	\end{align}
	If 
	$d = 1 \mmod  3$ and $\tcaN_0$, $\tcaN_1$ are the two defect subspaces defined in \eref{eq:TwoDefectSpaces}, then
	\begin{align}
	O \tcaN_j &=\tcaN_j O =\tcaN_j \quad \forall O \in O_3^\even(d), \quad 
	O \tcaN_j =\tcaN_j O =\tcaN_{\bar{j}} \quad \forall O \in O_3^\odd(d),\quad  j=0,1, \label{eq:ONNO}
	\end{align}	
	where $\bar{j}=1-j$. 
\end{lemma}

\subsection{\label{sec:SLSconstruct1}Construction of stochastic Lagrangian subspaces}
Here we determine the set $\Sigma(d)$  of stochastic Lagrangian subspaces in $\bbF_d^6$. 
For the convenience of the following discussions we define
\begin{equation}\label{eq:scrTSIEO}
\begin{aligned}
\scrT_\sym&:=\{\caT_O\ |\ O\in S_3\}, & 	\scrT_\ns&:=\Sigma(d)\setminus\scrT_\sym,\\
\scrT_\iso&:= \{T_O \; | \; O \in O_3(d)\}=O_3(d)\Delta O_3(d), &	\scrT_\defe&:=\Sigma(d)\setminus \scrT_\iso,
\\
\scrT_{\even}&:=\{T_O\ | \ O\in O_3^\even(d) \}, &
\scrT_{\odd}&:=\{T_O\ | \ O\in O_3^\odd(d) \}=\scrT_\iso\setminus \scrT_\even.
\end{aligned}
\end{equation}
Note that $\Delta\in \scrS_\sym$ by definition.
By virtue of \eref{eq:SigmattOrder} and \lref{lem:O3Dh} we can deduce that
\begin{align}\label{eq:Sigma33DefectNum}
|\Sigma(d)|=2d+2, \quad |\scrT_\defe|=|\Sigma(d)|-|O_3(d)|=
\begin{cases}
2   & d=3,\\
4   & d = 1\mmod 3,\\
0 & d=2\mmod 3.
\end{cases}
\end{align}
When $d = 2\mmod 3$,  $\scrT_\defe$ is empty and  $\Sigma(d)=\scrT_\iso$ contains only one double coset.
When  $d \neq 2 \mmod 3$, according to 
\lsref{lem:defect} and \ref{lem:ONNO} as well as the discussion in \sref{sec:Defect}, $\Sigma(d)$ contains two double cosets: all stochastic Lagrangian subspaces  in $\scrT_\iso$ (with trivial defect subspaces) form one double coset, and those in $\scrT_\defe$ (with nontrivial defect subspaces)
form another double coset.
By virtue of \lsref{lem:defect} and \eref{eq:Sigma33DefectNum} we can further determine all  stochastic Lagrangian subspaces in  $\scrT_\defe$ as shown in the following lemma and proved in \aref{app:SLagrangTproof}. Here we use a generating matrix to represent a  stochastic Lagrangian  subspace, which corresponds to the column span of the generating matrix. Note that the choice of the generating matrix is not unique.

\begin{lemma}\label{lem:defectT}
	When  $d=3$, $\scrT_\defe$ is composed of  the  two  stochastic Lagrangian subspaces:
	\begin{equation}\label{eq:SLStwo}
\tcaT_0=\left(\begin{array}{ccc}
	1 & 0 & 1 \\
	2 & 0 & 1 \\
	0 & 0 & 1 \\
	\hline 
	1 & 1 & 0 \\
	2 & 1 & 0 \\
	0 & 1 & 0 
	\end{array}
	\right),\quad
\tcaT_1=\left(\begin{array}{ccc}
	1 & 0 & 1 \\
	2 & 0 & 1 \\
	0 & 0 & 1 \\
	\hline 
	2 & 1 & 0 \\
	1 & 1 & 0 \\
	0 & 1 & 0
	\end{array}
	\right),
	\end{equation}
	in which $\tcaT_0$ is of CSS type. 	
	When $d = 1 \mmod 3$, $\scrT_\defe$ is composed of  the four  stochastic Lagrangian subspaces:
	\begin{equation}\label{eq:SLSfour}
	\tcaT_{ij} = \left(\begin{array}{ccc}
	\bfg_i & \mathbf{0}_3 & \mathbf{1}_3 \\
	\mathbf{0}_3 & \bfg_j & \mathbf{1}_3 \\
	\end{array}
	\right), \quad i,j=0,1,
	\end{equation}
	in which $\tcaT_{00}$ and $\tcaT_{11}$ are of CSS type.
\end{lemma}

When $d=3$, the left and right defect subspaces of $\tcaT_0$ and $\tcaT_1$ are all equal to $\tcaN$  defined in \eref{eq:DefectSpaced3}. 
When $d = 1 \mmod 3$, the left and right defect subspaces   of $\tcaT_{ij}$ read
\begin{equation}\label{eq:tTijDefect}
(\tcaT_{ij})_\LD=\tcaN_i, \quad (\tcaT_{ij})_\RD=\tcaN_j, \quad i,j=0,1,
\end{equation}
where $\tcaN_0$,  $\tcaN_1$ are defined in \eref{eq:TwoDefectSpaces}. In this case, each stochastic Lagrangian subspace is completely determined by its left and right defect subspaces.

For the convenience of the following discussions, define
\begin{subequations}\label{eq:scrT01}
	\begin{align}
	\scrT_0&:=\begin{cases}
	\scrT_{\even}\cup \{\tcaT_1\} &\! \mbox{if $d= 3$,}\\
	\scrT_{\even}\cup \{\tcaT_{01},\tcaT_{10}\} &\! \mbox{if $d=1\mmod 3$},\\
	\scrT_{\even} &\! \mbox{if $d=2\mmod 3$},
	\end{cases}\\
	\scrT_1&:=\Sigma(d)\setminus \scrT_0=\begin{cases}
	\scrT_{\odd}\cup \{\tcaT_0\} &\! \mbox{if $d=3$,}\\
	\scrT_{\odd}\cup \{\tcaT_{00},\tcaT_{11}\} &\! \mbox{if $d=1\mmod 3$},\\
	\scrT_{\odd} &\! \mbox{if $d=2\mmod 3$},\\
	\end{cases}
	\end{align}
\end{subequations}
where $\scrT_{\even}$ and $\scrT_{\odd}$ are defined in \eref{eq:scrTSIEO}. By definition we have
\begin{align}
|\scrT_0|=|\scrT_1|=d+1,\quad
\scrT_0\cap\scrT_\iso =\scrT_\even,\quad \scrT_1\cap\scrT_\iso =\scrT_\odd. 
\end{align}
In the rest of this section, we clarify the 
action of $O_3(d)$ on  $\scrT_\defe$, $\Sigma(d)$ 
and the intersections of stochastic Lagrangian subspaces. \Lsref{lem:SemigroupT}-\ref{lem:TDelta2} below are also proved in \aref{app:SLagrangTproof}.
Additional results can be found in \aref{app:SigmadAux}.

\begin{lemma}\label{lem:SemigroupT}
	Suppose $d$ is an odd prime,  $i,j=0,1$, $\bar{i}=1-i$, and $\bar{j}=1-j$. If  $d=3$, then 
	\begin{align} \label{eq:OTTO3}
	O\tcaT_i=\tcaT_i O=\tcaT_i \quad \forall O\in 	O_3^\even(d),\quad 
	O\tcaT_i=\tcaT_i O=\tcaT_{\bar{i}} \quad \forall O\in 	O_3^\odd(d).
	\end{align}	
	If $d = 1 \mmod  3$, then
	\begin{align}\label{eq:OTTOd}
	O\tcaT_{ij}=\tcaT_{ij} O=\tcaT_{ij} \quad \forall O\in 	O_3^\even(d);\quad 
	O\tcaT_{ij}=\tcaT_{\bar{i}j},\quad  \tcaT_{ij} O=\tcaT_{i\bar{j}} \quad \forall O\in 	O_3^\odd(d). 
	\end{align}
\end{lemma}
\Lref{lem:SemigroupT} means all stochastic Lagrangian subspaces  in $\scrT_\defe$ form one orbit under the action of  $O_3(d)$ or~$S_3$. 

\begin{lemma}\label{lem:SemigroupscrT}
	Suppose $d$ is an odd prime, $i=0,1$, and $\bar{i}=1-i$. Then
	\begin{align}	
	O\scrT_i=\scrT_i O=\scrT_i \quad \forall O\in O_3^\even(d), \quad  
	O\scrT_i=\scrT_i O=\scrT_{\bar{i}} \quad \forall O\in O_3^\odd(d).
	\end{align}	
\end{lemma}

\begin{lemma}\label{lem:TDelta}
	Suppose $d$ is  an odd prime and $\caT\in \Sigma(d)$. Then $\spa(\mathbf{1}_6)\leq \caT_\Delta$. If  $\caT\in \scrT_0$ and $\caT\neq \Delta$, then 
	$\caT_\Delta=\spa(\mathbf{1}_6)$. If  $\caT\in \scrT_1$, then $\dim \caT_\Delta=2$. 
\end{lemma}
\begin{lemma}\label{lem:TDelta2}
	Suppose  $d$ is an odd prime and $\caT_1\in \scrT_i$, $\caT_2\in \scrT_j$ with $i,j=0,1$. Then $\spa(\mathbf{1}_6)\leq \caT_1\cap \caT_2$ and
	\begin{align}\label{eq:T1capT2}
	\dim (\caT_1\cap \caT_2)=\begin{cases}
	3 & \mbox{if}\;\; \caT_1=\caT_2, \\
	1 & \mbox{if}\;\; \caT_1\neq \caT_2, i=j, \\
	2 &  \mbox{if}\;\; i\neq j. 
	\end{cases}
	\end{align}
\end{lemma}

\subsection{\label{sec:SLSconstruct2}Alternative construction of stochastic Lagrangian subspaces}

In this section we introduce an alternative construction of stochastic Lagrangian subspaces  in $\Sigma(d)$, which will be very useful to understanding the third moments of  Clifford orbits based on magic states and their application in shadow estimation.

Given any nonzero vector $\bfv$ in $\mathbf{1}_3^\perp=\{\bfx\in \bbF_d^3\,|\, \bfx \cdot \mathbf{1}_3=0 \}$, we can construct a stochastic Lagrangian subspace in $\bbF_d^6$ as follows, 
\begin{align}\label{eq:Tv}
\caT_\bfv:=\{(\bfx+a\bfv; \bfx-a\bfv)\ |\  \bfx\in \bfv^\perp,  a\in \bbF_d \}.
\end{align}
By definition we have 
\begin{align}\label{eq:TvDeltaDefect}
(\caT_\bfv)_\tDelta=\bfv^\perp,\quad  \caT_\bfv\in \scrT_1, \quad (\caT_\bfv)_\LD=(\caT_\bfv)_\RD=\begin{cases}
\spa(\bfv) & \mbox{if}\; \bfv\cdot \bfv= 0,\\
\mathbf{0}_3 &\mbox{if}\;  \bfv\cdot \bfv\neq 0,
\end{cases}
\end{align}
where the relation $\caT_\bfv\in \scrT_1$ follows from \lref{lem:TDelta},  given that $\dim(\caT_\bfv)_\Delta=\dim(\caT_\bfv)_\tDelta=2$. In addition, $\bfv$ is a defect vector iff $\bfv\cdot\bfv=0$. If $\bfv\cdot \bfv= 0$, then $\caT_\bfv$ is a stochastic Lagrangian subspace of CSS type [cf. \eref{eq:Tcss}]. If $\bfv\cdot \bfv\neq 0$, then  we can define a stochastic isometry $O_\bfv\in O_3^\odd(d)$ via the action
\begin{align}
O_\bfv (\bfx-a\bfv)=\bfx+a\bfv,\quad   \bfx\in \bfv^\perp,  a\in \bbF_d,
\end{align}
which satisfies $\caT_{O_\bfv}=\caT_\bfv$.
Note that $\bfv$ is an eigenvector of $O_\bfv$ with eigenvalue $-1$, while any nonzero vector in $\bfv^\perp$ is an eigenvector with eigenvalue 1. The following lemma is proved in \aref{app:lem:tauTvProof}

\begin{lemma}\label{lem:Tv}
	Suppose	$d$ is an odd prime and $\bfv=(v_1, v_2, v_3)^\top\in \mathbf{1}_3^\perp$. 
	Then $3v_1v_2v_3=v_1^3+v_2^3+v_3^3$. In addition, the following three conditions are equivalent,
	\begin{align}\label{eq:TvdefEqui}
	\caT_\bfv\in \scrT_\defe, \quad 
	\caT_\bfv\in \scrT_\defe\cap\scrT_1,\quad  \bfv\cdot\bfv=0.  
	\end{align}
	The following three conditions are also equivalent,
	\begin{align}\label{eq:TvsymEqui}
	\caT_\bfv\in \scrT_\sym, \;\; \caT_\bfv\in \scrT_\sym\cap\scrT_\odd, \;\;  v_1v_2v_3=0.
	\end{align}
\end{lemma}

Let $\tau_{12}$ be the transposition defined in \eref{eq:zetatau}, then $\tau_{12}\caT_\bfv$ is also a stochastic Lagrangian subspace. The following lemma 
is a simple corollary of \eref{eq:TvDeltaDefect} and \lref{lem:SemigroupscrT}. See \aref{app:lem:tauTvProof}
for more details. 

\begin{lemma}\label{lem:tauTv}
	Suppose  $\bfv\in \mathbf{1}_3^\perp$ is a nonzero vector. Then $\caT_\bfv\in \scrT_1$ and $\tau_{12} \caT_\bfv\in \scrT_0$. Let $\bfu\in \mathbf{1}_3^\perp$ be another nonzero vector; then  $\caT_\bfu=\caT_\bfv$ iff $\bfu$ and $\bfv$ are proportional to each other.
\end{lemma}

Thanks to \lref{lem:tauTv}, $\caT_\bfv$ is completely determined by  $\spa(\bfv)$. In addition, $\mathbf{1}_3^\perp$ contains $d+1$ one-dimensional subspaces, which can be expressed as $\spa(\bfv_y)$  for $y=0,1,\ldots, d$ with
\begin{align}\label{eq:vy}
\bfv_y&:=(1,y,-1-y)^\top, \quad y=0,1,\ldots, d-1;\quad  \bfv_d:=(0,1,-1)^\top, 
\end{align}
where $-1-y$ is determined modulo $d$.
By virtue of these subspaces we can construct
$d+1$ stochastic Lagrangian subspaces of the form $\caT_{\bfv_y}$ and $d+1$ stochastic Lagrangian subspaces of the form $\tau_{12}\caT_{\bfv_y}$, which exhaust all elements in
$\scrT_1$ and $\scrT_0$, given that 
$|\scrT_1|=|\scrT_0|=d+1$. In this way, we can construct all stochastic Lagrangian subspaces in $\Sigma(d)$ using the alternative approach proposed above. To be specific, we have
\begin{equation}\label{eq:TvSigma}
\scrT_1=\{\caT_{\bfv_y}\}_{y=0}^d,\quad 
\scrT_0=\{\tau_{12}\caT_{\bfv_y}\}_{y=0}^d,\quad 
\Sigma(d)=\{\caT_{\bfv_y}\}_{y=0}^d\sqcup\{\tau_{12}\caT_{\bfv_y}\}_{y=0}^d,
\end{equation}
where the symbol "$\sqcup$" denotes the union of disjoint sets. Note that the product of the three entries of $\bfv_y$ is equal to 0 iff $y=0$, $y=d-1$, or $y=d$. In conjunction with \lref{lem:Tv} we can deduce that
\begin{align}
\scrT_1\cap\scrT_\ns=\{\caT_{\bfv_y}\}_{y=1}^{d-2},\quad  \scrT_\ns=\{\caT_{\bfv_y}\}_{y=1}^{d-2}\sqcup\{\tau_{12}\caT_{\bfv_y}\}_{y=1}^{d-2}. \label{eq:TvNS}
\end{align}

\subsection{\label{sec:CharIndex} Indices  and cubic characters of stochastic Lagrangian subspaces}
According to \lref{lem:tauTv} and \eref{eq:TvSigma},
each stochastic Lagrangian subspace $\caT$ in $\scrT_1$ can be expressed as $\caT=\caT_\bfv $ for some $\bfv$ in $\mathbf{1}_3^\perp$. The vector $\bfv$, which is unique up to a nonzero scalar multiple, is called a \emph{characteristic vector} of $\caT$.  If $\caT\in\scrT_0$, then a vector in $\bbF_d^3$ is called a characteristic vector of $\caT$ if it is a characteristic vector of $\tau_{12}\caT$, where $\tau_{12}$ is the transposition defined in \eref{eq:zetatau}; note that  $\tau_{12}\caT\in \scrT_1$ whenever  $\caT\in\scrT_0$. 
If $\caT=\caT_O$ for $O\in O_3(d)$, then any characteristic vector of $\caT$ is also called a characteristic vector of $O$. When $O\in O_3^\odd$, 
any characteristic vector of $O$ is an eigenvector with eigenvalue $-1$, and vice versa; when $O\in O_3^\even$, however, there is no simple connection between a characteristic vector  and an eigenvector.

Let $\nu$ be a given primitive element of $\bbF_d$; see \tref{tab:nu} for a specific choice. 
The \emph{index} of $a\in \bbF_d$, denoted by $\ind(a)$, is defined as the smallest nonnegative  integer $i$ such that
$a/\nu^i$ is a cubic residue, that is, a cube of another element in $\bbF_d$. By definition $\ind(a)$ is equal to one of the three integers $0, 1, 2$. In addition, $\ind(a)=0$ iff $a$ is a cubic residue. The index may depend on the choice of the primitive element, but the condition $\ind(a)=0$ is independent of this choice. 

\begin{table*}[bt]
	\renewcommand{\arraystretch}{1.8}
	\caption{\label{tab:nu} Specific choice of the primitive element $\nu$ of the finite field $\bbF_d$ with $d<100$.}	
	\begin{math}
	\begin{array}{c|cccccccccccccc}
	\hline\hline
	d & \quad 3\quad & \quad 5 \quad & \quad 7  \quad & \quad 11  \quad & \quad 13  \quad & \quad 17  \quad & \quad 19  \quad & \quad 23  \quad & \quad 29  \quad & \quad 31  \quad & \quad 37  \quad & \quad 41  \quad \\
	\hline		
	\nu & \quad 2 \quad & \quad 2 \quad & \quad 3  \quad & \quad 2  \quad & \quad 2  \quad & \quad 3  \quad & \quad 2  \quad & \quad 5  \quad & \quad 2  \quad & \quad 3  \quad & \quad 2  \quad & \quad 6  \quad \\
	\hline\hline
	d & \quad 43\quad & \quad 47 \quad & \quad 53  \quad & \quad 59  \quad & \quad 61  \quad & \quad 67  \quad & \quad 71  \quad & \quad 73  \quad & \quad 79  \quad & \quad 83  \quad & \quad 89  \quad & \quad 97  \quad \\
	\hline		
	\nu & \quad 3 \quad & \quad 5 \quad & \quad 2  \quad & \quad 2  \quad & \quad 2  \quad & \quad 2  \quad & \quad 7  \quad & \quad 5  \quad & \quad 3  \quad & \quad 2  \quad & \quad 3  \quad & \quad 5  \quad \\
	\hline\hline
	\end{array}	
	\end{math}	
\end{table*}

Now suppose $d\geq 5$, $\caT\in \scrT_\ns$, and $\bfv=(v_1,v_2,v_3)^\top$ is a characteristic vector of $\caT$. 
Then $v_1^3+v_2^3+v_3^3\neq 0$ by \lref{lem:Tv} and is  thus  a power of $\nu$. The \emph{index} of $\caT$, denoted by $\ind(\caT)$, is defined as the index of $v_1^3+v_2^3+v_3^3$,
\begin{align}\label{eq:TindexDef}
\ind(\caT):=\ind(v_1^3+v_2^3+v_3^3). 
\end{align}
This definition is independent of the choice of the characteristic vector~$\bfv$. The index of $O\in O_3(d)\setminus S_3$ is defined as the index of $\caT_O$. By definition we have $\ind(\tau_{12} \caT)=\ind(\caT)$ and $\ind(\tau_{12} O)=\ind(O)$.
When  $d=2\mmod 3$,  any element in $\bbF_d$ is a cubic residue,  so we have $\ind(\caT)=0$ for all $\caT\in \scrT_\ns$.

When $d=1\mmod 3$, the situation is more complicated and interesting.  Let $\eta_3$ be a given cubic character of $\bbF_d$ (see \aref{app:MultiChar} for a brief introduction). The cubic character of $\caT$ is defined as the cubic character of $v_1^3+v_2^3+v_3^3$, that is,
\begin{align}\label{eq:CharIndex}
\eta_3(\caT):=\eta_3(v_1^3+v_2^3+v_3^3)=\eta_3\bigl(\nu^{\ind(\caT)}\bigr),
\end{align}
where the second equality follows from the definition of $\ind(\caT)$ and the fact that $\eta_3(a)=1$ whenever $a$ is a cubic residue. Let $\caT'\in \Sigma(d)$ be another stochastic Lagrangian subspace; then $\eta_3(\caT')=\eta_3(\caT)$ iff $\ind(\caT')=\ind(\caT)$. As we shall see in \sref{sec:MagicOrbitd1}, the existence of stochastic Lagrangian subspaces with different cubic characters has a profound implication for understanding the third moments of magic orbits.  The following lemma proved in \aref{app:TindexProof} clarifies the symmetry of $\ind(\caT)$ and $\eta_3(\caT)$.
\begin{lemma}\label{lem:Tindex}
	Suppose $d\geq 7$ is an odd prime satisfying $d=1\mmod 3$; then 
	\begin{gather}
	\ind(O\caT)=\ind(\caT O)=\ind(\caT)\quad \forall \caT\in \scrT_\ns,  O\in S_3;   \quad 
	\ind(\caT)=\ind(3)\quad \forall \caT\in \scrT_\defe; \label{eq:Tindex}
	\end{gather}
	and the same result still holds if $\ind$ is replaced by $\eta_3$. 	
\end{lemma}

Next, we determine the number of stochastic Lagrangian subspaces that have a given cubic character or index. 
Define
\begin{align}\label{eq:muj} 
\mu_j=\mu_j(d)&:=\left|\left\{\caT\in \scrT_\ns\,:\, \eta_3(\caT)= \eta_3(3\nu^j)\right\}\right|=\left|\left\{\caT\in \scrT_\ns\,:\, \ind(\caT)= \ind(3\nu^j)\right\}\right|, \quad j=0,1,2.  
\end{align}
By virtue of \eref{eq:TvNS} we can deduce the following result (see \lref{lem:muj} in \aref{app:SigmadAux}),
\begin{align}\label{eq:mujFormula}
\mu_j=	\frac{2d(d-2)+4\Re\bigl[\eta_3^2(\nu^j)G^3(\eta_3)\bigr] }{3d},\quad \mu_0+\mu_1+\mu_2=|\scrT_\ns|=2(d-2), 
\end{align}
where $G(\eta_3)$ is a  Gauss sum reviewed in \aref{app:GaussJacobi}, which satisfies $|G(\eta_3)|=\sqrt{d}$. When $d$ is large, we have $\mu_j=2d/3+\caO(\sqrt{d}\lsp)$ and $\mu_0\approx\mu_1\approx\mu_2$. These observations will be instructive to understanding the third moments of magic orbits.

\section{\label{sec:Commutant}Commutant of the third Clifford  tensor power}
In this section we first clarify the basic properties of the operators  $r(\caT)$ and $R(\caT)=r(\caT)^{\otimes n}$ for $\caT\in \Sigma(d)$. Then we show that $\{R(\caT)\}_{\caT\in \Sigma(d)}$ spans the commutant of the third Clifford  tensor power and construct a dual operator frame. Next, we determine the Schatten norms of $R(\caT)$ and related operators. Finally, we introduce shadow  maps based on $R(\caT)$ and clarify their properties, which are instrumental to studying the shadow norms of Clifford orbits. 

\subsection{\label{app:rRTbasic}Basic properties of $r(\caT)$ and $R(\caT)$}
When $\caT\in \scrT_\iso$, both $r(\caT)$ and $R(\caT)$ are unitary operators, and their properties are relatively simple. When  $\caT\in \scrT_\defe$, which is relevant only if $d\neq 2\mmod 3$, the situation is more complicated. According to \lref{lem:defectT}, when  $d=3$, we have $\scrT_\defe=\{\tcaT_0, \tcaT_1\}$  and
\begin{equation}\label{eq:rT01}
\begin{aligned}
\!\!r(\tcaT_0) = &\ 3 \bigl(|\tcaN,[000]\> \<\tcaN,[000]|+ |\tcaN,[120]\> \<\tcaN,[120]| + |\tcaN,[210]\> \<\tcaN,[210]| \bigr)=3P_{\mathrm{CSS}(\tcaN)},  \\
\!\! r(\tcaT_1) = &\ 3 \bigl(|\tcaN,[000]\> \<\tcaN,[000]|+ |\tcaN,[120]\> \<\tcaN,[210]| + |\tcaN,[210]\> \<\tcaN,[120]| \bigr),   
\end{aligned}
\end{equation}
where the defect subspace $\tcaN$ is defined in \eref{eq:DefectSpaced3}, while the projector $P_{\mathrm{CSS}(\tcaN)}$ and corresponding coset states are defined in \eref{eq:CSSprojCosetS}. Note that both $r(\tcaT_0)$ and $r(\tcaT_1)$ are Hermitian. When $d = 1 \mmod 3$,  we have $\scrT_\defe=\{\tcaT_{ij}\}_{i,j=0,1}$ and
\begin{equation}\label{eq:lambda4}
\!\! r(\tcaT_{ij})=d \sum_{x\in \bbF_d} |\tcaN_i,[x\mathbf{1}_3]\> \<\tcaN_j, [x\mathbf{1}_3]|,\\
\end{equation}
where the defect subspaces $\tcaN_0, \tcaN_1$ are defined in \eref{eq:TwoDefectSpaces}. According to \eref{eq:rTCSS} we have
\begin{equation}
r(\tcaT_{00})= d P_{\mathrm{CSS}(\tcaN_0)}, \quad r(\tcaT_{11})= d P_{\mathrm{CSS}(\tcaN_1)}.
\end{equation}

Next, we clarify the basic properties of the operators $r(\caT)$ and $R(\caT)$ for $\caT\in \Sigma(d)$. 
\Lsref{lem:SemigroupR}-\ref{lem:RTpt} below are proved in \aref{app:rRTbasicProofs}. 
\begin{lemma}\label{lem:SemigroupR}
	Suppose $d$ is an odd prime,  $n\in \bbN$,  and $D=d^n$. If $d=3$, then
	\begin{gather}
	R(\tcaT_i)^\dag =R(\tcaT_i), \quad 
	R(\tcaT_i)R(\tcaT_j)=3^{n}R(\tcaT_{i+j}),  \quad i,j=0,1, \label{eq:RTiTjd3}
	\end{gather}	
	where the addition $i+j$ is modulo $2$. If $d = 1 \mmod  3$, then
	\begin{gather}
	R(\tcaT_{ij})^\dag=R(\tcaT_{ji}),\quad	R(\tcaT_{ij})R(\tcaT_{kl})=D^{1-|j-k|}R(\tcaT_{il}), \quad i,j,k,l=0,1. \label{eq:RTiTjd}
	\end{gather} 	
\end{lemma}

\begin{lemma}\label{lem:TDeltaSum}
	Suppose  $d$ is an odd prime. Then 
	\begin{gather}
	\sum_{\caT\in \Sigma(d)} r(\caT_\Delta)=2\sum_{\caT\in \scrT_1} r(\caT_\Delta)=
	2\sum_{\caT\in \scrT_0} r(\caT_\Delta)=2\left(\bbI+d\sum_{x\in \bbF_d} |xxx\>\<xxx|\right). \label{eq:TDeltaSum}
	\end{gather}
\end{lemma}

The following proposition  determines the traces of $R(\caT)$ for $\caT\in \Sigma(d)$ and Hilbert-Schmidt inner products between these operators. It is a simple 	 corollary of \eqsref{eq:RTtr}{eq:RT1T2tr} as well as  \lsref{lem:TDelta} and \ref{lem:TDelta2}. 

\begin{proposition}\label{pro:RTT1T2tr}
	Suppose $d$ is an odd   prime, $n\in \bbN$,  $D=d^n$, and $\caT\in \Sigma(d)$.   Then	
	\begin{align}\label{eq:RscrTtr}
	\tr[R(\caT)]&=
	\begin{cases}
	D^3 &\mbox{if}\;\;  \caT=\Delta , \\
	D &\mbox{if}\;\; \caT\neq \Delta, \caT\in \scrT_0, \\
	D^2 &\mbox{if}\;\;  \caT\in \scrT_1.
	\end{cases}
	\end{align}	
	If $\caT_1\in \scrT_i$ and  $\caT_2\in \scrT_j$ with $i,j=0,1$, then 
	\begin{align} \label{eq:RscrT0T1tr}
	\tr[R(\caT_1)^\dag R(\caT_2)]&=
	\begin{cases}
	D^3 &\mbox{if}\;\;  \caT_1=\caT_2, \\
	D &\mbox{if}\;\;  \caT_1\neq \caT_2, i=j, \\
	D^2 &\mbox{if}\;\;  i\neq j. 
	\end{cases}
	\end{align}
\end{proposition}

Next,  we regard $R(\caT)$ as an operator on $\bigl(\caH_d^{\otimes n}\bigr)^{\otimes 3}$ and label the three copies of $\caH_d^{\otimes n}$ by $A, B, C$, respectively. The partial traces of $R(\caT)$ with respect to one or two  subsystems are clarified in the  following lemma.
\begin{lemma}\label{lem:RTpt}
	Suppose $d$ is an odd   prime, $n\in \bbN$,  $D=d^n$, and $\caT\in \Sigma(d)$.  Then 
	\begin{align}
	\tr_A R(\caT)&=\begin{cases}
	D \bbI & \mbox{if}\;\; \caT=\Delta, \\ 	
	\mathrm{SWAP} & \mbox{if}\;\; \caT\in \scrT_0\setminus \{\Delta\},\\
	\;\;\; \bbI & \mbox{if}\;\; \caT\in \scrT_1\cap \scrT_\ns,
	\end{cases}\label{eq:RTpt} \\
	\tr_{AB} R(\caT)&=\begin{cases}
	D^2 \bbI & \mbox{if}\;\; \caT=\Delta, \\
	\;\;\;\bbI  & \mbox{if}\;\; \caT\in \scrT_0\setminus \{\Delta\},\\
	\; D\bbI & \mbox{if}\;\; \caT\in \scrT_1. 
	\end{cases} \label{eq:RTpt2}
	\end{align}
	\Eref{eq:RTpt} still holds  if the partial trace is taken with respect to party $B$ or $C$; \eref{eq:RTpt2} still holds  if the partial trace is taken with respect to any other two parties.
\end{lemma}
According to \lref{lem:RTpt}, if $\caT$ does not correspond to a transposition, then the result of the partial trace is independent of the party or parties with respect to which the partial trace is taken.

\subsection{\label{sec:Spanning}A spanning set and dual frame}
By virtue of \pref{pro:RTT1T2tr}, here we show that the set $\{R(\caT)\}_{\caT \in \Sigma(d)}$ spans the commutant of the third Clifford tensor power. When $n\geq 2$, this conclusion  was already established by GNW \cite{GrosNW21}. Our analysis shows that the restriction $n\geq 2$ is not necessary. The following theorem is proved in \aref{app:SpanProof}.

\begin{theorem}\label{thm:spanR}
	Suppose $d$ is an odd prime and $n\in \bbN$. Then $\{R(\caT)\}_{\caT \in \Sigma(d)}$ spans the commutant of $\Cl(n,d)^{\totimes 3}$. If in addition  $n \ge 2$, then $\{R(\caT)\}_{\caT \in \Sigma(d)}$ forms a basis of the commutant; in the case $n=1$, any subset of $\{R(\caT)\}_{\caT \in \Sigma(d)}$ with cardinality $2d+1$ forms a basis of the commutant. 
\end{theorem}

When  $n\geq 2$, $\{R(\caT)\}_{\caT \in \Sigma(d)}$  forms a basis of the commutant of $\Cl(n,d)^{\totimes 3}$ thanks to \thref{thm:spanR}. By virtue of \pref{pro:RTT1T2tr}, we can further determine the dual basis. Define
\begin{align}\label{eq:DualBasis}
\tR(\caT):=\frac{1}{D^3-D}\left[R(\caT)+\frac{d}{D^2-d^2}R(\scrT_j) -\frac{D}{D^2-d^2}R(\scrT_{\bar{j}})\right]\quad \forall \caT\in \scrT_j,
\end{align}
where $\bar{j}=1-j$. 
Then $\{\tR(\caT)\}_{\caT \in \Sigma(d)}$ is the dual basis of $\{R(\caT)\}_{\caT \in \Sigma(d)}$ and satisfies the duality relation,
\begin{align}\label{eq:DualOrtho}
\tr[R(\caT_1)^\dag \tR(\caT_2)]=\delta_{\caT_1, \caT_2}\quad \forall \caT_1, \caT_2\in \Sigma(d).
\end{align}
Thanks to \pref{pro:RTT1T2tr} and \thref{thm:spanR}, 
as $n$ gets large, $\{R(\caT)/D^{3/2}\}_{\caT\in \Sigma(d)}$ becomes a better and better approximation to an orthonormal basis in the commutant of the third Clifford tensor power; consequently, we have $\tR(\caT)\approx R(\caT)/D^3$.

When $n=1$, the operators in $\{R(\caT)\}_{\caT \in \Sigma(d)}$ are not linearly independent, although they still span the commutant of $\Cl(n,d)^{\totimes 3}$. In this case,  $D=d$ and $\tR(\scrT_0)=\tR(\scrT_1)$, so $\tR(\caT)$ in  \eref{eq:DualBasis} is ill defined. Nevertheless, we can still construct a dual frame as follows,
\begin{align}\label{eq:DualBasisn1}
\tR(\caT):=\frac{1}{d^3-d}\left[R(\caT)-\frac{1}{2d}R(\scrT_0) \right], 
\end{align}
which can be understood as the limit of $\tR(\caT)$ in \eref{eq:DualBasis} as $D$ approaches $d$.  

In both cases, by virtue of \Pref{pro:RTT1T2tr} it is straightforward to verify that
\begin{equation}
\sum_{\caT'\in \Sigma(d)}\tr[\tR(\caT')^\dag R(\caT)]R(\caT')=\sum_{\caT'\in \Sigma(d)}\tr[R(\caT')^\dag R(\caT)]\tR(\caT')=R(\caT)\quad \forall \caT\in \Sigma(d),
\end{equation}
so $\{\tR(\caT)\}_{\caT\in \Sigma(d)}$ indeed forms a dual frame of $\{R(\caT)\}_{\caT\in \Sigma(d)}$. Consequently, any operator $\Ob$ in the commutant of  $\Cl(n,d)^{\totimes 3}$ can be expressed as follows,
\begin{align}
\Ob=\sum_{\caT\in \Sigma(d)}\tr[\tR(\caT)^\dag \Ob]R(\caT)=\sum_{\caT\in \Sigma(d)}\tr[R(\caT)^\dag \Ob]\tR(\caT). 
\end{align}

\subsection{Schatten norms of  $R(\caT)$ and related operators}

Here we clarify the Schatten norms of $R(\caT)$  for $\caT\in \Sigma(d)$ and related operators.  If $\caT\in \scrT_\iso$, that is, $\caT=\caT_O$ for $O\in O_3(d)$, then $R(\caT)$ is a permutation matrix, and all singular values of $R(\caT)$ are equal to~1. If $\caT\in \scrT_\defe$, then $|\caT_\LD|=|\caT_\RD|=d$; in addition, either $\caT$ or $\caT\tau$  is a stochastic Lagrangian subspace of CSS type according to \lsref{lem:defectT} and~\ref{lem:SemigroupT}, where $\tau$ is a transposition in $S_3$. Consequently, either $R(\caT)/D$ or $R(\caT)R(\tau)/D$ is a stabilizer projector of rank $D$ by \eref{eq:rTCSS}. Therefore,  $R(\caT)$ has $D$ nonzero singular values, all of which are equal to $D$, and
the Schatten $\ell$-norm of $R(\caT)$ for $\ell\geq 1$ (including $\ell=\infty$) reads
\begin{align}
\!\!\| R(\caT) \|_\ell=\begin{cases} 
D^{3/\ell} &\! \caT\in \scrT_\iso,\\
D^{(l+1)/l} &\! \caT\in \scrT_\defe. 
\end{cases}
\end{align}

\begin{table*}
	\renewcommand{\arraystretch}{1.5}
	\caption{\label{tab:RisoRdefe}Basic properties of the operators $R(\scrT_\iso)$, $ R(\scrT_\defe)$, and $R(\Sigma(d))$. }	
	\begin{math}
	\begin{array}{c|ccc}
	\hline\hline
	& d=3 &d=1 \mmod 3& d=2\mmod 3 \\
	\hline
	\tr R(\scrT_\iso) &	D(D+1)(D+d-1) & D(D+1)(D+d-2) &	D(D+1)(D+d)\\[1ex]	
	\|R(\scrT_\iso)\| & 2d  & 2d-2 & 2d+2\\	 		
	\rk R(\Sigma(d))=\rk R(\scrT_\iso) &\frac{D(D+1)(D+d-1)}{2d}	&\frac{D(D+1)(D+d-2)}{2d-2} &\frac{D(D+1)(D+d)}{2d+2}	 \\
	\tr R(\scrT_\defe)&D^2+D &2D^2+2D &0\\	
	\|R(\scrT_\defe)\|  & 2D  & 2D+2  & 0\\		
	\rk  R(\scrT_\defe) & \frac{D+1}{2} & D &0\\
	\tr R(\Sigma(d)) & D(D+1)(D+d) & D(D+1)(D+d) & D(D+1)(D+d)\\
	\|R(\Sigma(d))\|  & 2D+2d &2D+2d&2d+2\\
	\hline\hline
	\end{array}	
	\end{math}	
\end{table*}

Next, we clarify the basic properties of $R(\scrT_\iso)$,  $R(\scrT_\defe)$, and $R(\Sigma(d))=R(\scrT_\iso)+R(\scrT_\defe)$ defined according to \eref{eq:RscrT}. 
Note that $\scrT_\defe$ is empty when $d=2\mmod 3$, in which case $R(\scrT_\defe)=0$. The following lemma proved in \aref{app:RTisodefeProof} will be very useful to understanding the third moments of Clifford orbits, including the orbit of stabilizer states in particular.

\begin{lemma}\label{lem:RTisodefe}
	Suppose $d$ is an odd prime and $n\in \bbN$.  Then $R(\scrT_\sym)=6P_\sym$,  $R(\scrT_\iso)$, and $R(\scrT_\defe)$ are proportional to commuting projectors, and $R(\Sigma(d))$ is a positive operator. In addition,
	\begin{align}\label{eq:RTisodefeSupp}
	\supp R(\scrT_\defe)\leq \supp R(\scrT_\iso)=\supp R(\Sigma(d))\leq\Sym_3\bigl(\caH_d^{\otimes n}\bigr).
	\end{align}	
	The basic properties of $R(\scrT_\iso)$,  $R(\scrT_\defe)$, and   $R(\Sigma(d))$ are  summarized in \tref{tab:RisoRdefe}.
\end{lemma}

\subsection{\label{sec:ShadowMapR}Shadow maps associated with $R(\caT)$}

As in \lref{lem:RTpt}, here we regard $R(\caT)$ for $\caT\in \Sigma(d)$ as  operators on $\bigl(\caH_d^{\otimes n}\bigr)^{\otimes 3}$ and label the three copies of $\caH_d^{\otimes n}$ by $A, B, C$, respectively.
In analogy  to \eref{eq:ShadowMap}, the shadow maps $\caR_\caT(\cdot)$ tied to $R(\caT)$ are defined as
\begin{gather}
\caR_\caT(\Ob)  :=\tr_{BC}\bigl[R(\caT)\bigl(\bbI\otimes \Ob\otimes \Ob^\dag\bigr)\bigr],\quad \Ob\in \caL\bigl(\caH_d^{\otimes n}\bigr). \label{eq:ShadowMapR}
\end{gather}
As we shall see later, properties of $\caR_\caT(\Ob)$, including its operator norm in particular, will be instrumental to studying the shadow norms of Clifford orbits.

When $\Ob$ is a stabilizer projector on $\caH_d^{\otimes n}$, we can derive an  explicit formula for $\caR_\caT(\Ob)$  as shown in the following lemma and proved in \aref{app:RTOproof}. 
\begin{lemma}\label{lem:RTObStabProj}
	Suppose $d$ is an odd prime and $\Ob$ is a rank-$K$ stabilizer projector on  $\caH_d^{\otimes n}$. Then 
	\begin{align}
	\caR_\caT(\Ob)&=\begin{cases}
	K^2\bbI &\mbox{if}\quad  \caT=\Delta,\\
	K \bbI &\mbox{if}\quad \caT=\caT_{\tau_{23}}, \\
	\Ob &\mbox{if}\quad \caT\in \scrT_0\setminus \{\Delta\},\\
	K\Ob &	\mbox{if}\quad \caT\in \scrT_1\setminus\{\caT_{\tau_{23}}\}. 
	\end{cases} \label{eq:RTOstabProj}
	\end{align}
\end{lemma}

Next, we turn to the general situation in which $\Ob$ is an arbitrary linear operator  on $\caH_d^{\otimes n}$. If $\caT\in \scrT_\sym$, then $\caT=\caT_O$ with $O\in S_3$, and the operator norm of $\caR_{\caT}(\Ob)$ can be determined by straightforward calculation,
\begin{align}\label{eq:RTOS3norm}
\|\caR_{\caT}(\Ob)\|=
\begin{cases}
|\tr(\Ob)|^2 &\mbox{if} \; O=\mathds{1},\\
|\tr(\Ob)|\|\Ob\| &\mbox{if} \; O=\tau_{12}, \tau_{13},\\
\|\Ob\|_2^2  & \mbox{if}\; O=\tau_{23},\\
\|\Ob\|^2& \mbox{if}\; O=\zeta, \zeta^2.
\end{cases}
\end{align}
In addition, $\sum_{\caT\in \scrT_\sym}\caR_\caT(\Ob)$ is determined by \eref{eq:PsymOb} given that $\sum_{\caT\in \scrT_\sym}R(\caT)=6P_\sym$.

If $\caT\in \scrT_\ns$, then the situation is more complicated. Nevertheless, we can derive  universal upper bounds for the operator norms of  $\caR_\caT(\Ob)$, which are crucial to studying the shadow norms of Clifford orbits. \Lsref{lem:RTO} and \ref{lem:RTOdiag} below are proved in  \aref{app:RTOproof}. 
\begin{lemma}\label{lem:RTO}
	Suppose $d$ is an odd prime and $\Ob\in \caL\bigl(\caH_d^{\otimes n}\bigr)$. Then 
	\begin{align}\label{eq:RTOUB1}
	\|\caR_\caT(\Ob)\|\leq \|\Ob\|_2^2\quad \forall  \caT\in \scrT_\ns. 
	\end{align}
	If in addition $\Ob$ is traceless, then 
	\begin{align}\label{eq:RTOUB2}
	\|\caR_\caT(\Ob)\|\leq \|\Ob\|_2^2\quad \forall  \caT\in \Sigma(d). 
	\end{align}	
\end{lemma}

\begin{lemma}\label{lem:RTOdiag}
	Suppose $d$ is an odd prime
	and $\Ob\in \caL\bigl(\caH_d^{\otimes n}\bigr)$ is diagonal in a stabilizer basis. Then 
	\begin{align}
	\|\caR_\caT(\Ob)\|= \|\Ob\|^2\quad \forall \caT\in \scrT_0\setminus\{\Delta\}. 	
	\label{eq:RTOeven}
	\end{align}
\end{lemma}

\section{\label{sec:3rdmomentStab}Third moment of  stabilizer states}
In this section we first clarify the basic properties of the third normalized moment operator  $\bQ(n,d,3)$ of the ensemble $\Stab(n,d)$ of stabilizer states, including its spectrum, rank, and operator norm, assuming that $d$ is a prime. Then we  determine the shadow norms of stabilizer projectors with respect to $\Stab(n,d)$ and the shadow norm  of $\Stab(n,d)$ itself.


\subsection{Basic properties}

According to \eref{eq:MomentOtth}, the third moment operator $Q(n,d,3)$ and normalized moment operator $\bQ(n,d,3)$ of $\Stab(n,d)$ read
\begin{align}
Q(n,d,3)&=\frac{\sum_{\caT\in \Sigma(d)} R(\caT)}{D(D+1)(D+d)}, \quad 
\bQ(n,d,3)=\frac{D+2}{6(D+d)}\sum_{\caT\in \Sigma(d)} R(\caT)=\frac{(D+2)[R(\scrT_\iso) +R(\scrT_\defe)]}{6(D+d)},   \label{eq:bQnd3}
\end{align}
where the last equality holds because $\Sigma(d)=\scrT_\iso\sqcup\scrT_\defe$. Note that $\scrT_\defe$ is empty when $d=2\mmod 3$. The above equation holds even if $d=2$, in which case $\Sigma(d)=\scrT_\iso=\scrT_\sym$ and $\bQ(n,d,3)$ coincides with the projector $P_\sym$ onto the tripartite symmetric subspace $\Sym_3\bigl(\caH_d^{\otimes n}\bigr)$; these conclusions are consistent with the fact that  $\Stab(n,d)$ is a 3-design when $d=2$ \cite{KuenG13,Zhu17MC,Webb16}. 
In general, $\bQ(n,d,3)$ and $Q(n,d,3)$ are supported in  $\Sym_3\bigl(\caH_d^{\otimes n}\bigr)$
and they are proportional to each other by definition, so their properties are tied to each other. Here we shall focus on $\bQ(n,d,3)$ for simplicity.  The spectrum of $\bQ(n,d,3)$ within $\Sym_3\bigl(\caH_d^{\otimes n}\bigr)$ is clarified in \pref{pro:Qspectra}  and \tref{tab:bQnd3}, which follow from \eref{eq:bQnd3}  and \lref{lem:RTisodefe}.

\begin{table*}
	\renewcommand{\arraystretch}{1.8}
	\caption{\label{tab:bQnd3} Eigenvalues of $\bQ(n,d,3)$ within $\Sym_3\bigl(\caH_d^{\otimes n}\bigr)$, where $d$ is a prime. Three potential eigenvalues are denoted by  $\lambda_1, \lambda_2, \lambda_3$, which are arranged in decreasing order; their multiplicities are denoted by $m_1, m_2, m_3$. The rank of $\bQ(n,d,3)$ and operator norms of $\bQ(n,d,3)$	and $\bQ(n,d,3)-P_\sym$ are also listed for completeness.}	
	\begin{math}
	\begin{array}{c|ccc}
	\hline\hline
	& d=3 &d=1 \mmod 3& d=2\mmod 3 \\
	\hline		
	\lambda_1 & \frac{D+2}{3} & \frac{D+2}{3}&  \frac{(d+1)(D+2)}{3(D+d)}\\	
	m_1  & \frac{D+1}{2} & D &\frac{D(D+1)(D+d)}{2d+2}\\
	\lambda_2 & \frac{D+2}{D+3} &\frac{(d-1)(D+2)}{3(D+d)} &0\\
	m_2 & \frac{D(D+1)(D+2)}{6}-\frac{D+1}{2} & \frac{D(D+1)(D+d-2)}{2d-2}-D &   
	\frac{D(D+1)(D+2)}{6}-\frac{D(D+1)(D+d)}{2d+2}
	\\
	\lambda_3 &/ & 0 &/  \\
	m_3  &/ & \frac{D(D+1)(D+2)}{6}-\frac{D(D+1)(D+d-2)}{2d-2} &/  \\
	\rk(\bQ(n,d,3))	& \frac{D(D+1)(D+2)}{6}& \frac{D(D+1)(D+d-2)}{2d-2} & \frac{D(D+1)(D+d)}{2d+2}\\	
	\|\bQ(n,d,3)\| & \frac{D+2}{3}	& \frac{D+2}{3} & \frac{(d+1)(D+2)}{3(D+d)}\\
	\|\bQ(n,d,3)-P_\sym\| &\frac{D-1}{3}	& \frac{D-1}{3} & \begin{cases} 1 &\mbox{if  } d=5\\
	\frac{(d+1)(D+2)}{3(D+d)}-1 &\mbox{if  } d\geq 11
	\end{cases} \\[1.8ex]
	\hline\hline
	\end{array}	
	\end{math}	
\end{table*}

\begin{proposition}\label{pro:Qspectra}
	Suppose $d$ is a prime, then the spectrum of $\bQ(n,d,3)$ within the symmetric subspace $\Sym_3\bigl(\caH_d^{\otimes n}\bigr)$  is given in \tref{tab:bQnd3}.
\end{proposition}

When $d=2\mmod 3$, the operator $\bQ(n,d,3)$ is proportional to a projector. When $d=3$ or $d=1\mmod 3$, by contrast, $\bQ(n,d,3)$ has two distinct nonzero eigenvalues. 
When  $d=2$ or $d=3$,  $\bQ(n,d,3)$ has full rank in $\Sym_3\bigl(\caH_d^{\otimes n}\bigr)$.  When $d\geq 5$, by contrast, $\bQ(n,d,3)$ is rank deficient; moreover, the ratio
$\rk(\bQ(n,d,3))/\pi_{[3]}$ decreases monotonically with $n$ and approaches the following limit when  $n\to \infty$, 
\begin{equation}
\lim_{n\to \infty} \frac{\rk(\bQ(n,d,3))}{\pi_{[3]}}=
\begin{cases}
\frac{6}{2d+2} \quad & \text{if} \ d= 2\mmod 3,\\
\frac{6}{2d-2} \quad & \text{if} \ d= 1\mmod 3.
\end{cases}
\end{equation}

By virtue of \pref{pro:Qspectra}, it is easy to compute the Schatten $\ell$-norm
of $\bQ(n,d,3)$ for $\ell\geq 1$, in which the operator norm
is presented in \tref{tab:bQnd3}.
In the case $d= 2\mmod 3$, the general formula is also quite simple, 
\begin{align}\label{eq:Qnd3HSnorm}
&\|\bQ(n,d,3)\|_\ell =\frac{D(D+1)(D+2)}{6} \left(\frac{2d+2}{D(D+1)(D+d)}\right)^{1-1/\ell},\quad d= 2\mmod 3.
\end{align}
This formula is still applicable in the case $d=3$ or $d= 1\mmod 3$ if $\ell=1$ or $\ell=2$. Based on this observation, we can determine the third normalized frame potential, 
\begin{align}\label{eq:Phi3Stab}
\bPhi_3(n,d,3):=\bPhi_3(\Stab(n,d))=\frac{6\|\bQ(n,d,3)\|_2^2}{D(D+1)(D+2)}=\frac{(d+1)(D+2)}{3(D+d)},
\end{align}
which agrees with the result derived in \rcite{KuenG13}.

The dependence of  $\bPhi_3(n,d,3)$  on the local dimension $d$ and qudit number $n$ is illustrated in \fref{fig:framepot_stab}. Note  that $\bPhi_3(n,d,3)$ increases linearly with the local dimension $d$, but is almost independent of the qudit number $n$ when $n \geq 5$. The behavior of  $\|\bQ(n,d,3)\|$ illustrated in  \fref{fig:Qnorm_stab} is similar to  $\bPhi_3(n,d,3)$ when $d =2 \mmod 3$. When  $d \neq 2 \mmod 3$, however, $\|\bQ(n,d,3)\|$ grows linearly with the total dimension $D$ and exponentially with $n$, which means $\Stab(n,d)$ deviates exponentially from a 3-design with respect to $\|\bQ(n,d,3)\|$ as $n$ increases.

\begin{figure}[tb]
	\centering 
	\includegraphics[width=0.8\textwidth]{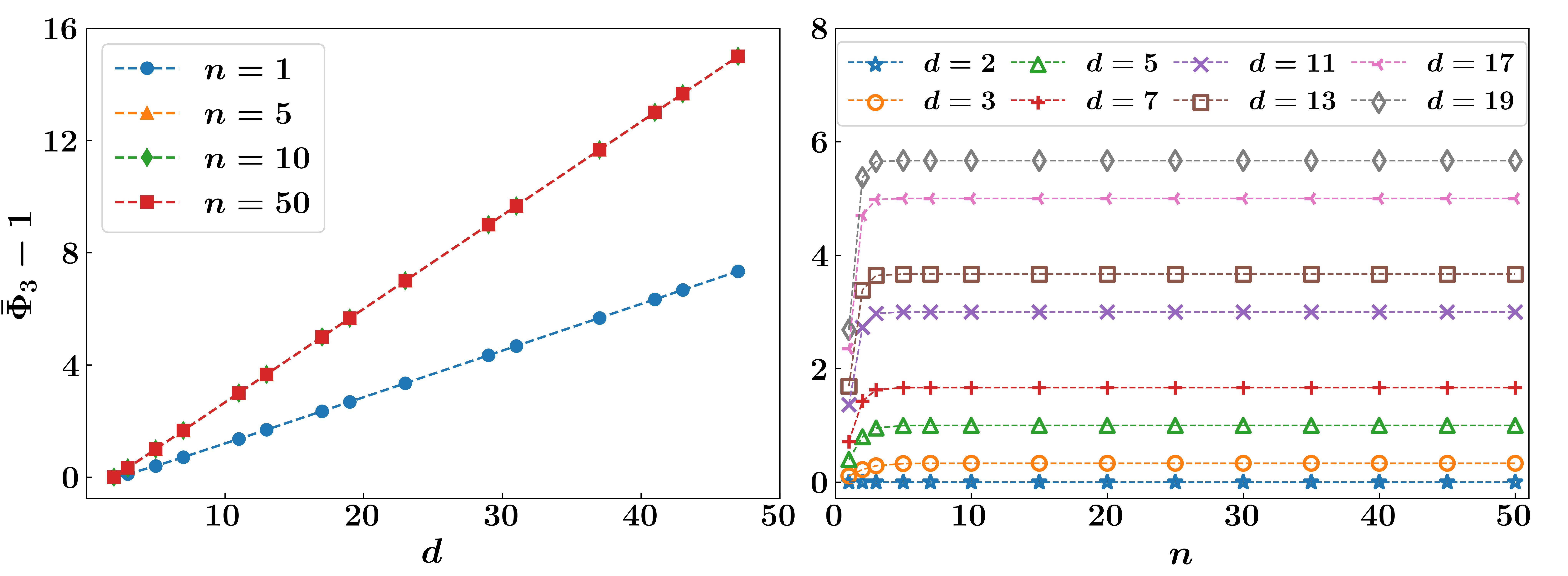}
	\caption{The deviation of the third normalized frame potential $\bPhi_3(n,d,3)$ of the ensemble $\Stab(n,d)$ of stabilizer states  as a function of the local dimension $d$  and qudit number $n$. Note that $d$ is a prime in this figure and all other figures in this paper. The dashed lines/curves are guides for the eye, which is similar for other figures. In the left plot, the results on $n=5,10,50$ almost coincide.} 
	\label{fig:framepot_stab}
\end{figure}

\begin{figure}[bt]
	\centering 
	\includegraphics[width=0.8\textwidth]{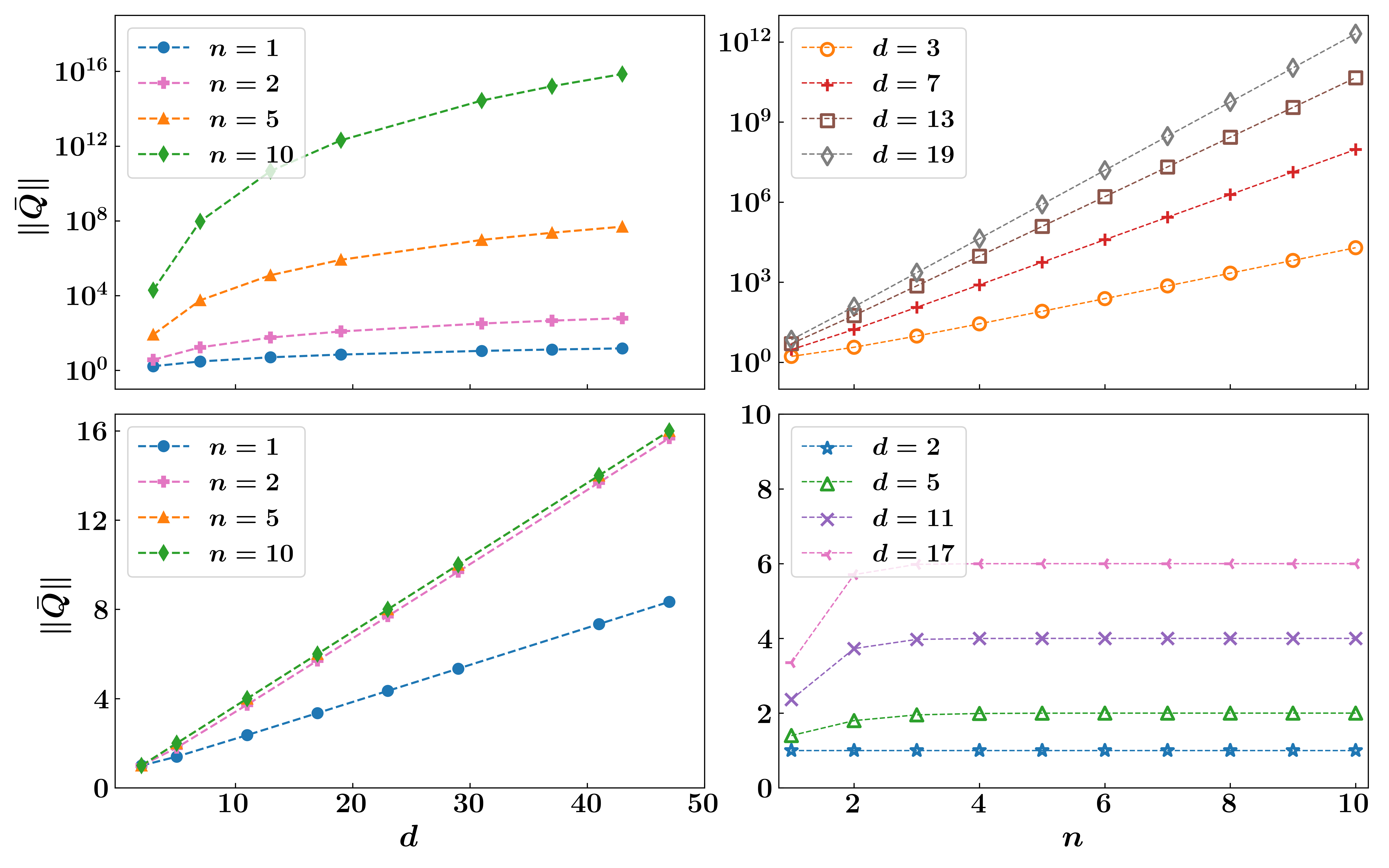}
	\caption{The operator norm of the third normalized moment operator $\bQ(n,d,3)$ of $\Stab(n,d)$  as a function of  $d$ and $n$. Here $d$ is a prime with $d \neq 2 \mmod 3$  in the first row, while it is a prime with 
		$d=2 \mmod 3$  in the second row because the properties of  $\|\bQ(n,d,3)\|$ in the two cases are dramatically different. } 
	\label{fig:Qnorm_stab}
\end{figure}

\subsection{Shadow norm}
Let  $\bcaQ_{n,d}(\cdot)$ be the shadow map associated with $\bQ(n,d,3)$ as defined in \eref{eq:ShadowMap}. Given  $\Ob\in \caL\bigl(\caH_d^{\otimes n}\bigr)$, let  $\|\Ob\|_\Stab$ be the shadow norm of $\Ob$ with respect to   $\Stab(n,d)$.
By virtue of \eqsref{eq:ObShNormDef}{eq:bQnd3} we can deduce that
\begin{align}
\bcaQ_{n,d}(\Ob)&=\frac{D+2}{6(D+d)}\sum_{\caT\in \Sigma(d)}\caR_\caT(\Ob), \quad \|\Ob\|^2_\Stab= \frac{6(D+1)}{D+2}\|\bcaQ_{n,d}(\Ob)\|. 
\label{eq:QndObShNormStab}
\end{align}
Here we are mainly interested in the shadow norms of traceless operators because of their connection with variances in shadow estimation \cite{HuanKP20}. Recall that the traceless part of $\Ob$ reads $\Ob_0=\Ob-\tr(\Ob)\bbI/D$. 
When $\Ob$ is a rank-$K$ stabilizer projector, we can derive  explicit formulas for $6\bcaQ_{n,d}(\Ob)$ and $6\bcaQ_{n,d}(\Ob_0)$ and thereby determine the shadow norm $\|\Ob_0\|_{\Stab}$ as long as $d$ is a prime (including $d=2$). The following lemma is proved in \aref{app:ThirdMomentStab}. 
\begin{lemma}\label{lem:QObStabProj}
	Suppose $d$ is a prime and $\Ob$ is a rank-$K$ stabilizer projector on  $\caH_d^{\otimes n}$. Then 
	\begin{align}
	6\bcaQ_{n,d}(\Ob)&=\frac{(D+2)[(K^2+K)\bbI+d(K+1)\Ob]}{(D+d)}, \label{eq:QObStabProj}\\
	6\bcaQ_{n,d}(\Ob_0)
	&=\frac{K(D+2)[D^2-(dD-D-d)K]\bbI}{D^2(D+d)}+\frac{(D+2)[dD+(dD-2D-2d)K]\Ob}{D(D+d)}.
	\label{eq:QO0stabProj}
	\end{align}
\end{lemma}   
The following theorem is a simple corollary of \eref{eq:QndObShNormStab} and  \lref{lem:QObStabProj}  given that   $\|\Ob_0\|_2^2=K-(K^2/D)$. 
\begin{theorem}\label{thm:StabShNormStab}
	Suppose $d$ is a prime and $\Ob$ is a rank-$K$ stabilizer projector on  $\caH_d^{\otimes n}$ with $1\leq K\leq D/d$. Then 
	\begin{align}
	\|\Ob_0\|^2_\Stab
	&=\frac{(D+1)(D-K)(dDK+dD-DK -dK)}{D^2(D+d)},  \label{eq:StabProjShNorm}\\
	\frac{\|\Ob_0\|^2_\Stab}{\|\Ob_0\|_2^2}&
	=\frac{D+1}{D+d}\left(d-1-\-\frac{d}{D}+\frac{d}{K}\right)\leq d-1-\frac{d}{D}+\frac{d}{K}\leq 2d-1. \label{eq:StabProjShNormRatio}
	\end{align}
\end{theorem}

\begin{figure}[tb]
	\centering 
	\includegraphics[width=0.8\textwidth]{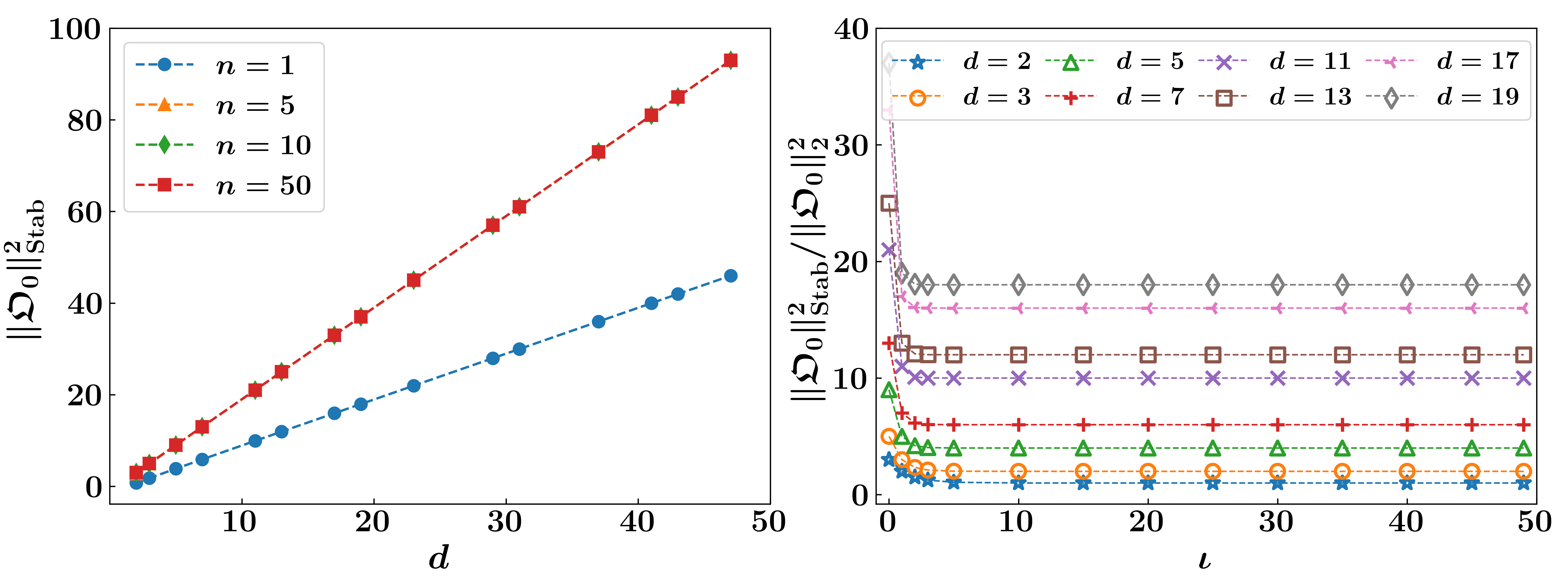}
	\caption{The squared shadow norm of $\Ob_0=\Ob-\tr(\Ob)\bbI/D$  with respect to $\Stab(n,d)$. In the left plot,  $\Ob$ is the projector onto a stabilizer state; in the right plot,   $n=50$ and  $\Ob$ is a  stabilizer projector of rank $K=d^\iota$.} 
	\label{fig:sn_stab}
\end{figure}

The dependence of the squared shadow norm $\|\Ob_0\|^2_\Stab$ on $d$ and $n$ is illustrated in  \fref{fig:sn_stab} (see the companion paper \cite{MaoYZ24}  for more figures). 
When $\Ob$ is the projector onto a stabilizer state ($K=1$), $\|\Ob_0\|^2_\Stab$ increases linearly with  $d$, but is almost independent of $n$ when $n\geq 5$. For a general rank-$K$ stabilizer projector $\Ob$, 
the ratio $\|\Ob_0\|^2_\Stab/\|\Ob_0\|^2$ decreases monotonically with the rank, but levels off quickly.

Next, we derive nearly tight  bounds for the shadow norm of a general  observable $\Ob\in \caL_0\bigl(\caH_d^{\otimes n}\bigr)$ and the shadow norm of the ensemble $\Stab(n,d)$ itself. The following theorem is also proved in \aref{app:ThirdMomentStab}.
\begin{theorem}\label{thm:ShNormStab}
	Suppose $d$ is a prime and  $\Ob\in \caL_0\bigl(\caH_d^{\otimes n}\bigr)$. Then 
	\begin{gather}
	\!\!\!	\frac{(D+1)\| \Ob \|_2^2}{D}\leq 	\|\Ob\|^2_\Stab \leq \frac{(D+1)[(2d-3) \| \Ob \|_2^2 +2\| \Ob \|^2]}{D+d}\leq (2d-3) \| \Ob \|_2^2 +2\| \Ob \|^2\leq (2d-1) \| \Ob \|_2^2, \label{eq:StabObShadowLUB} \\
	\frac{D+1}{D+d}\left(2d-1-\frac{d}{D}\right)\leq\|\Stab(n,d)\|_\sh\leq \frac{(D+1)(2d-1)}{D+d}\leq 2d-1. \label{eq:StabShadowLUB}
	\end{gather}	
	If  $\Ob$ is diagonal in a stabilizer basis, then 
	\begin{align}\label{eq:StabObShadowDiagn}
	\|\Ob\|^2_\Stab 
	&\leq \frac{(D+1)[(d-1) \| \Ob \|_2^2 +d\| \Ob \|^2]}{(D+d)}\leq (d-1) \| \Ob \|_2^2 +d\| \Ob \|^2. 	
	\end{align}	
	If  $n=1$ and $\Ob$ is diagonal in a stabilizer basis, then 
	\begin{align}\label{eq:StabObShadowDiag1}
	\|\Ob\|^2_\Stab =(d+1)\| \Ob \|^2.
	\end{align}		
\end{theorem}
As a simple corollary of \eref{eq:StabShadowLUB} we have
\begin{align}
\|\Stab(n,d)\|_\sh=\frac{(D+1)(2d-1)}{D+d}+\caO(d/D)=2d-1+\caO(d^2/D). 
\end{align}
The bounds in \eref{eq:StabShadowLUB} are almost tight when $n$ is large. In addition,  the lower bounds for the shadow norm based on the third normalized frame potential in \eref{eq:ShNormPhi3LB2} [see also \eref{eq:Phi3Stab}] are almost tight for the ensemble of stabilizer states. Actually, 
the lower bound in \eref{eq:StabShadowLUB} is saturated when $d=2$, in which case the set of stabilizer states forms a 3-design \cite{KuenG13,Zhu17MC,Webb16}. We believe that this lower bound is  saturated whenever $d$ is a prime. 
\begin{conjecture}\label{con:StabShadow}
	Suppose $d$ is a prime.  Then 
	\begin{gather}
	\|\Stab(n,d)\|_\sh=	\frac{D+1}{D+d}\left(2d-1-\frac{d}{D}\right).  \label{eq:StabShadowConj}
	\end{gather}	
\end{conjecture}

When $n=1$, the set of stabilizer states forms a complete set of mutually unbiased bases (MUB) \cite{DurtEBZ10}. If  $\Ob$ is the projector onto a basis state in the MUB, then \thref{thm:StabShNormStab} implies that 
\begin{align}
\|\Ob_0\|^2_\Stab
&=\frac{(d^2-1)(d-1)}{d^2}\leq d-1, \quad 
\frac{\|\Ob_0\|^2_\Stab}{\|\Ob_0\|_2^2}
=\frac{d^2-1}{d}\leq d. \label{eq:MUBstateShNorm}
\end{align}
If instead $\Ob\in \caL_0(\caH_d)$, then
\thref{thm:ShNormStab} implies that
\begin{gather}
\!\!\!	\frac{(d+1)\| \Ob \|_2^2}{d}\leq 	\|\Ob\|^2_\Stab \leq \frac{(d+1)[(2d-3) \| \Ob \|_2^2 +2\| \Ob \|^2]}{2d}\leq \frac{(d+1)(2d-1)}{2d}\|\Ob\|_2^2\leq (d+1)\|\Ob\|_2^2, \label{eq:MUBObShadowLUB} \\
\frac{d^2-1}{d}\leq\|\Stab(n,d)\|_\sh\leq \frac{(d+1)(2d-1)}{2d}\leq d+1. \label{eq:MUBShadowLUB}
\end{gather}

\section{\label{sec:ThirdMomentGen}Third moment of a general Clifford orbit}

In this section we clarify the basic properties of the third moment of a general Clifford orbit, assuming that the local dimension $d$ is an odd prime. We show that the third normalized frame potential and shadow norm of any Clifford orbit are $\caO(d)$ irrespective of the number $n$ of qudits. If $d=2\mmod 3$, then the same conclusion applies to the operator norm of the third normalized moment operator; if instead $d\neq 2 \mod 3$, however, then this norm may grow exponentially with $n$. To this end, we derive exact formulas or nearly tight bounds for the three figures of merit.

\subsection{Setting the stage}
Let $|\Psi\>$ be a pure state in $\caH_d^{\otimes n}$ and $\orb(\Psi)$ the Clifford orbit generated from $|\Psi\>$, which determines an ensemble of pure states (with uniform distribution on the orbit by default). The third moment operator of $\orb(\Psi)$ reads 
\begin{align}
Q(\orb(\Psi))=Q_3(\orb(\Psi))=\mathbb{E}_{U \sim \oCl(n,d)}\bigl(U|\Psi\>\<\Psi|U^\dag\bigr)^{\otimes 3}. 
\end{align}
By definition the third moment operator $Q(\orb(\Psi))$  belongs to the commutant of the third Clifford tensor power and thus can be expanded in terms of the operator frame $\{R(\caT)\}_{\caT\in \Sigma(d)}$ or its dual frame $\{\tilde{R}(\caT)\}_{\caT\in \Sigma(d)}$ studied in \sref{sec:Commutant},
\begin{align}\label{eq:MomentQpsi}
Q(\orb(\Psi))&=\sum_{\caT\in \Sigma(d)}\tka(\Psi,\caT)R(\caT)=\sum_{\caT\in \Sigma(d)}\kappa(\Psi,\caT)\tR(\caT),
\end{align}
where
\begin{equation}\label{eq:kappapsiT}
\begin{aligned}
\kappa(\Psi,\caT)&:=\tr[R(\caT)(|\Psi\>\<\Psi|)^{\otimes 3}] =\tr[R(\caT)Q(\orb(\Psi))],\\
\tka(\Psi,\caT)&:=\tr\bigl[\tR(\caT)(|\Psi\>\<\Psi|)^{\otimes 3}\bigr]=\tr\bigl[\tR(\caT)Q(\orb(\Psi))\bigr]. 
\end{aligned}
\end{equation}
Accordingly, the third normalized moment operator of $\orb(\Psi)$ reads
\begin{align}\label{eq:MomentQpsiBar}
\bQ(\orb(\Psi))=\bQ_3(\orb(\Psi))=\frac{1}{6}\sum_{\caT\in \Sigma(d)}\hka(\Psi,\caT)R(\caT),
\end{align}
where
\begin{align}
\hka(\Psi,\caT):=6\tr\bigl[\tR(\caT)\bQ(\orb(\Psi))\bigr]=D(D+1)(D+2)\tka(\Psi,\caT)= D(D+1)(D+2) \tr\bigl[\tR(\caT)(|\Psi\>\<\Psi|)^{\otimes 3}\bigr]. \label{eq:hkapsiT}
\end{align}
To simplify the notation, the argument $\Psi$ in $\kappa(\Psi,\caT)$ and $\hka(\Psi,\caT)$ may  be omitted if there is no danger of confusion. 

Let $\scrT$ be a subset of  $\Sigma(d)$;  define
\begin{align}\label{eq:kappascrT}
\kappa(\Psi,\scrT,j)&:=\sum_{\caT\in \scrT}\kappa^j(\Psi,\caT),\quad
|\kappa|(\Psi,\scrT,j):=\sum_{\caT\in \scrT}|\kappa(\Psi,\caT)|^j,\quad j\in \bbN,
\end{align}
which can be abbreviated as $\kappa(\Psi,\scrT)$ and  $|\kappa|(\Psi,\scrT)$ when $j=1$. The functions
$\hka(\Psi,\scrT)$ and $|\hka|(\Psi,\scrT)$ can be defined in a similar way.  Here the argument $\Psi$ may  be omitted if there is no danger of confusion. 
Since $\Sigma(d)=\scrT_\iso\sqcup \scrT_\defe$, it follows that 
\begin{align}
\kappa(\Psi,\Sigma(d))=\kappa(\Psi,\scrT_\iso)+\kappa(\Psi,\scrT_\defe),\label{eq:kappaSigIsoDef}
\end{align}
and this equality still holds if $\kappa$ is replaced by $|\kappa|$, $\hka$, or $|\hka|$.

The following proposition clarifies the symmetry of  $\kappa(\Psi,\caT)$ and $\hka(\Psi,\caT)$; it is a simple corollary of the fact that $R(O)$ for each $O\in S_3$ is a permutation.
\begin{proposition}\label{pro:kappaSym}
	Suppose $d$ is an odd prime, $|\Psi\>\in \caH_d^{\otimes n}$, and $\caT\in \Sigma(d)$. Then 
	\begin{align}\label{eq:kappaSym}
	\kappa(\Psi,\caT)&=\kappa(\Psi,O\caT)=\kappa(\Psi,\caT O)\quad \forall O\in S_3,
	\end{align}
	and this equation still holds if $\kappa$ is replaced by $|\kappa|$, $\hka$, or $|\hka|$ or if $\caT$ is replaced by any subset $\scrT$ of $\Sigma(d)$. If in addition $|\Psi\>\in \Stab(n,d)$  or $\caT\in \scrT_\sym$, then $\kappa(\Psi,\caT)=1$.
\end{proposition}

As a simple corollary of \pref{pro:kappaSym}, we can derive the following relations,
\begin{align}\label{eq:kappaSigNSscrT}
6+\kappa(\Psi,\scrT_\ns,j)=\kappa(\Psi,\Sigma(d),j)=2\kappa(\Psi,\scrT_0,j)=2\kappa(\Psi,\scrT_1,j),
\end{align}
given that $\Sigma(d)=\scrT_\sym\sqcup\scrT_\ns =\scrT_0 \sqcup  \scrT_1$, $|\scrT_\sym|=6$, and $\scrT_1=\tau_{12}\scrT_0$. Next, we establish informative lower and upper bounds for $\kappa(\Psi,\caT)$ and related functions. The following lemma is proved in  \aref{app:hkakaTLUBproof}. 
\begin{lemma}\label{lem:kappaTLUB}
	Suppose $d$ is an odd prime and $|\Psi\>
	\in\caH_d^{\otimes n}$. Then  $\kappa(\Psi,\caT)$ is real for  $\caT\in \Sigma(d)$ and
	\begin{gather}
	-1\leq \kappa(\Psi,\caT)\leq 1 \quad \forall \caT\in \Sigma(d), \quad 
	0\leq \kappa(\Psi,\caT)\leq 1\quad  \forall \caT\in \scrT_\defe,\quad  \kappa(\Psi,\caT)= 1\quad  \forall \caT\in \scrT_\sym, \label{eq:kappaTLUB}\\
	\frac{4(D+d)}{D+1}\leq \kappa(\Psi,\Sigma(d))\leq |\kappa|(\Psi,\Sigma(d))\leq 2d+2,  \label{eq:kappaSigLUB}\\	
	\frac{4d-2D-6}{D+1}\leq \kappa(\Psi,\scrT_\ns)\leq |\kappa|(\Psi,\scrT_\ns)\leq 2d-4,
	\label{eq:kappaNsLUB}\\
	0\leq \kappa(\Psi, \scrT_\defe)\leq |\scrT_\defe|\leq 4,\quad 	\frac{4(d-1)}{D+1}\leq \kappa(\Psi, \scrT_\iso)\leq |\kappa|(\Psi, \scrT_\iso)\leq |\kappa|(\Psi,\Sigma(d))\leq 2d+2. \label{eq:kappaDefIsoLUB}
	\end{gather} 
\end{lemma}

If $|\Psi\>$ can be expressed as a tensor product of the form $|\Psi\>=|\Psi_1\>\otimes |\Psi_2\>$,  then 
\begin{align}\label{eq:kappaPsiProd}
\kappa(\Psi,\caT)=\kappa(\Psi_1,\caT)\kappa(\Psi_2,\caT),\quad |\kappa(\Psi,\caT)|\leq \min\{|\kappa(\Psi_1,\caT)|,\; |\kappa(\Psi_2,\caT)|\}\quad \forall \caT\in \Sigma(d). 
\end{align}
Here the equality follows from the definition in \eref{eq:kappapsiT} given that $R(\caT)=r(\caT)^{\otimes n}$, and the inequality follows from \lref{lem:kappaTLUB}.

\begin{conjecture}\label{con:kappaTLUB}
	Suppose $d$ is an odd prime and $|\Psi\>\in\caH_d^{\otimes n}$. Then
	\begin{gather}
	0\leq \kappa(\Psi,\caT)\leq 1 \quad \forall \caT\in \Sigma(d), \quad  
	6\leq \kappa(\Psi,\Sigma(d))\leq 2d+2, \quad
	0\leq \kappa(\Psi,\scrT_\ns)\leq 2d-4,
	\quad
	\kappa(\Psi, \scrT_\iso)\geq 6. \label{eq:kappaTLUBcon}
	\end{gather} 
\end{conjecture}
If the first inequality in \eref{eq:kappaTLUBcon} holds,  then all inequalities  hold thanks to \lref{lem:kappaTLUB}. 
When  $d=3$, \cref{con:kappaTLUB} holds by \pref{pro:kappaSym} and \lref{lem:kappaTLUB} given that $\Sigma(d)=\scrT_\sym\sqcup \scrT_\defe$.

By virtue of \eref{eq:DualBasis} and the above results
we can clarify the basic properties of $\hka(\Psi,\caT)$, $\hka(\Psi,\Sigma(d))$, $\hka(\Psi,\scrT_\ns)$ and their relations with $\kappa(\Psi,\caT)$, $\kappa(\Psi,\Sigma(d))$, $\kappa(\Psi,\scrT_\ns)$ as summarized in the following proposition and proved in  \aref{app:hkakaTLUBproof}.
\begin{proposition}\label{pro:hkaka}
	Suppose $d$ is an odd prime and $|\Psi\>\in \caH_d^{\otimes n}$. Then $\hka(\Psi,\caT)$ is real for $\caT\in \Sigma(d)$ and
	\begin{gather}
	\hka(\Psi,\caT)
	=\frac{D+2}{D-1}\left[\kappa(\Psi,\caT) -\frac{\kappa(\Psi,\Sigma(d))}{2D+2d}\right]\leq \frac{D+2}{D+1}\quad \forall \caT\in \Sigma(d), 
	\quad  \hka(\Psi,\caT)\geq 	\frac{D+2}{D+d} \quad \forall \caT\in \scrT_\sym, \label{eq:hkakaTLUB}\\
	\hka(\Psi,\Sigma(d))=\frac{D+2}{D+d}\kappa(\Psi,\Sigma(d))=\frac{D+2}{D+d}[\kappa(\Psi,\scrT_\ns)+6]=6+\frac{D-1}{D+2}\hka(\Psi,\scrT_\ns),\label{eq:hkakaSig}\\	
	|\hka|(\Psi,\Sigma(d))=\frac{D+2}{D-1}\left[6-\frac{3\kappa(\Psi,\Sigma(d))}{D+d}\right]+|\hka|(\Psi,\scrT_\ns). \label{eq:hkaSigNsAbs}
	\end{gather}
\end{proposition}
By virtue of \eref{eq:hkakaTLUB} we can also express $\kappa(\Psi,\caT)$  in terms of $\hka(\Psi,\caT)$ and $\hka(\Psi,\Sigma(d))$, 
\begin{align}
\kappa(\Psi,\caT)&=\frac{D-1}{D+2}\hka(\Psi,\caT)+\frac{\hka(\Psi,\Sigma(d))}{2(D+2)}. 
\end{align}
\Lref{lem:kappaTLUB} and \pref{pro:hkaka}  together imply that
\begin{gather}
\hka(\Psi,\caT)=\kappa(\Psi,\caT)+\caO(d/D),\quad \hka(\Psi,\Sigma(d))=\kappa(\Psi,\Sigma(d))+\caO(d^2/D). \label{eq:hkakabigO}
\end{gather}
When $n$ is large, $\hka(\Psi,\caT)$  is approximately equal to $\kappa(\Psi,\caT)$, and $\hka(\Psi,\Sigma(d))$ is approximately equal to $\kappa(\Psi,\Sigma(d))$.
If in addition $\kappa(\Psi,\caT)\geq 0$ for all $\caT\in \Sigma(d)$, then $\kappa(\Psi,\Sigma(d))\geq0$ and 
\begin{equation}\label{eq:hkakabigOpos}
\begin{aligned}
|\hka(\Psi,\caT)|&=\hka(\Psi,\caT)+\caO(d/D)=
\kappa(\Psi,\caT)+\caO(d/D),\\ |\hka|(\Psi,\Sigma(d))&=\hka(\Psi,\Sigma(d))+\caO(d^2/D)=\kappa(\Psi,\Sigma(d))+\caO(d^2/D).
\end{aligned}
\end{equation}
These results are not surprising given that $\{R(\caT)/D^{3/2}\}_{\caT\in \Sigma(d)}$ becomes a better and better approximation to an orthonormal basis in the commutant of the third Clifford tensor power as $n$ gets large (see \sref{sec:Commutant}). More precise relations can be found in \aref{app:ThirdMomentAux}.

\subsection{Third normalized frame potential and normalized moment operator}
Here we derive necessary and sufficient conditions on when the Clifford orbit $\orb(\Psi)$ forms a 3-design. In addition, we clarify the deviation of  $\orb(\Psi)$ from a 3-design with respect to the third normalized frame potential and the operator norm of the third normalized moment operator.

First, we clarify the conditions on $\kappa(\Psi, \caT)$ and $\hka(\Psi,\caT)$ under which  $\orb(\Psi)$ forms a 3-design. The following proposition is proved in \aref{app:3designPhi3LUB}. 
\begin{proposition}\label{pro:orbit3designCon}
	Suppose $d$ is an odd prime and $|\Psi\>\in \caH_d^{\otimes n}$. Then the  three statements are equivalent: 
	\begin{enumerate}
		\item $\orb(\Psi)$ is a 3-design. 
		
		\item $\kappa(\Psi, \caT)=3/(D+2)$ for all $\caT\in \scrT_\ns$. 
		
		\item $\hka(\Psi,\caT)=0$ for all $\caT\in \scrT_\ns$.

	\end{enumerate}
\end{proposition}

Next, we determine the third normalized frame potential of $\orb(\Psi)$ and show that $\bar{\Phi}_3(\orb(\Psi))=\caO(d)$ irrespective of the  number $n$ of qudits.  \Thref{thm:Phi3LUB} and \coref{cor:Phi3UB} below are also proved in \aref{app:3designPhi3LUB}.
\begin{theorem}\label{thm:Phi3LUB}
	Suppose $d$ is an odd prime and  $|\Psi\>\in\caH_d^{\otimes n}$. Then 
	\begin{align}
	\bar{\Phi}_3(\orb(\Psi))&=
	\frac{1}{6}\sum_{\caT\in \Sigma(d)}\kappa(\Psi,\caT)\hka(\Psi,\caT)
	=1+\frac{1}{6}\sum_{\caT\in \scrT_\ns}\left[\kappa(\Psi,\caT)-\frac{3}{D+2}\right]\hka(\Psi,\caT)\nonumber\\
	&=\frac{D+2}{6(D-1)}\left[\kappa(\Psi,\Sigma(d),2) -\frac{\kappa^2(\Psi,\Sigma(d))}{2(D+d)}\right]=\frac{2(D-1)\hka(\Psi,\Sigma(d),2)+\hka^2(\Psi,\Sigma(d))}{12(D+2)}\nonumber\\
	&=1+\frac{D-1}{6(D+2)}\hka(\Psi,\scrT_\ns,2)+\frac{D-1}{12(D+2)^2}\hka^2(\Psi,\scrT_\ns) \leq\frac{d+2}{3}. \label{eq:Phi3psi}
	\end{align}	
\end{theorem}

\begin{corollary}\label{cor:Phi3UB}
	Suppose $\kappa(\Psi,\caT)\geq 0$ for all
	$\caT\in \Sigma(d)$ in \thref{thm:Phi3LUB}. Then 
	\begin{align}\label{eq:Phi3PosUB1}
	\bar{\Phi}_3(\orb(\Psi))\leq \frac{D^2+3(d-2)}{6D^2}\kappa(\Psi,\Sigma(d),2)=\frac{D^2+3(d-2)}{6D^2}\sum_{\caT\in \Sigma(d)}\kappa^2(\Psi,\caT).
	\end{align}
	If in addition  $\kappa(\Psi,\Sigma(d))\geq 6(D+d)/(D+2)$, then  $\bar{\Phi}_3(\orb(\Psi))\leq\kappa(\Psi,\Sigma(d),2)/6$. 
\end{corollary}

Next, we derive nearly tight upper bounds for the operator norms of  $\bQ(\orb(\Psi))$ and  $\bQ(\orb(\Psi))-P_\sym$. 
\Thsref{thm:MomentNormd2} and \ref{thm:MomentNormd13} below are proved in \aref{app:thm:MomentNormProof}.  To simplify the notation, the argument $\Psi$ in $\kappa(\Psi,\cdot)$ and $\hka(\Psi,\cdot)$ will be omitted.

\begin{theorem}\label{thm:MomentNormd2}
	Suppose $d$ is an odd prime satisfying $d=2\mmod 3$ and $|\Psi\>\in\caH_d^{\otimes n}$. Then $\hka(\Sigma(d))/6$ is an eigenvalue of $\bQ(\orb(\Psi))$ and
	\begin{gather}
	\frac{\hka(\Sigma(d))}{6}\leq \|\bQ(\orb(\Psi))\|\leq \frac{|\hka|(\Sigma(d))}{6}\leq  \frac{d+2}{3}, \label{eq:MomentNormd2a} \\
	\|\bQ(\orb(\Psi))-P_\sym\|\leq \frac{D+5}{6(D+2)}|\hka|(\scrT_\ns)< \frac{d+2}{3}. \label{eq:MomentNormd2b}
	\end{gather}	
\end{theorem}

\begin{theorem}\label{thm:MomentNormd13}
	Suppose $d$ is an odd prime satisfying $d\neq 2\mmod 3$,  $|\Psi\>\in\caH_d^{\otimes n}$,  and $\caT_\defe\in \scrT_\defe$. Then $\hka(\scrT_\iso)/6$ and	
	$(D+2)\kappa(\caT_\defe)/3$ are eigenvalues of 
	$\bQ(\orb(\Psi))$ and
	\begin{gather}
	\frac{1}{6}\max\left\{\hka(\scrT_\iso), 2(D+2)\kappa(\caT_\defe)\right\}\leq
	\|\bQ(\orb(\Psi))\|\leq \frac{1}{6}\max\left\{|\hka|(\scrT_\iso), 2(D+2)\kappa(\caT_\defe)\right\}\leq \frac{D+2}{3}, \label{eq:MomentNormd13a}\\
	\|\bQ(\orb(\Psi))-P_\sym\| 
	\leq 	\max\left\{\frac{D+5}{6(D+2)}|\hka|(\scrT_\ns), \left|\frac{D+2}{3}\kappa(\caT_\defe)-1\right|\right\}< \frac{D+2}{3}. \label{eq:MomentNormd13b}
	\end{gather}	 	
	If  $\hka(\caT_\defe)\leq 0$ or $\kappa(\caT_\defe)	\leq \kappa(\Psi,\Sigma(d))/[2(D+d)]$, then
	\begin{gather}
	\frac{1}{6}\hka(\scrT_\iso)\leq
	\|\bQ(\orb(\Psi))\|\leq \frac{1}{6}|\hka|(\scrT_\iso)\leq \frac{d+2}{3}, \label{eq:MomentNormd13c}\\
	\|\bQ(\orb(\Psi))-P_\sym\| 
	\leq 	\max\left\{\frac{D+5}{6(D+2)}|\hka|(\scrT_\ns), 1-\frac{D+2}{3}\kappa(\caT_\defe)\right\}< \frac{d+2}{3}. \label{eq:MomentNormd13d}
	\end{gather}
	If in addition $\hka(\caT_\defe)= 0$ and $\hka(\caT)\geq 0$ for all $\caT\in \Sigma(d)$. Then 
	\begin{gather}
	\|\bQ(\orb(\Psi))\|=\frac{\hka(\Sigma(d))}{6}=1+\frac{D-1}{6(D+2)}\hka(\scrT_\ns), \quad 
	\|\bQ(\orb(\Psi))-P_\sym\| 
	\leq \frac{D+5}{6(D+2)}\hka(\scrT_\ns). \label{eq:MomentNormd13e}
	\end{gather}	
\end{theorem}

By virtue of \eqsref{eq:MomentNormd2a}{eq:MomentNormd13a} in \thsref{thm:MomentNormd2} and \ref{thm:MomentNormd13} we can  deduce that
\begin{align}
\|\bQ(\orb(\Psi))-P_\sym\|\leq \max\{\|\bQ(\orb(\Psi))\|-1, 1\}
\leq  \frac{1}{3}\begin{cases}
d & d=2\mmod 3,\\
D & d\neq 2\mmod 3,
\end{cases}
\end{align}
which complements \eqsref{eq:MomentNormd2b}{eq:MomentNormd13b}. 
If $d=2\mmod 3$, then both $\|\bQ(\orb(\Psi))\|$ and $\|\bQ(\orb(\Psi))-P_\sym\| $ are $\caO(d)$  irrespective of the  number $n$ of qudits. If $d\neq 2\mmod 3$, by contrast, then $\|\bQ(\orb(\Psi))\|$ and $\|\bQ(\orb(\Psi))-P_\sym\| $ may increase exponentially with $n$.  
When $|\Psi\>$ is a stabilizer state for example,  $\kappa(\caT)=1$ and $\hka(\caT)=(D+2)/(D+d)$ for  $\caT\in \Sigma(d)$, which means $\kappa(\caT_\defe)=1$ and $|\hka|(\scrT_\iso)=\hka(\scrT_\iso)\leq 2d$, so the first two inequalities in \eref{eq:MomentNormd2a} and all three inequalities in \eref{eq:MomentNormd13a} are saturated, which is consistent with \pref{pro:Qspectra} and \tref{tab:bQnd3}.
Nevertheless, as we shall see later, a good approximate 3-design can be constructed from a suitable Clifford orbit as long as the fiducial state $|\Psi\>$ is chosen properly. To this end,  $|\hka|(\Psi,\scrT_\ns)$  and $\hka(\Psi,\caT_\defe)$  should be sufficiently small,  which means $\kappa(\Psi, \caT_\defe)$ should be close to $3/(D+2)$.

\subsection{\label{sec:ShadowNormGen}Shadow norm}

Here we first determine the  shadow norms of stabilizer projectors with respect to a general Clifford orbit $\orb(\Psi)$ and clarify  the operational significance of $\hka(\Psi,\Sigma(d))$ and $\kappa(\Psi,\Sigma(d))$. Then we derive nearly tight upper and lower bounds for the shadow norm of a general traceless observable  $\Ob\in \caL_0\bigl(\caH_d^{\otimes n}\bigr)$
and the shadow norm of  $\orb(\Psi)$ itself.

Let $\bcaQ_{\orb(\Psi)}(\cdot)$ be the shadow map associated with $\bQ(\orb(\Psi))$ as defined in \eref{eq:ShadowMap} and  $\|\Ob\|_{\orb(\Psi)}$ the shadow norm of $\Ob$ with respect to  $\orb(\Psi)$.  By virtue of \eqsref{eq:ObShNormDef}{eq:MomentQpsiBar} we can deduce that
\begin{align}
\bcaQ_{\orb(\Psi)}(\Ob)&=\frac{1}{6}\sum_{\caT\in \Sigma(d)}\hka(\Psi, \caT)\caR_\caT(\Ob), \quad \|\Ob\|^2_{\orb(\Psi)}= \frac{6(D+1)}{D+2}\|\bcaQ_{\orb(\Psi)}(\Ob)\|. \label{eq:ObShNormPsiDef}
\end{align}
When $\Ob$ is a rank-$K$ stabilizer projector, we can derive explicit formulas for $\|\bcaQ_{\orb(\Psi)}(\Ob)\|$ and $\|\Ob_0\|^2_{\orb(\Psi)}$. 
Define
\begin{equation} \label{eq:QPsiProj0eig}
\begin{aligned}
\upsilon_1
&:=\frac{(D-K)[D^2+2D+(D^2+D-2)K]}{D^2}-(D-K)(K+1)\hka(\Psi,\Delta)\\
&\;=\frac{(D-K)(K+1)}{2(D-1)}\hka(\Psi,\Sigma(d))-\frac{(D+2)(D-K)(2DK+D-K)}{D^2(D-1)}, \\
\upsilon_2&:=(K^2+K)\hka(\Psi,\Delta)-\frac{(D-1)(D+2)K^2}{D^2}\\
&\;=-\frac{K(K+1)}{2(D-1)}\hka(\Psi,\Sigma(d))+\frac{(D+2)K(D^2+2DK-K)}{D^2(D-1)},
\end{aligned}	
\end{equation}
where the second equality in each definition follows from \lref{lem:hkaDeltaNS}. 
Thanks to \pref{pro:hkaka}, here $\hka(\Psi,\Sigma(d))$ can also be expressed in terms of $\hka(\Psi,\scrT_\ns)$, $\kappa(\Psi,\Sigma(d))$, or $\kappa(\Psi,\scrT_\ns)$. \Thref{thm:StabShNormGen} and \coref{cor:Stab1ShNormGen} below are proved in \aref{app:thm:StabShNormPsi}. 
\begin{theorem}\label{thm:StabShNormGen}
	Suppose $d$ is an odd prime,  $|\Psi\>\in\caH_d^{\otimes n}$, and $\Ob$ is a rank-$K$ stabilizer projector on $\caH_d^{\otimes n}$ with $1\leq K\leq D/d$. Then 
	\begin{gather}
	6\|\bcaQ_{\orb(\Psi)}(\Ob)\|=\frac{K+1}{2} \hka(\Psi,\Sigma(d))+(K^2-1)\hka(\Psi,\Delta),
	\label{eq:QPsiProjNorm}\\
	\frac{(D+1)(D-K)K}{D^2}\leq \|\Ob_0\|^2_{\orb(\Psi)}=\frac{(D+1)\max\{\upsilon_1, \upsilon_2\}}{D+2}\leq \frac{(D+1)(D-K)[dD+(dD-D-d)K]}{D^2(D+d)}, \label{eq:StabShNormGen}\\
	\frac{D+1}{D}\leq 	\frac{\|\Ob_0\|^2_{\orb(\Psi)}}{\|\Ob_0\|_2^2}\leq \frac{D+1}{D+d}\left(d-1+\frac{d}{K}-\frac{d}{D}\right). \label{eq:StabShNormRatioGen}
	\end{gather}
	The upper bounds in \eqsref{eq:StabShNormGen}{eq:StabShNormRatioGen} are saturated when $|\Psi\>$ is a stabilizer state. If  $\kappa(\Psi,\scrT_\ns)\geq 2(d-2)/D$ or $\hka(\Psi,\scrT_\ns)\geq -4(d-2)(D+2)/[D(D+d)]$, then $\|\Ob_0\|^2_{\orb(\Psi)}=(D+1)\upsilon_1/(D+2)$. 
	If instead  $\hka(\Psi,\scrT_\ns)\leq0$, then $\|\Ob_0\|^2_{\orb(\Psi)}\leq K+2$.
\end{theorem}

\begin{corollary}\label{cor:Stab1ShNormGen}
	Suppose $d$ is an odd prime,  $|\Psi\>\in\caH_d^{\otimes n}$, and $\Ob$ is the projector onto a stabilizer state in  $\caH_d^{\otimes n}$. Then 
	\begin{gather}
	6\|\bcaQ_{\orb(\Psi)}(\Ob)\|= \hka(\Psi,\Sigma(d)),  \label{eq:QPsiProjNormK1}\\
	\frac{D+1}{D+2}\hka(\Psi,\Sigma(d))-\frac{(D+1)(3D-1)}{D^2}\leq \|\Ob_0\|^2_{\orb(\Psi)}\leq \frac{D+1}{D+2}\hka(\Psi,\Sigma(d))-\frac{3D^2-1}{D^2}\leq  \hka(\Psi,\Sigma(d))-3,  \label{eq:Stab1ShNormGen}\\
	\hka(\Psi,\Sigma(d))-3-\frac{5}{D}\leq \frac{\|\Ob_0\|^2_{\orb(\Psi)}}{\|\Ob_0\|_2^2}\leq  \hka(\Psi,\Sigma(d))-3,	  \label{eq:Stab1ShNormRatioGen}
	\end{gather}
	If $\kappa(\Psi,\scrT_\ns)\geq -2(D^2-2dD+D+d)/D^2$, then  the lower bound in \eref{eq:Stab1ShNormGen} is saturated. 
\end{corollary}

Note that the lower bound in \eref{eq:Stab1ShNormGen} is saturated whenever $n\geq 2$, $\kappa(\Psi,\scrT_\ns)\geq 0$, and $\Ob$ is the projector onto a stabilizer state. 
In addition, $4\leq \hka(\Psi,\Sigma(d))=(D+2)\kappa(\Psi,\Sigma(d))/(D+d)\leq 2(d+1)$ according to \lref{lem:kappaTLUB} and \pref{pro:hkaka},  so \coref{cor:Stab1ShNormGen} means
\begin{align}
\|\Ob_0\|^2_{\orb(\Psi)}&= \hka(\Psi,\Sigma(d))-3+\caO(d/D)=\kappa(\Psi,\Sigma(d))-3 +\caO(d^2/D),
\\ \frac{\|\Ob_0\|^2_{\orb(\Psi)}}{\|\Ob_0\|_2^2}&=\hka(\Psi,\Sigma(d))-3+\caO(1/D)=\kappa(\Psi,\Sigma(d))-3 +\caO(d^2/D).	  
\end{align}
When $n$ is large, $\hka(\Psi,\Sigma(d))-3$ and $\kappa(\Psi,\Sigma(d))-3$ are very accurate approximations of $ \|\Ob_0\|^2_{\orb(\Psi)}$ and $ \|\Ob_0\|^2_{\orb(\Psi)}/\|\Ob_0\|_2^2$. This observation  highlights  the operational significance of $\hka(\Psi,\Sigma(d))$ and $\kappa(\Psi,\Sigma(d))$.

Next, we turn to the shadow norm of a general traceless observable $\Ob$ and the shadow norm of $\orb(\Psi)$. 
\Thref{thm:ShNormGen} and \coref{cor:ShNormGenPos} below are proved in \aref{app:thm:orbitShadowProof}.

\begin{theorem}\label{thm:ShNormGen}
	Suppose $d$ is an odd prime, $|\Psi\>\in\caH_d^{\otimes n}$, and $\Ob\in \caL_0\bigl(\caH_d^{\otimes n}\bigr)$. Then 
	\begin{align}
	&\|\Ob\|^2_{\orb(\Psi)}	
	\leq \frac{D+1}{D+2}\left\{[\hka(\Delta)+|\hka|(\scrT_\ns) ]\|\Ob\|_2^2+2\hka(\Delta)\|\Ob\|^2\right\}\nonumber\\
	&\hphantom{\|\Ob\|^2_{\orb(\Psi)}}\leq \left[1+\frac{D+1}{D+2}|\hka|(\scrT_\ns) \right]\|\Ob\|_2^2+2\|\Ob\|^2
	\leq (2d-3)\|\Ob\|_2^2+2\|\Ob\|^2\leq (2d-1)\|\Ob\|_2^2,   \label{eq:OShNormGen}\\[1ex]
	&	
	\|\orb(\Psi)\|_\sh\leq\frac{D+1}{D+2}[3\hka(\Delta)+|\hka|(\scrT_\ns)] 	
	\leq 3+\frac{D+1}{D+2}|\hka|(\scrT_\ns)
	\leq 2d-1, \label{eq:ShNormGen}\\	
	&\hka(\Sigma(d))-3-\frac{5}{D}\leq \|\orb(\Psi)\|_\sh
	\leq  \frac{D+1}{D+2}[|\hka|(\Sigma(d))-3\hka(\Delta)] \leq |\hka|(\Sigma(d))-3+\frac{d-5}{D+d}. \label{eq:ShNormGenB}
	\end{align}
	If  in addition $\Ob$ is  diagonal in some stabilizer basis, then 		
	\begin{align}\label{eq:DiagOShNormGen}
	\|\Ob\|^2_{\orb(\Psi)} &\leq \frac{D+1}{2(D+2)} \left\{\left[2\hka(\Delta)+|\hka|(\scrT_\ns)\right]\|\Ob\|_2^2 +\left[4\hka(\Delta)+|\hka|(\scrT_\ns)\right] \|\Ob\|^2\right\}\leq (d-1)\|\Ob\|_2^2 +d \|\Ob\|^2.
	\end{align}		
\end{theorem}

\begin{corollary}\label{cor:ShNormGenPos}
	Suppose  $\hka(\Psi,\caT)\geq 0$ for all $\caT\in \Sigma(d)$ in \thref{thm:ShNormGen}. Then 
	\begin{gather}
	\frac{D+d}{D+1}\|\Ob\|^2_{\orb(\Psi)}	\leq [1+\kappa(\scrT_\ns)]\|\Ob\|_2^2+2\|\Ob\|^2\leq (2d-3)\|\Ob\|_2^2+2\|\Ob\|^2\leq (2d-1)\|\Ob\|_2^2,   \label{eq:OShNormGenPos}\\[1ex]
	\frac{D+d}{D+1}\|\orb(\Psi)\|_\sh\leq  3+\kappa(\scrT_\ns)  \leq (2d-1). \label{eq:ShNormGenPos}
	\end{gather}	
	If  in addition $\Ob$ is  diagonal in some stabilizer basis, then 		
	\begin{align}\label{eq:DiagOShNormGenPos}
	\frac{D+d}{D+1}\|\Ob\|^2_{\orb(\Psi)} &\leq  \frac{1}{2}\left[2+\kappa(\scrT_\ns)\right]\|\Ob\|_2^2 +\frac{1}{2}\left[4+\kappa(\scrT_\ns)\right] \|\Ob\|^2\leq (d-1)\|\Ob\|_2^2 +d \|\Ob\|^2.
	\end{align}		
\end{corollary}

Surprisingly, $\|\orb(\Psi)\|_\sh$ is $\caO(d)$ irrespective of the  number $n$ of qudits and the Clifford orbit chosen. If  $\kappa(\Psi,\caT)\geq 0$ for $\caT\in \Sigma(d)$ (see \cref{con:kappaTLUB}),
then \thref{thm:ShNormGen} and \eref{eq:hkakabigOpos} further imply that
\begin{align}
\|\orb(\Psi)\|_\sh
=\hka(\Sigma(d)-3+\caO(d^2/D)=\kappa(\Sigma(d)-3+\caO(d^2/D). 
\end{align}
If  in addition $\hka(\Psi,\caT)\geq 0$ for $\caT\in \Sigma(d)$,
then $|\hka|(\Sigma(d))=\hka(\Sigma(d))$ and we have
\begin{align}
\|\orb(\Psi)\|_\sh
=\hka(\Sigma(d)-3+\caO(d/D). 
\end{align}
These results further highlight the operational significance of $\hka(\Sigma(d))$ and $\kappa(\Sigma(d))$.

\section{\label{sec:Balance}Clifford orbits based on balanced ensembles}

Let $\scrE=\{|\Psi_j\>,p_j\}_j$ be an ensemble of pure states in $\caH_d^{\otimes n}$, where $\{p_j\}_j$ forms a probability distribution. Let $\orb(\scrE)$ be the ensemble constructed from  $\orb(\Psi_j)$ according to the  probability  distribution $\{p_j\}_j$. Then  the third normalized moment operator of $\orb(\scrE)$ reads
\begin{align}\label{eq:MomentQEbar}
\bQ(\orb(\scrE))=\sum_j p_j \bQ(\orb(\Psi_j))
=\frac{1}{6}\sum_{\caT\in \Sigma(d)}\hka(\scrE,\caT)R(\caT),
\end{align}
where $\bQ(\orb(\scrE))$ and $\bQ(\orb(\Psi_j))$ are abbreviations of $\bQ_3(\orb(\scrE))$ and $\bQ_3(\orb(\Psi_j))$, respectively,  and
\begin{align}
\kappa(\scrE,\caT)&=\sum_j p_j\kappa(\Psi_j,\caT),\quad \hka(\scrE,\caT)=\sum_j p_j\hka(\Psi_j,\caT). \label{eq:kahkaScrE}
\end{align}
Most results on $\kappa(\Psi,\caT)$ and $\hka(\Psi,\caT)$, including \psref{pro:kappaSym}, \ref{pro:hkaka} and \lref{lem:kappaTLUB} for example,  also apply to $\kappa(\scrE,\caT)$ and $\hka(\scrE,\caT)$.

The ensemble $\scrE=\{|\Psi_j\>,p_j\}_j$ is \emph{balanced} if $\kappa(\scrE,\caT)$ for all $\caT\in \scrT_\ns$ are equal to each other. In this case, $\hka(\scrE,\caT)$ for all $\caT\in \scrT_\ns$ are also equal to each other according to \pref{pro:hkaka}. 
In view of this fact, $\kappa(\scrE,\caT)$ and $\hka(\scrE,\caT)$ for $\caT\in \scrT_\ns$ can be abbreviated as $\kappa(\scrE)$ and $\hka(\scrE)$, respectively, when there is no danger of confusion. For example, any ensemble composed of one or more stabilizer states is balanced thanks to \pref{pro:kappaSym}.  A balanced ensemble $\scrE$ is interesting because the third moment of $\orb(\scrE)$ is quite simple, which may serve as a benchmark for understanding more general ensembles. In addition, such ensembles are very useful to constructing 3-designs as we shall see in \sref{sec:3designExact}. The following lemma is proved in \aref{app:BalanceProofs}.

\begin{lemma}\label{lem:kahkaBalance}
	Suppose $d$ is an odd prime   and $\scrE$ is a balanced ensemble on  $\caH_d^{\otimes n}$. Then 
	\begin{gather}
	\kappa(\scrE,\scrT_\ns)=2(d-2)\kappa(\scrE),\quad \kappa(\scrE,\Sigma(d))=6+2(d-2)\kappa(\scrE), 
	\label{eq:kaSigBalance}\\
	\hka(\scrE)=\frac{(D+2)^2}{(D-1)(D+d)}\left[\kappa(\scrE)-\frac{3}{D+2}\right],\quad   \hka(\scrE,\Delta)=1-\frac{(d-2)(D+2)}{(D-1)(D+d)}\left[\kappa(\scrE)-\frac{3}{D+2}\right], \label{eq:hkaBalance}   \\ \hka(\scrE,\scrT_\ns)=\frac{2(d-2)(D+2)^2}{(D-1)(D+d)}\left[\kappa(\scrE)-\frac{3}{D+2}\right], 
	\quad
	\hka(\scrE,\Sigma(d))                               =6+\frac{D-1}{D+2}\hka(\scrE,\scrT_\ns), \label{eq:hkaSigBalance}  \\
	|\hka|(\scrE,\Sigma(d)) 
	=6+\frac{2(d-2)(D+2)(D+2-3s)}{(D-1)(D+d)}\left|\kappa(\scrE)-\frac{3}{D+2}\right|\leq 6+\frac{20d}{9}\left|\kappa(\scrE)-\frac{3}{D+2}\right|,  \label{eq:hkaSigAbsBalance}
	\end{gather}
	where $s=1$ if  $\hka(\scrE)\geq 0$ and $s=-1$ if $\hka(\scrE)< 0$. If in addition $d\geq 7$,  or $n\geq 2$, or $\kappa(\scrE)\geq 3/(D+2)$, then the coefficient $20/9$ in \eref{eq:hkaSigAbsBalance} can be replaced by 2.
\end{lemma}

The following theorem determines the third normalized frame potential of a balanced ensemble; it is a simple corollary of \eqsref{eq:Phi3psi}{eq:hkaBalance} given that $|\scrT_\ns|=2(d-2)$.
\begin{theorem}\label{thm:Phi3Balance}
	Suppose $d$ is an odd prime  and $\scrE$ is a balanced ensemble on $\caH_d^{\otimes n}$. Then 
	\begin{align} \bar{\Phi}_3(\orb(\scrE))=1+\frac{(d-2)(D+2)^2}{3(D-1)(D+d)}\left[\kappa(\scrE)-\frac{3}{D+2}\right]^2	\leq 1+\frac{d}{3}\left[\kappa(\scrE)-\frac{3}{D+2}\right]^2.	
	\end{align}			
\end{theorem}

According to \eref{eq:MomentQEbar}, the third normalized moment operator $\bQ(\orb(\scrE))$ can be expressed as 
\begin{align}
\bQ(\orb(\scrE))&=\frac{1}{6}\sum_{\caT\in \Sigma(d)}\hka(\scrE,\caT)R(\caT)=\frac{1}{6}\{[\hka(\scrE,\Delta)-\hka(\scrE)]R(\scrT_\sym)+ \hka(\scrE) R(\Sigma(d)) \}\nonumber\\
&=[\hka(\scrE,\Delta)-\hka(\scrE)]P_\sym+ \frac{1}{6}\hka(\scrE) R(\scrT_\iso)+ \frac{1}{6}\hka(\scrE) R(\scrT_\defe), \label{eq:MomentBal}
\end{align}
given that $\Sigma(d)=\scrT_\sym\sqcup\scrT_\ns=\scrT_\iso\sqcup\scrT_\defe$ and
\begin{align}
R(\scrT_\sym)=6P_\sym,\quad  R(\Sigma(d))=R(\scrT_\iso)+R(\scrT_\defe),\quad \hka(\scrE,\caT)=\begin{cases}
\hka(\scrE,\Delta) & \caT\in \scrT_\sym,\\ 
\hka(\scrE) & \caT\in \scrT_\ns.
\end{cases}
\end{align}
Thanks to \lref{lem:RTisodefe}, both  $R(\scrT_\iso)$ and $R(\scrT_\defe)$ are proportional to projectors, and
the support of 
$R(\scrT_\defe)$ is contained in the support of $R(\scrT_\iso)$, which in  turn is contained in the support of $P_\sym$. Based on these observations, we can determine the eigenvalues of $\bQ(\orb(\scrE))$ together with their  multiplicities as summarized in \thref{thm:Moment3Balance} and \tref{tab:Qbalance}; see \aref{app:BalanceProofs} for more details. If $\hka(\scrE)=0$, that is, $\kappa(\scrE)= 3/(D+2)$, then $\bQ(\orb(\scrE))=P_\sym$  and $\orb(\scrE)$ is a 3-design. Otherwise, $\bQ(\orb(\scrE))$ may have  two or three distinct eigenvalues within the tripartite symmetric subspace $\Sym_3\bigl(\caH_d^{\otimes n}\bigr)$ depending on the value of $d$.

\begin{theorem}\label{thm:Moment3Balance}
	Suppose $d$ is an odd prime and  $\scrE$ is a balanced ensemble on $\caH_d^{\otimes n}$. Then the eigenvalues (with multiplicities) of the third normalized moment operator $\bQ(\orb(\scrE))$ within $\Sym_3\bigl(\caH_d^{\otimes n}\bigr)$ are given in \tref{tab:Qbalance}. The operator norms of $\bQ(\orb(\scrE))$ and $\bQ(\orb(\scrE))-P_\sym$ are also given in  \tref{tab:Qbalance}.			
\end{theorem}

\begin{table*}
	\renewcommand{\arraystretch}{1.8}
	\caption{\label{tab:Qbalance} Eigenvalues of the third normalized moment operator $\bQ=\bQ(\orb(\scrE))$ (within the tripartite symmetric subspace) associated with the balanced ensemble $\scrE$.  Three potential eigenvalues are denoted by  $\lambda_1, \lambda_2, \lambda_3$, which are arranged in decreasing order when $\kappa(\scrE)> 3/(D+2)$ and increasing order when $\kappa(\scrE)< 3/(D+2)$.   Their multiplicities, denoted by $m_1, m_2, m_3$, coincide with the counterparts in \tref{tab:bQnd3}. In the special case $\kappa(\scrE)= 3/(D+2)$, all eigenvalues are equal to 1.}	
	\begin{math}
	\begin{array}{c|ccc}
	\hline\hline
	& d=3 &d=1 \mmod 3& d=2\mmod 3 \\
	\hline		
	\lambda_1 &  \frac{(D+2)}{3}\kappa(\scrE) & \frac{D+2}{3}\kappa(\scrE)&   \frac{(d-2)(D+2)}{3(D+d)}\kappa(\scrE)+\frac{D+2}{D+d}\\	
	m_1  & \frac{D+1}{2} & D &\frac{D(D+1)(D+d)}{2d+2}\\
	\lambda_2 &  \hka(\scrE,\Delta) &	
	\frac{D+2}{D+d}\bigl[\frac{dD-4D-d-2}{3(D-1)} \kappa(\scrE)+  \frac{D+1}{D-1}\bigr]
	&\frac{D+2}{D-1}[1-\kappa(\scrE)]\\
	m_2 & \frac{D(D+1)(D+2)}{6}-\frac{D+1}{2} & \frac{D(D+1)(D+d-2)}{2d-2}-D &   
	\frac{D(D+1)(D+2)}{6}-\frac{D(D+1)(D+d)}{2d+2}
	\\
	\lambda_3 &/ & \frac{D+2}{D-1}[1-\kappa(\scrE)] &/  \\
	m_3  &/ & \frac{D(D+1)(D+2)}{6}-\frac{D(D+1)(D+d-2)}{2d-2} &/  \\[1ex]
	\|\bQ\| &\begin{cases}
	\lambda_1 &\mbox{if  } \kappa(\scrE)\geq 3/(D+2)\\
	\lambda_2 & \mbox{if  } \kappa(\scrE)\leq 3/(D+2)\\
	\end{cases} 
	& \begin{cases}
	\lambda_1 &\mbox{if  } \kappa(\scrE)\geq 3/(D+2)\\
	\lambda_3 & \mbox{if  } \kappa(\scrE)\leq 3/(D+2)\\
	\end{cases}  &\begin{cases}
	\lambda_1 &\mbox{if  } \kappa(\scrE)\geq 3/(D+2)\\
	\lambda_2 & \mbox{if  } \kappa(\scrE)\leq 3/(D+2)\\
	\end{cases}  \\[2.5ex]	
	\|\bQ-P_\sym\| &|\lambda_1-1|	& |\lambda_1-1| &   \begin{cases} |\lambda_2-1| &\mbox{if  } d=5\\
	|\lambda_1-1| &\mbox{if  } d\geq 11
	\end{cases} \\[1ex]	\hline\hline
	\end{array}	
	\end{math}	
\end{table*}

Next, we clarify the shadow norm of the  ensemble $\scrE$. The squared shadow norms of stabilizer projectors can be determined by virtue of \thref{thm:StabShNormGen} and \coref{cor:Stab1ShNormGen} with $\Psi$ replaced by $\scrE$. The following theorem offers informative upper bounds for the squared shadow norm of a general traceless observable and the shadow norm of the ensemble $\orb(\scrE)$; see \aref{app:BalanceProofs} for a proof. 
\begin{theorem}\label{thm:ShNormBalance}
	Suppose $d$ is an odd prime, $\scrE$ is a balanced ensemble on $\caH_d^{\otimes n}$, and  $\Ob\in \caL_0\bigl(\caH_d^{\otimes n}\bigr)$. Then 
	\begin{align}
	\|\Ob\|^2_{\orb(\scrE)}
	&\leq\left[1+	
	\frac{2(d-2)(D+1)(D+2)}{(D-1)(D+d)}\left|\kappa(\scrE)-\frac{3}{D+2}\right|\right]
	\|\Ob\|_2^2+
	2\|\Ob\|^2\nonumber\\
	& \leq\left[1+	
	2d\left|\kappa(\scrE)-\frac{3}{D+2}\right|\right]
	\|\Ob\|_2^2+
	2	\|\Ob\|^2, \label{eq:OShNormBalance}\\
	\|\orb(\scrE)\|_\sh
	&\leq 3+\frac{2(d-2)(D+1)(D+2)}{(D-1)(D+d)}\left|\kappa(\scrE)-\frac{3}{D+2}\right|
	\leq 	3+2d\left|\kappa(\scrE)-\frac{3}{D+2}\right|. \label{eq:ShNormBalance}
	\end{align}
\end{theorem}

\section{\label{sec:MagicState}Qudit magic gates and magic states}
In this section we review basic results about qudit magic gates and magic states, where the local dimension $d$ is an odd prime. By convention, these magic gates belong to the third Clifford hierarchy \cite{GottC99}.  We are mainly interested in magic
gates that are diagonal in the computational basis, which were clarified by Howard and  Vala \cite{HowaV12}. When $d>3$, let $\tscrP_3(\bbF_d)$ be the set of cubic polynomials on $\bbF_d$ with nonzero cubic coefficients; when $d=3$, let  $\tscrP_3(\bbF_d)$ be the set of functions from  $\bbF_3$ to $\bbZ_9$ that have the form
\begin{align}\label{eq:f39}
f(u)=c_3u^3+3c_2u^2,\quad c_2\in \bbF_3, \quad c_3\in \bbZ_9, \quad c_3\neq 0\mmod 3. 
\end{align}	
In either case, the special cubic polynomial/function $f(u)=u^3$ is called \emph{canonical}.

\begin{proposition}\label{pro:MagicGate}
	Suppose $d$ is an odd prime. Then a diagonal unitary $U$ is a magic gate iff it has the form $U=\rme^{\rmi \phi}T_f$,  where $\phi$ is a real phase, $f\in \tscrP_3(\bbF_d)$, and 
	\begin{align}\label{eq:MagicGate}
	T_f:=\sum_{u\in \bbF_d}\tomega^{f(u)}|u\>\<u|,\quad \tomega:=\begin{cases}
	\omega_d=\rme^{2\pi\rmi/d} & d\geq 5,\\
	\omega_9=\rme^{2\pi\rmi/9} & d=3.
	\end{cases}
	\end{align} 
\end{proposition}
The overall phase factor $\rme^{\rmi \phi}$ in \pref{pro:MagicGate} is irrelevant to our study and will be ignored in the following discussion. The diagonal magic gate $T_f$ defined in \pref{pro:MagicGate} will be referred to as a $T$ gate henceforth. The single-qudit magic state associated with  $T_f$ reads 
\begin{align}\label{eq:psif}
|\psi_f\>:=T_f|+\>=\frac{1}{\sqrt{d}} \sum_{u\in \bbF_d}\tomega^{f(u)}|u\>,
\end{align}
where $|+\>:=\sum_{u\in\bbF_d} |u\>/\sqrt{d}$. By definition $|\psi_f\>$ can be generated from $|0\>$ by the Fourier gate $F$ (which belongs to the Clifford group) followed by the magic gate $T_f$.
The magic gate $T_f$ and magic state $|\psi_f\>$ are \emph{canonical} if $f$ is canonical. Define $\scrM(d)$ as the set of all magic states in $\caH_d$ that have the form in \eref{eq:psif}, that is, 
\begin{align}
\scrM(d):=\{|\psi_f\>\, |\, f \in \tscrP_3(\bbF_d) \}. 
\end{align}

\begin{figure}[tb]
	\centering 
	\includegraphics[width=0.6\textwidth]{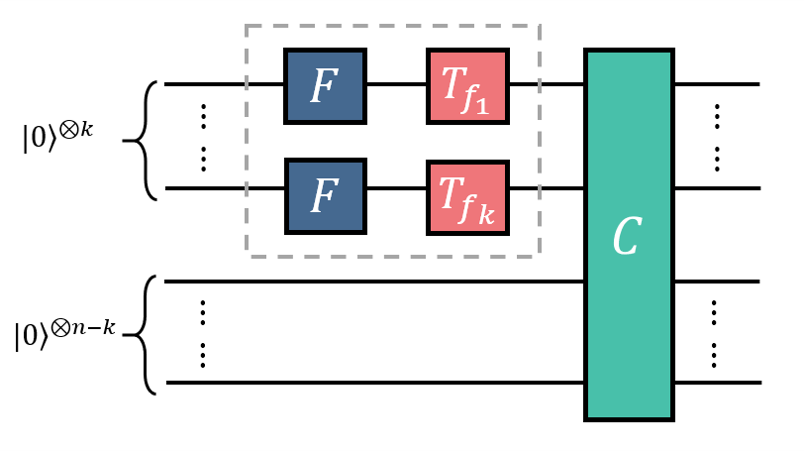}
	\caption{A circuit for sampling from a magic (Clifford) orbit. Here $F$ is the Fourier gate,  $T_{f_1},\ldots, T_{f_k}$ are diagonal magic gates defined in \eref{eq:MagicGate}, and $C$ is a random Clifford gate. The gates in the dashed box generate a fiducial state of the magic  orbit; such a fiducial state belongs to the set $\scrM_{n,k}(d)$.
	} 
	\label{fig:magic_ensemble}
\end{figure}

Next, we turn to $n$-qudit magic states.
Given  $n,k\in \bbN$  with $k\leq n$, define $\scrM_{n,k}(d)$ as  the set of magic states that can be  expressed as tensor products of $k$ single-qudit magic states and $n-k$ copies of  the basis state $|0\>$. Each magic state $|\Psi\>$ in $\scrM_{n,k}(d)$ can be generated by $k$ $T$ gates after $k$ Fourier gates, as illustrated in \fref{fig:magic_ensemble}, and 
has the  form 
\begin{align}\label{eq:MagicStatenk}
|\Psi\>=|\psi_{f_1}\>\otimes |\psi_{f_2}\>\otimes \cdots \otimes |\psi_{f_k}\> \otimes |0\>^{\otimes (n-k)}, \quad f_1, f_2, \ldots, f_k\in  \tscrP_3(\bbF_d)
\end{align}
up to a permutation of tensor factors.   
A magic state in $\scrM_{n,k}(d)$ is canonical if all cubic functions underlying single-qudit magic states in its tensor decomposition are canonical. Such a magic state can be generated by  $k$ \emph{canonical} $T$ gates after $k$ Fourier gates and  has the form $|\psi_f\>^{\otimes k}\otimes |0\>^{\otimes (n-k)}$
up to a permutation of tensor factors, where $f$ is canonical. Denote by $\scrM_{n,k}^\can(d)$
the set of all canonical magic states in $\scrM_{n,k}(d)$.
For the special case $k=0$, we define $\scrM_{n,k}^\can(d)=\scrM_{n,k}(d):=\Stab(n,d)$.

When $d=1\mmod 3$, we need to distinguish three types of $T$ gates and three types of single-qudit magic states, depending on  the   cubic character of the cubic coefficient of the underlying cubic polynomial. Let $\eta_3$ be a given cubic character of $\bbF_d$ as mentioned in \sref{sec:CharIndex} (see also \aref{app:MultiChar} for a brief introduction) and  $f\in \tscrP_3(\bbF_d)$ a cubic polynomial with cubic coefficient $c$. The cubic character of $f\in \tscrP_3(\bbF_d)$, denoted by $\eta_3(f)$, is defined as the cubic character of the cubic coefficient $c$, that is, $\eta_3(f)=\eta_3(c)$, which is in general different from $\eta_3(f(u))$ for a given $u\in \bbF_d$. The cubic character of $T_f$ and that of $|\psi_f\>$ are defined as the cubic character of $f$. The indices of $f$, $T_f$, and $|\psi_f\>$ can be defined in a similar way.

Next, we turn to classifying $n$-qudit magic states. Suppose $|\Psi\>\in \scrM_{n,k}(d)$ has the form in 
\eref{eq:MagicStatenk} up to the positions of tensor factors that are in the state $|0\>$. 
Define
\begin{align}
\eta_3(\Psi):=(\eta_3(f_1),\eta_3(f_2),\ldots, \eta_3(f_k))
\end{align}
and let $\scrC_l=\scrC_l(\Psi)$ be the number of entries of $\eta_3(\Psi)$ that are equal to $\omega_3^l$ for $l=0,1,2$. Then the vector $\scrC(\Psi)=(\scrC_0,\scrC_1,\scrC_2)$ only depends on the cubic characters of cubic coefficients of the polynomials $f_1, f_2, \ldots, f_k$ and 
is invariant under any permutation of the tensor factors of $|\Psi\>$. This vector can be used to distinguish magic states that are fundamentally different. In this way, magic states in $|\Psi\>\in \scrM_{n,k}(d)$ can be decomposed  into $(k+1)(k+2)$ classes, which reduce to three classes when $k=1$ as expected. 
As we shall see in \sref{sec:MagicOrbitd1}, magic orbits associated with different types of magic states may have dramatically different properties.

Next, we highlight certain special magic states.
A magic state $|\Psi\>\in \scrM_{n,k}(d)$ is uniform if all single-qudit magic states in its tensor decomposition have the same cubic character, that is, $\eta_3(\Psi)=\omega_3^l(1,1,\ldots, 1)$ for some $l=0,1,2$. In this case, $\scrC(\Psi)$ has only one nonzero entry. Such a magic state can be generated by  $k$ $T$ gates that have the same cubic character in addition to Fourier gates.
Denote by $\scrM_{n,k}^\id(d)$ the subset of uniform magic states in $\scrM_{n,k}(d)$, that is,
\begin{align}
\scrM_{n,k}^\id(d):=\{|\Psi\>\in \scrM_{n,k}(d): \eta_3(\Psi)=\omega_3^l(1,1,\ldots, 1) \quad \mbox{for} \quad l=0,1,2\}. 
\end{align}
Note that $\scrM_{n,k}^\id(d)$ is a union of three classes of magic states.

By contrast, a magic state  $|\Psi\>\in \scrM_{n,k}(d)$ is $k$-quasi-balanced if $|\scrC_l(\Psi)-k/3|\leq 2/3$ for $l=0,1,2$; in other words, $\scrC_l(\Psi)$ for each $l$ is equal to either $\lfloor k/3\rfloor$ or  $\lceil k/3\rceil$. 
The magic state $|\Psi\>$ is $k$-balanced if in addition 
$\scrC_l(\Psi)=k/3$ for $l=0,1,2$, which can happen only if $k$ is divisible by 3. As we shall see later in \sref{sec:BalanceMagic}, a $k$-balanced magic state is automatically balanced according to the definition in \sref{sec:Balance}. 
Denote by $\scrM_{n,k}^\rmB(d)$  the set of $k$-balanced magic states and by $\scrM_{n,k}^\QB(d)$ the set of $k$-quasi-balanced magic states. 

For completeness, when $d\neq 1 \mmod 3$, we take the convention 
$\scrM_{n,k}^\QB(d)=\scrM_{n,k}^\rmB(d)=\scrM_{n,k}^\id(d)=\scrM_{n,k}(d)$; in other words, 
all magic states in $\scrM_{n,k}(d)$ are regarded as uniform and $k$-balanced.

\section{\label{sec:MagicOrbitd23}Magic orbits and 3-designs  when $d\neq 1\mmod 3$}

In this section and the next section, we clarify the properties of magic (Clifford) orbits and propose simple methods for constructing approximate and exact 3-designs, assuming that the local dimension $d$ is an odd prime.  Our discussion is based on three main figures of merit, namely, the  third normalized frame potential, shadow norm, and operator norm of the third normalized moment operator. Throughout the two sections we assume that $n\in \bbN$, $k\in \bbN_0$, and $k\leq n$. Two cases are discussed separately depending on the congruence class of $d$ modulo 3.

In this section we focus on the case $d\neq 1\mmod 3$, which means  $d=3$ or  $d=2\mmod 3$. Suppose $|\Psi\>$ is a magic state in $\scrM_{n,k}(d)$ as defined in \sref{sec:MagicState}. 
Here  show that the deviation of $\orb(\Psi)$  from a 3-design decreases exponentially with $k$ with respect to all three figures of merit. If $k$ is sufficiently large, then $\orb(\Psi)$ is a good approximate 3-design with respect to all three  figures of merit when $d=2\mmod 3$. 
When $d=3$, by contrast,  $\orb(\Psi)$ is a good approximate 3-design with respect to the  third normalized frame potential and shadow norm, but the operator norm of the third normalized moment operator increases exponentially with $n$ even if $k=n$. Nevertheless, it is not difficult to find other Clifford orbits that form exact 3-designs \cite{GrosNW21}.

\subsection{Technical basis}

To start our analysis, we need to determine the value of $\kappa(\psi,\caT)$ for $\caT\in \scrT_\ns$ when $|\psi\>$ is a magic state of a single qudit, that is, $|\psi\>\in \scrM(d)$. When $d>3$, $|\psi\>$ is determined by a cubic polynomial $f$, and this problem amounts to  counting the number of solutions of a  cubic equation associated with $f$ on the stochastic Lagrangian subspace $\caT$. When $d=3$, the underlying counting problem  is slightly different. Nevertheless, 
in both cases we can derive a uniform result as summarized in \lref{lem:kappaMagicOned23} below and proved in \aref{app:lem:Qud23MagicProof}. As a corollary, \lref{lem:kahkaMagicd23} is also proved in \aref{app:lem:Qud23MagicProof}. The two lemmas are the  key to understanding Clifford orbits based on magic states. 
\begin{lemma}\label{lem:kappaMagicOned23}
	Suppose $d$ is an odd prime, $d\neq1 \mmod 3$, and  $|\psi\>\in \scrM(d)$ is a single-qudit magic state. Then
	\begin{align}\label{eq:kappaMagicOned23}
	\kappa(\psi, \caT)=\frac{2}{d}\quad \forall \caT\in \scrT_\ns.
	\end{align}
\end{lemma}

\begin{lemma}\label{lem:kahkaMagicd23}
	Suppose $d$ is an odd prime, $d\neq1 \mmod 3$, $\caT\in \scrT_\ns$, $j\in \bbN_0$,   and	$|\Psi\>\in \scrM_{n,k}(d)$.  Then
	\begin{gather}
	\kappa(\Psi, \caT)=\frac{2^k}{d^k},\quad   \kappa(\Psi, \scrT_\ns,j)=\frac{2^{jk+1}(d-2)}{d^{jk}}, \quad 
	\kappa(\Psi,\Sigma(d),j)=6+\frac{2^{jk+1}(d-2)}{d^{jk}},  \label{eq:kappaSigd23}   \\
	\hka(\Psi, \caT)=\frac{ (D+2) \bigl[\frac{2^k}{d^k}(D+2)-3\bigr]}{(D-1)(D+d)}, \quad 
	|\hka(\Psi, \caT)|\leq \frac{2^k(D+2)}{d^k(D+d)}< \frac{2^k}{d^k},  \label{eq:hkaAbsUBd23}\\
	\hka(\Psi,\Sigma(d))=\frac{D+2}{D+d}\kappa(\Psi, \Sigma(d)),\quad 	|\hka|(\Psi,\Sigma(d))=\begin{cases}
	\frac{(d+2)(7d^2-21d+20)}{d^2(d-1)} & \mbox{if}\quad n=k=1, d\geq 5,\\
	\hka(\Psi,\Sigma(d)) &\mbox{otherwise}.
	\end{cases}\label{eq:hkaAbsSigd23}
	\end{gather}
	In addition, $\hka(\Psi, \caT)<0$ if $n=k=1$ and $d\geq 5$, while  $\hka(\Psi, \caT)>0$ otherwise. 	
\end{lemma}

By virtue of \lsref{lem:kappaMagicOned23} and \ref{lem:kahkaMagicd23} it is straightforward to deduce the following limits,
\begin{gather}\label{eq:kahkaSigLimd23}
\begin{gathered}
\lim_{k\to \infty}|\hka|(\Psi,\Sigma(d))=\lim_{k\to \infty}\hka(\Psi,\Sigma(d))=\lim_{k\to \infty} \kappa(\Psi,\Sigma(d))=6,\\
\lim_{d\to \infty}|\hka|(\Psi,\Sigma(d))=\lim_{d\to \infty}\hka(\Psi,\Sigma(d))=\lim_{d\to \infty} \kappa(\Psi,\Sigma(d))=\begin{cases}
10 & n\geq 2,k=1,\\
6 & n\geq k\geq 2, 
\end{cases}\\
\lim_{d\to \infty} \kappa(\Psi,\Sigma(d))=10,\quad 
\lim_{d\to \infty}\hka(\Psi,\Sigma(d))=5,\quad
\lim_{d\to \infty}|\hka|(\Psi,\Sigma(d))=7,\quad n=k= 1. 
\end{gathered}
\end{gather}
Note that the specific choice of the magic state $|\Psi\>$ in $\scrM_{n,k}(d)$  does not affect these limits.  Similar remarks  apply to other limits in this section.

\subsection{Approximate 3-designs based on magic orbits }

Here we clarify the deviations of magic  orbits from 3-designs with respect to three figures of merit. \Thsref{thm:Phi3Magicd23}-\ref{thm:ShNormMagicd23} below are proved in \aref{app:thm:Phi3Shadow32Proof}.

\begin{theorem}\label{thm:Phi3Magicd23}
	Suppose $d$ is an odd prime, $d\neq 1 \mmod 3$,   and	$|\Psi\>\in \scrM_{n,k}(d)$. Then 
	\begin{gather}
	\bar{\Phi}_3(\orb(\Psi))=1+\frac{(d-2)(D+2)^2}{3(D-1)(D+d)}\left(\frac{2^k}{d^k}-\frac{3}{D+2}\right)^2\leq 1+\frac{2^{2k}(d-2)(D-1)}{3d^{2k}(D+d)}
	\leq 1+\frac{2^{2k}}{3d^{2k-1}}. \label{eq:PHi332}
	\end{gather}
	If in addition $k\geq 1$, then $\bar{\Phi}_3(\orb(\Psi))\leq 29/25$. 	
\end{theorem}

Suppose $|\Psi\>\in \scrM_{n,k}(d)$ with $n\geq k\geq 1$; then \thref{thm:Phi3Magicd23} implies that
\begin{align}\label{eq:PhiLimd23}
\lim_{d\to \infty} \bar{\Phi}_3(\orb(\Psi))=\lim_{k\to \infty}\bar{\Phi}_3(\orb(\Psi))=1.  
\end{align}
Note that $\orb(\Psi)$ is a good approximation to a 3-design with respect to the third normalized frame potential as long as $k\geq 1$. 

Next, we turn to the shadow norms of stabilizer projectors and general observables. 
\begin{theorem}\label{thm:StabShNormMagicd23}
	Suppose $d$ is an odd prime,   $d\neq 1 \mmod 3$,	$|\Psi\>\in \scrM_{n,k}(d)$, $\Ob$ is a rank-$K$ stabilizer projector on $\caH_d^{\otimes n}$ with $1\leq K\leq D/d$, and $\Ob_0=\Ob-\tr(\Ob)\bbI/D$. Then 
	\begin{align}
	\|\Ob_0\|^2_{\orb(\Psi)}&=\frac{(D+1)(D-K)(K+1)}{2(D-1)(D+2)}\hka(\Psi,\Sigma(d))-\frac{(D+1)(D-K)(2DK+D-K)}{D^2(D-1)}\nonumber\\
	&=\frac{(D+1)(D-K)[2D^2-dD+(D^2-2dD+D+d)K]}{D^2(D-1)(D+d)}+\frac{2^k(d-2)(D+1)(D-K)(K+1)}{d^k(D-1)(D+d)},\nonumber\\
	&\leq \frac{D+1}{D+d}\left[\frac{K+2}{K}+\frac{2^k(d-2)(K+1)}{Kd^k}
	\right]\|\Ob_0\|_2^2\leq \left[3+\frac{2^{k+1}(d-2)}{d^k}\right]	\|\Ob_0\|_2^2.
	\label{eq:StabShNormMagicd23}
	\end{align}
	In addition, $\|\Ob_0\|^2_{\orb(\Psi)}/\|\Ob_0\|_2^2$ decreases monotonically with $k$ and $K$. 
\end{theorem}

As a simple corollary of \thref{thm:StabShNormMagicd23}, we can deduce the following limits,
\begin{equation}
\begin{aligned}
\lim_{n\to \infty}\|\Ob_0\|^2_{\orb(\Psi)}&=K+2+\frac{2^k(d-2)(K+1)}{d^k},\quad \lim_{n,k\to \infty}\|\Ob_0\|^2_{\orb(\Psi)}=K+2,\\ \lim_{d\to \infty}\|\Ob_0\|^2_{\orb(\Psi)}&=\begin{cases}
2 & n=k=K=1,\\
3K+4 & n\geq 2, k=1, \\
K+2 &n\geq k\geq 2.
\end{cases}	  \label{eq:StabShNormMagicd23Lim}
\end{aligned}
\end{equation}
When $K=1$, so that $\Ob$ is the projector onto a stabilizer state, \eref{eq:StabShNormMagicd23} can be simplified as follows,
\begin{align}
\|\Ob_0\|^2_{\orb(\Psi)}&=\frac{D+1}{D+2}\hka(\Psi,\Sigma(d))-\frac{(D+1)(3D-1)}{D^2}=
\frac{(D+1)(3D^2-3dD+D+d)}{D^2(D+d)}+\frac{2^{k+1}(d-2)(D+1)}{d^k(D+d)}\nonumber\\
&\leq \frac{D+1}{D+d} \left[3+\frac{2^{k+1}(d-2)}{d^k}\right]	\|\Ob_0\|_2^2\leq 3+\frac{2^{k+1}(d-2)}{d^k}. \label{eq:Stab1ShNormMagicd23}
\end{align}

\begin{theorem}\label{thm:ShNormMagicd23}
	Suppose $d$ is an odd prime, $d\neq1 \mmod 3$, 	$|\Psi\>\in \scrM_{n,k}(d)$,
	and  $\Ob\in \caL_0\bigl(\caH_d^{\otimes n}\bigr)$. Then 
	\begin{gather}
	\|\Ob\|^2_{\orb(\Psi)}\leq	
	\left[1+\frac{2^{k+1}(d-2)}{d^k}\right]	\|\Ob\|_2^2+2\|\Ob\|^2\leq \left[3+\frac{2^{k+1}(d-2)}{d^k}\right]	\|\Ob\|_2^2,  \label{eq:OShNormMagicd23}\\[1ex]
	\hka(\Sigma(d))-3-\frac{5}{D}\leq	\|\orb(\Psi)\|_\sh\leq 	3+\frac{2^{k+1}(d-2)}{d^k}\leq 3+ \frac{2^{k+1}}{d^{k-1}}. \label{eq:ShNormMagicd23}
	\end{gather}		
\end{theorem}

Suppose $|\Psi\>\in \scrM_{n,k}(d)$ with  $0\leq k\leq n$; then \thref{thm:ShNormMagicd23} implies that
\begin{align}
\|\orb(\Psi)\|_\sh\leq \begin{cases}
2d-1 & k=0,\\
7 &k=1,\\
4 &k\geq 2.
\end{cases}
\end{align}
In addition,
\thref{thm:ShNormMagicd23} and \eref{eq:kahkaSigLimd23} together imply that
\begin{align}\label{eq:lim73}
\lim_{k\to \infty}\|\orb(\Psi)\|_\sh=3, \quad \lim_{d\to \infty} \|\orb(\Psi)\|_\sh=\begin{cases}
7 & n\geq 2, k=1,\\
3 &n\geq k\geq 2. 
\end{cases}
\end{align}

Next, we clarify the basic properties of the third normalized moment operator $\bQ(\orb(\Psi))$; two cases are discussed separately. \Thsref{thm:Moment3Magicd2} and \ref{thm:Moment3Magicd3} below are simple corollaries of \lref{lem:kahkaMagicd23} and \thref{thm:Moment3Balance}.
\begin{theorem}\label{thm:Moment3Magicd2}
	Suppose $d$ is an odd prime, $d=2 \mmod 3$, and $|\Psi\>\in \scrM_{n,k}(d)$. Then $\bQ(\orb(\Psi))$ has two distinct eigenvalues in $\Sym_3\bigl(\caH_d^{\otimes n}\bigr)$, namely, 	
	\begin{align}\label{eq:MomEigMagicd2}
	\lambda_1=\frac{\hka(\Psi,\Sigma(d))}{6}=\frac{[3d^k+2^k(d-2)](D+2)}{3d^k(D+d)},\quad 
	\lambda_2=\frac{D+2}{D-1}\left(1-\frac{2^k}{d^k}\right).
	\end{align}
	The multiplicities $m_1, m_2$ of $\lambda_1,\lambda_2$ are the same as that given in \tsref{tab:bQnd3} and \ref{tab:Qbalance}. Moreover,
	\begin{gather}
	\| \bQ(\orb(\Psi))\|=\begin{cases}
	\lambda_2 & \mbox{if $n=k=1$},\\
	\lambda_1 & \mbox{otherwise},
	\end{cases}  \\ 
	\| \bQ(\orb(\Psi))-P_\sym\|=
	\begin{cases}
	|\lambda_2-1| &\mbox{if $d=5$},\\
	|\lambda_1-1| &\mbox{if $d>5$},
	\end{cases}   \\
	\| \bQ(\orb(\Psi))-P_\sym\| \leq  \frac{d-2}{3}\times \frac{2^k}{d^k}. 
	\end{gather}
\end{theorem}
By virtue of \thref{thm:Moment3Magicd2} it is straightforward to deduce that $\| \bQ(\orb(\Psi))-P_\sym\|\leq 2/3$ when $k\geq 1$. In addition,
\begin{align}
\lim_{k\to \infty}	\| \bQ(\orb(\Psi))-P_\sym\|&=0,\quad 
\lim_{d\to \infty}	\| \bQ(\orb(\Psi))-P_\sym\|=
\begin{cases}
0 &n\geq k\geq 2,\\[0.5ex]
\frac{1}{6} & n=k=1,\\[0.5ex]
\frac{2}{3} &k=1,n\geq 2.
\end{cases}
\end{align}

\begin{theorem}\label{thm:Moment3Magicd3}
	Suppose $d=3$  and	$|\Psi\>\in \scrM_{n,k}(d)$. Then $\bQ(\orb(\Psi))$ has two distinct eigenvalues in $\Sym_3\bigl(\caH_d^{\otimes n}\bigr)$, namely, 
	\begin{align}\label{eq:MomEigMagicd3}
	\begin{aligned}
	\lambda_1=\frac{2^k(D+2)}{3d^k}, \quad  \lambda_2=\frac{(D+2)[d^kD-2^k]}{d^k(D-1)(D+3)}.
	\end{aligned}
	\end{align}
	The multiplicities $m_1, m_2$ of $\lambda_1,\lambda_2$ are the same as that given in  \tsref{tab:bQnd3} and \ref{tab:Qbalance}. Moreover,
	\begin{align}
	\| \bQ(\orb(\Psi))\|=\lambda_1,   \quad \| \bQ(\orb(\Psi))-P_\sym\|=\lambda_1-1.
	\end{align}
\end{theorem} 
The eigenvalues in \eref{eq:MomEigMagicd3} satisfy the following inequalities
\begin{align}
\frac{5}{6}\leq 1-\frac{2^k}{d^k(D+3)}< \lambda_2< 1< \lambda_1,\quad  \lambda_1-1>1-\lambda_2. 
\end{align}
In the case $d=3$, the orbit of any magic state  is in general far from a 3-design with respect to the operator norm of the third normalized moment operator, although it is  a good approximation with respect to the third normalized frame potential and shadow norm. Nevertheless, an exact 3-design can be constructed from a suitable Clifford orbit according to \rcite{GrosNW21} and \sref{sec:3designConstructd23} below.

\subsection{\label{sec:3designConstructd23}Construction of exact 3-designs}
For a single qudit with $d=2\mmod 3$ and $d\geq 5$, we can construct an exact 3-design using only two Clifford orbits, namely, the stabilizer orbit and a magic orbit.  By contrast, $\caO(d)$ orbits are required according to the construction method proposed by GNW \cite{GrosNW21}.

Let $|\Psi\>\in \scrM_{n,k}(d)$; then $\kappa(\Psi,\caT)=2^k/d^k$ for $\caT\in \scrT_\ns$ by \lref{lem:kahkaMagicd23}. Suppose we choose the stabilizer orbit with weight $w$ and $\orb(\Psi)$ with weight $1-w$. Then the resulting ensemble forms a  3-design if 
\begin{align}
w+\frac{2^k}{d^k}(1-w)=\frac{3}{D+2},
\end{align}
which means 
\begin{align}\label{eq:wdnk}
w=w(d,n,k):=\frac{\frac{3}{D+2}-\frac{2^k}{d^k}}{1-\frac{2^k}{d^k}}=\frac{3d^k-2^k(D+2)}{(d^k-2^k)(D+2)}.
\end{align}
Note that  $w(d,n,k)<1$; in addition,  $w(d,n,k)>0$ iff $n=k=1$ and $d\geq 5$, so an exact 3-design can be constructed when $n=k=1$ and $d\geq 5$.  In general, according to \rcite{GrosNW21}, it is possible to  construct an exact 3-design using $\caO(d)$ orbits, but it is still an open problem whether the orbit number can be reduced to a constant that is independent of $d$.

In the case $d=3$, the above construction does not work. Nevertheless, an exact 3-design for an $n$-qutrit system can be constructed from a single orbit  if the fiducial state is chosen properly \cite{GrosNW21}. Let
\begin{align}
|\psi(\phi)\>=\frac{\cos\phi |0\>-\sin\phi |1\>}{\sqrt{2}},\quad |\Psi_n(\phi)\>=|\psi(\phi)\>^{\otimes n}. 
\end{align}
By virtue of \lref{lem:defectT} we can deduce that
\begin{align}
\kappa(\Psi_n(\phi),\caT)=\bigl(\cos^3\phi-\sin^3\phi\bigr)^{2n}\quad \forall \caT\in \scrT_\ns=\scrT_\defe. 
\end{align}
So $|\Psi_n(\phi)\>$ is a fiducial state of a 3-design if 
\begin{align}
\cos^3\phi-\sin^3\phi=\pm \biggl(\frac{3}{D+2}\biggr)^{1/(2n)},
\end{align}
in which case $\kappa(\Psi_n(\phi),\caT)=3/(D+2)$ for all $\caT\in \scrT_\defe$.
Alternatively, we can construct a fiducial state of the form $|\psi(\phi)\>\otimes|0\>^{\otimes (n-1)}$, where $\phi$ satisfies the  condition
\begin{align}
\cos^3\phi-\sin^3\phi=\pm \biggl(\frac{3}{D+2}\biggr)^{1/2}. 
\end{align}

\section{\label{sec:MagicOrbitd1}Magic orbits and 3-designs  when $d= 1\mmod 3$}
In this section we assume  that $d$ is an odd prime satisfying $d=1\mmod 3$ (which means $d\geq 7$) unless stated otherwise (results on other prime local dimensions are presented as a benchmark); in addition, $n\in \bbN$, $k\in \bbN_0$,  $k\leq n$, and $D=d^n$ as before. Recall that
$\tscrP_3(\bbF_d)$ is the set of cubic polynomials over $\bbF_d$ with nonzero cubic coefficients;  $\scrM_{n,k}^\can(d)\subseteq \scrM_{n,k}^\id(d)\subseteq \scrM_{n,k}(d)$ are three sets of magic states defined in \sref{sec:MagicState}.

After clarifying the basic properties of a general magic  orbit $\orb(\Psi)$ with $|\Psi\>\in \scrM_{n,k}(d)$,  we show that  the deviation of $\orb(\Psi)$ from a 3-design usually decreases exponentially with $k$ with respect to the third normalized frame potential, operator norm of the third normalized moment operator, and shadow norm.
If $k$ is sufficiently large, then $\orb(\Psi)$ is a good approximate 3-design with respect to the third normalized frame potential and shadow norm. 
However, the operator norm  strongly depends on the fiducial state $|\Psi\>$ chosen and may even increase exponentially with $n$.
In addition, we propose  a general method for constructing accurate approximate 3-designs with respect to all three figures of merit simultaneously.  By virtue of balanced magic ensembles, we further construct exact 3-designs for  an infinite family of local dimensions. 
As we shall see shortly, two Clifford orbits suffice to construct an accurate approximate 3-design;  four Clifford orbits suffice to construct an exact 3-design for an infinite family of local dimensions.

To streamline the presentation, the proofs of most technical lemmas and main results, including \Lsref{lem:kappaMagicOned1}-\ref{lem:BalMagicExactd1} and \thsref{thm:Phi3MagicUBd1}-\ref{thm:StabShNormMagicd1}, \ref{thm:AccurateDesignMagic},  are relegated to \aref{app:MagicOrbitd1Proof}.

\subsection{\label{sec:MagicOrbitTechbasis}Technical basis}

To understand the  basic properties of the  Clifford orbit generated from a magic state $|\Psi\>\in \scrM_{n,k}(d)$,  we need to determine the function $ \kappa(\Psi,\caT)$ defined in  \eref{eq:kappapsiT}. To this end, first we need to    consider the case in which $|\Psi\>$ is a single-qudit magic state. Then $|\Psi\>$ is  specified by a cubic polynomial $f\in \tscrP_3(\bbF_d)$ according to \eref{eq:psif}, and $\kappa(\Psi,\caT)$ can be written as $\kappa(f,\caT)$.

\begin{lemma}\label{lem:kappaMagicOned1}
	Suppose $d$ is an odd prime, $d=1\mmod 3$,	
	$f\in \tscrP_3(\bbF_d)$ has cubic coefficient $c$, $\caT\in \scrT_\ns$, and $a\in \bbF_d$ satisfies $\eta_3(a)=\eta_3(\caT)$.  Then 
	\begin{gather}\label{eq:kappaMagicOned1}
	\frac{1}{9d^2}< \kappa(f,\caT)=\frac{2}{d}+\frac{2}{d^2}\Re\bigl[\eta_3(ac)G^2(\eta_3)\bigr]=\frac{g^2(3,ac)}{d^2}< \frac{4d-3}{d^2}<\frac{4}{d}.
	\end{gather}
\end{lemma}

Here $\eta_3$ is any given cubic character on $\bbF_d$ as in \sref{sec:CharIndex}; $G(\eta_3)$ and $g(3,ac)$ are two types of Gauss sums \cite{LidlN97book,BernEW98book} as explained in \aref{app:GaussJacobi}. The cubic character $\eta_3(\caT)$ is defined in \sref{sec:CharIndex}. Thanks to \eref{eq:gGaussSumCubic}, the value of $\kappa(f,\caT)$  is independent of the specific choice of the cubic character $\eta_3$. 
In addition, $\kappa(f,\caT)$  is completely determined by the cubic character of  the product $ac$ and can take on three possible values for any given prime $d$ satisfying $d=1\mmod 3$ (cf. \lref{lem:GaussSum3LUB} in \aref{app:GaussJacobi}).  The three possible values of 
$\kappa(f,\caT)$,  arranged in nondecreasing order, are denoted by $\kappa_0^\uparrow(d), \kappa_1^\uparrow(d), \kappa_2^\uparrow(d)$ henceforth, which can be abbreviated as  $\kappa_0^\uparrow, \kappa_1^\uparrow, \kappa_2^\uparrow$ if there is no danger of confusion. 
Suppose  $\nu$ is a given primitive element in $\bbF_d$ (see \tref{tab:nu}); for $j=0,1,2$ define
\begin{align}
\kappa_j(f)=\kappa_j(f,d)&:=\frac{2}{d}+\frac{2}{d^2}\Re\bigl[\eta_3(3\nu^jc)G^2(\eta_3)\bigr]=\frac{g^2(3,3\nu^j c)}{d^2}. \label{eq:kappaMagid1j}
\end{align}
Then  $\kappa(f,\caT)=\kappa_j(f)$ if $\eta_3(\caT)=\eta_3(3\nu^j)$ and  $\{\kappa_j(f)\}_{j=0,1,2}$ for a given $f$ is identical to $\{\kappa_j^\uparrow\}_{j=0,1,2}$
up to a permutation. In conjunction with \lref{lem:Tindex} we can deduce that
\begin{align}\label{eq:kappaTdefMagicd1}
\kappa(f,\caT)=\kappa_0(f)\quad \forall \caT\in \scrT_\defe. 
\end{align}

To better understand the properties of $\kappa_0^\uparrow, \kappa_1^\uparrow, \kappa_2^\uparrow$, define
\begin{align}\label{eq:tgd}
\tg(d):=\frac{g(3,1)g(3,\nu)g(3,\nu^2)}{d}.
\end{align}
Note that $\tg(d)$ is independent of the specific choice of $\nu$. According to \lref{lem:GaussSum3LUB} [see also \tref{tab:tgd}] in \aref{app:GaussJacobi}, $\tg(d)$ is an integer and satisfies 
\begin{align}\label{eq:tgdsq}
1\leq \tg(d)^2=2d[1+\cos(6\phi)]\leq 4d-27,
\end{align}
where $\phi$ is the phase of the Gauss sum $G(\eta_3)$.

\begin{lemma}\label{lem:kappaOrder}
	Suppose $d$ is an odd prime and $d=1\mmod 3$; then	$\kappa_0^\uparrow(d), \kappa_1^\uparrow(d), \kappa_2^\uparrow(d)$ satisfy the relations
	\begin{gather}
	\kappa_0^\uparrow+\kappa_1^\uparrow
	+\kappa_2^\uparrow=\frac{6}{d},\quad \kappa_0^\uparrow\kappa_1^\uparrow+\kappa_1^\uparrow\kappa_2^\uparrow+\kappa_2^\uparrow\kappa_0^\uparrow=\frac{9}{d^2}, \quad 
	\frac{1}{d^4}\leq \kappa_0^\uparrow\kappa_1^\uparrow\kappa_2^\uparrow=\frac{\tg(d)^2}{d^4}
	\leq \frac{4d-27}{d^4}, \label{eq:kappaSumProd} \\ 		
	0<\kappa_0^\uparrow< \frac{1}{d}< \kappa_1^\uparrow < \frac{3}{d}< \kappa_2^\uparrow <\frac{4}{d},\quad 
	\bigl(d\kappa_0^\uparrow\bigr)^k+\bigl(d\kappa_1^\uparrow\bigr)^k> 2\quad \forall k\in \bbN. \label{eq:kappaOrder} 
	\end{gather}	
\end{lemma}
Suppose $f_1, f_2, f_3\in \tscrP_3(\bbF_d)$ and  $\caT_1,\caT_2\in \scrT_\ns$, then $\kappa(f_1, \caT_1)=\kappa(f_2, \caT_2)$ iff $\eta_3(f_1)\eta_3(\caT_1)=\eta_3(f_2)\eta_3(\caT_2)$ thanks to \lsref{lem:kappaMagicOned1} and \ref{lem:kappaOrder}. In addition, if $f_1, f_2, f_3$ have pairwise different cubic characters, then the product $\kappa(f_1, \caT)\kappa(f_2, \caT)\kappa(f_3, \caT)$ is independent of $\caT$ for $\caT\in \scrT_\ns$. This observation is crucial to constructing good approximate 3-designs and exact 3-designs as we shall see in Secs.~\ref{sec:BalanceMagic} and \ref{sec:3designExact}.

Next, we  consider the properties of $\kappa(\Psi,\caT)$ and related functions for a general magic state $|\Psi\>$ in $\scrM_{n,k}(d)$. Thanks to \lref{lem:kappaTLUB},  we have $-1\leq \kappa(\Psi,\caT)\leq 1$, and the upper bound is saturated when $|\Psi\>$ is a stabilizer state (which means $k=0$) or $\caT\in \scrT_\sym$.  Suppose  $|\Psi\>$ has the form in \eref{eq:MagicStatenk}; by virtue of  \eqsref{eq:kappaPsiProd}{eq:kappaTdefMagicd1}, we can deduce that  
\begin{align}\label{eq:kappaMagicDecomd1}
\kappa(\Psi,\caT)=\prod_{i=1}^k\kappa(f_i,\caT)\quad \forall \caT\in \scrT_\ns,
\quad \kappa(\Psi,\caT)=\prod_{i=1}^k\kappa_0(f_i)
\quad \forall \caT\in \scrT_\defe.
\end{align}
Note that $\kappa(f_i,\caT)=\kappa_j(f_i)$ when $\eta_3(\caT)=\eta_3(3\nu^j)$. In conjunction with Eqs.~\eqref{eq:muj}, \eqref{eq:mujFormula}, and \pref{pro:hkaka} we can immediately deduce the following proposition. 

\begin{proposition}\label{pro:kahkaMagicd1}
	Suppose $d$ is an odd prime, $d=1\mmod 3$, $|\Psi\>\in \scrM_{n,k}(d)$ has the form in \eref{eq:MagicStatenk}, and $m\in \bbN$. Then 
	\begin{gather}
	\kappa(\Psi,\Sigma(d),m)-6	
	=\kappa(\Psi,\scrT_\ns,m)=\sum_{j=0}^2\mu_j\prod_{i=1}^k\kappa_j^m(f_i),\quad 	\kappa(\Psi,\scrT_\defe,m)=4\prod_{i=1}^k\kappa_0^m(f_i),		\\
	|\hka|(\Psi,\scrT_\ns)= \frac{D+2}{D-1}\sum_{j=0}^2\mu_j\left|\prod_{i=1}^k\kappa_j(f_i)-\frac{\kappa(\Psi,\Sigma(d))}{2D+2d}\right|, \\
	\hka(\Psi,\scrT_\defe)= \frac{4(D+2)}{D-1}\left[\prod_{i=1}^k\kappa_0(f_i)-\frac{\kappa(\Psi,\Sigma(d))}{2D+2d}\right],\quad |\hka|(\Psi,\scrT_\defe)=|\hka(\Psi,\scrT_\defe)|,
	\end{gather}	
	where $\mu_j$ and $\kappa_j(f_i)$ are determined by \eqsref{eq:mujFormula}{eq:kappaMagid1j}, respectively. 	
\end{proposition}
Thanks to \psref{pro:hkaka}, \ref{pro:kahkaMagicd1} and \eref{eq:kappaSigIsoDef}, it is also straightforward to calculate $\hka(\Psi,\scrT_\ns)$, $	\hka(\Psi,\Sigma(d))$,  $	|\hka|(\Psi,\Sigma(d))$, $\kappa(\Psi,\scrT_\iso)$, $\hka(\Psi,\scrT_\iso)$, and $|\hka|(\Psi,\scrT_\iso)$.  In several special cases, more explicit formulas on $\kappa(\Psi,\scrT_\ns)$ are presented in 
\lsref{lem:CubicStateSum} and \ref{lem:CubicState12Sum} in \aref{app:MagicOrbitd1Aux}.

By virtue of \lref{lem:kappaMagicOned1} and \eref{eq:kappaMagicDecomd1} we can further deduce the following lemma. 
\begin{lemma}\label{lem:kappaMagicLUBd1}
	Suppose $d$ is an odd prime, $d=1\mmod 3$, and	$|\Psi\>\in \scrM_{n,k}(d)$. Then
	\begin{gather}
	\frac{1}{9^kd^{2k}}\leq  \kappa(\Psi,\caT) \leq \left(\frac{4d-3}{d^2}\right)^k \leq \frac{4^k}{d^k}\quad \forall \caT\in \scrT_\ns,
	\quad  \kappa(\Psi,\caT)= 1 \quad    \forall \caT\in \scrT_\sym, \label{eq:kappaMagicLUBd1} \\	
	\frac{2(d-2)}{9^kd^{2k}}\leq  \kappa(\Psi,\Sigma(d))-6=\kappa(\Psi,\scrT_\ns) \leq  \frac{2^{2k+1}}{d^{k-1}}, \label{eq:kaNsSigMagicd1}\\   
	\frac{D+2}{D+d}\left[6+\frac{2(d-2)}{9^kd^{2k}}\right]\leq  \hka(\Psi,\Sigma(d))\leq |\hka|(\Psi,\Sigma(d)) \leq 6+ \frac{2^{2k+1}}{d^{k-1}}+\frac{14d}{D},\quad |\hka|(\Psi,\scrT_\ns)\leq \frac{2^{2k+1}}{d^{k-1}}+\frac{14d}{D}. 	\label{eq:hkaAbsNsSigMagicd1}
	\end{gather}
	If $\caT,\caT'\in \scrT_\ns$ have the same cubic character, then $\kappa(\Psi,\caT')=\kappa(\Psi,\caT)$. 
\end{lemma}

Thanks to \lref{lem:kappaMagicLUBd1}, any magic state $|\Psi\>$ in $\scrM_{n,k}(d)$ satisfies $\kappa(\Psi,\caT)\geq 0$, which agrees with \cref{con:kappaTLUB}, so    
$|\kappa|(\Psi,\scrT)=\kappa(\Psi,\scrT)$ for any subset $\scrT$ in $\Sigma(d)$. Notably, this is the case when $\scrT=\Sigma(d), \scrT_\ns, \scrT_\iso,\scrT_\defe$.
The following limits are  simple corollaries of 
\pref{pro:kahkaMagicd1} and \lref{lem:kappaMagicLUBd1} (cf. \lref{lem:CubicStateSum} in \aref{app:MagicOrbitd1Aux}),
\begin{gather}\label{eq:kahkaSigLimd1}
\begin{gathered}
\lim_{k\to \infty}|\hka|(\Psi,\Sigma(d))=\lim_{k\to \infty}\hka(\Psi,\Sigma(d))=\lim_{k\to \infty} \kappa(\Psi,\Sigma(d))=6,\\
\lim_{d\to \infty}|\hka|(\Psi,\Sigma(d))=\lim_{d\to \infty}\hka(\Psi,\Sigma(d))=\lim_{d\to \infty} \kappa(\Psi,\Sigma(d))=\begin{cases}
10 & n\geq 2,k=1,\\
6 & n\geq k\geq 2, 
\end{cases}\\
\lim_{d\to \infty} \kappa(\Psi,\Sigma(d))=10,\quad 
\lim_{d\to \infty}\hka(\Psi,\Sigma(d))=5,\quad n=k= 1. 
\end{gathered}
\end{gather}
Note that the specific choice of the magic state $|\Psi\>$ in $\scrM_{n,k}(d)$  does not affect these limits,  which is  analogous to the counterparts for the case $d=2\mmod 3$ as presented in \eref{eq:kahkaSigLimd23}.  On the other hand, $|\hka|(\Psi,\Sigma(d))$ has no limit when $n=k=1$ and $d\to \infty$, which is in sharp contrast with \eref{eq:kahkaSigLimd23}.

\Lsref{lem:kappaMagicOned1}-\ref{lem:kappaMagicLUBd1} are instructive to understanding the basic properties of magic orbits. Nevertheless, some bounds in \lref{lem:kappaMagicLUBd1} are suboptimal. 
By virtue of more sophisticated analysis we can derive
better bounds as presented in 
\lsref{lem:hkaAuxMagicd1} and \ref{lem:kahkaNsSigMagicSLUBd1} below. 
They are proved in \aref{app:kahkaUBd1} after we establish some auxiliary results in \aref{app:MagicOrbitd1Aux}.
\begin{lemma}\label{lem:hkaAuxMagicd1}
	Suppose $d$ is an odd prime, $d=1\mmod 3$, and	$|\Psi\>\in \scrM_{n,k}(d)$. Then 
	\begin{align}\label{eq:hkaSigNsAbsd1A}
	|\hka|(\Psi,\Sigma(d))\leq\kappa(\Sigma(d))+\frac{6d}{D},\quad |\hka|(\Psi,\scrT_\ns)\leq\kappa(\scrT_\ns)+\frac{6d}{D}.
	\end{align}
	If in addition $n\geq [2+(\log_9d)^{-1}]k+(\log_6d)^{-1}$,  then $0\leq \hka(\Psi,\caT)\leq \kappa(\Psi,\caT)$ for all $\caT\in \Sigma(d)$ and
	\begin{align}\label{eq:hkaSigNsAbsd1B}
	|\hka|(\Psi,\Sigma(d))= \hka(\Psi,\Sigma(d))=\frac{D+2}{D+d}\kappa(\Psi,\Sigma(d)), \quad |\hka|(\Psi,\scrT_\ns)=\hka(\Psi,\scrT_\ns)\leq \frac{D+2}{D+d}\kappa(\scrT_\ns).
	\end{align} 
\end{lemma}
Next, define
\begin{align}\label{eq:gammadk}
\gamma_{d,k}&:= \frac{4^k\gamma_k}{d^{k-1}},\qquad
\gamma_k:=\begin{cases}
9/8 &k=1, \\
15/16 &k= 2,\\
13/16 &k\geq 3.
\end{cases}	
\end{align}
The significance of $\gamma_{d,k}$ is manifested in \lref{lem:kahkaNsSigMagicSLUBd1} and \thsref{thm:Phi3MagicUBd1}-\ref{thm:AccurateDesignMagic} below.

\begin{lemma}\label{lem:kahkaNsSigMagicSLUBd1}
	Suppose $d$ is an odd prime, $d=1\mmod 3$, and $|\Psi\>\in \scrM_{n,k}(d)$ with $k\geq 1$. Then 
	\begin{align}
	\kappa(\Psi,\Sigma(d))-6=\kappa(\Psi,\scrT_\ns)&\leq 	\gamma_{d,k}\leq \frac{9}{8}\times\frac{4^k}{d^{k-1}}, \label{eq:kaNsSigUBd1} 
	\end{align}
	and the same bounds apply to $|\hka|(\Psi,\scrT_\ns)$ and $|\hka|(\Psi,\Sigma(d))-6$. If in addition $|\Psi\>\in \scrM^\id_{n,k}(d)$, then 
	\begin{align}\label{eq:kappaNsSigLB}
	\kappa^2(\Psi,\Sigma(d))\geq \frac{6(D+d)}{(D+j)}\kappa(\Psi,\Sigma(d),2),\quad 	d^k\kappa(\Psi,\scrT_\ns)\geq \begin{cases}
	2(d-2) & k=0,1,\\
	3\times 2^{k-1}(d-2) &k\geq 2,
	\end{cases}
	\end{align}	
	where $j=3$ if $n=k=1$ and $d=7$, while  $j=2$ otherwise. 		
\end{lemma}

\subsection{\label{sec:MagicOrbitKeyd1}Approximate 3-designs based on magic orbits }

Here we  clarify the deviations of a general magic orbit $\orb(\Psi)$ for	$|\Psi\>\in \scrM_{n,k}(d)$ with respect to three figures of merit, namely, the third normalized frame potential $\bar{\Phi}_3(\orb(\Psi))$, operator norm $\|\bar{Q}(\orb(\Psi))\|$ of the third normalized moment operator, and shadow norm $\|\orb(\Psi)\|_\sh$. Note that $\|\orb(\Psi)\|_\sh$ is approximately equal to $\|\Ob_0\|^2_\sh=\|\Ob_0\|^2_{\orb(\Psi)}$, where $\Ob$ is the projector onto a stabilizer state and $\Ob_0=\Ob-\tr(\Ob)\bbI/D$. \Thsref{thm:Phi3MagicUBd1}-\ref{thm:StabShNormMagicd1} below are proved in \aref{app:MagicKeyd1Proofs}.

By virtue of \thref{thm:Phi3LUB} and \pref{pro:kahkaMagicd1} we can calculate
the third normalized frame potential $\bar{\Phi}_3(\orb(\Psi))$ and  derive informative bounds as presented in the following theorem. 
\begin{theorem}\label{thm:Phi3MagicUBd1}
	Suppose $d$ is an odd prime, $d=1\mmod 3$, and $|\Psi\>\in \scrM_{n,k}(d)$ with $k\geq 1$. Then 
	\begin{align}\label{eq:Phi3MagicUBd1} \bar{\Phi}_3(\orb(\Psi))\leq  1+\min\left\{\frac{2}{11},\frac{\kappa(\Psi,\scrT_\ns,2)}{6}+\frac{3d}{D^2}\right\},\quad  \bar{\Phi}_3(\orb(\Psi))\leq 1+ \frac{\gamma_{d,2k}}{6}\leq  1+\frac{5}{32}\times\frac{16^k}{d^{2k-1}}. 
	\end{align}
\end{theorem}
As a simple corollary, \thref{thm:Phi3MagicUBd1} yields the following limits,
\begin{align}\label{eq:PhiLimd1}
\lim_{d\to \infty} \bar{\Phi}_3(\orb(\Psi))=\lim_{k\to \infty}\bar{\Phi}_3(\orb(\Psi))=1. 
\end{align} 
The deviation  $\bar{\Phi}_3(\orb(\Psi))-1$ vanishes quickly as $k$  and $d$ increase (assuming $k\geq 1$) as illustrated in \fsref{fig:Phi_Q_sn_k10} and \ref{fig:framepot}, where  $|\Psi\>$ belongs to the set $\scrM_{n,k}^\can(d)$ of canonical magic states defined in \sref{sec:MagicState}.  The two figures also cover the case with
$d \neq 1 \mmod 3$ for completeness, in which  $\bar{\Phi}_3(\orb(\Psi))$ is determined by \thref{thm:Phi3Magicd23} [see also \eref{eq:PhiLimd23}] and shows a similar behavior. 
Notably, $\bar{\Phi}_3(\orb(\Psi))-1$ 
deceases exponentially with $k$.  In addition,  $\bar{\Phi}_3(\orb(\Psi))$ tends to increase with $n$  initially, but levels off quickly when $n\geq 5$, as illustrated in the right plot in \fref{fig:framepot}.

\begin{figure}[tb]
	\centering 
	\includegraphics[width=0.95\textwidth]{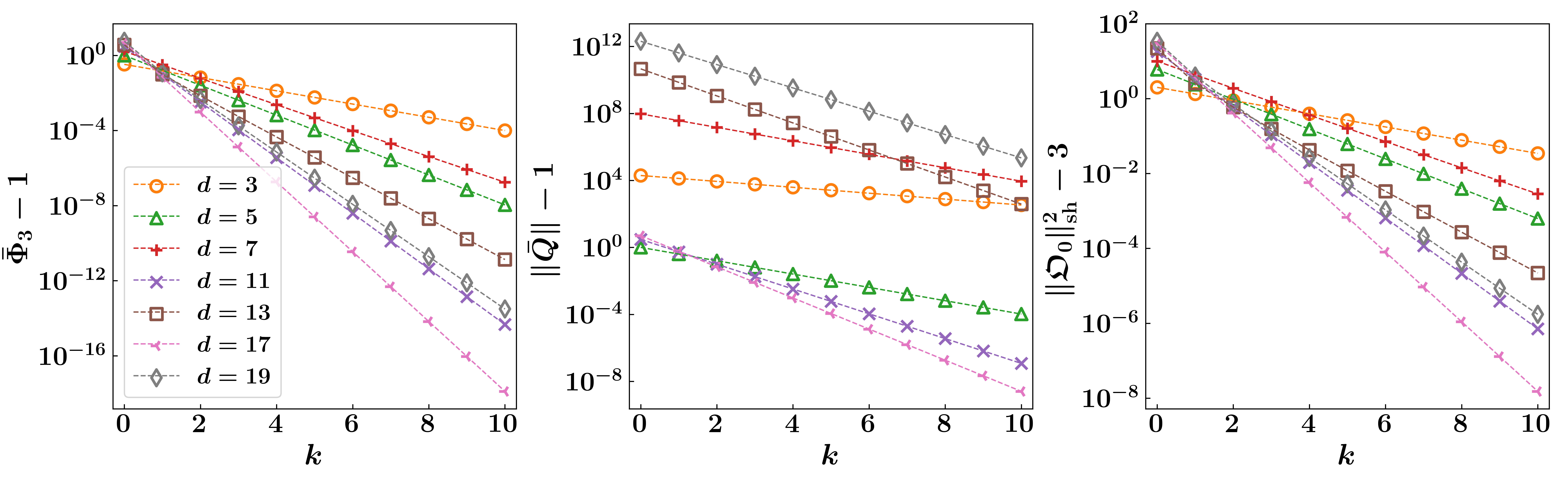}
	\caption{
		The deviations of $\bPhi_3(\orb(\Psi))$, $\|\bQ(\orb(\Psi))\|$, and $\|\Ob_0\|_\sh^2=\|\Ob_0\|_{\orb(\Psi)}^2$ associated with the canonical magic state $|\Psi\>$ in 
		$\scrM_{n,k}^\can(d)$. 	Here $n=10$ and $\Ob$ is the projector onto a stabilizer state. 	A  small deviation means the Clifford orbit  is close to a  3-design with respect to the corresponding figure of merit.	
	}\label{fig:Phi_Q_sn_k10}
\end{figure}

\begin{figure}[tb]
	\centering 
	\includegraphics[width=0.8\textwidth]{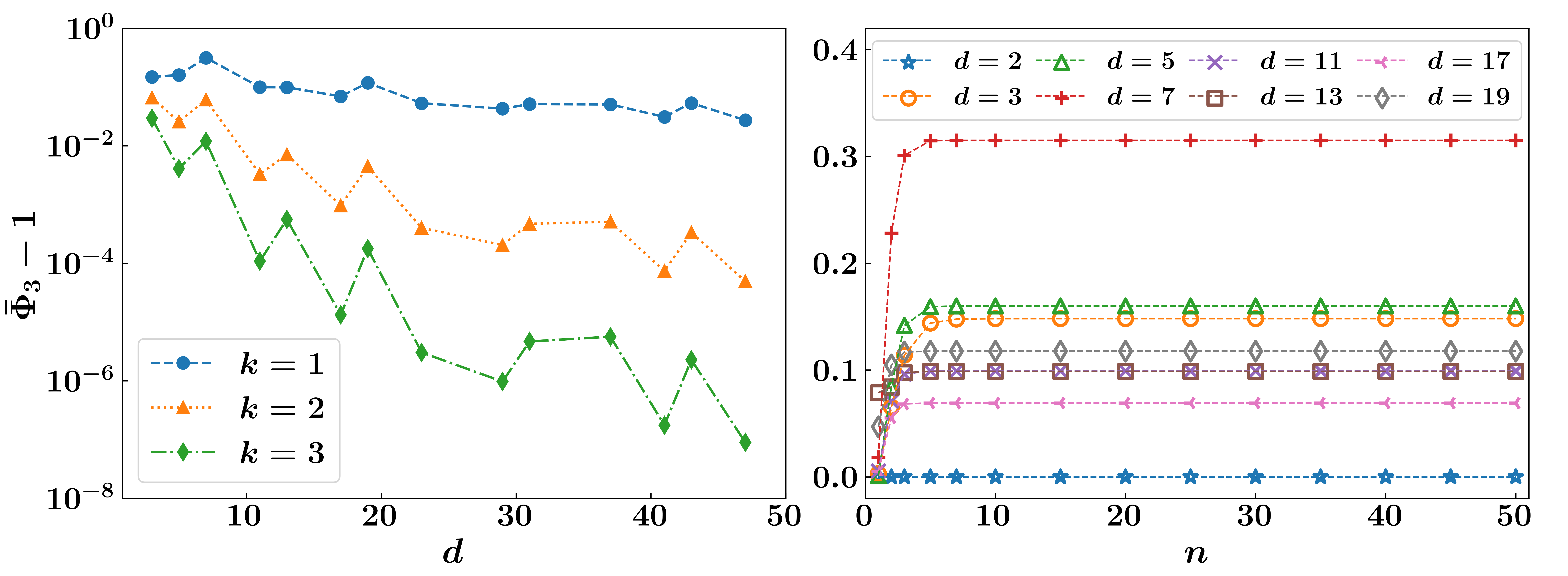}
	\caption{The deviation of  $\bPhi_3(\orb(\Psi))$ with $|\Psi\> \in \scrM_{n,k}^\can(d)$ as a function of $d$  and $n$.  Here  $n=50$ in the left plot and $k=1$ in the right plot. } 
	\label{fig:framepot}
\end{figure}

\begin{figure}[tb]
	\centering 
	\includegraphics[width=0.8\textwidth]{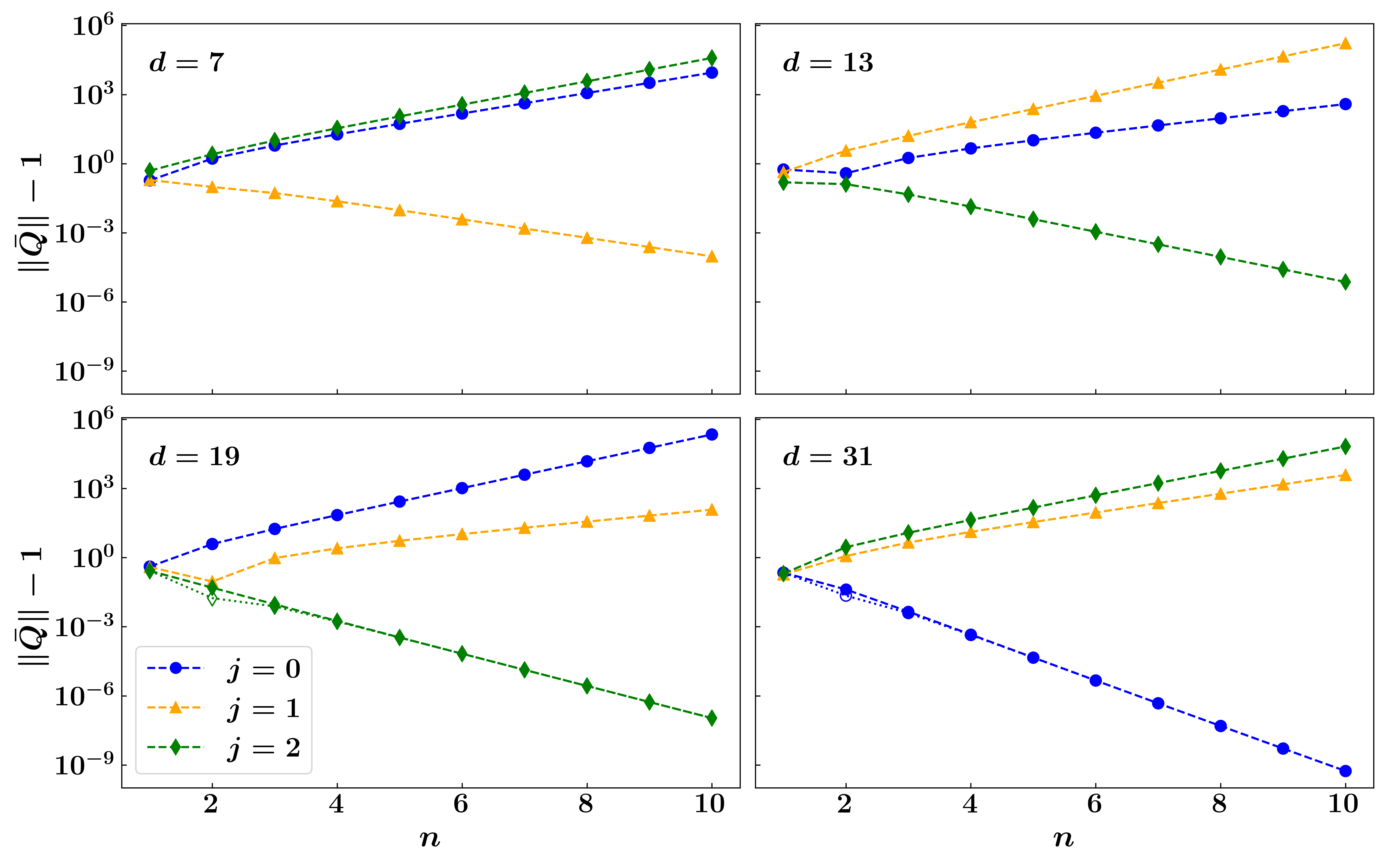}
	\caption{The operator norms  $\|\bQ(\orb(\Psi_j))\|$ associated with some magic states in $\scrM_{n,n}^\id(d)$, where $|\Psi_j\> =|\psi_{f_j}\>^{\otimes n}$ with $f_j(x)=\nu^jx^3$ and $j=0,1,2$. Results on four different local dimensions are shown in four plots, respectively. When $d=19$ and $j=2$, the figure shows an upper bound (solid mark) and lower bound (hollow mark) for $\|\bQ(\orb(\Psi_j))\|$ as presented in \eref{eq:MomentNormd13a} in \thref{thm:MomentNormd13}; the two bounds almost coincide although it is difficult to compute the exact value in this special case. The situation is similar when $d=31$ and $j=0$. 
	} 
	\label{fig:Qnormnu3}
\end{figure}

\begin{figure}[tb]
	\centering 
	\includegraphics[width=0.8\textwidth]{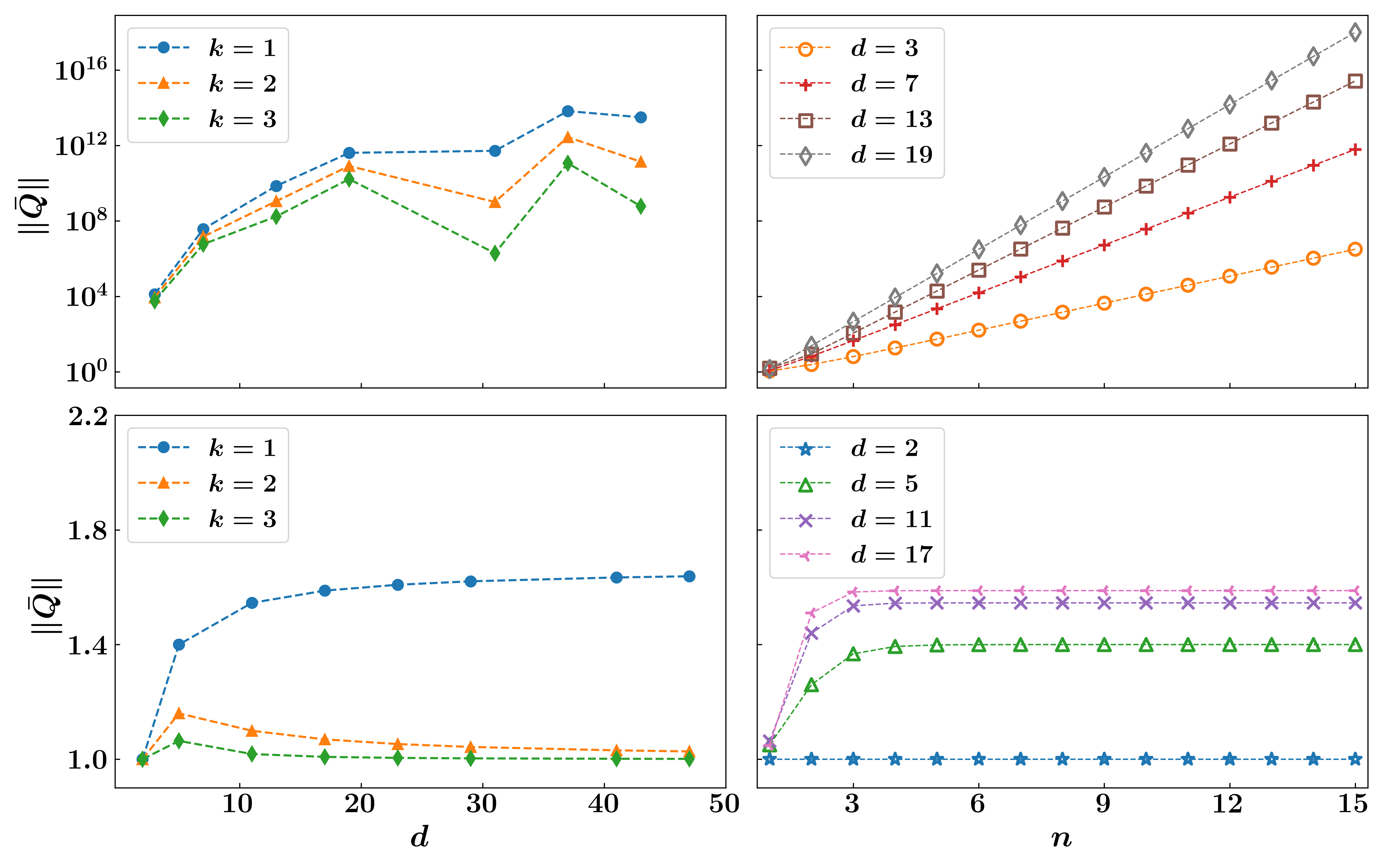}
	\caption{The operator norm  $\|\bQ(\orb(\Psi))\|$ with $|\Psi\> \in \scrM_{n,k}^\can(d)$ as a function of $d$  and $n$. Here $n=10$
		in the left column and  $k=1$ in the right column; $d$ is a prime with $d \neq 2 \mmod 3$  in the first row, while it is a prime with	$d=2 \mmod 3$  in the second row because the properties of  $\|\bQ(\orb(\Psi))\|$ in the two cases are dramatically different. 
	} 
	\label{fig:Qnorm}
\end{figure}
\begin{figure}[tb]
	\centering 
	\includegraphics[width=0.8\textwidth]{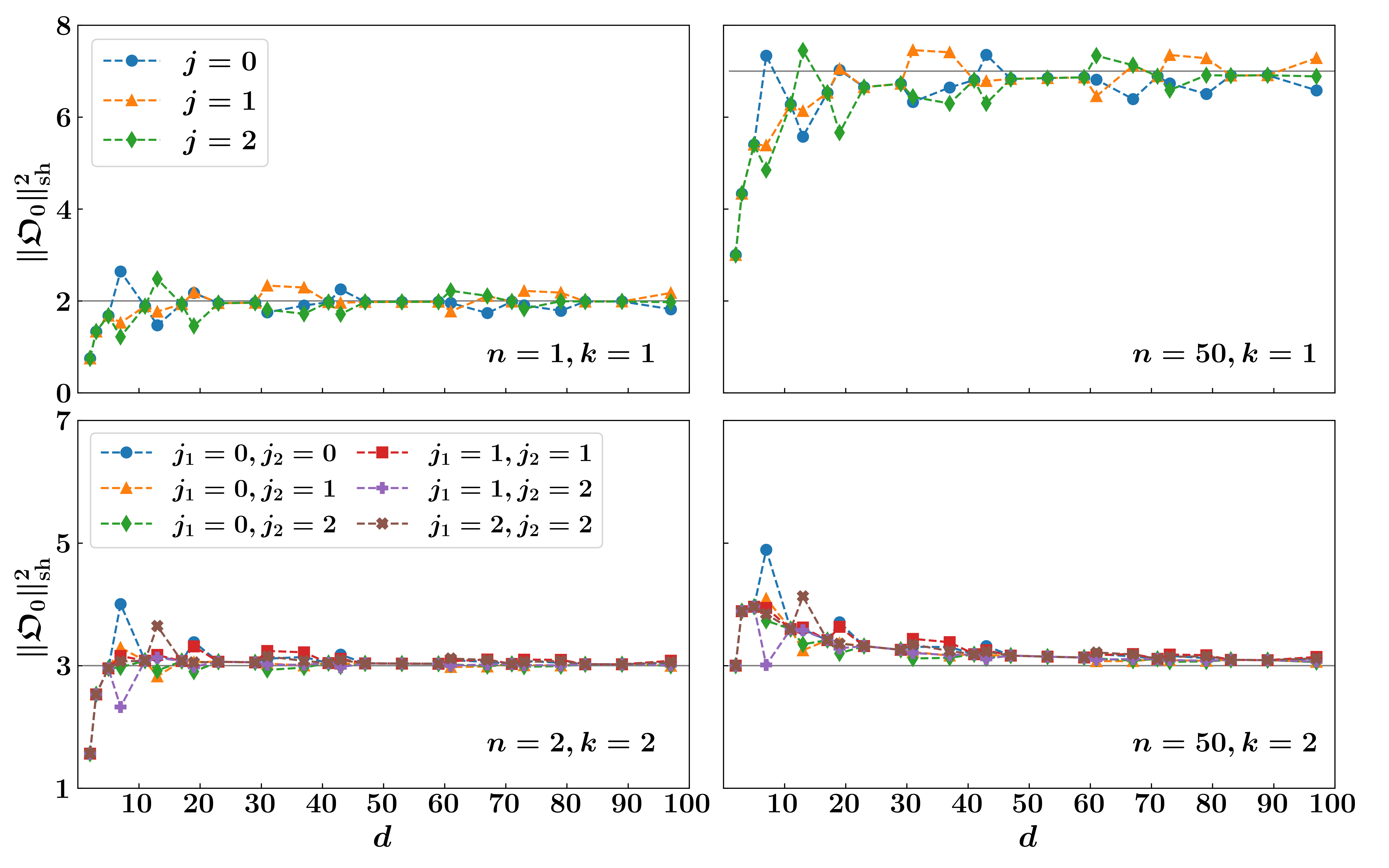}
	\caption{The squared shadow norm of $\Ob_0$ with respect to  $\orb(\Psi)$ for $|\Psi\> \in \scrM_{n,k}(d)$, where $\Ob$ is the projector onto a stabilizer state. Here the values of  $n,k$ are shown in each plot. In the first row, $k=1$ and the fiducial state $|\Psi\>$ is determined by the cubic function $f_j(x)=\nu^j x^3$ for  $j=0,1,2$, where $\nu$ is a primitive element of $\bbF_d$ as specified in \tref{tab:nu}. In the second row, $k=2$ and 
		$|\Psi\>$ is determined by the two cubic functions $f_{j_1}(x)$ and  $f_{j_2}(x)$. As a benchmark, the horizontal line in each plot denotes the asymptotic value of $\|\Ob_0\|_\sh^2$ when $d \to \infty$ as determined in \eqsref{eq:StabShNormMagicd23Lim}{eq:StabShNormMagicd1Lim}.  } 
	\label{fig:stabstate_k12}
\end{figure}

Next, by virtue of \thref{thm:MomentNormd13}, \psref{pro:hkaka}, \ref{pro:kahkaMagicd1}, and \eref{eq:kappaSigIsoDef} we can derive pretty good lower and upper bounds for  $\|\bQ(\orb(\Psi))\|$. Moreover,  an exact formula can be obtained when  $|\hka|(\Psi,\scrT_\iso)=\hka(\Psi,\scrT_\iso)$ or $|\hka|(\Psi,\scrT_\iso)\leq 2(D+2)\kappa(\caT_\defe)$, where $\caT_\defe$ is any stochastic Lagrangian subspace in $\scrT_\defe$. Note that the first condition holds
if  $n\geq [2+(\log_9d)^{-1}]k+(\log_6d)^{-1}$ by \lref{lem:hkaAuxMagicd1}. For small $d$ and $n$, $\|\bQ(\orb(\Psi))\|$ can also be calculated directly. In addition, we can derive simple explicit bounds for $\|\bQ(\orb(\Psi))\|$ as  presented in the following theorem, which are instructive to understanding the general trends. 
\begin{theorem}\label{thm:Moment3MagicLUBd1}
	Suppose $d$ is an odd prime, $d=1\mmod 3$, and $|\Psi\>\in \scrM_{n,k}(d)$ with $k\geq 1$. Then 
	\begin{gather}\label{eq:Moment3MagicLUBd1}
	\frac{d^{n-2k}}{3^{2k+1}}\leq \frac{D+2}{3(9d^2)^k}
	\leq \|\bQ(\orb(\Psi))\|\leq \max\left\{\frac{7}{4},\frac{4^kd^{n-k}}{3}\right\}. 
	\end{gather}
	If in addition $\hka(\Psi,\caT_\defe)\leq 0$ or $\kappa(\Psi,\caT_\defe)	\leq \kappa(\Psi,\Sigma(d))/[2(D+d)]$ for any stochastic Lagrangian subspace $\caT_\defe$  in $\scrT_\defe$, then 
	\begin{gather}\label{eq:Moment3MagicDefNegLUBd1}
	\frac{d^{n-2k}}{3^{2k+1}}
	\leq \|\bQ(\orb(\Psi))\|\leq 
	\frac{1}{6}|\hka|(\Psi,\Sigma(d))
	\leq  1+\frac{\gamma_{d,k}}{6}\leq 
	1+ \frac{3\times 4^{k-2}}{d^{k-1}}. 
	\end{gather}
\end{theorem}
\Thref{thm:Moment3MagicLUBd1} indicates that the norm $\|\bQ(\orb(\Psi))\|$ decreases exponentially with $k$, which is similar to the case  $d\neq1 \mmod 3$ determined  by \thsref{thm:Moment3Magicd2} and \ref{thm:Moment3Magicd3}, as illustrated in \fref{fig:Phi_Q_sn_k10}. However, there is a crucial difference depending on the congruence class of $d$ modulo 3. When $d=2\mmod 3$, the norm $\|\bQ(\orb(\Psi))\|$ approaches~1 exponentially with $k$. When $d\neq 2\mmod 3$, by contrast, $\|\bQ(\orb(\Psi))\|$ cannot approach 1 unless $k$ is  $\caO(n)$. Even in the case $k=n$, $\|\bQ(\orb(\Psi))\|$ may still increase exponentially with $n$ if the magic state $|\Psi\>$ is not chosen properly as illustrated in \fref{fig:Qnormnu3}; 
nevertheless,  $\|\bQ(\orb(\Psi))\|$ approaches  1 exponentially with $k=n$ if the magic state $|\Psi\>$ is chosen properly  (cf. \thref{thm:AccurateDesignMagic} and \fref{fig:two_orbits} in \sref{sec:MoreAccurateDesign}). The dependence of $\|\bQ(\orb(\Psi))\|$ on $d$ and $n$ (with $k$ fixed) is further illustrated in \fref{fig:Qnorm}, which  shows dramatically different behaviors depending on whether $d=2\mmod 3$ or $d\neq 2\mmod3$. In the first case,  $\|\bQ(\orb(\Psi))\|$ tends to level off quickly as $d$ and $n$ increases; notably, the deviation $\|\bQ(\orb(\Psi))\|-1$ vanishes quickly
as $d$ increases whenever $k\geq 2$. In the second case, by contrast, $\|\bQ(\orb(\Psi))\|$ tends to increase exponentially with $n$ and oscillates quite randomly with $d$.

Next, by virtue of \thref{thm:StabShNormGen} and \psref{pro:hkaka}, \ref{pro:kahkaMagicd1} we can calculate the shadow norm $\|\Ob_0\|^2_{\orb(\Psi)}$ associated with any stabilizer projector $\Ob$. 
When  the projector $\Ob$  has rank 1 and thus corresponds to a stabilizer state, we can derive more explicit formulas as presented in the following theorem. 
\begin{theorem}\label{thm:StabShNormMagicd1}
	Suppose $d$ is an odd prime, $d=1\mmod 3$, $|\Psi\>\in \scrM_{n,k}(d)$ with $k\geq 1$, and $\Ob$ is the projector onto a stabilizer state in $\caH_d^{\otimes n}$. Then 
	\begin{gather}
	\|\Ob_0\|^2_{\orb(\Psi)}=\frac{D+1}{D+2}\hka(\Psi,\Sigma(d))-\frac{(D+1)(3D-1)}{D^2}=
	\frac{(D+1)(3D^2-3dD+D+d)}{D^2(D+d)}+\frac{D+1}{D+d}\kappa(\Psi,\scrT_\ns), \label{eq:StabShNormMagicd1}\\
	\|\Ob_0\|^2_{\orb(\Psi)}< \frac{\|\Ob_0\|^2_{\orb(\Psi)}}{\|\Ob_0\|_2^2}\leq 3+\kappa(\Psi,\scrT_\ns)\leq 3+\gamma_{d,k}\leq 3+\frac{9}{8}\times\frac{4^k}{d^{k-1}}.
	\label{eq:StabShNormMagicd1UB}
	\end{gather}	
\end{theorem}
As a simple corollary of \thref{thm:StabShNormMagicd1} and \eref{eq:kahkaSigLimd1} we can deduce the following limits [cf. \eref{eq:StabShNormMagicd23Lim}],
\begin{align}
\lim_{n\to \infty} 	\|\Ob_0\|^2_{\orb(\Psi)}=3+\kappa(\Psi,\scrT_\ns),
\quad \lim_{n,k\to \infty} 	\|\Ob_0\|^2_{\orb(\Psi)}=3,\quad 
\lim_{d\to \infty} 	\|\Ob_0\|^2_{\orb(\Psi)}=\begin{cases}
2 & n=k=1,\\
7 & n\geq 2, k=1, \\
3 &n\geq k\geq 2.
\end{cases}	  \label{eq:StabShNormMagicd1Lim}
\end{align}
The shadow norm $\|\Ob_0\|^2_{\orb(\Psi)}$ 
levels off quickly as $k$, $d$, or $n$ increases (assuming $k\geq 1$),  which is similar to the case  $d\neq1 \mmod 3$ determined  by \thref{thm:StabShNormMagicd23} and \eref{eq:Stab1ShNormMagicd23}, as illustrated in \fsref{fig:Phi_Q_sn_k10} and~\ref{fig:stabstate_k12}. Notably, the deviation
$\|\Ob_0\|^2_{\orb(\Psi)}-3$  decreases exponentially with $k$ and vanishes quickly.

Finally, we offer informative bounds for the shadow norm of a general traceless observable and the shadow norm $\|\orb(\Psi)\|_\sh$  itself.

\begin{theorem}\label{thm:ShNormMagicd1}
	Suppose $d$ is an odd prime, $d=1\mmod 3$,  $|\Psi\>\in \scrM_{n,k}(d)$ with $k\geq 1$, and $\Ob\in \caL_0\bigl(\caH_d^{\otimes n}\bigr)$. Then 
	\begin{gather}
	\|\Ob\|^2_{\orb(\Psi)}\leq	
	[1+|\hka|(\scrT_\ns)  ]\|\Ob\|_2^2+2\|\Ob\|^2	
	\leq(1 +\gamma_{d,k})\|\Ob\|_2^2+2\|\Ob\|^2,  \\
	\hka(\Psi,\Sigma(d))-3-\frac{5}{D}\leq \|\orb(\Psi)\|_\sh\leq 3+|\hka|(\scrT_\ns)\leq  3+\gamma_{d,k}\leq 3+\frac{9}{8}\times\frac{4^k}{d^{k-1}}. 
	\end{gather}
\end{theorem}
\Thref{thm:ShNormMagicd1} is a simple corollary of \thref{thm:ShNormGen}  and \lref{lem:kahkaNsSigMagicSLUBd1}.  
Thanks to \thref{thm:ShNormMagicd1} and \lref{lem:kappaMagicLUBd1}, the shadow norm $\|\orb(\Psi)\|_\sh$ converges to 3 exponentially as $k$ increases; if $k\geq 2$, then  $\|\orb(\Psi)\|_\sh$ converges to 3 as $d$ increases too.

\subsection{\label{sec:MoreAccurateDesign}Construction of more accurate approximate 3-designs}
Here we propose a simple recipe for constructing accurate approximate 3-designs with respect to three figures of merit, including the operator norm of the third normalized moment operator.
Only two Clifford orbits are required to apply this recipe.

Let $\scrE$ be an ensemble of pure states in $\caH_d^{\otimes n}$
and define
\begin{align}\label{eq:kappaMin}
\kappa_{\min}(\scrE):=\min_{\caT\in \Sigma(d)} \kappa(\scrE,\caT),\quad \hka_{\min}(\scrE):=\min_{\caT\in \Sigma(d)} \hka(\scrE,\caT). 
\end{align}
Let $\caT_\defe$ be any given stochastic Lagrangian subspace in $\scrT_\defe$; then $\kappa(\scrE, \caT)=\kappa(\scrE, \caT_\defe)$ for all $\caT\in \scrT_\defe$ by \pref{pro:kappaSym} and \lref{lem:SemigroupT}.
Suppose $\scrE$ satisfies the following two equivalent conditions [cf. \pref{pro:hkaka}],
\begin{align}\label{eq:DesignFiducialCon}
\kappa_{\min}(\scrE)=\kappa(\scrE, \caT_\defe)\leq \frac{\kappa(\scrE,\Sigma(d))}{2(D+d)},\quad \min_{\caT\in \Sigma(d)}\hka(\scrE, \caT)= \hka(\scrE, \caT_\defe)\leq 0.
\end{align}
Let $\scrE_p$ be the ensemble constructed by choosing the ensembles 
$\scrE$ and $\{|0\>^{\otimes n}\}$ with probabilities $1-p$ and $p$
respectively. 
Then the normalized moment operator of $\orb(\scrE_p)$ reads
\begin{equation}
\begin{gathered}
\bQ(\orb(\scrE_p))=(1-p)\bQ(\orb(\scrE))+p\bQ(n,d,3)=\frac{1}{6}\sum_{\caT\in \Sigma(d)} 
\left[\hka(\scrE_p,\caT)R(\caT)\right],\\
\hka(\scrE_p,\caT):=(1-p)\hka(\scrE,\caT)+\frac{p(D+2)}{D+d},
\end{gathered}
\end{equation}
where $\bQ(n,d,3)$ is the third normalized moment operator of stabilizer states. 
Here $\hka(\scrE_p,\caT)$ enjoys the same symmetry of $\kappa(\Psi,\caT)$ as clarified in \pref{pro:kappaSym}. Notably, we have $\hka(\scrE_p,\caT)=\hka(\scrE_p,\caT_\defe)$ for all $\caT\in \scrT_\defe$.

Next, suppose $p$ satisfies the following equation,
\begin{align}\label{eq:designProb}
\hka(\scrE_p, \caT_\defe)=(1-p)\hka(\scrE, \caT_\defe)+\frac{D+2}{D+d}p=0.
\end{align}
Then
\begin{align}
\hka(\scrE_p,\caT)=\frac{(1-p)(D+2)}{D-1}\kappa'(\scrE,\caT)\leq \frac{D+2}{D-1}\kappa'(\scrE,\caT),   \quad  \kappa'(\scrE,\caT):=\kappa(\scrE,\caT)-\kappa_{\min}(\scrE)\geq 0,
\label{eq:hkapDesign}
\end{align}
and the second inequality is saturated iff $\kappa(\scrE,\caT)=\kappa(\scrE, \caT_\defe)=\kappa_{\min}(\scrE)$. Notably, $\hka(\scrE_p,\caT)=\kappa'(\scrE,\caT)=0$ for all $\caT\in \scrT_\defe$. Define
\begin{align}\label{eq:kappaNsP}
\kappa'(\scrE,\scrT_\ns):=\sum_{\caT\in \scrT_\ns}\kappa'(\scrE,\caT)=
\sum_{\caT\in \scrT_\ns} [\kappa(\scrE,\caT)-\kappa_{\min}(\scrE)]
=\kappa(\scrE,\scrT_\ns)-2(d-2)\kappa_{\min}(\scrE),
\end{align}
where the second equality holds because $|\scrT_\ns|=2(d-2)$. Now the following proposition is a simple corollary of \thref{thm:MomentNormd13}.
\begin{proposition}\label{pro:AccurateDesign}
	Suppose $d$ is an odd prime, $d=1\mmod 3$,  $\scrE$ is an ensemble of pure states in $\caH_d^{\otimes n}$ that satisfies the condition in \eref{eq:DesignFiducialCon}, and $p$ satisfies the condition in \eref{eq:designProb}. Then 
	\begin{align}
	\|\bQ(\orb(\scrE_p))\|&=1+\frac{D-1}{6(D+2)}\hka(\scrE_p,\scrT_\ns)=1+\frac{(1-p)}{6}\kappa'(\scrE,\scrT_\ns)\leq 1+\frac{1}{6}\kappa'(\scrE,\scrT_\ns), \label{eq:QscrEp1}\\
	\|\bQ(\orb(\scrE_p))-P_\sym\| 
	&\leq \frac{D+5}{6(D+2)}\hka(\scrE_p,\scrT_\ns)
	=\frac{(1-p)(D+5)}{6(D-1)}\kappa'(\scrE,\scrT_\ns)
	\leq \frac{D+5}{6(D-1)}\kappa'(\scrE,\scrT_\ns).  \label{eq:QscrEp2}
	\end{align}
\end{proposition}
\Pref{pro:AccurateDesign} implies the following relations,
\begin{align}\label{eq:AccurateDevRel}
\|\bQ(\orb(\scrE_p))\|-1\leq \|\bQ(\orb(\scrE_p))-P_\sym\| \leq \frac{D+5}{D-1}(\|\bQ(\orb(\scrE_p))\|-1)\leq  2(\|\bQ(\orb(\scrE_p))\|-1),
\end{align}
where the last inequality holds because $D\geq 7$.

\begin{theorem}\label{thm:AccurateDesignMagic}
	Suppose  the ensemble $\scrE$  in \pref{pro:AccurateDesign} is built on $\scrM_{n,k}(d)$ with $k\geq 1$. Then	$p\leq 21/(4D)$ and
	\begin{align}
	\|\bQ(\orb(\scrE_p))\|&\leq  1+\frac{\gamma_{d,k}}{6}\leq 
	1+ \frac{3\times 4^{k-2}}{d^{k-1}}, \label{eq:QscrEpMagic1} \\
	\|\bQ(\orb(\scrE_p))-P_\sym\| 
	&\leq \frac{(D+5)\gamma_{d,k}}{6(D-1)}
	\leq  
	\frac{\gamma_{d,k}}{3}\leq 
	\frac{6\times 4^{k-2}}{d^{k-1}}. \label{eq:QscrEpMagic2}
	\end{align}
\end{theorem}

\begin{figure}[tb]
	\centering 
	\includegraphics[width=0.8\textwidth]{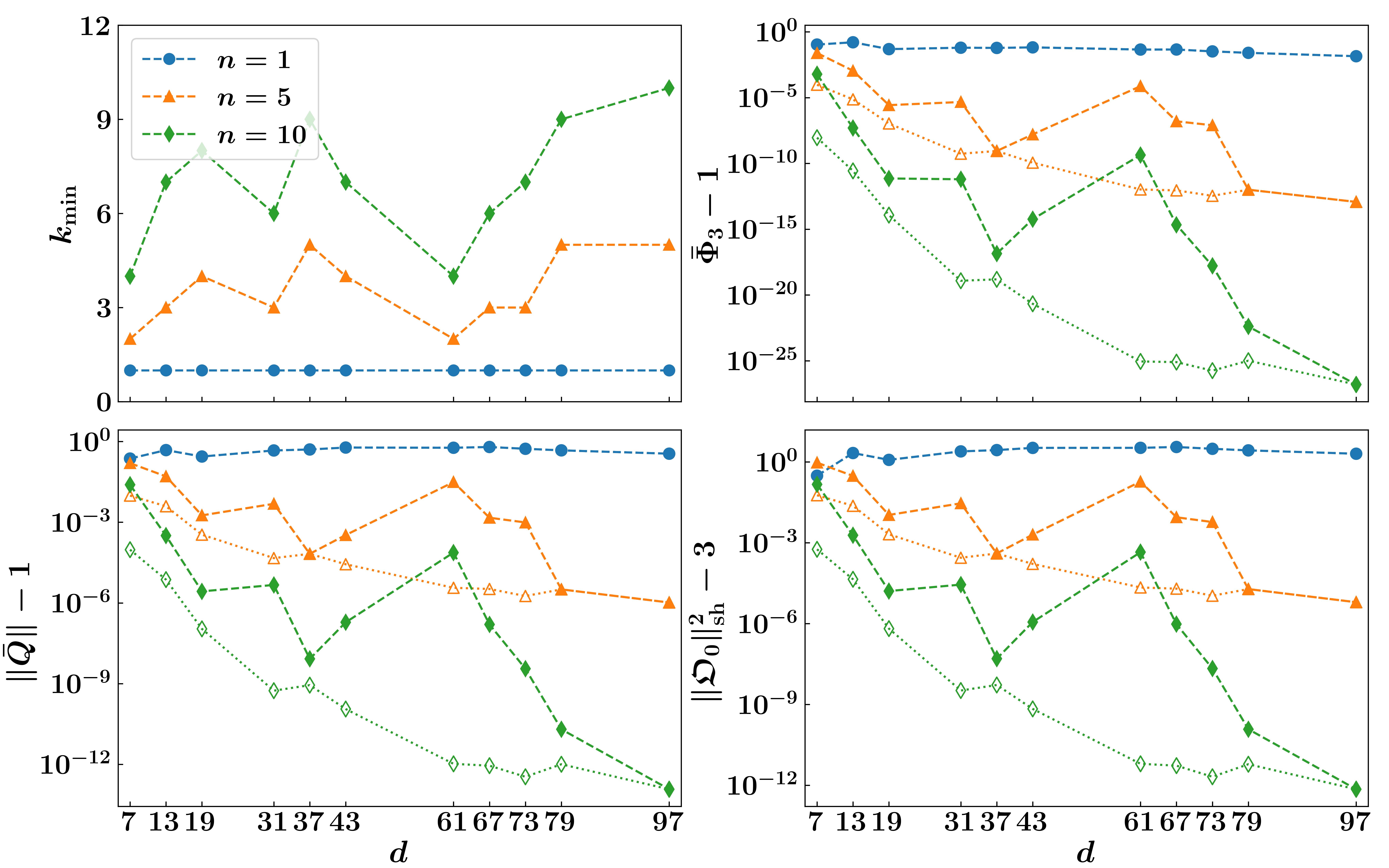}
	\caption{The deviations of an approximate 3-design $\scrE_p$ composed of two Clifford orbits with respect to three figures of merit. Here the local dimension $d$ is a prime with $d=1\mmod 3$ and $d<100$, $\scrE$ is composed of a single magic state of the form $|\Psi_j\>=|\psi_{f_j}\>^{\otimes  k}\otimes|0\>^{\otimes (n-k)}$ with $f_j(x)=\nu^j x^3$ and $\kappa(f_j, \caT_\defe)<1/d$, and $p$ is determined by \eref{eq:designProb}.
		The upper left plot shows the  minimum value $k_{\min}$ of $k$ such that $\hka(\Psi_j,\caT_\defe)\leq 0$ [cf. \eqsref{eq:DesignFiducialCon}{eq:DesignFiducialConPsij}]. The other three plots show the deviations of 
		$\bPhi_3(\orb(\scrE_p))$ (upper right), $\|\bQ(\orb(\scrE_p))\|$ (lower left), and $\|\Ob_0\|_\sh^2=\|\Ob_0\|_{\orb(\scrE_p)}^2$ (lower right), respectively, where $\Ob$ is the projector onto a stabilizer state;  the curves in the three plots have similar shapes, which means the three figures of merit are correlated with each other.   
		Two cases $k=k_{\min}$ and $k=n$ are labeled with solid markers  and hollow markers, respectively; when $n=1$, we have $k_{\min}=n=1$, so the two cases coincide. Note that $\|\bQ(\orb(\scrE_p))-P_\sym\| \approx \|\bQ(\orb(\scrE_p))\|-1$ by \pref{pro:AccurateDesign} and \eref{eq:AccurateDevRel}. }
	\label{fig:two_orbits}
\end{figure}

\Thref{thm:AccurateDesignMagic} follows from \pref{pro:AccurateDesign} as shown in 
\aref{app:thm:AccurateDesignMagicProof}. Note that the third normalized frame potential and shadow norm of $\scrE_p$ are comparable to the counterparts of $\scrE$ since $p=\caO(1/D)$.  
Therefore, an accurate approximate 3-design with respect to all three figures of merit can be constructed as illustrated in \fref{fig:two_orbits} if an ensemble $\scrE$ that satisfies the conditions in \eref{eq:DesignFiducialCon} can be constructed from $\scrM_{n,k}(d)$, which is the case.

Let $f_j(x)=\nu^j x^3$ for $j=0,1,2$ be three cubic polynomials on $\bbF_d$. Then $\min_{j=0,1,2} \kappa(f_j, \caT) <1/d$ for all $\caT\in \scrT_\ns$ thanks to \lsref{lem:kappaMagicOned1} and \ref{lem:kappaOrder}.
Choose a  cubic polynomial $f_j$ such that $\kappa(f_j, \caT_\defe)<1/d$, then   $\kappa(f_j, \caT_\defe)=\kappa_{\min}(f_j):=\min_{\caT\in \scrT_\ns}\kappa(f_j, \caT)$. 
Let $|\Psi_j\>=|\psi_{f_j}\>^{\otimes  k}\otimes|0\>^{\otimes (n-k)}$; then $|\Psi_j\>\in \scrM_{n,k}(d)$ and
\begin{align}
\kappa_{\min}(\Psi_j):=\min_{\caT\in \scrT_\ns} \kappa(\Psi_j,\caT)=\kappa(\Psi_j, \caT_\defe)=[\kappa(f_j, \caT_\defe)]^k.
\end{align}
In addition, the ratio $\kappa(\Psi_j, \caT_\defe)/\kappa(\Psi_j, \Sigma(d))$ decreases monotonically with $k$, and we have
\begin{align}\label{eq:DesignFiducialConPsij}
\kappa(\Psi_j, \caT_\defe)\leq \frac{1}{D}\leq \frac{\kappa(\Psi_j,\Sigma(d))}{2(D+d)}
\end{align}
when $k=n$ given that $\kappa(\Psi_j,\Sigma(d))\geq 6$ by \lref{lem:kappaMagicLUBd1}. Let $k_{\min}$ be the minimum value of $k$ (depending on $d$ and $n$) such that \eref{eq:DesignFiducialConPsij} holds; then  the ensemble $\scrE$ composed of the single state $|\Psi_j\>$  satisfies the conditions in \eref{eq:DesignFiducialCon} as long as  $k_{\min}\leq k\leq n$. The resulting ensemble $\orb(\scrE_p)$, which is composed of only two Clifford orbits, is usually  an accurate approximate 3-design with respect to all three figures of merit simultaneously, and the accuracy tends to increase quickly with  $n$ and $k$, as illustrated in \fref{fig:two_orbits}.

\subsection{\label{sec:BalanceMagic}Approximate 3-designs from balanced magic ensembles}

\subsubsection{Balanced  magic ensembles} 
Here we generalize the concept of $k$-balanced magic states introduced in \sref{sec:MagicState} to ensembles and clarify their basic properties. Suppose $\scrE=\{|\Psi_j\>,p_j\}_j$ is an ensemble of magic states constructed from $\scrM_{n,k}(d)$ in which $|\Psi_j\>\in \scrM_{n,k}(d)$ is chosen with probability $p_j$. Then $\scrE$ is $k$-balanced if 
\begin{align}
|\Psi_j\>\in \scrM_{n,k}^\QB(d)\quad \forall j,\quad 3\sum_j p_j\scrC_l(\Psi_j)=k\quad \forall l=0,1,2.
\end{align}
In this case, $\scrE$ is indeed balanced according to the definition in \sref{sec:Balance} as we shall see shortly. Let $f_j(x)=\nu^j x^3$ for $j=0,1,2$ be three cubic polynomials on $\bbF_d$ and let 
\begin{align}\label{eq:Psij}
|\tPsi\>&=\bigl(|\psi_{f_0}\> \otimes |\psi_{f_1}\> \otimes |\psi_{f_2}\>\bigr )^{\otimes \lfloor k/3\rfloor}\otimes |0\>^{\otimes (n-k)},\quad  
|\Psi_j\>=\begin{cases}
|\psi_{f_j}\> \otimes |\tPsi\> &k=1 \mmod 3, \\
|\psi_{f_{j+1}}\> \otimes |\psi_{f_{j+2}}\> \otimes|\tPsi\>  &k=2 \mmod 3, \\
\end{cases}
\end{align}
where the addition in each subscript is modulo 3. Then $|\tPsi\>\in \scrM_{n,k}^\rmB(d)$  if $k=0\mmod 3$ and $|\Psi_j\>\in \scrM_{n,k}^\QB(d)$ for $j=0,1,2$ if $k=1,2\mmod 3$. In addition, a $k$-balanced magic ensemble can be constructed as follows, 
\begin{align}\label{eq:kbalanceEx}
\scrE&=\begin{cases}
\{|\tPsi\>\} & k=0 \mmod 3,\\
\{|\Psi_0\>, |\Psi_1\>, |\Psi_2\>\} & k\neq 0 \mmod 3.
\end{cases}
\end{align}

Let  $s_k$ be the remainder of $k$ modulo 3 and define
\begin{align}\label{eq:kappatgammadk}
\kappa_{d,k}:=\frac{(s_k+1)\tg(d)^{2\lfloor k/3\rfloor}}{d^{k+\lfloor k/3\rfloor}},\quad 
\gamma^\rmB_{d,k}:=\frac{4^{(k+2)/3}(d-2)}{d^k},		
\end{align}
where $\tg(d)$ is defined in \eref{eq:tgd}.  \Lsref{lem:kahkaBalMagicd1} and \ref{lem:kahkaQBalMagicd1} below are
proved in \aref{app:BalanceMagicProof}.

\begin{lemma}\label{lem:kahkaBalMagicd1}
	Suppose $d$ is an odd prime, $d=1\mmod 3$,  $\scrE$ is a $k$-balanced magic ensemble 
	on  $\scrM_{n,k}(d)$ with $k\geq 1$, and $\caT\in \scrT_\ns$. Then 
	\begin{gather}
	\frac{s_k+1}{d^{k+\lfloor k/3\rfloor}}\leq 	\kappa(\scrE,\caT)=\kappa_{d,k}  \leq \frac{s_k+1}{d^k}\left(\frac{4d-27}{d}\right)^{\lfloor k/3\rfloor}\leq \frac{4^{(2k+1)/6}}{d^k}, \label{eq:kaBalance} \\
	\hka(\scrE,\caT)=\frac{(D+2)^2}{(D-1)(D+d)}\left[\kappa(\scrE,\caT)-\frac{3}{D+2}\right],\quad |\hka(\scrE,\caT)|\leq \frac{4^{(2k+1)/6}(D+2)}{d^k(D+d)},
	\label{eq:hkaBalMagic}\\
	\kappa(\scrE,\scrT_\ns)	=2(d-2)\kappa_{d,k}\leq \gamma^\rmB_{d,k},\quad  	
	|\hka|(\scrE,\scrT_\ns)=2(d-2)|	\hka(\scrE,\caT)|
	\leq \frac{(D+2)\gamma^\rmB_{d,k}}{D+d}. \label{eq:kahkaNsAbsBalMagic}	
	\end{gather}
	If in addition $k=1,2$ or  $k<3n/4$, then 
	\begin{align}
	|\hka(\scrE,\caT)|\leq \frac{D+2}{D+d}	\kappa(\scrE,\caT),\quad |\hka|(\scrE,\scrT_\ns)\leq \frac{D+2}{D+d}	\kappa(\scrE,\scrT_\ns). \label{eq:hkaNsAbsBalMagicUB}
	\end{align}
\end{lemma}

Based on \lref{lem:kahkaBalMagicd1} we can introduce a special subset of $k$-quasi-balanced magic states,  \begin{align}\label{eq:MagicQBset}
\tscrM_{n,k}^\QB(d):=\bigl\{|\Psi\>\in \scrM_{nk}^\QB(d) \,:\, \kappa(\Psi,\scrT_\ns)\leq 2(d-2)\kappa_{d,k} \bigr\}. 
\end{align}
If $k=0 \mmod 3$, then the state $|\tPsi\>$ defined in \eref{eq:Psij} belongs to $\tscrM_{n,k}^\QB(d)$; if instead $k\neq 0 \mmod 3$, then at least one of the three states $|\Psi_0\>, |\Psi_1\>, |\Psi_2\>$ belongs to $\tscrM_{n,k}^\QB(d)$, given that $\{|\Psi_j\>\}_{j=0}^2$ is a $k$-balanced magic ensemble  on $\scrM_{n,k}(d)$, which means $\sum_{j=0,1,2}\kappa(\Psi_j,\scrT_\ns)/3=2(d-2)\kappa_{d,k}$.

\begin{lemma}\label{lem:kahkaQBalMagicd1}
	Suppose $d$ is an odd prime, $d=1\mmod 3$,  $|\Psi\>\in \tscrM_{n,k}^\QB(d)$ with $k\geq 1$, and  $\gamma^\rmB_{d,k,n}:=\gamma^\rmB_{d,k}+(6d/D)$. Then 
	\begin{gather}	
	|\hka|(\Psi,\scrT_\ns)\leq 	\kappa(\Psi,\scrT_\ns)+\frac{6d}{D}	
	\leq 2(d-2)\kappa_{d,k}+\frac{6d}{D} \leq \gamma^\rmB_{d,k,n}\leq \frac{4^{(k+2)/3}}{d^{k-1}}+\frac{6d}{D}. \label{eq:kahkaNsAbsQBalMagicA}
	\end{gather}
	If in addition $n\geq(4k+8)/3$, then 
	\begin{align}\label{eq:kahkaNsAbsQBalMagicB}
	|\hka|(\Psi,\scrT_\ns)=\hka(\Psi,\scrT_\ns)\leq 	\frac{D+2}{D+d}\kappa(\Psi,\scrT_\ns)\leq 
	\frac{2(d-2) (D+2)\kappa_{d,k}}{D+d}\leq \frac{(D+2)\gamma^\rmB_{d,k}}{D+d}. 
	\end{align} 	
\end{lemma}

\subsubsection{Deviations from 3-designs} 

By virtue of \lref{lem:kahkaBalMagicd1} we can clarify the deviation of $\orb(\scrE)$ from a 3-design
with respect to three main figures of merit (cf. \thsref{thm:Phi3MagicUBd1}-\ref{thm:ShNormMagicd1} in \sref{sec:MagicOrbitKeyd1}). The following theorem is 
a simple corollary of \thref{thm:Phi3Balance} and   \lref{lem:kahkaBalMagicd1}. 
\begin{theorem}
	Suppose $d$ is an odd prime, $d=1\mmod 3$, and $\scrE$ is a $k$-balanced magic ensemble 
	on  $\scrM_{n,k}(d)$ with $k\geq 1$. Then
	\begin{align} 
	\bar{\Phi}_3(\orb(\scrE))-1= \frac{(d-2)(D+2)^2}{3(D-1)(D+d)}\left[\kappa(\scrE)-\frac{3}{D+2}\right]^2	\leq 
	\frac{4^{(2k+1)/3}(d-2)(D-1)}{3d^{2k}(D+d)}
	\leq \frac{4^{(2k+1)/3}}{3d^{2k-1}}. 	
	\end{align}			
\end{theorem}
In the special cases $n=k=1,2$, we can derive  more explicit formulas for $\bar{\Phi}_3(\orb(\scrE))$, given that $\kappa(\scrE)=2/d$ when $k=1$ and $\kappa(\scrE)=3/d^2$ when $k=2$, 
\begin{align}\label{eq:Phi3d1Balancen12}
\bar{\Phi}_3(\orb(\scrE))=1+
\begin{cases}
\frac{(d-2)(d-4)^2}{6d^3(d-1)} & n=k=1, \\[1ex]
\frac{12(d-2)}{d^5(d+1)^2(d-1)}& n=k=2. 
\end{cases}
\end{align}

Next, we consider 
the operator norms of $\bQ(\orb(\scrE))$ and $\bQ(\orb(\scrE))-P_\sym$. 
\begin{theorem}\label{thm:Moment3BalMagicLUBd1}
	Suppose $d$ is an odd prime, $d=1\mmod 3$, and $\scrE$ is a $k$-balanced magic ensemble 
	on  $\scrM_{n,k}(d)$ with $k\geq 1$. Then	
	\begin{gather} 
	\frac{(s_k+1)d^{n-k-\lfloor k/3\rfloor}}{3}	\leq \|\bQ(\orb(\scrE))\|	=\max\left\{\frac{D+2}{3}\kappa(\scrE), \frac{D+2}{D-1}[1-\kappa(\scrE)] \right\}
	\leq\frac{4^{(2k+1)/6}(d^n+2)}{3d^k},\\
	\|\bQ(\orb(\scrE))-P_\sym\|=\left|\frac{D+2}{3}\kappa(\scrE)-1\right|	\leq \frac{4^{(2k+1)/6}d^{n-k}}{3}.
	\end{gather}			
\end{theorem}
\Thref{thm:Moment3BalMagicLUBd1} is a simple corollary of \thref{thm:Moment3Balance}, \lref{lem:kahkaBalMagicd1}, and the following fact,
\begin{align}
\frac{D+2}{D-1}[1-\kappa(\scrE)]\leq \frac{(d+2)(d-2)}{d(d-1)}\leq \frac{15}{14}. 
\end{align}	
Here the first inequality is saturated when $n=k=1$, in which case $D=d$ and $\kappa(\scrE)=2/d$; the second inequality is saturated when $d=7$. In the special cases $n=k=1,2$, we can derive more explicit formulas,
\begin{align}
\|\bQ(\orb(\scrE))\|=\begin{cases}
\frac{(d+2)(d-2)}{d(d-1)} &n=k=1,\\[1ex]
\frac{d^2+2}{d^2} & n=k=2; 
\end{cases}\quad 
\|\bQ(\orb(\scrE))-P_\sym \|=\begin{cases}
\frac{d-4}{3d} &n=k=1,\\[1ex]
\frac{2}{d^2} & n=k=2. 
\end{cases}
\end{align}	
In general, $k$ should satisfy the condition $k\geq 3n/4$ to suppress the operator norm $\|\bQ(\orb(\scrE))\|$.

\begin{figure}[tb]
	\centering 
	\includegraphics[width=0.75\textwidth]{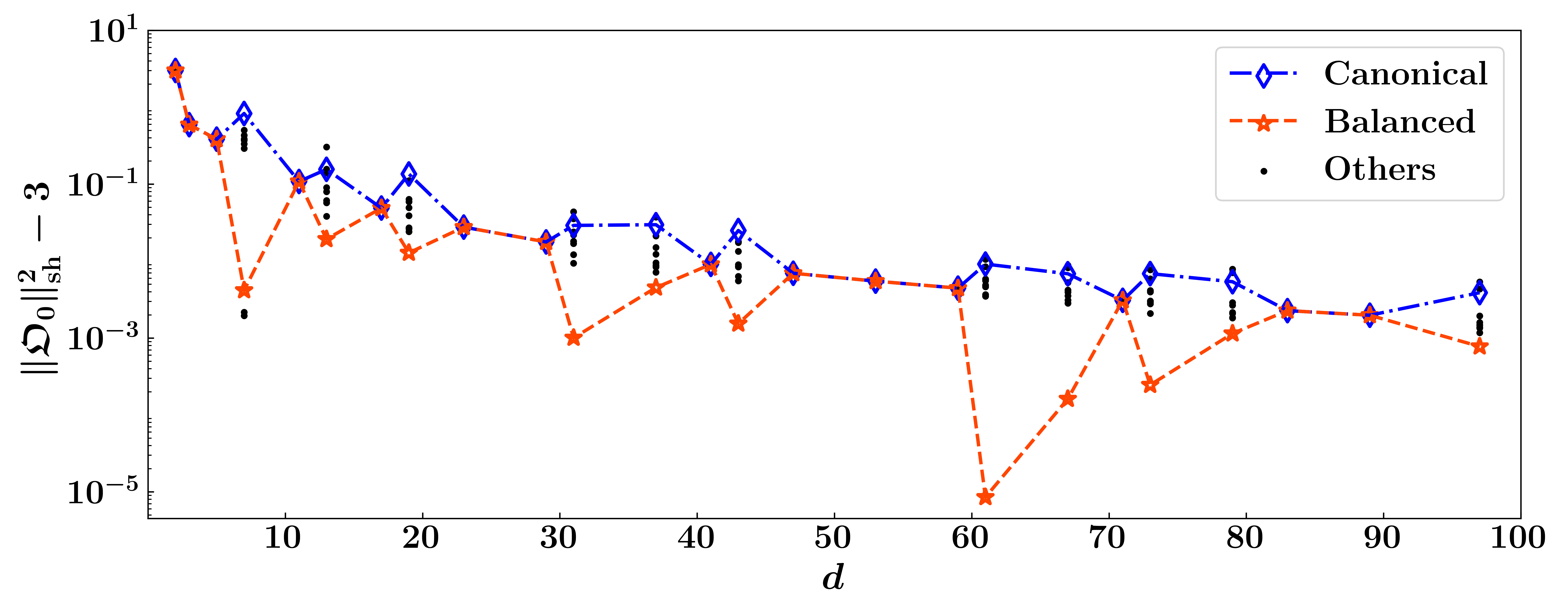}
	\caption{The squared shadow norm of $\Ob_0$ with respect to  magic orbits $\orb(\Psi)$ for all $|\Psi\>\in \scrM_{n,k}(d)$ with $n=50$ and $k=3$, where $\Ob$ is the projector onto a stabilizer state. Orbits based on  canonical magic states and $k$-balanced magic states are highlighted.} 
	\label{fig:stabstate_canbal}
\end{figure}

\begin{figure}[tb]
	\centering 
	\includegraphics[width=0.8\textwidth]{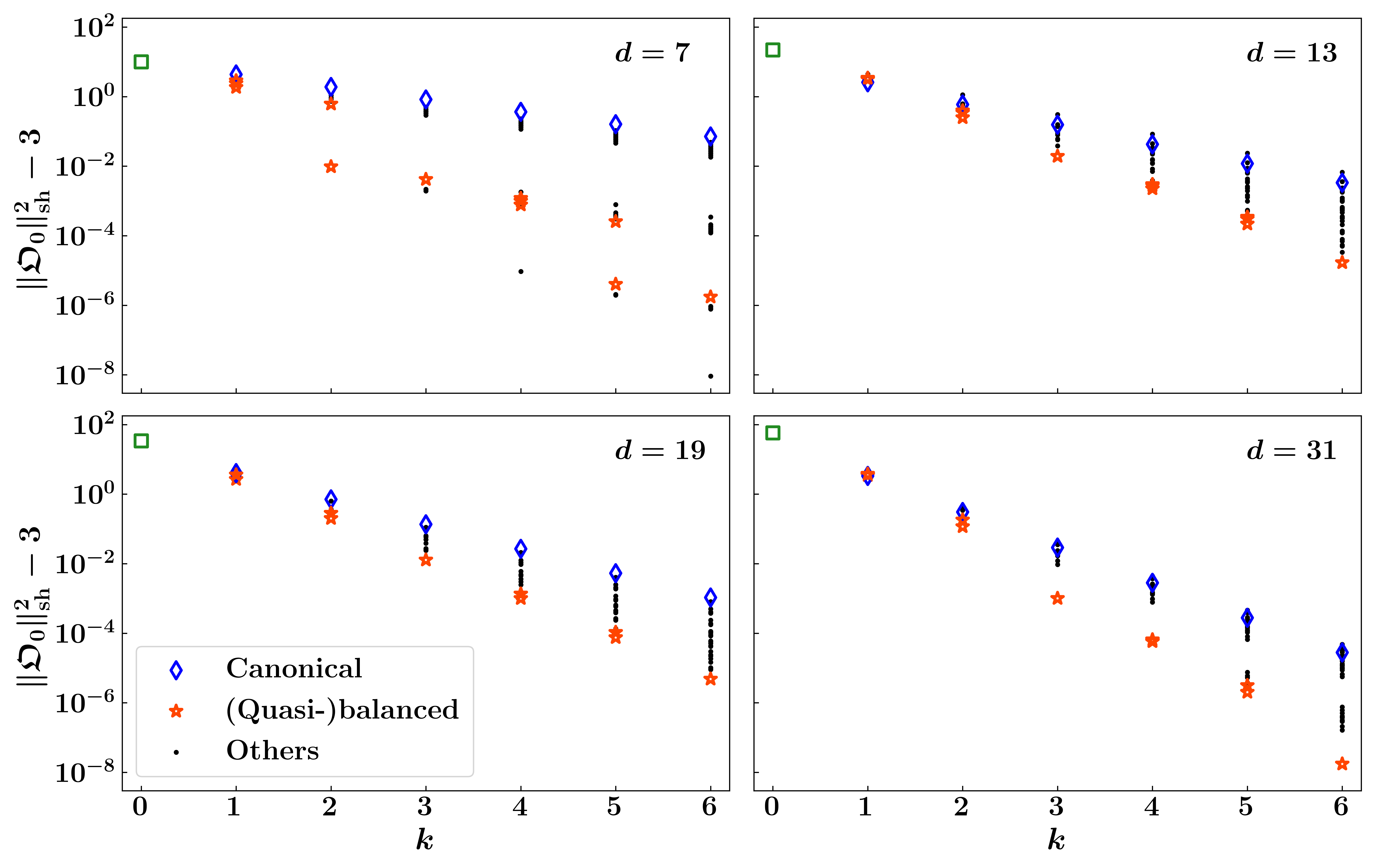}
	\caption{The squared shadow norm of $\Ob_0$ with respect to $\orb(\Psi)$ for all $|\Psi\> \in \scrM_{n,k}(d)$ with $n=50$ and $0 \leq k \leq 6$, where $\Ob$ is the projector onto a stabilizer state. Results on four local dimensions ($d=7,13,19,31$) are presented in four plots, respectively; orbits based on  canonical magic states and $k$-quasi-balanced magic sates  are highlighted as in \fref{fig:stabstate_canbal}.
		In each plot, the shadow norm associated with the stabilizer orbit $\Stab(n,d)$ (corresponding to the case $k=0$) is marked with a green square.	
	} 
	\label{fig:d713k6}
\end{figure}

\begin{theorem}\label{thm:ShNormBalMagicd1}
	Suppose $d$ is an odd prime, $d=1\mmod 3$,  $\scrE$ is a $k$-balanced magic ensemble 
	on $\scrM_{n,k}(d)$ with $k\geq 1$, and $\Ob\in \caL_0\bigl(\caH_d^{\otimes n}\bigr)$. Then 
	\begin{gather}
	\|\Ob\|^2_{\orb(\scrE)}
	\leq\left[1+	\frac{(D+1)\gamma^\rmB_{d,k}}{D+d}\right]
	\|\Ob\|_2^2+
	2\|\Ob\|^2 \leq 	
	(1+\gamma^\rmB_{d,k})
	\|\Ob\|_2^2+
	2\|\Ob\|^2,							\\
	\hka(\scrE,\Sigma(d))-3-\frac{5}{D}\leq \|\orb(\scrE)\|_\sh
	\leq 3+	\frac{(D+1)\gamma^\rmB_{d,k}}{D+d}
	\leq 3+\gamma^\rmB_{d,k}\leq 3+\frac{4^{(k+2)/3}}{d^{k-1}}. 
	\end{gather}
\end{theorem}
\Thref{thm:ShNormBalMagicd1} is a simple corollary of \thref{thm:ShNormGen} and \lref{lem:kahkaBalMagicd1}. If instead $\Ob$ is the projector onto a stabilizer state, then \thref{thm:StabShNormMagicd1} is still applicable when the magic state $|\Psi\>\in \scrM_{n,k}(d)$ is replaced by the $k$-balanced magic ensemble $\scrE$. Notably, we have
\begin{align}
\|\Ob_0\|^2_{\orb(\scrE)}&=\frac{D+1}{D+2}\hka(\scrE,\Sigma(d))-\frac{(D+1)(3D-1)}{D^2}=
\frac{(D+1)(3D^2-3dD+D+d)}{D^2(D+d)}+\frac{D+1}{D+d}\kappa(\scrE,\scrT_\ns).  
\end{align}	
When the $k$-balanced magic ensemble $\scrE$ is replaced by a $k$-quasi-balanced magic state  $|\Psi\>$ in $\tscrM^\qb_{n,k}(d)$, the above results still hold after minor modification.
For example, the following theorem is  an analog of \thref{thm:ShNormBalMagicd1} and a simple corollary of \thref{thm:ShNormGen}  and \lref{lem:kahkaQBalMagicd1}. Compared with the result on a general magic orbit presented in \thref{thm:ShNormMagicd1},  the shadow norm $\|\orb(\Psi)\|_\sh$ converges to 3 more quickly as $k$ and $d$ increases.

\begin{theorem}
	Suppose $d$ is an odd prime, $d=1\mmod 3$, $|\Psi\>\in \tscrM^\qb_{n,k}(d)$ with $k\geq 1$, and $\Ob\in \caL_0\bigl(\caH_d^{\otimes n}\bigr)$. Then 
	\begin{gather}
	\|\Ob\|^2_{\Psi} \leq 	
	(1+\gamma^\rmB_{d,n,k})
	\|\Ob\|_2^2+
	2\|\Ob\|^2,							\\
	\hka({\Psi},\Sigma(d))-3-\frac{5}{D}\leq \|\orb({\Psi})\|_\sh
	\leq 3+\gamma^\rmB_{d,n,k}\leq 3+\frac{4^{(k+2)/3}}{d^{k-1}}+\frac{6d}{D}.  
	\end{gather} 
	If in addition $n\geq(4k+8)/3$, then \thref{thm:ShNormBalMagicd1} is still applicable when $\scrE$ is replaced by $|\Psi\>$.
\end{theorem}

At this point,
it is instructive to compare shadow norms  of  stabilizer states with respect to $\orb(\Psi)$ for all $|\Psi\>\in \scrM_{n,k}(d)$, which can be calculated by virtue of \thref{thm:StabShNormMagicd1} and \pref{pro:kahkaMagicd1}. \Fsref{fig:stabstate_canbal} and \ref{fig:d713k6}
illustrate some numerical results, which  highlight the performances of  magic obits based on $k$-quasi-balanced magic states.  Note that a $k$-quasi-balanced magic state is automatically $k$-balanced when $k=0\mmod 3$. 
Except for the special case $d=7$, 
orbits based on  $k$-quasi-balanced magic states are nearly optimal, while orbits based on canonical magic states are nearly worst, although the difference is not so significant for practical applications.  \Fref{fig:stabstate_canbal} also covers the case $d\neq 1 \mmod 3$, in which all magic states in $\scrM_{n,k}(d)$ are $k$-balanced and the common shadow norm is  determined  by \thref{thm:StabShNormMagicd23} and \eref{eq:Stab1ShNormMagicd23}.

\subsection{\label{sec:3designExact}Construction of exact 3-designs}

Suppose $\scrE$ is a balanced ensemble and satisfies
the conditions in \eref{eq:DesignFiducialCon}, which  can now be simplified as 
\begin{align}\label{eq:ExDesignFiducialCon}
\kappa(\scrE)\leq \frac{3}{D+2},
\end{align}
which is equivalent to the condition  $\hka(\scrE)\leq 0$. Let $p$ be the probability determined by \eref{eq:designProb}. Then $\scrE_p$ is a balanced ensemble with $\hka(\scrE_p)=0$, so $\orb(\scrE_p)$  forms an exact 3-design by \eref{eq:QscrEp2} (cf. \pref{pro:orbit3designCon}).

Next, suppose
$\scrE$ is a $k$-balanced magic ensemble. Then $\kappa(\scrE)=\kappa_{d,k}$ by
\lref{lem:kahkaBalMagicd1}, where  $\kappa_{d,k}$  is defined in \eref{eq:kappatgammadk} and
decreases monotonically with $k$ given that $1\leq  \tg(d)^2\leq 4d-27$ by \eref{eq:tgdsq}. Define
\begin{align}\label{eq:kstardn}
k_*(d,n):=\min\left\{k\in \bbN\,\Big|\, \kappa_{d,k}\leq \frac{3}{D+2}\right\},\quad \scrS(d):=\{n\,|\,k_*(d,n)\leq n \};
\end{align} 
note that we do not impose the requirement $k\leq n$ in the definition of $k_*(d,n)$. Then the condition in  \eref{eq:ExDesignFiducialCon} is satisfied iff $k_*(d,n)\leq k\leq n$. If this condition holds, then $n\in \scrS(d)$ and an exact 3-design can be constructed from a $k$-balanced magic ensemble. On the other hand, if $\kappa_{d,n}>3/(D+2)$, then $k_*(d,n)\geq n+1$, $n\notin \scrS(d)$,  and exact 3-designs cannot be constructed  in this way for  all $k\leq n$. Direct calculation based on \eqsref{eq:kappatgammadk}{eq:kstardn}  yields
\begin{align}\label{eq:kstardn123}
k_*(d,n)=\begin{cases}
1 &\mbox{if}\quad  n=1,\\
3 &\mbox{if}\quad  n=2,\\
3 &\mbox{if}\quad n=3,\; \tg(d)^2\leq 3d^4/(d^3+2), \\
4 &\mbox{if}\quad n=3, \; \tg(d)^2> 3d^4/(d^3+2).
\end{cases}
\end{align}
So the condition in  \eref{eq:ExDesignFiducialCon} is satisfied when $k=n=1$, but cannot be satisfied when $n=2$, which means $1\in \scrS(d)$ and $2\notin \scrS(d)$. 

\begin{table*}
	\caption{\label{tab:Exact3design} The set $\scrS(d)$ defined in \eref{eq:kstardn} and its cardinality for odd prime $d$ with $d=1\mmod 3$ and $d\leq 1000$. Exact 3-designs   can be constructed from $k$-balanced magic ensembles on $\scrM_{n,k}(d)$ with $k= n$ if $n\in \scrS(d)$. }
	\renewcommand{\arraystretch}{1.25}
	\setstretch{1.2}
	\begin{tabularx}{\textwidth}{c X cl}
		\hline\hline
		No. & $d$ & $|\scrS(d)|$ &  $\scrS(d)$\\ \hline\hline
		\hphantom{xx}1 \hphantom{xx}  &  7, 31, 61, 67, 73, 109, 181, 199, 211, 307, 331, 337, 373, 421, 433, 457, 547, 577, 619, 691, 739, 823, 829, 859, 907,  997  &  $\hphantom{xxx}\infty$\hphantom{xxx}  &  $\bbN\setminus\{2\}$\\
		
		2  &  37, 79, 97, 127, 139, 163, 223, 271, 283, 313, 349, 367, 523, 607, 613, 661, 709, 733, 769, 811, 877, 937, 991 &  1  &  {1}\\
		
		3  &  13, 19, 151, 193, 229, 277, 379, 397, 409, 439, 499, 631, 643, 673, 727, 751, 883, 919, 967  &  2  &  {1, 3}\\
		
		4  &  {103, 541, 571}  &  3  &  {1, 3, 6}\\
		
		5  &  {43}  &  4  &  {1, 3, 4, 6}\\
		
		6  &  {853}  &  5  &  {1, 3, 4, 6, 9}\\
		
		7  &  {157, 487}  &  6  &  {1, 3, 4, 6, 9, 12}\\
		
		8  &  {787}  &  8  &  {1, 3, 4, 6, 7, 9, 12, 15}\\
		
		9  &  {241}  &  9  &  {1, 3, 4, 6, 7, 9, 12, 15, 18}\\
		
		10  &  {463}  &  12  &  {1, 3, 4, 6, 7, 9, 10, 12, 15, 18, 21, 24}\\
		
		11  &  {601}  &  13  &  {1, 3, 4, 6, 7, 9, 10, 12, 15, 18, 21, 24, 27}\\
		
		12  &  {757}  &  14  &  {1, 3, 4, 6, 7, 9, 10, 12, 15, 18, 21, 24, 27, 30}\\			
		\hline	\hline 
	\end{tabularx}
\end{table*}

\Lref{lem:BalMagicExactd1} and \thref{thm:BalMagicExactd1} below
clarify the basic properties of $k_*(d,n)$ and $\scrS(d)$. \Lref{lem:BalMagicExactd1} can be proved by virtue of \lref{lem:GaussSum3LUB} and  the definitions in \eqsref{eq:kappatgammadk}{eq:kstardn} as shown in 
\aref{app:BalMagicExactd1Proof}. \Thref{thm:BalMagicExactd1} is a simple corollary of \eref{eq:kstardn123} and \lref{lem:BalMagicExactd1}. 

\begin{lemma}\label{lem:BalMagicExactd1}
	Suppose $d$ is an odd prime and $d=1\mmod 3$; then 
	\begin{align}\label{eq:kstardnLimLB}
	k_*(d,n)\geq 3n/4,\quad \frac{3}{4}\leq \lim_{n\to \infty } \frac{k_*(d,n)}{n}=\frac{3\ln d}{4\ln d-2\ln |\tg(d)|}< \frac{3\ln d}{3\ln d-2\ln 2}. 
	\end{align}	 
	If  $\tg(d)^2\geq 3d$ and $n\geq 2$, then $k_*(d,n)\geq n+1$. If $\tg(d)^2\leq 3d$;  then  $k_*(d,3)= 3$.	
	If  $\tg(d)^2\geq d$,  then
	$k_*(d,n)\geq n$;  in addition, $k_*(d,n)\geq n+1$ when  $n= 2\mmod 3$ or $n\geq d\ln 3+5$. 
	If instead $\tg(d)^2\leq d$, then  $k_*(d,n)\leq n$ except when $n= 2$, in which case  $k_*(d,n)= n+1$. If in addition $\tg(d)^2=1$, then 
	\begin{align}\label{eq:kstartgd1}
	k_*(d,n)=\begin{cases}
	\left\lceil 3n/4\right\rceil+1 & \mbox{if}\quad n=2 \mmod 4, \\
	\left\lceil 3n/4\right\rceil &\mbox{if}\quad  n\neq 2 \mmod 4. 
	\end{cases}
	\end{align}		 
\end{lemma}

\begin{theorem}\label{thm:BalMagicExactd1}
	Suppose $d$ is an odd prime and $d=1\mmod 3$; then
	\begin{align}
	\scrS(d)=\begin{cases}
	\bbN\setminus \{2\} & \tg(d)^2\leq d,\\
	\{1\} & \tg(d)^2\geq 3d,
	\end{cases} \quad 
	|\scrS(d)|=
	\begin{cases}
	\infty& \tg(d)^2\leq d,\\
	1 & \tg(d)^2\geq 3d.
	\end{cases}
	\end{align}	
	If  $d<\tg(d)^2< 3d$, then $\{1,3\}\subseteq \scrS(d)$, $\max \scrS(d)<d\ln 3+5$,  and $n\neq 2\mmod 3$ for each $n\in \scrS(d)$; in addition, $2\leq |\scrS(d)|\leq d$. 		 
\end{theorem}

Suppose $d$ is an odd prime satisfying $d=1 \mmod 3$. 
If $\tg(d)^2\leq d$, then we can construct an exact 3-design in $\caH_d^{\otimes n}$ with $n\neq 2$ using at most four Clifford orbits (only two orbits when $n$ is divisible by 3) thanks to \thref{thm:BalMagicExactd1} and the balanced magic ensemble  constructed in \eref{eq:kbalanceEx}.
If  instead $\tg(d)^2\geq d$,  then an exact 3-design based on a $k$-balanced magic ensemble can be constructed only for limited choices of $n$ as specified by the set $\scrS(d)$. 
If in addition  $\tg(d)^2\geq 3d$, then such a 3-design  can be constructed only when $n=1$ since $\scrS(d)=\{1\}$. By virtue of \eqsref{eq:kappatgammadk}{eq:kstardn} above and \tref{tab:tgd} in \aref{app:GaussJacobi} we can determine 
the set $\scrS(d)$ for each odd prime $d$ satisfying   $d=1\mmod 3$ and $d\leq 1000$. The results are summarized in \tref{tab:Exact3design}.

Next, we consider the asymptotic behavior of the set $\scrS(d)$ as $d$ gets large. Note that the value of $\tg(d)^2$ is tied  to the phase $\phi$ of the Gauss sum $G(\eta_3)$ according to \eref{eq:tgdsq}. The distribution of the phase $\phi$  was originally studied by  Kummer in the 19th century, who formulated  the famous  Kummer conjecture (falsified later). 
Now, it is well known that the phase of the Gauss sum $G(\eta_3)$ is asymptotically uniformly distributed in the  interval $[0,2\pi)$  if we consider all primes $d$ with $d=1\mmod 3$, although there is a bias 
if we set an upper bound for the primes
\cite{HeatP79}. The bias in  the distribution remains an active research topic in number theory until now \cite{DunnR21}.  
Asymptotically, about one third of these primes satisfy each of the following three conditions,
\begin{align}
0\leq \tg(d)^2\leq d, \quad d<\tg(d)^2< 3d, \quad 3d\leq \tg(d)^2\leq 4d,
\end{align}
which correspond to the three cases considered in \thref{thm:BalMagicExactd1}, respectively.

If instead $d$ is an odd prime satisfying $d\neq 1\mmod 3$, then every magic state is balanced according to \lref{lem:kahkaMagicd23} in \sref{sec:MagicOrbitd23}. In addition, an exact 3-design can be constructed from a mixture of magic orbits and the orbit of stabilizer states iff $n=1$ and $d\geq 5$. Denote by $\scrS(d)$ the set of $n$ for which an exact 3-design can be constructed in this way; then
$\scrS(d)$ is empty when $d=3$ and $\scrS(d)=\{1\}$ when $d\geq 5$ and $d=2\mmod 3$. In the special case $d=3$, although exact 3-designs cannot be constructed in this way, many Clifford orbits form exact 3-designs  (see \rcite{GrosNW21} or \sref{sec:3designConstructd23}). Asymptotically, about one half primes are equal to 1 modulo 3 and the other half primes are equal to 2 modulo 3. 
If we consider all primes, then the following three cases
\begin{align}
\scrS(d)=\bbN\setminus\{2\}, \quad 2\leq |\scrS(d)|\leq d,\quad \scrS(d)=\{1\}
\end{align}
occur with probabilities about $1/6$, $1/6$, and $2/3$, respectively.

\section{Summary}
When the local dimension $d$ is an odd prime, the qudit Clifford group is only a 2-design, but not a 3-design, which is in sharp contrast with the counterpart based on a qubit system. This  distinction extends to the orbit of stabilizer states and other generic Clifford orbits. To clarify the implications of this distinction, in this work we  studied   the third moments of general qudit Clifford orbits  systematically and in depth. First, we introduced the shadow norm for quantifying the deviation of an ensemble of quantum states from a 3-design and clarified its connection with the third normalized frame potential. Then, we showed that  the third normalized frame potential and shadow norm of every Clifford orbit are both $\caO(d)$, although the operator norm of the third normalized moment operator may increase exponentially with $n$ when $d\neq 2\mmod 3$. Notably, this is the case for the orbit of stabilizer states. Therefore, the overhead of qudit stabilizer measurements in shadow estimation compared with qubit stabilizer measurements is only $\caO(d)$, which is independent of $n$, even if the ensemble of stabilizer states is far from a 3-design when  $d\neq 2\mmod 3$.

Furthermore, we showed that magic orbits based on $T$ gates are much better approximations to 3-designs and thus more appealing to many applications. 
Notably, the third normalized frame potential, operator norm of the third normalized moment operator, and shadow norm all decrease exponentially with the number of $T$ gates. Therefore,  good approximate 3-designs   with respect to 
the third normalized frame potential and shadow norm can be constructed using  only a few $T$ gates. Actually, the third normalized frame potential and shadow norm
of any magic orbit are upper bounded by the constants $13/11$ and $15/2$, respectively, so a single $T$ gate can already bridge the gap between qudit systems and qubit systems with respect to the two figures of merit.
When $d=2\mod 3$, the operator norm of the third normalized moment operator is upper bounded by the constant $5/3$. 
When $d\neq 2\mmod 3$, however, this operator norm may grow exponentially with $n$ depending on the magic orbit under consideration. Nevertheless, we can construct good approximate 3-designs   with respect to all three figures of merit  from only two Clifford orbits. For an infinite family of local dimensions, we can construct exact 3-designs  from only two or four Clifford orbits. Interestingly, the existence of such exact 3-designs is tied to a deep result on the phase distribution of cubic Gauss sums, which was originally studied by Kummer in the 19th century and remains an active research topic in number theory until today \cite{HeatP79,DunnR21}.

Our work  offers valuable insights on the (qudit) Clifford group, Clifford orbits, and $t$-designs and  highlights the power of a single magic gate in quantum information processing, which is quite unexpected. 
It may have profound implications for various topics in quantum information processing, including shadow estimation in particular (see the companion paper \cite{MaoYZ24} for more discussions). Notably, the gap between qudit systems and qubit systems can be circumvented with only a few $T$ gates.  In the course of study we clarify the key properties of the commutant of the third Clifford tensor power and the underlying mathematical structures. The ideas and technical tools we introduce here are also useful to studying higher moments of Clifford orbits and their applications.

\section*{Acknowledgements}
This work is supported by the National Natural Science Foundation of China (Grant No.~92165109), National Key Research and Development Program of China (Grant No.~2022YFA1404204), and Shanghai Municipal Science and Technology Major Project (Grant No.~2019SHZDZX01).

\appendix
\numberwithin{lemma}{section}
\numberwithin{proposition}{section}
\numberwithin{corollary}{section}
\renewcommand{\thelemma}{\thesection\arabic{lemma}}
\renewcommand{\theproposition}{\thesection\arabic{proposition}}
\renewcommand{\thecorollary}{\thesection\arabic{corollary}}

\section{\label{app:Shadow3design}Proof of \thref{thm:ShNormPhi3LB}}

\begin{proof}[Proof of \thref{thm:ShNormPhi3LB}]
	The inequality $\|\bQ\|_\sh\geq \|P_\sym\|_\sh$ in \eref{eq:QbarShNormPhi3LB} can be proved as follows. Let $\bbT$ be the twirling superoperator; then  $\bbT(\bQ)=P_\sym$. Let $\Ob$ be an arbitrary traceless operator on $\caH$ that is normalized with respect to the Hilbert-Schmidt norm. Then 
	\begin{align}
	\|\bQ\|_\sh&\geq \left\|\tr_{BC}\bigl[\bQ \bbT\bigl(\bbI\otimes   \Ob\otimes \Ob^\dag\bigr)\bigr]\right\|
	=\left\|\tr_{BC}\bigl[\bbT(\bQ)\bigl(\bbI\otimes   \Ob\otimes \Ob^\dag\bigr)\bigr]\right\|
	=\left\|\tr_{BC}\bigl[P_\sym \bigl(\bbI\otimes   \Ob\otimes \Ob^\dag\bigr)\bigr]\right\|,
	\end{align}
	which implies that $\|\bQ\|_\sh\geq \|P_\sym\|_\sh$.

	To prove \eref{eq:QbarShNormPhi3LB} completely, we need to resort to an alternative approach. Let 
	\begin{align}
	M=\sum_j w_j |\psi_j\>\<\psi_j|\otimes (D |\psi_j\>\<\psi_j|-\bbI)\otimes (D |\psi_j\>\<\psi_j|-\bbI);
	\end{align}
	note that $\|D |\psi_j\>\<\psi_j|-\bbI\|_2^2=D(D-1)$. 
	By virtue of \lref{lem:MQ} below and \eref{eq:ShadowNormSym}  we can deduce that
	\begin{align}\label{eq:QbarShNormLBproof}
	\|\bQ\|_\sh&\geq \frac{\tr(\bQ M) }{D(D-1)}=\frac{(D+1)(D+2)}{6(D-1)}(D^2\Phi_3-2D\Phi_2+\Phi_1)
	\geq \frac{3D-2}{6D}=\|P_\sym\|_\sh
	,
	\end{align}
	which confirms \eref{eq:QbarShNormPhi3LB}.  \Eref{eq:ShNormPhi3LB} is a simple corollary of  \eqsref{eq:EnsembleShNormDef}{eq:QbarShNormPhi3LB}.

	If $\scrE$ forms a 3-design, then $\bQ=P_\sym$, so the inequality $\|\bQ\|_\sh\geq \|P_\sym\|_\sh$ is saturated. 
	Conversely, if this inequality is saturated, then both inequalities in \eref{eq:QbarShNormLBproof} are necessarily saturated,  and the saturation of the second inequality implies that $\scrE$ forms a 3-design by \lref{lem:MQ}. Thanks to \eref{eq:EnsembleShNormDef},  the final   bound for $\|\scrE\|_\sh$ in \eref{eq:ShNormPhi3LB}	
	is saturated iff the final bound for $\|\bQ\|_\sh$ in  \eref{eq:QbarShNormPhi3LB} is saturated.

	Note that $\bar{\Phi}_3=D(D+1)(D+2)\Phi_3/6$. If  $\scrE$ is a 2-design, then $\Phi_1=1/D$ and $\Phi_2=2/[D(D+1)]$, so \eref{eq:QbarShNormLBproof} implies that
	\begin{align}
	\|\bQ\|_\sh&\geq \frac{D}{D-1}\bar{\Phi}_3-\frac{(D+2)(3D-1)}{6D(D-1)}\geq \bar{\Phi}_3+\frac{1}{D-1}-\frac{(D+2)(3D-1)}{6D(D-1)}=\bar{\Phi}_3-\frac{3D+2}{6D},\\
	\|\scrE\|_\sh&\geq	\frac{6D(D+1)}{(D-1)(D+2)}\bar{\Phi}_3-\frac{(D+1)(3D-1)}{D(D-1)}\geq 6\bar{\Phi}_3+\frac{12}{(D-1)(D+2)}-\frac{(D+1)(3D-1)}{D(D-1)}\nonumber\\
	&\geq  6\bar{\Phi}_3-3-\frac{1}{D}-\frac{4}{D+2}
	> 6\bar{\Phi}_3-3-\frac{5}{D},	
	\end{align}
	which confirm \eqsref{eq:QbarShNormPhi3LB2}{eq:ShNormPhi3LB2}. Here the second inequalities in both equations hold because $\bar{\Phi}_3\geq 1$. 
\end{proof}

In the rest of this appendix, we prove an  auxiliary lemma employed in the proof of \thref{thm:ShNormPhi3LB}. 

\begin{lemma}\label{lem:MQ}
	Suppose $\scrE=\{|\psi_j\>, w_j \}_j$ is an ensemble of pure states in $\caH$, $Q=Q_3(\scrE)$ is the third moment operator, and 
	\begin{align}\label{eq:Mdef}
	M=\sum_j w_j |\psi_j\>\<\psi_j|\otimes (D |\psi_j\>\<\psi_j|-\bbI)\otimes (D |\psi_j\>\<\psi_j|-\bbI). 
	\end{align}
	Then 	
	\begin{align}\label{eq:MQ}
	\tr(M^2)=D^2\tr(MQ)=D^4\Phi_3-2D^3\Phi_2+D^2\Phi_1\geq \frac{D(3D-2)(D-1)}{(D+1)(D+2)},
	\end{align}
	and the lower bound is saturated iff
	$\scrE$ forms a 3-design. 
\end{lemma}
Here $\Phi_t$ for $t=1,2,3$ are shorthands for $\Phi_t(\scrE)$.

\begin{proof}[Proof of \lref{lem:MQ}]
	Let $Q_t=Q_t(\scrE)$ for $t=1,2,3$ be the $t$th moment operator; note that $Q_3=Q$. 
	By definition we can deduce that
	\begin{align}
	M&=D^2 Q_3-D Q_2\otimes \bbI-D \sum_j \left[w_j( |\psi_j\>\<\psi_j|\otimes \bbI\otimes  |\psi_j\>\<\psi_j|)\right]+Q_1\otimes \bbI\otimes \bbI,  \label{eq:Mproof} \\
	\tr(M^2)&=D^2\tr(MQ_3)=D^4\tr(Q_3^2)-2D^3\tr(Q_2^2)+D^2\tr(Q_1^2)=D^4\Phi_3-2D^3\Phi_2+D^2\Phi_1,
	\end{align}
	which imply the equalities in \eref{eq:MQ}. Note that \eref{eq:Mproof} still holds if we apply the twirling superoperator $\bbT$ on
	$M$ and $Q_t$ for $t=1,2,3$. Based on this observation we can deduce that
	\begin{align}
	\tr(M^2)&\geq \tr\bigl\{[\bbT(M)]^2\bigr\}	
	=D^2\tr\bigl\{\bbT(M)\bbT(Q_3)\bigr\}
	=D^4\tr\bigl\{[\bbT(Q_3)]^2\bigr\}-2D^3\tr\bigl\{[\bbT(Q_2)]^2\bigr\}+D^2\tr\bigl\{[\bbT(Q_1)]^2\bigr\}\nonumber\\
	&=\frac{6D^4}{D(D+1)(D+2)}-\frac{4D^3}{D(D+1)}+\frac{D^2}{D}=\frac{D(3D-2)(D-1)}{(D+1)(D+2)},
	\end{align}
	which confirms the inequality in \eref{eq:MQ}. Here the third equality holds because $\bbT(Q_t)=P_{[t]}/\pi_{[t]}$. The inequality is saturated iff $M$ is invariant under twirling, that is, $\bbT(M)=M$. If $\scrE$ forms a 3-design, then $M$ is invariant under twirling (cf. \pref{pro:tdesignEqui}), so the inequality in \eref{eq:MQ} is saturated.

	Conversely, if the inequality in \eref{eq:MQ} is saturated, then $M$ is invariant under twirling,  and so are the three operators $\tr_A M$,  $\Pi(\tr_A M)\Pi$, $\tr_B[\Pi(\tr_A M)\Pi]$, where $\Pi$ is the projector onto the antisymmetric subspace in $\caH^{\otimes 2}$. Meanwhile, direct calculation yields the following results,
	\begin{align}
	\tr_AM&=D^2Q_2 -DQ_1\otimes \bbI -D\bbI\otimes Q_1+\bbI\otimes \bbI,\\
	\Pi(\tr_A M)\Pi&= \Pi- \Pi(DQ_1\otimes \bbI +D\bbI\otimes Q_1)\Pi=\Pi-2D\Pi(Q_1\otimes \bbI)\Pi,\\
	\tr_B[\Pi(\tr_AM)\Pi]&=-\frac{\bbI}{2}-\frac{D(D-2)}{2}Q_1. 
	\end{align}
	Here the last equation means $Q_1$ is invariant under twirling, which in turn implies that $Q_2$ is invariant under twirling thanks to the first equation. In conjunction with \eref{eq:Mproof} we can deduce that $Q_3$ is invariant under twirling and is thus proportional to the projector $P_\sym$ onto the symmetric subspace in $\caH^{\otimes 3}$ (cf. \pref{pro:tdesignEqui}). 
	So $\scrE$ forms a 3-design, which completes the proof of \lref{lem:MQ}.
\end{proof}

\section{\label{app:sec:O3dAux}Proofs of \lref{lem:O3Dh} and auxiliary results on $O_3(d)$}
In this appendix we prove \lref{lem:O3Dh} and introduce some auxiliary results on $O_3(d)$. 
\subsection{Proof of \lref{lem:O3Dh}}\label{appendix:O3Dh}
\begin{proof}[Proof of \lref{lem:O3Dh}]
	When  $d=3$, direct calculation shows that $O_3(d)$ coincides with the symmetric group $S_3$, which is isomorphic to the dihedral group $\caD_3$.
	
	Next, we assume $d\geq 5$. By definition any stochastic isometry $O\in O_3(d)$ satisfies
	\begin{equation}
	O\mathbf{1}_3 = \mathbf{1}_3,\quad 	O\bfx \cdot O\bfx = \bfx \cdot \bfx \quad \forall \bfx\in\bbF_d^3,
	\end{equation}
	which implies that 
	\begin{gather}
	O\bfx_1\cdot O\bfx_2 = \bfx_1 \cdot \bfx_2 \quad \forall \bfx_1,\bfx_2 \in\bbF_d^3,\\
	\bfx \cdot \mathbf{1}_3=O\bfx \cdot O\mathbf{1}_3 = O\bfx \cdot \mathbf{1}_3 \quad \forall \bfx\in\bbF_d^3.
	\end{gather}
	Note that both $\mathrm{span}(\mathbf{1}_3)$ and  $\mathbf{1}_3^\perp=\{\bfx\in \bbF_d^3\,|\, \bfx \cdot \mathbf{1}_3=0 \}$ are  invariant  under $O_3(d)$. Denote by $O(2,d)$ the orthogonal group on  $\mathbf{1}_3^\perp$, which preserves the dot product. Then $O_3(d)$ is isomorphic to $O(2,d)$.

	According to \rcite{Came00book}, there exist two types of orthogonal groups on $\bbF_d^2$, with orders $2(d-1)$ and $2(d+1)$, respectively, and  
	both of them are isomorphic to dihedral groups. So both  $O(2,d)$ and $O_3(d)$
	are isomorphic to a dihedral group of order either $2(d-1)$ or $2(d+1)$. Since $O_3(d)$ contains $S_3$ as a subgroup, its order is necessarily divisible by 6, which means 
	\begin{equation}\label{eq:O_3(d)}
	|O_3(d)|=
	\begin{cases}
	2(d+1) \quad \text{if} \  d=2 \mmod 3,\\
	2(d-1) \quad \text{if} \  d=1 \mmod 3.
	\end{cases}
	\end{equation}
	Therefore, $O_3(d)$ is isomorphic to  $\caD_h$ with
	\begin{equation}
	h = \begin{cases}
	d+1 \quad \text{if} \  d=2 \mmod 3,\\
	d-1 \quad \text{if} \  d=1 \mmod 3,
	\end{cases}
	\end{equation}
	which completes the proof of \lref{lem:O3Dh}. 
\end{proof}

\subsection{Alternative approach for determining the order of $O_3(d)$}\label{app:O3element}

Suppose $O\in O_3(d)$ is a stochastic isometry and $(a,b,c)^\top$ is the first column of $O$. Then 
\begin{equation}\label{aeq:abc}
a^2+b^2+c^2=a+b+c=1,
\end{equation}
where the arithmetic is done modulo $d$.  Similar conclusion  applies to each column and each row of $O$. Since $d$ is an odd prime by assumption, to satisfy \eref{aeq:abc}, $a,b,c$ cannot equal  a same number. Once the first column of $O$  is fixed,  there exist two and only two stochastic isometries. To be specific,  $O$ must equal one of the  two stochastic isometries presented in \eref{eq:Oevenodd} (see \rcite{NezaW20} for the case  $d=2\mmod 3$) as reproduced here,
\begin{equation}\label{aeq:Oevenodd}
O_\even=aI + b \zeta + c \zeta^{-1}=
\begin{pmatrix}
a & c & b \\
b & a & c \\
c & b & a 
\end{pmatrix},
\quad O_\odd=a\tau_{23} + b \tau_{12} + c \tau_{13}=
\begin{pmatrix}
a & b & c \\
b &c & a \\
c & a & b    
\end{pmatrix}. 
\end{equation}
The determinants of $O_\even$ and $O_\odd$ read
\begin{align}
\det(O_\even)=-\det(O_\odd)=a^3+b^3+c^3-3abc=\frac{3 (a + b + c)(a^2 + b^2 + c^2) - (a + b + c)^3}{2}=1.
\end{align}

According to the above analysis,  each solution to \eref{aeq:abc} corresponds to a pair of stochastic isometries in $O_3(d)$. Denote by $N_d$ the number of solutions $(a,b,c)^\top\in \bbF_d^3$ to \eref{aeq:abc};  then $|O_3(d)|=2N_d$. 
To determine the order of $O_3(d)$, it suffices to determine $N_d$. 
Note that any permutation of a solution $(a,b,c)^\top$ is also a solution and that there does not exist a solution with  $a=b=c$, so both  $N_d$ and $|O_3(d)|$ are multiples of 3, which is consistent with the fact that $O_3(d)$ contains $S_3$ as a subgroup.

When $d=3$, direct calculation shows  that the solution set to \eref{aeq:abc}
is $\{(1,0,0),(0,1,0),(0,0,1)\}$. Accordingly,  $O_3(3)=\{\mathds{1},\zeta,\zeta^{-1},\tau_{12},\tau_{13},\tau_{23}\}$ coincides with $S_3$, which means $|O_3(3)|=6$.

Next, suppose  $d \geq 5$. By eliminating the variable $c$   in \eref{aeq:abc} we can deduce that $a^2+b^2+ab-a-b=0$, that is, 
\begin{equation}
3(2a+b-1)^2+(3b-1)^2=4.
\end{equation}
To determine $N_d$, it suffices to determine the number of solutions to this diagonal quadratic equation on $a,b\in\bbF_d$. According to Lemma 6.24 in \rcite{LidlN97book}, we have
\begin{align}
N_d&=d-\eta_2(-3,d)=d-\eta_2(-1,d)\eta_2(3,d)=d-(-1)^{(d-1)/2}\eta_2(3,d)=d-\eta_2(d,3)\nonumber\\
&=\begin{cases}
d+1 \quad \text{if} \  d=2 \mmod 3,\\
d-1 \quad \text{if} \  d=1 \mmod 3,
\end{cases}
\end{align}
which implies the order formula in \lref{lem:O3Dh} [see also \eref{eq:O_3(d)}] given that $|O_3(d)|=2N_d$. Here $\eta_2(\cdot,d)$ and $\eta_2(\cdot,3)$
denote  the quadratic characters  over the finite fields $\bbF_d$ and $\bbF_3$, respectively, and the fourth equality follows from the law of quadratic reciprocity \cite{LidlN97book} as reproduced in \lref{lem:QuadraticReciprocity} in \aref{app:MultiChar}.

\subsection{Auxiliary results on  stochastic isometries in $O_3(d)$}

When $b=1$ and $a=c=0$, $O_\even$ and $O_\odd$ in \eref{aeq:Oevenodd} reduce to $\zeta$ and $\tau_{12}$, respectively. In addition, if any of $a,b,c$ is equal to 0 or 1, then only one of them is nonzero according to \eref{aeq:abc}, in which case $O_\even, O_\odd\in S_3$. This observation leads to the following proposition. 
\begin{proposition}\label{pro:O0entry}
	Suppose $d$ is  an odd prime and $O\in O_3(d)\setminus S_3$; then no entry of $O$ is equal to 0 or 1. 
\end{proposition}

Next, we clarify the basic properties of  $\ker(O-\mathds{1})$ for $O\in O_3(d)$, which is tied to the number of fixed points of $O$ and is useful to computing the trace of $R(O)$ according to the following equation [cf. \eref{eq:RTtr}]
\begin{align}\label{eq:ROtr}
\tr R(O)=d^{n\dim \ker (O-\mathds{1})}. 
\end{align}

\begin{lemma}\label{lem:kerOI}
	Let $d$ be an odd prime and $O\in O_3(d)$. Then 
	$\mathbf{1}_3\in \ker(O-\mathds{1})$. If $O\neq \mathds{1}$ and $O\in O_3^\even(d)$, then  $\ker(O-\mathds{1})=\spa(\mathbf{1}_3)$. If  $O\in O_3^\odd(d)$, then $\dim \ker(O-\mathds{1})=2$. 
	If  $O\in O_3(d)\setminus S_3$, then any nonzero vector in $\ker(O-\mathds{1})$  has at least two nonzero entries. 
\end{lemma}

\begin{lemma}\label{lem:kerOI2}
	Let $d$ be an odd prime and $O\in O_3(d)$. Then any row (column) of $O-\mathds{1}$ is a linear combination of the other two rows (columns). If in addition $O\in O_3^\odd(d)\setminus S_3$, then the three rows (columns) of  $O-\mathds{1}$ are nonzero vectors that are proportional to each other.
\end{lemma}

\begin{proof}[Proof of \lref{lem:kerOI}]
	By definition $\ker(O-\mathds{1})$ is the eigenspace of $O$ with eigenvalue 1. 
	In addition, $\mathbf{1}_3$ is an  eigenvector of $O$  with  eigenvalue 1. 
	In conjunction with \eref{eq:Odet} we can deduce that
	\begin{align}\label{eq:kerOI}
	\!\!\!\dim \ker(O-\mathds{1})=\begin{cases}
	3 &\! \mbox{if} \; O=\mathds{1},\\
	1 &\! \mbox{if} \; O\neq \mathds{1}, \; O\in O_3^\even(d), \\
	2 & \! \mbox{if} \;  O\in O_3^\odd(d). 
	\end{cases}
	\end{align}
	In the three cases, $O$ has $d^3$, $d$, and $d^2$ fixed points, respectively. In addition, $\ker(O-\mathds{1})=\spa(\mathbf{1}_3)$ when $O\neq \mathds{1}$ and $O\in O_3^\even(d)$.
	
	It remains to prove the last claim in \lref{lem:kerOI}.
	If $O\in O_3(d)\setminus S_3$, then $O$ has no entry equal to 0 according to  \pref{pro:O0entry}. So $\ker(O-\mathds{1})$  cannot contain any nonzero vector that has only one nonzero entry, which completes the proof of \lref{lem:kerOI}. 
\end{proof}

\begin{proof}[Proof of \lref{lem:kerOI2}]
	If $O\in S_3$,  then it is straightforward to verify that  any row  of $O-\mathds{1}$ is a linear combination of the other two rows.	
	
	If  $O\in O_3^\even(d)\setminus S_3$, 
	then $\ker(O-\mathds{1})=\spa(\mathbf{1}_3)$ according to \lref{lem:kerOI}. So the row span of $O-\mathds{1}$ is two dimensional, and the three rows of $O-\mathds{1}$ are linearly dependent. On the other hand, the three rows of $O-\mathds{1}$ are related to each other by cyclic permutations according to  \eref{aeq:Oevenodd}, so any two rows must be linearly independent. Therefore, any row of $O-\mathds{1}$ is a linear combination of the other two rows.

	If   $O\in O_3^\odd(d)$, 
	then $\dim\ker(O-\mathds{1})=2$ according to \lref{lem:kerOI}, so the row span of $O-\mathds{1}$ is one dimensional.  If in addition  $O\in O_3^\odd(d)\setminus S_3$,  then  no entry of $O-\mathds{1}$ is equal to 0 according to \pref{pro:O0entry}. Therefore, the three rows of  $O-\mathds{1}$ are nonzero vectors that are proportional to each other, which means any row  of $O-\mathds{1}$ is a linear combination of the other two rows. 
	
	The above conclusions still hold if rows are replaced by columns because $O^\top=O^{-1}\in O_3(d)$ whenever $O\in O_3(d)$. 
\end{proof}

\section{Proofs of \lsref{lem:defect}-\ref{lem:Tindex} and auxiliary results on stochastic Lagrangian subspaces}

In this appendix we prove \lsref{lem:defect}-\ref{lem:Tindex} and introduce some auxiliary results on stochastic Lagrangian subspaces in \aref{app:SigmadAux}.

\subsection{\label{app:DefectProof}Proofs of \lsref{lem:defect} and  \ref{lem:ONNO}}

\begin{proof}[Proof of \lref{lem:defect}]
	Note that any nontrivial defect subspace in $\bbF_d^3$ is one dimensional and is thus determined by a defect vector. Suppose  $\bfg=(a,b,c)^\top\in \bbF_d^3$  is a defect vector; then $a^2+b^2+c^2=a+b+c=0$,	which implies that $a,b,c\neq 0$ and $c=-a-b$ given that $d$ is an odd prime by assumption. 
	By eliminating the variable $c$, we can deduce that
	\begin{equation}\label{eq:Neq}
	a^2+ab+b^2 = 0, \quad a,b \neq 0. 
	\end{equation}	
	If $d=3$, then $\bfg$ is necessarily proportional to $\mathbf{1}_3$, so  $\tcaN$ defined in \eref{eq:DefectSpaced3} is the only nontrivial defect subspace in 
	$\bbF_d^3$.

	Next, suppose $d>3$. Then \eref{eq:Neq}  has no  solution with $a=b$ and is equivalent to the equation $a^3-b^3=0$ with $a\neq b$,	that is,
	\begin{align}\label{eq:Neq2}
	(b/a)^3=1, \quad a,b\neq 0, \; a\neq b. 
	\end{align}
	When $d=2 \mmod 3$, this equation has no solution 
	because the order of the multiplicative group $\bbF_d^{\times}=\bbF_d\setminus \{0\}$ is not divisible by 3. 
	So there are no nontrivial defect subspaces.

	When $d=1 \mmod 3$, \eref{eq:Neq2} means $b/a$ is an order-3 element in the multiplicative group $\bbF_d^{\times}$. So $b/a=\xi$ or $b/a=\xi^2=-1-\xi$, where 
	$\xi$ is the order-3 element in $\bbF_d^\times$
	that appears in \eref{eq:TwoDefectSpaces}.   It follows that $\bfg$ is proportional to either $\bfg_0$ or $\bfg_1$ defined in \eref{eq:TwoDefectSpaces}. Meanwhile, $\tcaN_0=\spa(\bfg_0)$ and $\tcaN_1=\spa(\bfg_1)$ are indeed two distinct nontrivial defect subspaces, which completes the proof of \lref{lem:defect}. 
\end{proof}

\begin{proof}[Proof of \lref{lem:ONNO}]
	Note that $O\caN$ and $\caN O$ are  defect subspaces if  $\caN$ is a defect subspace and $O\in O_3(d)$. If  $d=3$, then
	$\tcaN$ defined in  \eref{eq:DefectSpaced3} is the only nontrivial defect subspace by \lref{lem:defect}, so \eref{eq:ONNOd3} holds.

	If 
	$d = 1\mmod  3$, then  $\tcaN_0=\spa(\bfg_0)$ and $\tcaN_1=\spa(\bfg_1)$  defined in \eref{eq:TwoDefectSpaces} are the only two nontrivial defect subspaces. Note that $\bfg_0$ and $\bfg_1$ are the eigenvectors of the cyclic permutation $\zeta$ defined in \eref{eq:zetatau} with eigenvalues $\xi^2$ and $\xi$, respectively, where $\xi$ is an order-3  element in $\bbF_d$. If $O\in O_3^\even(d)$, then $O$ commutes with  $\zeta$ by \eref{eq:Odet}, which means  $O\bfg_j$, $O^\top\bfg_j$, and $\bfg_j$ for each $j=0,1$ are eigenvectors of $\zeta$ with the same eigenvalue. Since $\zeta$ is nondegenerate, we conclude that  $O\bfg_j, O^\top\bfg_j \in \spa(\bfg_j)$, which implies \eref{eq:ONNO}.

	If $O\in O_3^\odd(d)$,  then $O^2=\mathds{1}$ and  $O\zeta O=\zeta^{-1}=\zeta^2$ by \eref{eq:Odet}, which means 
	$O\bfg_j$, $O^\top\bfg_j$, and $\bfg_{\bar{j}}$  are eigenvectors of $\zeta$ with the same eigenvalue for each $j=0,1$. So 
	$O\bfg_j, O^\top\bfg_j \in \spa(\bfg_{\bar{j}})$, which implies \eref{eq:ONNO}. 	 	
\end{proof}

\subsection{\label{app:SLagrangTproof}Proofs of \lsref{lem:defectT}-\ref{lem:TDelta2}}

Before proving  \lref{lem:defectT}, we first 
clarify the conditions on when a matrix in $\bbF_d^{6,3}$ can encode a stochastic Lagrangian subspace. The following auxiliary lemma  is a simple corollary of the definition in \sref{sec:SLSspanningSet}. 
\begin{lemma}\label{lem:GenMatrix}
	Suppose $d$ is an  odd prime, $\bfx_1, \bfx_2,\bfx_3, \bfy_1,\bfy_2,\bfy_3\in \bbF_d^3$, and 
	\begin{align}
	M=\begin{pmatrix}
	\bfx_1&\bfx_2&\bfx_3\\
	\bfy_1&\bfy_2&\bfy_3 
	\end{pmatrix}.
	\end{align}
	Then the column span of $M$ is a stochastic Lagrangian subspace in $\bbF_d^6$ iff the following three conditions hold, 
	\begin{enumerate}
		\item The three vectors $(\bfx_1;\bfy_1), (\bfx_2;\bfy_2), (\bfx_3;\bfy_3)$ are linearly independent;
		
		\item $\mathbf{1}_6\in \spa(\{(\bfx_1;\bfy_1), (\bfx_2;\bfy_2), (\bfx_3;\bfy_3) \})$;
		
		\item $\bfx_j\cdot \bfx_k =\bfy_j\cdot \bfy_k\;\forall j,k=1,2,3$.  
	\end{enumerate}
\end{lemma}

\begin{proof}[Proof of \lref{lem:defectT}]
	When $d=3$, by virtue of \lref{lem:GenMatrix} it is straightforward to verify that $\tcaT_0$ and $\tcaT_1$ are two distinct stochastic Lagrangian subspaces  that have nontrivial left and right defect subspaces. Therefore,
	$\scrT_\defe=\{\tcaT_0, \tcaT_1\}$ given that $|\scrT_\defe|=2$ by \eref{eq:Sigma33DefectNum}. According to \eref{eq:Tcss},  $\tcaT_0$ is of CSS type.

	When $d = 1 \mmod 3$, $\tcaT_{ij}$ for $i,j=0,1$ are four distinct stochastic Lagrangian subspaces that have nontrivial left and right defect subspaces by \eref{eq:g12Properties} and \lref{lem:GenMatrix}. Therefore, $\scrT_\defe=\{\tcaT_{ij}\}_{i,j=0,1}$ given that $|\scrT_\defe|=4$ by \eref{eq:Sigma33DefectNum}. According to \eref{eq:Tcss},  $\tcaT_{00}$ and $\tcaT_{11}$ are of CSS type.
\end{proof}

\begin{proof}[Proof of \lref{lem:SemigroupT}] 
	When $d=3$, \eref{eq:OTTO3} can be verified by straightforward calculation based on \eref{eq:SLStwo}. Note that the left action of $O_3(d)=S_3$ on $\scrT_\defe$ is transitive, and the stabilizer of $\tcaT_0$ (and similarly for $\tcaT_1$) is  necessarily a subgroup of $O_3(d)$ of index 2. The same conclusion holds if  the left action is replaced by the right action. When $d = 1 \mmod  3$, each stochastic Lagrangian subspace is completely determined by its left and right defect subspaces. So \eref{eq:OTTOd} follows from \eref{eq:tTijDefect}, \lref{lem:ONNO}, and the following equation,
	\begin{align}
	(O_1\caT O_2)_{\LD}=O_1T_{\LD},\quad (O_1\caT O_2)_{\RD}=O_2^\top \caT_{\RD}=\caT_{\RD}O_2\quad \forall O_1, O_2\in O_3(d). 
	\end{align}  
\end{proof}

\begin{proof}[Proof of \lref{lem:SemigroupscrT}]
	\Lref{lem:SemigroupscrT} follows from \lref{lem:SemigroupT} and the following equation
	\begin{align}
	OO_3^\even(d)&=O_3^\even(d) O=O_3^\even(d), &\quad 
	OO_3^\odd(d)&=O_3^\odd(d) O=O_3^\odd(d)\quad \quad \forall O\in O_3^\even(d),\\
	OO_3^\even(d)&=O_3^\even(d) O=O_3^\odd(d),& \quad 
	OO_3^\odd(d)&=O_3^\odd(d) O=O_3^\even(d)\quad \quad \forall O\in O_3^\odd(d).
	\end{align}
\end{proof}

\begin{proof}[Proof \lref{lem:TDelta}] 
	By definition $\mathbf{1}_6\in \caT$ for all $\caT\in \Sigma(d)$, which implies that	$\spa(\mathbf{1}_6)\leq \caT_\Delta$. 
	To prove other results in  \lref{lem:TDelta}, it is instructive to recall the facts that $\Sigma(d)=\scrT_\iso\sqcup\scrT_\defe$ and
	\begin{equation}
	\begin{aligned}
	\scrT_0\cap \scrT_\iso&=\{T_O\ |\ O\in O_3^\even(d)\}, & 
	\scrT_0\cap \scrT_\defe&=\begin{cases}
	\tcaT_1 & d=3,\\
	\{\tcaT_{01},\tcaT_{10}\} &d=1\mmod 3,
	\end{cases}\\
	\scrT_1\cap \scrT_\iso&=\{T_O\ |\ O\in O_3^\odd(d)\},&
	\scrT_1\cap \scrT_\defe&=\begin{cases}
	\tcaT_0 & d=3,\\
	\{\tcaT_{00},\tcaT_{11}\} &d=1\mmod 3.
	\end{cases}
	\end{aligned}
	\end{equation}	
	Note that $\scrT_\defe$ is empty when $d=2\mmod 3$. Here $\tcaT_i$ and $\tcaT_{ij}$ for $i,j=0,1$ are determined by \lref{lem:defectT} and they have the following properties,
	\begin{equation}\label{eq:tTDelta}
	\begin{gathered}
	(\tcaT_{0})_\Delta=\{(\bfx;\bfx) \,|\, \bfx\in \tcaN^\perp\}, \quad  (\tcaT_{1})_\Delta=\spa(\mathbf{1}_6),\\
	(\tcaT_{01})_\Delta=(\tcaT_{10})_\Delta=\spa(\mathbf{1}_6),\quad  (\tcaT_{ii})_\Delta=\{(\bfx;\bfx) \ |\ \bfx\in \tcaN_i^\perp\},\quad i=0,1,
	\end{gathered}
	\end{equation}
	where $\tcaN$ is defined in \eref{eq:DefectSpaced3}, and $\tcaN_0, \tcaN_1$ are defined in \eref{eq:TwoDefectSpaces}. In addition, $\tcaN,\tcaN_0, \tcaN_1$ are one dimensional, while $\tcaN^\perp,\tcaN_0^\perp, \tcaN_1^\perp$ are two dimensional.

	If $\caT\in \scrT_0$ and $\caT\neq\Delta$, then either  $\caT=\caT_O$ for some $O\in O_3^\even(d)\setminus \{\mathds{1}\}$ or $\caT\in \scrT_0\cap \scrT_\defe$. In the first case we have $\caT_\Delta=\spa(\mathbf{1}_6)$ according to \lref{lem:kerOI} and the following equality 
	\begin{align}\label{eq:TODelta}
	(\caT_O)_\Delta=
	\{(\bfx;\bfx)\ |\ \bfx= \ker(O-\mathds{1})\}.
	\end{align}
	In the second case, we have $\caT_\Delta=\spa(\mathbf{1}_6)$ thanks to \eref{eq:tTDelta}.

	If  $\caT\in \scrT_1$; then either  $\caT=\caT_O$ for some $O\in O_3^\odd(d)$ or $\caT\in\scrT_1\cap \scrT_\defe$. In the first case we have $\dim \caT_\Delta=2$ according to \eref{eq:TODelta} and \lref{lem:kerOI}. In the second case,  we have $\dim \caT_\Delta=2$  thanks to \eref{eq:tTDelta}. 
\end{proof}

\begin{proof}[Proof of \lref{lem:TDelta2}]
	By definition $\mathbf{1}_6\in \caT_1,\caT_2$, which implies that $\spa(\mathbf{1}_6)\leq \caT_1\cap \caT_2$. 	If $\caT_1\in \scrT_\iso$, that is, $\caT_1=\caT_O$ with $O\in O_3(d)$, then 
	\begin{align}
	\dim(\caT_1\cap \caT_2)=\dim(\caT_O\cap \caT_2)=\dim(\Delta\cap O^\top \caT_2)=\dim(O^\top \caT_2)_\Delta,
	\end{align}
	and \eref{eq:T1capT2} follows from \lsref{lem:SemigroupscrT} and \ref{lem:TDelta}. By the same token \eref{eq:T1capT2} holds if $\caT_2\in \scrT_\iso$. If $\caT_1, \caT_2\notin \scrT_\iso$, then $\caT_1, \caT_2\in \scrT_\defe$, in which case \eref{eq:T1capT2} can be verified by virtue of \lref{lem:defectT}.
\end{proof}

\subsection{\label{app:lem:tauTvProof}Proofs of \lsref{lem:Tv} and \ref{lem:tauTv}}

Before proving \lref{lem:Tv}, we need to introduce an auxiliary lemma. 
\begin{lemma}\label{lem:abcCubic}
	Suppose $d$ is an odd prime and $a,b,c\in \bbF_d$ satisfy $a+b+c=0$.  Then $a^3+b^3+c^3=3abc$. 
	If in addition $d\geq 5$,  then $a^3+b^3+c^3=0$ iff  $abc=0$, that is, one of the three numbers $a,b,c$ is equal to 0.
\end{lemma}

\begin{proof}[Proof of \lref{lem:abcCubic}]
	\Lref{lem:abcCubic} is a simple corollary of   the following equation,
	\begin{align}
	a^3+b^3+c^3-3abc=\frac{3 (a + b + c)(a^2 + b^2 + c^2) - (a + b + c)^3}{2}=0.
	\end{align}
\end{proof}

\begin{proof}[Proof of \lref{lem:Tv}]
	By assumption $v_1+v_2+v_3=\bfv\cdot \mathbf{1}_3=0$, so the equality $3v_1v_2v_3=v_1^3+v_2^3+v_3^3$ follows from \lref{lem:abcCubic}. As a corollary, $v_1v_2v_3=0$ iff	$v_1^3+v_2^3+v_3^3=0$ when $d\geq 5$.
	
	The equivalence of the three conditions in \eref{eq:TvdefEqui} follows from \eqref{eq:TvDeltaDefect}. 
	
	The equivalence of the first two conditions in \eref{eq:TvsymEqui} follows from \eqsref{eq:scrT01}{eq:TvDeltaDefect} given that $\scrT_\sym\cap\scrT_\odd=\scrT_\sym\cap\scrT_1$.  If $ v_1v_2v_3=0$, then one of the three entries $v_1, v_2, v_3$ is zero, and $\bfv$ is proportional to $(1,-1,0)^\top$, $(1,0,-1)^\top$, or $(0,1,-1)^\top$.  So  $\bfv\cdot\bfv\neq0$ and $O_\bfv$ is a transposition, which means $\caT_\bfv\in \scrT_\sym\cap\scrT_\odd$. Conversely, if $\caT_\bfv\in \scrT_\sym\cap\scrT_\odd$, then $\bfv\cdot\bfv\neq0$ by \eref{eq:TvdefEqui}, $O_\bfv$ is a transposition, and  $\bfv$ is proportional to $(1,-1,0)^\top$, $(1,0,-1)^\top$, or $(0,1,-1)^\top$, given that 
	$\bfv$ is an eigenvector of $O_\bfv$ with eigenvalue $-1$. Therefore, $v_1v_2v_3=0$. 	
\end{proof}

\begin{proof}[Proof of \lref{lem:tauTv}]	
	The relations  $\caT_\bfv\in \scrT_1$ and $\tau_{12}\caT_\bfv\in \scrT_0$ follow from \eref{eq:TvDeltaDefect}  and \lref{lem:SemigroupscrT}.
	If $\bfu$ and $\bfv$ are proportional to each other, then 
	$\caT_\bfu=\caT_\bfv$ thanks to the definition in \eref{eq:Tv}. Otherwise, $(\caT_\bfu)_\tDelta\neq (\caT_\bfv)_\tDelta$ by \eref{eq:TvDeltaDefect},  which means $\caT_\bfu\neq \caT_\bfv$. Therefore,  $\caT_\bfu=\caT_\bfv$ iff $\bfu$ and $\bfv$ are proportional to each other. 
	
	Alternatively, the relation $\tau_{12}\caT_\bfv\in \scrT_0$ follows from \lref{lem:TDelta} and the equation below
	\begin{align}\label{eq:DeltatauTvProof}
	(\tau_{12} \caT_\bfv)_\tDelta=\begin{cases}
	\bbF_d^3 &\mbox{if} \quad \bfv \propto (1,-1,0), \\
	\spa(\mathbf{1}_3) &\mbox{otherwise},
	\end{cases}
	\end{align}
	which can be proved as follows. If $ \bfv \propto (1,-1,0)^\top$, then $\caT_\bfv=\caT_{\tau_{12}}$ and $\tau_{12} \caT_\bfv=\Delta$, so $(\tau_{12} \caT_\bfv)_\tDelta=\bbF_d^3$, which confirms \eref{eq:DeltatauTvProof}. Otherwise, by virtue of \eref{eq:Tv} we can deduce that
	\begin{align}
	(\tau_{12} \caT_\bfv)_\tDelta&=\{\bfx-a\bfv\ |\  \bfx\cdot \bfv=0,  a\in \bbF_d, \tau_{12} \bfx+a\tau_{12}\bfv=\bfx-a\bfv\} \nonumber\\
	&=\{\bfx-a\bfv\ |\  \bfx\cdot \bfv=0,  a\in \bbF_d, a(1+\tau_{12})\bfv=(1-\tau_{12})\bfx\} =\{\bfx\ |\  \bfx\cdot \bfv=0,  (1-\tau_{12})\bfx =0\}\nonumber\\
	&=\bigl\{\bfx\ |\  \bfx\cdot \bfv=0,  \bfx\cdot (1,-1,0)^\top=0 \bigr\}=\spa(\mathbf{1}_3),
	\end{align}
	which confirms \eref{eq:DeltatauTvProof}. 	Here the third equality holds because the third entry of $(1-\tau_{12})\bfx$ is zero, while the third entries of $\bfv$ and $(1+\tau_{12})\bfv$ are nonzero. The last equality  holds because $\bfv\cdot \mathbf{1}_3=(1,-1,0)^\top\cdot \mathbf{1}_3=0$ and  $\bfv$ is not proportional to $(1,-1,0)^\top$. 
\end{proof}

\subsection{\label{app:TindexProof}Proof of \lref{lem:Tindex}}

\begin{proof}[Proof of \lref{lem:Tindex}]	
	If $\caT\in \scrT_\defe$, then $\caT\in \scrT_\defe\cap\scrT_1$ or  $\tau_{12}\caT\in \scrT_\defe\cap\scrT_1$, and any characteristic vector of $\caT$ is  proportional to either $\bfg_0$ or $\bfg_1$ defined in \eref{eq:TwoDefectSpaces} by \lref{lem:defect} and \eref{eq:TvDeltaDefect}. Note that the third power of each entry of $\bfg_0$ or $\bfg_1$ is equal to 1, so 
	$\ind(\caT)=\ind(3)$,  which implies \eref{eq:Tindex}  given that $O\caT,\caT O\in \scrT_\defe$ for all $O\in S_3$ by \lref{lem:SemigroupT}.

	Next, suppose $\caT\in \scrT_\ns\setminus \scrT_\defe=\scrT_\ns\cap\scrT_\iso$. Then
	$\caT=\caT_O$ for some $O\in O_3(d)\setminus S_3$ and $\ind(\caT)=\ind(O)$. Now \eref{eq:Tindex} follows from \lref{lem:Oindex} below.
	
	Thanks to \eref{eq:CharIndex}, \eref{eq:Tindex} still holds if $\ind$ is replaced by $\eta_3$. 	
\end{proof}

\begin{lemma}\label{lem:Oindex}
	Suppose $d\geq 5$ is an odd prime and $O\in O_3(d)\setminus S_3$; then 
	\begin{align}\label{eq:Oindex}
	\ind(O'O )=\ind(OO')=\ind(O)\quad \forall O'\in S_3. 
	\end{align}
\end{lemma}
\begin{proof}[Proof of \lref{lem:Oindex}]
	Let $O_1$ be an arbitrary element in $O_3^\odd\setminus S_3$ and let $\bfv$ be a characteristic vector of $O_1$; then $\bfv$ is also an eigenvector of $O_1$ associated with eigenvalue $-1$. In addition, $O'\bfv$ is a characteristic vector of $O'O_1O'^\top$ and is related to $\bfv$ by a permutation. Therefore,
	\begin{align}\label{eq:OindexProof1}
	\ind(O'O_1O'^\top)&=\ind(O_1) \quad \forall O'\in S_3, \quad 
	\ind(\tau O_1\tau)=\ind(O_1),
	\end{align}
	where $\tau=\tau_{12}$ is the transposition defined in \eref{eq:zetatau}. 
	In addition, according to \lref{lem:O3Dh} and \eref{eq:Odet}, $O_3(d)$ is a dihedral group; 
	$O_1$ has order 2 and satisfies the relation
	\begin{align}
	\zeta^{2j}O_1\zeta^{-2j}=\zeta^{j}O_1 =O_1 \zeta^{-j},\quad j=0,1,2, 
	\end{align}
	where $\zeta$ is the  cyclic permutation defined in \eref{eq:zetatau}.
	The above two equations together imply  that
	\begin{align}\label{eq:OindexProof2}
	\ind(\zeta^j O_1 )=\ind(O_1\zeta^j)=\ind(O_1)\quad j=0,1,2. 
	\end{align}
	By virtue of \eqsref{eq:OindexProof1}{eq:OindexProof2} we can further deduce the following results for $j=0,1,2$,
	\begin{align}
	\begin{aligned}
	\ind(\zeta^j\tau O_1)&=\ind(\tau \zeta^j\tau O_1)=\ind(\zeta^{-j} O_1)=\ind(O_1),\\
	\ind(O_1\zeta^j\tau)&=\ind(\tau O_1\zeta^j\tau)=\ind(O_1\zeta^j)=\ind(O_1),
	\end{aligned}
	\end{align} 
	note that $O_1\zeta^j\in O_3^\odd\setminus S_3$. 
	So \eref{eq:Oindex} holds when  $O\in O_3^\odd\setminus S_3$
	given that $S_3=\{1,\zeta,\zeta^2, \tau, \zeta\tau,\zeta^2\tau\}$.

	If $O\notin O_3^\odd$, then $O$ can be expressed as $O=\tau O_1$ with  $O_1\in O_3^\odd\setminus S_3$, and  $\tau O, O\zeta^j \tau\in O_3^\odd$. By virtue of \eqsref{eq:OindexProof1}{eq:OindexProof2} we can deduce the following results for $j=0,1,2$,
	\begin{align}
	\begin{aligned}
	\ind(\zeta^j O)&=\ind(\tau \zeta^j O)=\ind(\zeta^{-j}\tau O)=\ind(\tau O)=\ind(O),\\
	\ind(O\zeta^j)&=\ind(\tau O  \zeta^j )=\ind(\tau O)=\ind(O),\\
	\ind(\zeta^j \tau O)&=\ind(\tau O)=\ind(O),\\
	\ind(O\zeta^j \tau)&=\ind(\tau O\zeta^j )=\ind(\tau O)=\ind(O),
	\end{aligned}
	\end{align}
	which confirm \eref{eq:Oindex} and complete the proof of \lref{lem:Oindex}. 
\end{proof}

\subsection{\label{app:SigmadAux}Auxiliary results on stochastic Lagrangian subspaces in $\Sigma(d)$}

\begin{lemma}\label{lem:TDelta3}
	Suppose  $d$ is an odd prime, and $\caT_1, \caT_2, \caT_3$ are three distinct stochastic Lagrangian subspaces in $\Sigma(d)$, then $\caT_1\cap \caT_2 \cap \caT_3=\spa(\mathbf{1}_6)$. 
	In addition, $(\caT_1)_\Delta=(\caT_2)_\Delta$ iff $\caT_1, \caT_2\in \scrT_0\setminus \{\Delta\}$. 
\end{lemma}

The following lemma clarifies the number of zero entries of a nonzero vector in a defect subspace; it is a simple corollary of \lref{lem:defect} and \eref{eq:NbotN}. 
\begin{lemma}\label{lem:defect0}
	Suppose $d$ is an odd prime and $\caN$ is a nontrivial defect subspace in $\bbF_d^3$.   Then any nonzero vector in $\caN$ has no entry equal to 0, and any nonzero vector in $\caN^\perp$ has at most one entry equal to 0. 
\end{lemma}

The following lemma clarifies the number of zero entries of a nonzero vector in a stochastic Lagrangian subspace in $\Sigma(d)$. 
\begin{lemma}\label{lem:TzeroEntry}
	Suppose $d$ is an odd prime, $\caT\in \scrT_\ns$, and $(\bfx;\bfy)$ is a nonzero vector in $\caT$. Then  	
	$(\bfx;\bfy)$ has at most three zero entries. If the upper bound is saturated, then $\caT\in \scrT_\defe$; in addition, either $\bfx$ or $\bfy$ has no zero entry. Furthermore,  
	any nonzero vector in $\caT_\Delta$ has at most two zero entries. 
\end{lemma}
Thanks to this lemma, for any  $\caT\in \scrT_\ns$, the projection of $\caT$ onto any three-dimensional coordinate subspace is injective and surjective unless $\caT\in \scrT_\defe$ and the coordinate subspace is associated with the first three coordinates or last three coordinates. This fact will be  crucial to proving \lref{lem:RTO} in \sref{sec:ShadowMapR}, which is a cornerstone for understanding third moments of Clifford orbits.

\begin{lemma}\label{lem:muj} 
	The functions  $\mu_j=\mu_j(d)$ for $j=0, 1,2$ defined in \eref{eq:muj} satisfy the following relations,	
	\begin{gather}\label{eq:mujLUB}
	-\frac{4\sqrt{d}}{3}\leq \mu_j-\frac{2(d-2)}{3}=	\frac{4\Re\bigl[\eta_3^2(\nu^j)G^3(\eta_3)\bigr] }{3d}
	\leq \frac{4\sqrt{d}}{3},\quad 
	\mu_j\leq \begin{cases}
	d-1 & d=7,13,\\
	d-3 & d=19, \\
	d-7 & d\geq 31, 
	\end{cases}\\	
	\mu_0+\mu_1+\mu_2=2(d-2).
	\label{eq:mu012Sum}
	\end{gather}	
	If $0\leq i<j\leq 2$ are integers and $\alpha, \beta$ are real numbers, then 
	\begin{gather}
	\left|\alpha \mu_i +\beta \mu_j-\frac{2(d-2)(\alpha+\beta)}{3}\right|\leq \frac{4\sqrt{d(\alpha^2+\beta^2-\alpha \beta)}}{3}.
	\label{eq:muijGenSum}
	\end{gather} 				
\end{lemma}

\begin{proof}[Proof of \lref{lem:TDelta3}]
	By assumption two of the three stochastic Lagrangian subspaces $\caT_1, \caT_2, \caT_3$ either belong to $\scrT_0$ simultaneously or belong to $\scrT_1$ simultaneously. Therefore, $\caT_1\cap \caT_2 \cap \caT_3=\spa(\mathbf{1}_6)$ by \lref{lem:TDelta2}.

	If $\caT_1, \caT_2\in \scrT_0\setminus \{\Delta\}$, then $(\caT_1)_\Delta=(\caT_2)_\Delta=\spa(\mathbf{1}_6)$  by \lref{lem:TDelta}. Conversely, if 
	$(\caT_1)_\Delta=(\caT_2)_\Delta$, then $\caT_1, \caT_2, \Delta$ are three distinct stochastic Lagrangian subspaces  and $(\caT_1)_\Delta=(\caT_2)_\Delta=\caT_1\cap \caT_2\cap \Delta=\spa(\mathbf{1}_6)$.
	Therefore, $\caT_1, \caT_2\in \scrT_0\setminus \{\Delta\}$ thanks to \lref{lem:TDelta} again. 
\end{proof}

\begin{proof}[Proof of \lref{lem:TzeroEntry}]
	First, suppose $\caT\in \scrT_\iso\setminus \scrT_\sym$, then $\caT=\caT_O$ for some $O\in O_3(d)\setminus S_3$,  and we have
	$\bfx=O\bfy$ and $\bfy=O^\top \bfx$. According to \pref{pro:O0entry},  $O$ has no entry equal to 0, so $(\bfx;\bfy)$ has at most two entries equal to zero given that $(\bfx;\bfy)$ is a nonzero vector by assumption.

	Next, suppose $\caT\in \Sigma(d)\setminus \scrT_\iso=\scrT_\defe$;  then $\bfx\in \caT_{\LD}^\perp$ and $\bfy\in \caT_{\RD}^\perp$. 
	According to \lref{lem:defect0},  any nonzero vector in 	$\caT_{\LD}$ or $\caT_{\RD}$ has no entry equal to 0, while any nonzero vector in 
	$\caT_{\LD}^\perp$ or $\caT_{\RD}^\perp$ has at most one entry equal to 0. If $\bfx$ has at least two entries equal to 0, then $\bfx=0$, and $\bfy\in \caT_\RD$ has no entry equal to 0. Similarly, if $\bfy$ has at least two entries equal to 0, then $\bfy=0$, and $\bfx\in \caT_\LD$ has no entry equal to 0. Therefore, $(\bfx;\bfy)$ has at most three entries equal to 0. If the upper bound is saturated, then either $\bfx$ or $\bfy$ has no entry equal to 0. 
	
	As a simple corollary of the above discussion, any nonzero vector in $\caT_\Delta$ has at most two zero entries, given that the number of zero entries is necessarily even. 
\end{proof}

\begin{proof}[Proof of \lref{lem:muj}]
	By virtue of \eref{eq:TvNS} we can deduce that
	\begin{align}
	\mu_j&=\left|\left\{\caT\in\scrT_\ns \,|\, \eta_3(\caT)= \eta_3(3\nu^j)\right\} \right|=2\left|\left\{\caT\in \scrT_1\cap\scrT_\ns \,|\, \eta_3(\caT)= \eta_3(3\nu^j)\right\} \right|\nonumber\\
	&=2\left|\left\{y\in \bbF_d\,|\, y=0, y\neq d-1, \eta_3(-3(y^2+y))
	= \eta_3(3\nu^j)\right\}\right|\nonumber\\
	&=2\left|\left\{y\in \bbF_d\,|\, y=0, y\neq d-1, \eta_3(y^2+y)=\eta_3(\nu^j)\right\}\right|=\frac{2}{3}\left|\left\{x,y\in \bbF_d\,|\, y=0, y\neq d-1, y^2+y=\nu^j x^3 \right\}\right|\nonumber\\
	&=\frac{2N(y^2+y=\nu^j x^3)-4}{3}=\frac{2d(d-2)+4\Re\bigl[\eta_3^2(\nu^j)G^3(\eta_3)\bigr] }{3d},
	\end{align} 
	which confirms the equality in \eref{eq:mujLUB}. 	Here the third equality holds because 
	$\scrT_1\cap\scrT_\ns=\{\caT_{\bfv_y}\}_{y=1}^{d-2}$, and the  last equality follows from \lref{lem:Cubic213SolNum} in \aref{sec:CubicEqFF}. The first two inequalities in  \eref{eq:mujLUB} hold because $|\eta_3^2(\nu^j)|=1$ and $|G^3(\eta_3)|=d\sqrt{d}$ (see \aref{app:GaussJacobi}). If $d\geq 44$, then the third inequality in  \eref{eq:mujLUB} follows from the second inequality; if $d< 44$, then this inequality holds by direct calculation.
	
	\Eqsref{eq:mu012Sum}{eq:muijGenSum} are simple corollaries of \eref{eq:mujLUB}.  
\end{proof}

\section{Proofs of results on the commutant of the third Clifford tensor power}

In this appendix we prove \lsref{lem:SemigroupR}-\ref{lem:RTOdiag} and \thref{thm:spanR}, which are tied to 
the commutant of the third Clifford tensor power.

\subsection{\label{app:rRTbasicProofs}Proofs of \lsref{lem:SemigroupR}-\ref{lem:RTpt}}

\begin{proof}[Proof of \lref{lem:SemigroupR}] 
	When $d=3$, the first equality in \eref{eq:RTiTjd3}  can be verified directly based on \eref{eq:SLStwo}; see also \eref{eq:rT01}. 
	The second equality in \eref{eq:RTiTjd3} follows from \eref{eq:OTTO3}
	and the fact that $R(\tcaT_0) = 3^nP_{\mathrm{CSS}(\tcaN)}$, given that any element in $O_3^\odd(d)$ has order 2. 
	
	When $d = 1 \mmod  3$, each stochastic Lagrangian subspace is completely determined by its left and right defect subspaces.
	The first equality in \eref{eq:RTiTjd}  follows from \eref{eq:tTijDefect} and the following facts
	\begin{align}
	R(\caT)=r(\caT)^{\otimes n},\quad r(\caT)=\sum_{(\bfx;\bfy)\in \caT} |\bfx\>\<\bfy|,\quad  r(\caT)^\dag=\sum_{(\bfx;\bfy)\in \caT} |\bfy\>\<\bfx|.
	\end{align}
	Alternatively, it can be verified directly based on \eref{eq:SLSfour}. 
	The second equality in 
	\eref{eq:RTiTjd} follows from \lref{lem:SemigroupT} and \eref{eq:tTijDefect} given that  $\tcaN_0,\tcaN_1$ are the only two nontrivial defect subspaces in $\bbF_d^3$, both of which are one dimensional.  
\end{proof}

\begin{proof}[Proof of \lref{lem:TDeltaSum}]
	The last equality in \eref{eq:TDeltaSum} can be proved as follows,
	\begin{align}
	\sum_{\caT\in \scrT_0} r(\caT_\Delta)=r(\Delta)
	+\sum_{\caT\in \scrT_0\setminus\{\Delta\}}
	r(\caT_\Delta)=\bbI+d\sum_{x\in \bbF_d} |xxx\>\<xxx|,
	\end{align}
	where the second equality holds because $r(\Delta)=\bbI$ and $\caT_\Delta=\spa(\mathbf{1}_6)$ for $\caT\in \scrT_0\setminus\{\Delta\}$ by \lref{lem:TDelta}. 
	
	To prove the second equality in \eref{eq:TDeltaSum}, note that $\mathbf{1}_3\in \caT_\tDelta$ for all $\caT\in \Sigma(d)$.  If in addition $\caT\in \scrT_1$, then $\dim \caT_\Delta=\dim \caT_\tDelta=2$ according to \eqsref{eq:TvDeltaDefect}{eq:TvSigma}  (cf. \lref{lem:TDelta}). Let $\scrV$ be the set of two dimensional subspaces of $\bbF_d^3$ that contain the vector $\mathbf{1}_3$; then $\caT_\tDelta\in \scrV$. In addition,
	\begin{align}
	|\scrV|=d+1=|\scrT_1|;  \quad  \cup_{V\in \scrV}V=\bbF_d^3;   \quad 	V_1 \cap V_2 =\spa(\mathbf{1}_3)\quad \forall V_1, V_2\in \scrV,\; V_1\neq V_2. 
	\end{align}
	Furthermore, any vector in $\bbF_d^3\setminus \spa(\mathbf{1}_3)$ belongs to a unique subspace in $\scrV$. 
	Now consider the map $\caT\mapsto \caT_\tDelta$ from $\scrT_1$ to $\scrV$. According to \eqsref{eq:TvDeltaDefect}{eq:TvSigma} (see also 
	\lref{lem:TDelta3}), this map is injective and is thus also surjective given that $|\scrV|=|\scrT_1|$. In conjunction with the definition in \eref{eq:TDeltaDef} we can deduce that
	\begin{align}
	\sum_{\caT\in \scrT_1} r(\caT_\Delta)	=\sum_{\caT\in \scrT_1} \sum_{\bfx\in \caT_\tDelta} |\bfx\>\<\bfx|=
	\sum_{V\in \scrV}\sum_{\bfx\in V} |\bfx\>\<\bfx|=\bbI+d\sum_{x\in \bbF_d} |xxx\>\<xxx|=\sum_{\caT\in \scrT_0} r(\caT_\Delta),
	\end{align}
	which confirms the second equality in \eref{eq:TDeltaSum}.

	Finally, the first equality in \eref{eq:TDeltaSum} is a simple corollary of the second equality because  $\Sigma(d)=\scrT_0\sqcup\scrT_1$ and  $\sum_{\caT\in \Sigma(d)} r(\caT_\Delta)=\sum_{\caT\in \scrT_0} r(\caT_\Delta)+\sum_{\caT\in \scrT_1} r(\caT_\Delta)$, which completes the proof of \lref{lem:TDeltaSum}. 
\end{proof}

\begin{proof}[Proof of \lref{lem:RTpt}]
	To simplify the notation, here we focus on the case $n=1$, which means $R(\caT)=r(\caT)$, but the basic idea admits straightforward generalization. 	
	First, suppose $\caT\in \scrT_1\cap \scrT_\ns$, then either $\caT=\caT_O$ with $O\in O_3^\odd\setminus S_3$ or $\caT\in \scrT_1\cap\scrT_\defe$. In the first case we have
	\begin{align}
	\tr_A R(\caT)=\sum_{\bfx\in \bbF_d^3} \tr_A (|O\bfx\>\<\bfx|)=\sum_{\bfx\in \bbF_d^3\,|\, [(O-\mathds{1})\bfx]_1=0} \tr_A (|O\bfx\>\<\bfx|)=\sum_{\bfx\in \ker(O-\mathds{1})} \tr_A (|\bfx\>\<\bfx|)=\bbI,
	\end{align}
	which confirms \eref{eq:RTpt}. 
	Here the third equality holds because the three rows of $O-\mathds{1}$ are proportional to each other by \lref{lem:kerOI2}; the last equality holds because $\ker(O-\mathds{1})$ is two dimensional and any nonzero vector has at most one entry equal to zero by \lref{lem:kerOI},  which means the projection of $\ker(O-\mathds{1})$ onto any two-dimensional coordinate subspace is both 
	injective and surjective.

	If $\caT\in \scrT_1\cap\scrT_\defe$, then $\caT$ is of CSS type and has the form in \eref{eq:Tcss}. Therefore, 
	\begin{align}
	\tr_A R(\caT)&=\sum_{[\bfx] \in \caN^\perp/\caN, \ \bfz, \mathbf{w} \in \caN} \tr_A(|\bfx+\bfz\>\<\bfx+\mathbf{w}|)
	=\sum_{[\bfx] \in \caN^\perp/\caN, \ \bfz\in \caN} \tr_A(|\bfx+\bfz\>\<\bfx+\bfz|)\nonumber\\
	&=\sum_{\bfy \in \caN^\perp} \tr_A(|\bfy\>\<\bfy|)
	=\bbI,
	\end{align}
	which confirms \eref{eq:RTpt}. 
	Here the second equality holds because any nonzero vector in $\caN$ has no entry equal to 0 by \lref{lem:defect0}; the last equality holds because $\caN^\perp$ is two dimensional and any nonzero vector has at most one entry equal to zero by  \lref{lem:defect0} again. Alternatively, the  result $\tr_A R(\caT)=\bbI$ can be verified directly by virtue of \lref{lem:defectT}.

	Next, suppose $\caT\in \scrT_0\setminus \{\Delta\}$. If $\caT=\caT_O$ with $O=\zeta$ or $O=\zeta^2$, then it is straightforward to verify that $\tr_A R(\caT)=
	\mathrm{SWAP}$, which confirms \eref{eq:RTpt}. 
	Otherwise,  $\caT\in \scrT_0\setminus \scrT_\sym$  can be expressed as $\caT=\tau_{23} \caT'$, where $\caT'\in \scrT_1\cap \scrT_\ns$. According to the above analysis, we have  $\tr_A R(\caT')=\bbI$, which means $\tr_A R(\caT)=
	\mathrm{SWAP}$ and confirms \eref{eq:RTpt}.
	
	When $\caT\in \scrT_\sym$, \eref{eq:RTpt2} can be verified by direct calculation; otherwise, \eref{eq:RTpt2} follows from \eref{eq:RTpt}. 
	
	Finally, when the partial trace is taken with respect to 
	other parties,  the same conclusions can be proved by virtue  of a similar reasoning as presented above. This observation completes the proof of \lref{lem:RTpt}.	
\end{proof}

\subsection{\label{app:SpanProof}Proof of  \thref{thm:spanR}}

\begin{proof}[Proof of \thref{thm:spanR}]
	The dimension of the commutant of $\Cl(n,d)^{\totimes 3}$ is equal to the third frame potential of the Clifford group, which is determined in \rcite{Zhu17MC} as reproduced in \eref{eq:CliffordFP3} in the main text.	
	Let $\caT_i$ for $i=1,2,\ldots, 2(d+1)$ be 
	the $2(d+1)$ elements in $\Sigma(d)$ and let 	
	$\Gamma$ be the Gram matrix of the set $\Sigma(d)$, whose entries read
	\begin{equation}\label{eq:Gammaij}
	\Gamma_{ij}: = \tr[R(\caT_i)^\dag R(\caT_j)]=\tr[R(\caT_i)^\top R(\caT_j)],\quad \caT_i, \caT_j\in \Sigma(d). 
	\end{equation}
	Then we have
	\begin{align}
	\dim\spa\bigl(\{R(\caT)\}_{\caT \in \Sigma(d)}\bigr)=\rk \Gamma. 
	\end{align}
	Similarly, the dimension of the span of any subset in $\{R(\caT)\}_{\caT \in \Sigma(d)}$ is equal to the rank of the corresponding principal submatrix of $\Gamma$. So \thref{thm:spanR} follows from \eref{eq:CliffordFP3}  and  \lref{lem:GammaRank} below.
\end{proof}

\begin{lemma}\label{lem:GammaRank}
	Suppose $d$ is an odd prime and $\Gamma$ is the Gram matrix of 
	$\{R(\caT)\}_{\caT \in \Sigma(d)}$ as defined in \eref{eq:Gammaij}. Then 
	\begin{equation}\label{eq:GammaRank}
	\rk \Gamma = \begin{cases}
	2d+1, \quad n=1,\\
	2d+2, \quad n \geq 2.
	\end{cases}
	\end{equation}
	In addition, any $(2d+1)\times (2d+1)$ principal submatrix of $\Gamma$ has rank $2d+1$. 
\end{lemma}
\begin{proof}[Proof of \lref{lem:GammaRank}]
	Without loss of generality, we can assume that $\caT_i\in \scrT_0$ for $i=1,2,\ldots, d+1$ and  $\caT_i\in \scrT_1$ for $i=d+2,d+3,\ldots, 2d+2$. Then by virtue of \pref{pro:RTT1T2tr}  we can deduce that
	\begin{align}
	\Gamma_{ij}=\begin{cases}
	D^3 & i=j; \\
	D &  i\neq j, i, j\leq d+1 \; \mbox{or}\; i\neq j, i,j> d+1; \\
	D^2 & i\leq d+1, j>d+1\;\mbox{or}\; i>d+1, j\leq d+1. 
	\end{cases}
	\end{align} 
	In other words, $\Gamma$ has the form
	\begin{equation}
	\Gamma = \left(\begin{array}{cccc|cccc}
	D^3 & D & \cdots & D  & D^2 & D^2 & \cdots & D^2 \\ 
	D & D^3 & \cdots & D  & D^2 & D^2 & \cdots & D^2  \\ 
	\vdots & \vdots & \ddots & \vdots  & \vdots & \vdots & \ddots & \vdots  \\ 
	D & D & \cdots & D^3 & D^2 & D^2 & \cdots & D^2 \\ 
	& & & & & & & \\
	\hline
	& & & & & & & \\
	D^2 & D^2 & \cdots & D^2 & D^3 & D & \cdots & D  \\ 
	D^2 & D^2 & \cdots & D^2 & D & D^3 & \cdots & D    \\ 
	\vdots & \vdots & \ddots & \vdots  & \vdots & \vdots & \ddots & \vdots  \\ 
	D^2 & D^2 & \cdots & D^2 & D & D & \cdots & D^3  \\ 
	\end{array}\right). 
	\end{equation} 
	It is easy to verify that
	$\Gamma$ has three distinct eigenvalues, $D^3-D$, $D(D+1)(D+d)$, and $D(D-1)(D-d)$ with multiplicities $2d$, $1$, and $1$, respectively. The first eigenspace is spanned by the $2d$ vectors
	\begin{align}
	e_i-e_{i+1}, \quad i=1,2,\ldots, 2d+1, \; i\neq d+1, 
	\end{align}
	where $e_1, e_2, \ldots, e_{2d+2}$ form the standard orthonormal basis of $\bbR^{2d+2}$. 
	The second and third eigenspaces are spanned by the two vectors $\sum_{i=1}^{2d+2} e_i$ and $\sum_{i=1}^{d+1} e_i-\sum_{i=d+2}^{2d+2} e_i$, respectively. These observations imply \eref{eq:GammaRank}, given that $\rk \Gamma$ is equal to the number of nonzero eigenvalues of $\Gamma$.

	Next, let $\Gamma'$ be the   $(2d+1)\times (2d+1)$ principal submatrix of $\Gamma$  obtained by deleting the last row  and last column. Then direct calculation shows that all eigenvalues of  $\Gamma'$ are positive, so $\Gamma'$ is nonsingular and has rank $2d+1$. This conclusion also follows from the above analysis and Cauchy interlacing theorem (see Theorem 4.3.17 in \rcite{HornJ13book}). Since all  $(2d+1)\times (2d+1)$ principal submatrices of $\Gamma$ are similar to each other, it follows that all $(2d+1)\times (2d+1)$ principal submatrices of $\Gamma$ have the same rank of $2d+1$. 
\end{proof}

\subsection{\label{app:RTisodefeProof}Proof of \lref{lem:RTisodefe}}

\begin{proof}[Proof of \lref{lem:RTisodefe}]
	By definition we have
	\begin{gather}\label{eq:RisoRdef}
	R(\scrT_\sym)=\sum_{\caT\in \scrT_\sym} R(\caT)=\sum_{O\in S_3}R(O)=6P_\sym, \\
	R(\Sigma(d))=R(\scrT_\iso)+R(\scrT_\defe),
	\quad 
	R(\scrT_\iso)=\sum_{\caT\in \scrT_\iso} R(\caT)=\sum_{O\in O_3(d)} R(O),\\   
	R(\scrT_\defe)=\sum_{\caT\in \scrT_\defe} R(\caT)=\begin{cases}
	R(\tcaT_0)+R(\tcaT_1) & d=3,\\
	R(\tcaT_{00})+R(\tcaT_{01})+R(\tcaT_{10})+R(\tcaT_{11}) & d=1\mmod 3,\\
	0 & d=2\mmod 3. 
	\end{cases}
	\end{gather}	
	So  $R(\scrT_\sym)/6$ and $R(\scrT_\iso)/|O_3(d)|$ are projectors, which means  $\|R(\scrT_\sym)\|=6$ and  $\|R(\scrT_\iso)\|=|O_3(d)|$, where $|O_3(d)|$ is determined in \lref{lem:O3Dh}. In addition,  $\supp R(\scrT_\iso)\leq\supp P_\sym= \Sym_3\bigl(\caH_d^{\otimes n}\bigr)$ given that $S_3$ is a subgroup of $O_3(d)$. 
	In conjunction with \lsref{lem:SemigroupT} and \ref{lem:SemigroupR} we can deduce that $R(\scrT_\defe)^\dag =R(\scrT_\defe)$ and 
	\begin{gather}
	R(\scrT_\defe)^2=\begin{cases}
	2D R(\scrT_\defe) & d=3,\\
	(2D+2)R(\scrT_\defe) & d=1\mmod 3,
	\end{cases}\quad 
	\|R(\scrT_\defe)\|=\begin{cases}
	2D  & d=3,\\
	2D+2& d=1\mmod 3,
	\end{cases}\\
	R(\scrT_\iso)R(\scrT_\defe)=R(\scrT_\defe)R(\scrT_\iso)=|O_3(d)|R(\scrT_\defe).
	\end{gather}	
	Therefore, $R(\scrT_\defe)$ is also proportional to a projector, and its support  is contained in the support of $R(\scrT_\iso)$. As a corollary, $R(\Sigma(d))$ is a positive operator that has  the same support as $R(\scrT_\iso)$, which confirms \eref{eq:RTisodefeSupp}. In addition, $\|R(\Sigma(d))\|=\|R(\scrT_\iso)\|+\|R(\scrT_\defe)\|=|O_3(d)|+\|R(\scrT_\defe)\|$. The above analysis  confirms the results on operator norms in \tref{tab:RisoRdefe}.

	The  traces of the three operators $R(\scrT_\iso)$, $R(\scrT_\defe)$, and $R(\Sigma(d))$  presented in   \tref{tab:RisoRdefe} follow from  \lref{lem:O3Dh}, \eref{eq:Sigma33DefectNum}, and
	\pref{pro:RTT1T2tr}, given that
	\begin{gather}
	\Sigma(d)=\scrT_0\sqcup\scrT_1,\quad |\Sigma(d)|=2|\scrT_0|=2|\scrT_1|=2(d+1),\\ |\scrT_\iso|=2|\scrT_\iso\cap\scrT_0|=2|\scrT_\iso\cap\scrT_1|=|O_3(d)|,\quad  |\scrT_\defe|=2|\scrT_\defe\cap\scrT_0|=2|\scrT_\defe\cap\scrT_1|. 
	\end{gather}
	Finally, the ranks of the three operators can be calculated as follows, 
	\begin{align}
	\rk R(\Sigma(d))=\rk R(\scrT_\iso)=\frac{\tr R(\scrT_\iso)}{\|R(\scrT_\iso)\|}, \quad \rk R(\scrT_\defe)=\frac{\tr R(\scrT_\defe)}{\|R(\scrT_\defe)\|},
	\end{align}
	which completes the proof of \lref{lem:RTisodefe}. 
\end{proof}

\subsection{\label{app:RTOproof}Proofs of  \lsref{lem:RTObStabProj}-\ref{lem:RTOdiag}}

\begin{proof}[Proof of \lref{lem:RTObStabProj}]
	Since all stabilizer projectors of  rank $K$ form one orbit under the action of the Clifford group, without loss of generality we can assume that $\Ob$ has the form  $\Ob=\Ob_1 \otimes \bbI_2$, where  $\bbI_2$ is the identity operator on  $\caH_d^{\otimes n_2}$ with  $n_2=\log_d K$, and $\Ob_1$ is a rank-1 stabilizer projector on $\caH_d^{\otimes n_1}$ with  $n_1=n-n_2$.
	Let  $\bbI_1$ be the identity operator on $\caH_d^{\otimes n_1}$; then \eref{eq:RTOstabProj} can be proved as follows,
	\begin{align}
	\caR_\caT(\Ob)&=\caR_\caT(\Ob_1\otimes \bbI_2)=\caR_\caT(\Ob_1)\otimes\caR_\caT(\bbI_2)=\begin{cases}
	K^2\bbI_1\otimes \bbI_2 &\mbox{if}\;\;  \caT=\Delta,\\
	K \bbI_1\otimes \bbI_2 &\mbox{if}\;\; \caT=\caT_{\tau_{23}}, \\
	\Ob_1 \otimes \bbI_2&\mbox{if}\;\; \caT\in \scrT_0\setminus \{\Delta\},\\
	K\Ob_1\otimes \bbI_2 &	\mbox{if}\;\; \caT\in \scrT_1\setminus\{\caT_{\tau_{23}}\}.
	\end{cases} 
	\end{align} 
	Here the third equality follows from \lref{lem:RTpt} in the main text and \lref{lem:RTObStab} below since $\caR_\caT(\bbI)=\tr_{BC}R(\caT)$.
\end{proof}

Next, we prove an auxiliary lemma used in the proof of \lref{lem:RTObStabProj}. 
\begin{lemma}\label{lem:RTObStab}
	Suppose $d$ is an odd prime and $\Ob$ is the projector onto a stabilizer state in $\caH_d^{\otimes n}$. Then 
	\begin{align}
	\caR_\caT(\Ob)&=\begin{cases}
	\bbI &\mbox{if}\quad  \caT\in\{\Delta,\caT_{\tau_{23}}\}, \\
	\Ob &\mbox{if}\quad \caT\in \Sigma(d)\setminus  \{\Delta,\caT_{\tau_{23}}\}.
	\end{cases} \label{eq:RTOstab}
	\end{align}
\end{lemma}

\begin{proof}[Proof of \lref{lem:RTObStab}]
	To simplify the notation, here we focus on the case $n=1$, but the basic idea admits straightforward generalization. Since all stabilizer states form one orbit under the action of the Clifford group, we can assume that  $\Ob=|0\>\<0|$ 
	without loss of generality. If in addition $\caT\in \scrT_\ns$, then
	\begin{align}
	\caR_\caT(|0\>\<0|) &= \sum_{(\bfx;\bfy)\in \caT}\tr_{BC}\bigl[\lsp|\bfx\>\<\bfy|(I\otimes |0\>\<0|\otimes |0\>\<0|) \bigr]=|0\>\<0|,
	\end{align}
	which confirms \eref{eq:RTOstab}. 
	Here in deriving the second equality we have noted that $(\bfx;\bfy)=\mathbf{0}$ whenever four entries are equal to 0 according to \lref{lem:TzeroEntry}. 
	
	When  $\caT\in \scrT_\sym$,  \eref{eq:RTOstab} can be verified by direct calculation.
\end{proof}

Before proving  \lref{lem:RTO}, we need to introduce an auxiliary lemma. Let $M$ be a $k\times k$ complex matrix. The maximum column sum  norm $\vertiii{M}_1$ and 
maximum row sum norm $\vertiii{M}_\infty$ are defined as follows \cite{HornJ13book},
\begin{align}
\vertiii{M}_1:=\max_{1\leq j\leq k}\sum_{i} |M_{ij}|,\quad \vertiii{M}_\infty:=\max_{1\leq i\leq k}\sum_{j} |M_{ij}|.
\end{align}

\begin{lemma}\label{lem:Matrix3norms}
	Any  $k\times k$ complex matrix $M$ satisfies 
	\begin{align}\label{eq:Matrix3norms}
	\|M\| \leq \sqrt{\vertiii{M}_1 \vertiii{M}_\infty}\leq \frac{\vertiii{M}_1+ \vertiii{M}_\infty}{2}. 
	\end{align}
\end{lemma}
\begin{proof}[Proof of \lref{lem:Matrix3norms}]
	The second inequality in \eref{eq:Matrix3norms}	follows from the well-known relation between the geometric mean and arithmetic mean. The first inequality in \eref{eq:Matrix3norms} follows from 5.6.P21 in \rcite{HornJ13book}. Alternatively, it is a simple corollary of the following equation,
	\begin{align}
	\|M\|^2 =\|M^\dag M\|\leq \vertiii{M^\dag M}_\infty\leq \vertiii{M^\dag}_\infty \vertiii{M}_\infty=\vertiii{M}_1 \vertiii{M}_\infty.
	\end{align}	
	Here the first inequality  holds because  $\|M^\dag M\|$ is equal to the spectral radius of the Hermitian matrix $M^\dag M$, which is upper bounded by $\vertiii{M^\dag M}_1$ and $\vertiii{M^\dag M}_\infty$ according to Corollary 6.1.5 in \rcite{HornJ13book}. The second inequality holds because the matrix norms $\vertiii{M}_1$ and  $\vertiii{M}_\infty$ are submultiplicative. 
\end{proof}

\begin{proof}[Proof of \lref{lem:RTO}]
	First, we focus on the case $n=1$. By definition in \eref{eq:ShadowMapR} we have
	\begin{align}
	\caR_\caT(\Ob) = \tr_{BC}\bigl[R(\caT)\bigl(I\otimes \Ob\otimes \Ob^\dag\bigr) \bigr]
	=\!\sum_{(\bfx;\bfy)\in \caT}\tr_{BC}\bigl[\lsp|\bfx\>\<\bfy|\bigl(I\otimes \Ob\otimes \Ob^\dag\bigr) \bigr]
	=\!\sum_{(\bfx;\bfy)\in \caT} \Ob_{y_2, x_2}\bigl(\Ob^\dag\bigr)_{y_3, x_3} |x_1\>\<y_1|,
	\end{align}
	where $x_1, x_2, x_3$ are the three entries of $\bfx$, and $y_1, y_2, y_3$ are the three entries of $\bfy$.

	To understand the entries of $	\caR_\caT(\Ob)$  in the computational basis, we introduce the following two linear maps from $\caT$ to $\bbF_d^3$, 
	\begin{align}\label{eq:CoordinateProjMap}
	f_i: (\bfx;\bfy)\mapsto (x_1;x_i;y_i),\quad i=2,3.
	\end{align}
	The kernels of the two maps read
	\begin{align}
	\ker(f_i)=\{(\bfx;\bfy)\in \caT\ |\ x_1=x_i=y_i=0\}, \quad i=2,3.
	\end{align}	
	According to \lref{lem:TzeroEntry}, $\ker(f_i)$ for $i=2,3$ are both trivial, so both maps are injective and surjective, given that each $\caT\in \Sigma(d)$ is three dimensional. Consequently, 
	\begin{align}
	\{(y_3;x_3)\ |\ (\bfx;\bfy)\in \caT, \;x_1 =a \}=\{(y_2;x_2)\ |\ (\bfx;\bfy)\in \caT, \;x_1 =a \}=\bbF_d^2\quad \forall a\in \bbF_d. 
	\end{align}
	Therefore,
	\begin{align}
	\sum_{b\in \bbF_d}|\caR_\caT(\Ob)_{a,b}|\leq \sum_{b,c\in \bbF_d}|\Ob_{b,c}| |\Ob_{\sigma(b,c)}|\leq \sum_{b,c\in \bbF_d}|\Ob_{b,c}|^2=\|\Ob\|_2^2\quad \forall a\in \bbF_d,
	\end{align}
	where $\sigma$ is a permutation on $\bbF_d^2$, which may depend on $a$. By a similar reasoning we can deduce that 
	\begin{align}
	\sum_{b\in \bbF_d}|\caR_\caT(\Ob)_{b,a}|\leq \|\Ob\|_2^2\quad \forall a\in \bbF_d. 
	\end{align}
	The above two equations together 	
	imply \eref{eq:RTOUB1}	thanks to \lref{lem:Matrix3norms}.

	Next, we consider the general situation with $n\geq 1$. Given $\caT\in \Sigma(d)$, let 
	\begin{align}
	\caT^n:=\{(\tbfx;\tbfy) \,|\, (\bfx_j;\bfy_j)\in \caT \quad \forall j=1,2,\ldots, n\},
	\end{align}
	where $\tbfx$ and $\tbfy$  are shorthands for $(\bfx_1; \bfx_2;\ldots; \bfx_n)$ and $(\bfy_1; \bfy_2;\ldots; \bfy_n)$, respectively, with $\bfx_j=(x_{j1}; x_{j2};x_{j3})$  and $\bfy_j=(y_{j1}; y_{j2};y_{j3})$. Note that $\bfx_j,\bfy_j\in \bbF_d^3$ for $j=1,2,\ldots, n$ and $\bfx,\bfy\in \bbF_d^{3n}$. In addition, the basis vectors in the computational basis of $\caH_d^{\otimes n}$ can be labeled by vectors in $\bbF_d^n$.
	Then 
	\begin{align}
	\caR_\caT(\Ob) &= \tr_{BC}\bigl[R(\caT)\bigl(I\otimes \Ob\otimes \Ob^\dag\bigr) \bigr]
	=\sum_{(\tbfx;\tbfy)\in \caT^{n}}\tr_{BC}\bigl[\lsp|\tbfx\>\<\tbfy|\bigl(I\otimes \Ob\otimes \Ob^\dag\bigr) \bigr]
	\nonumber\\
	&=\sum_{(\tbfx;\tbfy)\in \caT^n} \Ob_{\ty_2, \tx_2}\bigl(\Ob^\dag\bigr)_{\ty_3, \tx_3} |\tx_1\>\<\ty_1|,
	\end{align}
	where  $\tx_i$ and $\ty_i$ for  $i=1,2,3$ are shorthands for $(x_{1i}; x_{2i};\ldots; x_{ni})$ and  $(y_{1i}; y_{2i};\ldots; y_{ni})$, respectively.

	Now, we can introduce  two linear maps from $\caT^n$ to $\bbF_d^{3n}$ as follows,
	\begin{align}
	\tf_i: (\tbfx;\tbfy)\mapsto (\tx_1;\tx_i;\ty_i),\quad i=2,3,
	\end{align}
	which generalize the two linear maps $f_i$ defined in \eref{eq:CoordinateProjMap}.
	Moreover, both maps are injective and surjective as before. Consequently, 
	\begin{align}
	\{(\ty_3;\tx_3)\ |\ (\tbfx;\tbfy)\in \caT^n, \;\tx_1 =a \}=\{(\ty_2;\tx_2)\ |\ (\tbfx;\tbfy)\in \caT^n, \;\tx_1 =a \}=\bbF_d^{2n}\quad \forall a\in \bbF_d^n. 
	\end{align}
	Therefore,
	\begin{align}
	\sum_{b\in \bbF_d^n}|\caR_\caT(\Ob)_{a,b}|\leq \sum_{b,c\in \bbF_d^n}|\Ob_{b,c}| |\Ob_{\sigma(b,c)}|\leq \sum_{b,c\in \bbF_d^n}|\Ob_{b,c}|^2=\|\Ob\|_2^2\quad \forall a\in \bbF_d^n,
	\end{align}
	where $\sigma$ is a permutation on $\bbF_d^{2n}$, which may depend on $a$. By a similar reasoning we can deduce that 
	\begin{align}
	\sum_{b\in \bbF_d^n}|\caR_\caT(\Ob)_{b,a}|\leq \|\Ob\|_2^2\quad \forall a\in \bbF_d^n. 
	\end{align}
	The above two equations together 	
	imply \eref{eq:RTOUB1}	thanks to \lref{lem:Matrix3norms} again. 
	
	\Eref{eq:RTOUB2} follows from \eqsref{eq:RTOS3norm}{eq:RTOUB1}.
\end{proof}

Before proving \lref{lem:RTOdiag} we need to introduce an auxiliary lemma.  
\begin{lemma}\label{lem:TDeltaij}
	Suppose $d$ is an odd prime and 
	\begin{align}
	\Delta_{ij}=\bigl\{(\bfx; \bfy)\ |\ \bfx,\bfy\in \bbF_d^3,\; x_i=y_i,\; x_j=y_j\bigr\},\quad i,j=1,2,3,\; i\neq j,
	\end{align}
	where $x_1, x_2, x_3$ are the three entries of $\bfx$, and $y_1, y_2, y_3$ are the three entries of $\bfy$. Then 
	\begin{align}\label{eq:TDeltaij}
	\caT\cap\Delta_{ij}=\caT_\Delta\quad \forall i,j=1,2,3,\; i\neq j.
	\end{align}
\end{lemma}
This lemma is a simple corollary of the following equation
\begin{equation}\label{eq:bfxdotbfy}
\bfx\cdot \mathbf{1}_{t}=\bfy\cdot \mathbf{1}_{t}\quad \forall (\bfx;\bfy)\in \caT,
\end{equation}
which follows from  Conditions~1 and 3 
in the definition of stochastic Lagrangian subspaces in \sref{sec:SLSspanningSet}.

\begin{proof}[Proof of \lref{lem:RTOdiag}]
	To simplify the notation, here we focus on the case $n=1$, but the basic idea admits straightforward generalization (cf. the proof of \lref{lem:RTO}). Since all stabilizer bases form one orbit under the action of the Clifford group, we can assume that  $\Ob$ 
	is diagonal in the computational basis without loss of generality. Then 
	$\Ob$ has the form $\Ob=\sum_{x\in \bbF_d}o_x |x\>\<x|$,
	where $o_x$ for $x\in \bbF_d$ are the eigenvalues of $\Ob$. Therefore,
	\begin{align}
	\caR_\caT(\Ob) &= \sum_{(\bfx;\bfy)\in \caT}\tr_{BC}\bigl[\lsp|\bfx\>\<\bfy|\bigl(I\otimes \Ob\otimes \Ob^\dag\bigr) \bigr]=\sum_{(\bfx;\bfy)\in \caT\cap\Delta_{23}}\tr_{BC}\bigl[\lsp|\bfx\>\<\bfy|\bigl(I\otimes \Ob\otimes \Ob^\dag\bigr) \bigr]\nonumber\\
	&=\sum_{(\bfx;\bfx)\in \caT_\Delta}\tr_{BC}\bigl[\lsp|\bfx\>\<\bfx|\bigl(I\otimes \Ob\otimes \Ob^\dag\bigr) \bigr]=\sum_{x\in \bbF_d}\tr_{BC}\bigl[\lsp|x\mathbf{1}_3\>\<x\mathbf{1}_3|\bigl(I\otimes \Ob\otimes \Ob^\dag\bigr) \bigr]\nonumber\\
	&=\sum_{x\in \bbF_d}o_x^2 |x\>\<x|=\Ob\Ob^\dag=\Ob^\dag\Ob,
	\end{align}
	which implies \eref{eq:RTOeven}. Here the second inequality holds because $\Ob$ is diagonal in the computational basis, the third inequality follows from  \lref{lem:TDeltaij} above,
	and  the  fourth equality holds because $\caT_\Delta=\spa(\mathbf{1}_6)$ for $\caT\in \scrT_0\setminus\{\Delta\}$ by  \lref{lem:TDelta}. 
\end{proof}

\section{\label{app:ThirdMomentStab}Proofs of results on the third moment of stabilizer states}
In this appendix we prove \lref{lem:QObStabProj} and \thref{thm:ShNormStab} after introducing two auxiliary lemmas in \aref{app:QndAux}.

\subsection{\label{app:QndAux}Auxiliary results on the shadow map $\bcaQ_{n,d}(\cdot)$ }

\begin{lemma}\label{lem:QndObOb0}
	Suppose $d$ is a prime, $\Ob\in \caL^\rmH\bigl(\caH_d^{\otimes n}\bigr)$, and $\Ob_0=\Ob-\tr(\Ob)\bbI/D$. Then 
	\begin{align}
	6\bcaQ_{n,d}(\Ob_0)
	&=	6\bcaQ_{n,d}(\Ob)-\frac{(D-1)(D+2)(\tr \Ob)^2}{D^2}\bbI-\frac{2(D+2)\tr(\Ob)}{D}\Ob.
	\label{eq:QndObOb0}
	\end{align}
\end{lemma}

\begin{proof}[Proof of \lref{lem:QndObOb0}]
	The set of stabilizer states forms a 2-design, so   $3\tr_{C}[\bQ(n,d,3)]=(D+2)P_{[2]}$ and
	\begin{align}
	\begin{aligned}
	\tr_{BC}[\bQ(n,d,3)(\bbI\otimes \bbI\otimes \Ob)]&=\tr_{BC}[\bQ(n,d,3)(\bbI\otimes \Ob\otimes \bbI)],\quad 6\bcaQ_{n,d}(\bbI)=(D+1)(D+2)\bbI,\\
	6\tr_{BC}[\bQ(n,d,3)(\bbI\otimes \Ob\otimes \bbI)]&=2(D+2)\tr_B[P_{[2]}(\bbI\otimes \Ob)]=(D+2)[\tr(\Ob) \bbI+\Ob].
	\end{aligned}
	\end{align}	
	Consequently,
	\begin{align}
	6\bcaQ_{n,d}(\Ob_0)&=6\tr_{BC}\biggl\{\bQ(n,d,3)\biggl[\bbI\otimes \biggl(\Ob-\frac{\tr(\Ob)}{D}\bbI\biggr)\otimes \biggl(\Ob-\frac{\tr(\Ob)}{D}\bbI\biggr) \biggr]\biggr\}\nonumber\\
	&=6\bcaQ_{n,d}(\Ob)+\frac{6(\tr\Ob)^2}{D^2}\bcaQ_{n,d}(\bbI)-\frac{12\tr(\Ob)}{D}\tr_{BC}\bigl[\bQ(n,d,3)(\bbI\otimes \Ob\otimes \bbI)\bigr]\nonumber\\
	&=6\bcaQ_{n,d}(\Ob)+\frac{(D+1)(D+2)(\tr \Ob)^2}{D^2}\bbI-\frac{2(D+2)\tr(\Ob)[\tr(\Ob) \bbI+\Ob]}{D}  \nonumber\\
	&=	6\bcaQ_{n,d}(\Ob)-\frac{(D-1)(D+2)(\tr \Ob)^2}{D^2}\bbI-\frac{2(D+2)\tr(\Ob)}{D}\Ob,
	\end{align}
	which confirms \eref{eq:QndObOb0}. 
\end{proof}

\begin{lemma}\label{lem:QndObn1}
	Suppose $d$ is a prime and $\Ob\in \caL(\caH_d)$  is diagonal in some stabilizer basis. Then 
	\begin{align}\label{eq:Qnd1}
	\bcaQ_{1,d}(\Ob)=\frac{(d+2)(|\tr\Ob|^2\bbI+d\Ob^\dag\Ob)}{d}.
	\end{align}	
\end{lemma}

\begin{proof}[Proof of \lref{lem:QndObn1}]
	By assumption $\Ob$ is diagonal in a stabilizer basis, so  $\Ob\Ob^\dag=\Ob^\dag\Ob$. 
	When $d=2$, the set of stabilizer states forms a 3-design \cite{KuenG13,Zhu17MC,Webb16}, so  $\bQ(n,d,3)$ coincides with the projector $P_\sym$ onto the symmetric subspace $\Sym_3\bigl(\caH_d^{\otimes n}\bigr)$. In conjunction with \eref{eq:PsymOb} we can deduce that
	\begin{align}
	6\bcaQ_{1,d}(\Ob)&=\Ob\Ob^\dag+\Ob^\dag\Ob+\tr(\Ob\Ob^\dag)\bbI +|\tr(\Ob)|^2\bbI
	+\tr(\Ob)\Ob^\dag+\tr\bigl(\Ob^\dag\bigr)\Ob=2\bigl(|\tr\Ob|^2\bbI+d\Ob\Ob^\dag\bigr),   
	\end{align}		
	which confirms \eref{eq:Qnd1}. 
	Here the second equality follows from the  equation below,
	\begin{align}
	\tr(\Ob\Ob^\dag)\bbI+\tr(\Ob)\Ob^\dag+\tr\bigl(\Ob^\dag\bigr)\Ob=|\tr\Ob|^2\bbI+d\Ob^\dag\Ob, 
	\end{align}
	which is easy to verify given that $\Ob$ is diagonal in a stabilizer basis (first consider the case in which $\Ob$ is traceless).

	Next, suppose $d$ is an odd prime. 	Since all stabilizer bases form one orbit under the action of the Clifford group, we can assume that $\Ob$ is diagonal in the computational basis without loss of generality. Then by virtue of   \eref{eq:QndObShNormStab} we can deduce that
	\begin{align}
	\frac{2d}{d+2}\bcaQ_{1,d}(\Ob)	&=\sum_{\caT\in \Sigma(d)}\tr_{BC}\bigl[r(\caT)\bigl(\bbI\otimes \Ob\otimes \Ob^\dag\bigr)\bigr]=\sum_{\caT\in \Sigma(d)}\tr_{BC}\bigl[r(\caT\cap\Delta_{23})\bigl(\bbI\otimes \Ob\otimes \Ob^\dag\bigr)\bigr]\nonumber\\
	&=\sum_{\caT\in \Sigma(d)}\tr_{BC}\bigl[r(\caT_\Delta)\bigl(\bbI\otimes \Ob\otimes \Ob^\dag\bigr)\bigr]=
	2\bigl(|\tr\Ob|^2\bbI+d\Ob^\dag\Ob\bigr),
	\end{align}		
	which confirms \eref{eq:Qnd1}. 	
	Here the second and fifth equalities hold because $\Ob$ is diagonal in the computational basis; 
	the third equality holds because $\caT\cap\Delta_{23}=\caT_\Delta$ for $\caT\in \Sigma(d)$ by \lref{lem:TDeltaij}; the fourth equality follows from \lref{lem:TDeltaSum}. 
\end{proof}

\subsection{Proofs of \lref{lem:QObStabProj} and \thref{thm:ShNormStab}}

\begin{proof}[Proof of \lref{lem:QObStabProj}]
	When $d=2$, the set of stabilizer states forms a 3-design \cite{KuenG13,Zhu17MC,Webb16}, which means $\bQ(n,d,3)=P_\sym$, so \eqsref{eq:QObStabProj}{eq:QO0stabProj}  follow from \eqsref{eq:PsymOb}{eq:PsymObTraceless}.

	Next, suppose  $d$ is an odd prime. By virtue of \eref{eq:QndObShNormStab} and \lref{lem:RTObStabProj} we can deduce that
	\begin{align}
	\frac{6(D+d)}{D+2}\bcaQ_{n,d}(\Ob)&=\sum_{\caT\in \Sigma(d)} \caR_\caT(\Ob)=(K^2+K)\bbI+|\scrT_0\setminus \{\Delta\}|\Ob+ |\scrT_1\setminus\{\caT_{\tau_{23}}\}|K\Ob\nonumber\\
	&=(K^2+K)\bbI+d(K+1)\Ob,
	\end{align}	
	which implies \eref{eq:QObStabProj}. Here the last equality holds because $|\scrT_1\setminus\{\caT_{\tau_{23}}\}|=|\scrT_0\setminus \{\Delta\}|=d$. \Eref{eq:QO0stabProj} follows from \eref{eq:QObStabProj} and \lref{lem:QndObOb0} with $\tr(\Ob)=K$. 
\end{proof}

\begin{proof}[Proof of \thref{thm:ShNormStab}]
	The upper bounds in  \eref{eq:StabShadowLUB} follow from \eref{eq:StabObShadowLUB} and the fact that $\| \Ob \|\leq \| \Ob \|_2$; the lower bound in  \eref{eq:StabShadowLUB} follows from  \eref{eq:StabProjShNormRatio} with $K=1$. 	So it remains to prove  \eqssref{eq:StabObShadowLUB}{eq:StabObShadowDiagn}{eq:StabObShadowDiag1}. 
	
	When $d=2$, the set of stabilizer states forms a 3-design \cite{KuenG13,Zhu17MC,Webb16}, which means $\bQ(n,d,3)=P_\sym$.  In conjunction with \eref{eq:PsymObTraceless} we can deduce that $6\|\bcaQ_{n,d}(\Ob)\|=\|\Ob\|_2^2+\|\Ob\Ob^\dag+\Ob^\dag\Ob\|$. This equation implies the following result,
	\begin{equation}
	\frac{(D+2)\| \Ob \|_2^2}{D}\leq 	6\|\bcaQ_{n,d}(\Ob)\| \leq  \| \Ob \|_2^2 +2\| \Ob \|^2\leq 3 \| \Ob \|_2^2, 
	\end{equation}
	which in turn implies \eqsref{eq:StabObShadowLUB}{eq:StabObShadowDiagn} thanks to the formula in \eref{eq:QndObShNormStab}. 
	If in addition $n=1$ (which means $	D=d=2$) and $\Ob$ is diagonal in a stabilizer basis, then 
	\begin{align}
	\|\Ob\Ob^\dag+\Ob^\dag\Ob\|=2\|\Ob\|^2, \quad \|\Ob\|_2^2=2\|\Ob\|^2, \quad  6\|\caQ_{n,d}(\Ob)\|= 4\|\Ob\|^2, 
	\end{align}
	which implies  \eref{eq:StabObShadowDiag1}.

	From now on, we assume that $d$ is an odd prime. Then 	
	\begin{align}
	6\|\bcaQ_{n,d}(\Ob)\|&=6\max_\rho\tr\bigl[\bQ(n,d,3)\bigl(\rho\otimes \Ob \otimes \Ob^\dag\bigr)\bigr]\geq 
	\frac{6}{D}\tr\bigl[\bQ(n,d,3)\bigl(\bbI\otimes \Ob \otimes \Ob^\dag\bigr)\bigr]\nonumber\\
	&=\frac{D+2}{D}\tr\bigl[(\bbI+\mathrm{SWAP})\bigl(\Ob \otimes \Ob^\dag\bigr)\bigr]=\frac{(D+2)\tr(\Ob\Ob^\dag)}{D}=\frac{(D+2)\|\Ob\|_2^2}{D},   \label{eq:QndObProof1}
	\end{align}
	where the maximization is taken over all density operators (or all pure states) on $\caH_d^{\otimes n}$. Here the second equality holds because $6\tr_A \bQ(n,d,3)=(D+2)(\bbI+\mathrm{SWAP})$
	given that the set of stabilizer states forms a 2-design \cite{Zhu17MC,Webb16}. In addition, \eref{eq:RTOS3norm} and \lref{lem:RTO} imply  that
	\begin{align}	
	\sum_{\caT\in \scrT_\sym}\|\caR_\caT(\Ob)\|=  \| \Ob \|_2^2 +2\| \Ob \|^2,   \quad 
	\sum_{\caT\in \scrT_\ns}\|\caR_\caT(\Ob)\|\leq (2d-4) \| \Ob \|_2^2, 
	\end{align}	
	given that $|\scrT_\ns|=2d-4$ and $\Ob$ is traceless. In conjunction with  \eref{eq:QndObShNormStab} we can deduce that 
	\begin{align}
	\frac{6(D+d)}{D+2} \|\bcaQ_{n,d}(\Ob)\|&
	\leq\sum_{\caT\in \Sigma(d)}\|\caR_\caT(\Ob)\|	
	=\sum_{\caT\in \scrT_\sym}\|\caR_\caT(\Ob)\|+\sum_{\caT\in \scrT_\ns}\|\caR_\caT(\Ob)\|\nonumber\\
	&\leq (2d-3) \| \Ob \|_2^2 +2\| \Ob \|^2,   \label{eq:QndObProof2}
	\end{align}	
	which implies  \eref{eq:StabObShadowLUB} given 	
	\eqsref{eq:QndObShNormStab}{eq:QndObProof1}.

	Next, suppose $\Ob$ is diagonal in a stabilizer basis. By virtue of \lsref{lem:RTO}, \ref{lem:RTOdiag} and a similar reasoning that leads to \eref{eq:QndObProof2} we can deduce that
	\begin{align}\label{eq:QndObDiagn}
	\sum_{\caT\in \scrT_\ns}\|\caR_\caT(\Ob)\|\leq (d-2) \left(\| \Ob \|_2^2+\| \Ob \|^2\right),\quad 	\frac{6(D+d)}{D+2} \|\bcaQ_{n,d}(\Ob)\| 
	&\leq(d-1) \| \Ob \|_2^2 +d\| \Ob \|^2, 
	\end{align}		
	which implies \eref{eq:StabObShadowDiagn}. 
	If in addition $n=1$, then $6\|\bcaQ_{n,d}(\Ob)\| =(d+2)\| \Ob \|^2$ thanks to \lref{lem:QndObn1}, 
	which implies \eref{eq:StabObShadowDiag1} according to \eref{eq:QndObShNormStab}.	
\end{proof}

\section{\label{app:ThirdMomentGen}Proofs of results on the third moment of a general Clifford orbit}

In this appendix we prove our main results on the third moment of a general Clifford orbit, including
\lref{lem:kappaTLUB}, \psref{pro:hkaka}, \ref{pro:orbit3designCon}, \thsref{thm:Phi3LUB}-\ref{thm:ShNormGen},  and \cosref{cor:Phi3UB}-\ref{cor:ShNormGenPos}, assuming that the local dimension $d$ is an odd prime. To this end,  some auxiliary results (\lsref{lem:hkaDeltaNS}-\ref{lem:QPsiProjEig}) are introduced in \aref{app:ThirdMomentAux}. 

\subsection{\label{app:hkakaTLUBproof}Proofs of  \lref{lem:kappaTLUB} and \pref{pro:hkaka}}

\begin{proof}[Proof of \lref{lem:kappaTLUB}]
	If $\caT\in \scrT_1$, then  $R(\caT)$ is Hermitian, which means $\kappa(\caT)$ is real. So $\kappa(\caT)$ is real for all $\caT\in \Sigma(d)$ thanks to \pref{pro:kappaSym} given that $\Sigma(d)=\scrT_1\sqcup \tau_{12}\scrT_1$, where $\tau_{12}$ is the transposition in $S_3$ as defined in \eref{eq:zetatau}.  
	
	The  inequalities $-1\leq \kappa(\caT)\leq1$ in
	\eref{eq:kappaTLUB} follow from  \lref{lem:RTO}. If  $\scrT_\defe$ is nonempty, then it contains a stochastic Lagrangian subspace $\tcaT$ of CSS type and the operator $R(\tcaT)$ is proportional to a projector, which means $0\leq \kappa(\tcaT)\leq1$.  Therefore,  $0\leq \kappa(\caT)\leq1$ for all $\caT\in \scrT_\defe$ by \pref{pro:kappaSym}
	given that all stochastic Lagrangian subspaces  in $\scrT_\defe$ form one orbit under the action of  $S_3$ according to \lref{lem:SemigroupT}. If $\caT\in \scrT_\sym$, then $\kappa(\caT)=1$ thanks to \pref{pro:kappaSym} again.

	The second inequality in \eref{eq:kappaSigLUB}  is trivial; the last  inequality    follows from \eref{eq:kappaTLUB} and the fact that $|\Sigma(d)|=2d+2$. To prove the first inequality in \eref{eq:kappaSigLUB}, note that the set of stabilizer states in $\caH_d^{\otimes n}$ forms 2-design. Let $|\Psi'\>$ be a random stabilizer state in $\Stab(n,d)$; then 
	\begin{equation}
	\begin{gathered}
	\bbE_{\Stab(n,d)} |\<\Psi'|\Psi\>|^2=\frac{1}{D},\quad 
	\bbE_{\Stab(n,d)} |\<\Psi'|\Psi\>|^4=\frac{2}{D(D+1)},\\
	\left(\bbE_{\Stab(n,d)} |\<\Psi'|\Psi\>|^4\right)^2\leq \left(\bbE_{\Stab(n,d)} |\<\Psi'|\Psi\>|^2 \right)\left(\bbE_{\Stab(n,d)} |\<\Psi'|\Psi\>|^6 \right).
	\end{gathered}
	\end{equation}
	Therefore,
	\begin{align}
	\tr\bigl[Q(n,d,3)(|\Psi\>\<\Psi|)^{\otimes 3}\bigr]=\bbE_{\Stab(n,d)} |\<\Psi'|\Psi\>|^6 \geq \frac{\left(\bbE_{\Stab(n,d)} |\<\Psi'|\Psi\>|^4\right)^2}{\bbE_{\Stab(n,d)} |\<\Psi'|\Psi\>|^2}=\frac{4}{D(D+1)^2}.
	\end{align}
	In conjunction with  \eref{eq:bQnd3} we can deduce that
	\begin{align}
	\kappa(\Sigma(d))=D(D+1)(D+d)\tr\bigl[Q(n,d,3)(|\Psi\>\<\Psi|)^{\otimes 3}\bigr]\geq \frac{4(D+d)}{D+1}, 
	\end{align}
	which confirms the first inequality in \eref{eq:kappaSigLUB}. 
	
	\Eref{eq:kappaNsLUB} follows from \eref{eq:kappaSigLUB} given that $\kappa(\scrT_\ns)=\kappa(\Sigma(d))-6$ and $|\kappa|(\scrT_\ns)=|\kappa|(\Sigma(d))-6$. 
	\Eref{eq:kappaDefIsoLUB} follows from \eqsref{eq:kappaTLUB}{eq:kappaSigLUB} given that $\Sigma(d)=\scrT_\iso\sqcup\scrT_\defe$ and $|\scrT_\defe|\leq  4$ by \eref{eq:Sigma33DefectNum}.
\end{proof}

\begin{proof}[Proof of \pref{pro:hkaka}]
	The equality in \eref{eq:hkakaTLUB} follows from \eqssref{eq:DualBasis}{eq:hkapsiT}{eq:kappaSigNSscrT}, while the two inequalities follow from \lref{lem:kappaTLUB}.

	The first equality in \eref{eq:hkakaSig} follows from \eref{eq:hkakaTLUB}  given that  $|\Sigma(d)|=2d+2$; the second equality holds because $\kappa(\Sigma(d)=\kappa(\scrT_\ns)+6$ by \eref{eq:kappaSigNSscrT}. By virtue of \eref{eq:hkakaTLUB} and the fact $|\scrT_\ns|=2(d-2)$ we can deduce that
	\begin{align}
	\hka(\scrT_\ns)&=\frac{D+2}{D-1}\left[\kappa(\scrT_\ns)-\frac{d-2}{D+d}\kappa(\Sigma(d))\right]	
	=\frac{D+2}{(D-1)(D+d)}\left[(D+2)\kappa(\scrT_\ns)-6(d-2)\right],
	\end{align}
	which implies   the last equality in \eref{eq:hkakaSig}. 
	
	Finally, by virtue of \pref{pro:kappaSym} and \eref{eq:hkakaTLUB}  we can deduce that	
	\begin{align}
	|\hka|(\Sigma(d))=|\hka|(\scrT_\sym)+|\hka|(\scrT_\ns)=6\hka(\Delta)+|\hka|(\scrT_\ns)
	=\frac{D+2}{D-1}\left[6-\frac{3\kappa(\Sigma(d))}{D+d}\right]+|\hka|(\scrT_\ns)
	\end{align}	
	given that $|\scrT_\sym|=6$, where $\Delta$ is the diagonal stochastic Lagrangian subspace defined in
	\eref{eq:Delta}. This observation  confirms \eref{eq:hkaSigNsAbs} and completes the proof of \pref{pro:hkaka}.
\end{proof}

\subsection{\label{app:ThirdMomentAux}Auxiliary results on the third moments of general Clifford orbits}

\subsubsection{Auxiliary results \lsref{lem:hkaDeltaNS}-\ref{lem:QPsiProjEig}}

The following two lemmas further clarify the relations between $\hka(\Psi,\scrT_\ns)$, $\hka(\Psi,\Delta)$, $\hka(\Psi,\scrT_\ns)$, and related functions, where $\Delta$ is the diagonal stochastic Lagrangian subspace in  $\bbF_d^{2t}$ with $t=3$ as defined in \eref{eq:Delta}.

\begin{lemma}\label{lem:hkaDeltaNS}
	Suppose $d$ is an odd prime, $\scrT$ is any subset in $\Sigma(d)$ that contains $\scrT_\sym$,  and $|\Psi\>\in\caH_d^{\otimes n}$. Then 
	\begin{gather}
	\hka(\Psi,\caT)=|\hka(\Psi,\Delta)|=\hka(\Psi,\Delta)=\frac{D+2}{D-1}
	-\frac{\hka(\Psi,\Sigma(d))}{2(D-1)}=	
	1-\frac{\hka(\Psi,\scrT_\ns)}{2(D+2)}\quad \forall \caT\in \scrT_\sym, \label{eq:hkaDeltaNS}\\
	|\hka|(\Psi,\scrT_\sym)=\hka(\Psi,\scrT_\sym)=6-\frac{3}{D+2}\hka(\Psi,\scrT_\ns), \label{eq:hkaSym}\\	
	\hka(\Psi,\Sigma(d))=6+\frac{D-1}{D+2}\hka(\Psi,\scrT_\ns)=2D+4-2(D-1)\hka(\Psi,\Delta),  \label{eq:hkaSigNS}\\
	\hka(\Psi,\scrT_\ns)=
	\frac{(D+2)^2}{(D-1)(D+d)}\kappa(\Psi,\Sigma(d))-\frac{6(D+2)}{D-1}=\frac{(D+2)^2}{(D-1)(D+d)}\kappa(\Psi,\scrT_\ns)-\frac{6(D+2)(d-2)}{(D-1)(D+d)},  \label{eq:hkaNskappaNs}\\	
	6+\frac{D-1}{D+2}|\hka|(\Psi,\scrT_\ns)\leq |\hka|(\Psi,\Sigma(d))\leq 6+\frac{D+5}{D+2}|\hka|(\Psi,\scrT_\ns),\label{eq:hkaSigAbsNS}\\
	|\hka(\Psi,\scrT)-6|\leq \frac{D}{D+2}|\hka|(\Psi,\scrT_\ns),\quad
	6-\frac{3}{D+2}\hka(\Psi,\scrT_\ns)\leq |\hka|(\Psi,\scrT)\leq 6+\frac{D+5}{D+2}|\hka|(\Psi,\scrT_\ns). \label{eq:hkaGenAbsNS}
	\end{gather}
\end{lemma}

\begin{lemma}\label{lem:kahkaDef}
	Suppose $d$ is an odd prime satisfying  $d\neq 2\mmod 3$
	and $|\Psi\>\in\caH_d^{\otimes n}$. Then 
	\begin{gather}
	\hka(\Psi,\scrT_\iso) +\hka(\Psi,\caT)\|R(\scrT_\defe)\|=2(D+2)\kappa(\Psi,\caT)\quad \forall \caT\in \scrT_\defe. \label{eq:kahkaDef}
	\end{gather}
\end{lemma}

\begin{lemma}\label{lem:kahakRelation}
	Suppose $d$ is an odd prime and $|\Psi\>\in\caH_d^{\otimes n}$. Then the following seven statements are equivalent:
	\begin{enumerate}
		\item $\hka(\Psi,\scrT_\ns)\geq 0$.
		\item $\hka(\Psi,\Sigma(d))\geq 6$. 
		\item $\hka(\Psi,\Delta)\leq 1$.
		\item $\hka(\Psi,\caT)\leq 1$ for all $\caT\in \Sigma(d)$.
		\item 	$\hka(\Psi,\caT)\leq \kappa(\Psi,\caT)$ for all $\caT\in \Sigma(d)$.
		\item $\kappa(\Psi,\Sigma(d))\geq 6(D+d)/(D+2)$. 
		\item $\kappa(\Psi,\scrT_\ns)\geq 6(d-2)/(D+2)$. 		
	\end{enumerate}	
\end{lemma}

\begin{lemma}\label{lem:hkakaNsIneq}
	Suppose $d$ is an odd prime, $|\Psi\>\in\caH_d^{\otimes n}$, $0\leq j\leq 6$, and $\scrT$ is a subset of $\Sigma(d)$ that contains $\scrT_\ns$. Then 
	\begin{align}\label{eq:hkakascrT}
	j\hka(\Psi,\Delta)+\hka(\Psi,\scrT_\ns)\leq \frac{D+2}{D+d}[j+\kappa(\Psi,\scrT_\ns)],\quad \hka(\Psi,\scrT)\leq \frac{D+2}{D+d}\kappa(\Psi,\scrT).
	\end{align}
\end{lemma}

\begin{lemma}\label{lem:hkaLB}
	Suppose $d$ is an odd prime  and $|\Psi\>\in\caH_d^{\otimes n}$. Then    
	\begin{gather}
	\kappa(\Psi,\caT)\geq \frac{2\kappa(\Psi,\scrT_\iso)}{|O_3(d)|}-1\geq  \frac{2[\kappa(\Psi,\Sigma(d))-|\scrT_\defe|]}{|O_3(d)|}-1\geq -\frac{D-3}{D+1}\quad  \forall \caT\in\scrT_\iso, \label{eq:kappaTSLB}\\
	\kappa(\Psi,\caT)\geq -\frac{D-3}{D+1},\quad  \hka(\Psi,\caT)\geq -\frac{D+2}{D+1} \quad \forall \caT\in \Sigma(d). \label{eq:hkaLB}
	\end{gather}	
\end{lemma}

\begin{lemma}\label{lem:hkaLUB}
	Suppose $d$ is an odd prime  and $|\Psi\>\in\caH_d^{\otimes n}$. Then
	\begin{gather}
	-\frac{D+2}{D+1}\leq \hka(\Psi,\caT)\leq \frac{D+2}{D+1}\quad \forall \caT\in \Sigma(d), \quad 
	\hka(\Psi,\caT)\geq \frac{D+2}{D+d} \quad \forall \caT\in \scrT_\sym, \label{eq:hkaLUB}\\
	\frac{4(D+2)}{D+1}\leq \hka(\Psi,\Sigma(d))\leq |\hka|(\Psi,\Sigma(d))\leq \frac{2(d+1)(D+2)}{D+1}\leq 2(d+2),  \label{eq:hkaSigLUB}\\	
	\frac{D+2}{D+d}\kappa(\Psi,\Sigma(d))=	\hka(\Psi,\Sigma(d))	
	\leq |\hka|(\Psi,\Sigma(d))\leq \frac{D+d-1}{D}|\kappa|(\Psi,\Sigma(d)),  \label{eq:hkaSigAbsLUB}\\		
	-\frac{2(D+2)}{D+1}\leq \hka(\Psi,\scrT_\ns)\leq \frac{D+2}{D+d}\kappa(\Psi,\scrT_\ns), \quad 
	|\hka|(\Psi,\scrT_\ns)\leq \frac{2(d-2)(D+2)}{D+1}\leq 2(d-1). \label{eq:hkaNsLUB}
	\end{gather}	
\end{lemma}

The following lemma  clarifies the approximations made in \eqsref{eq:hkakabigO}{eq:hkakabigOpos}. 
\begin{lemma}\label{lem:kappaMhkaTLUB}
	Suppose $d$ is an odd prime and $|\Psi\>\in\caH_d^{\otimes n}$. Then  
	\begin{gather}
	-\frac{1}{D}<-\frac{1}{D+1} \leq \kappa(\Psi,\caT)-\hka(\Psi,\caT)\leq \frac{d}{D-1},  \label{eq:kappaMhkaTLUB}\\
	\frac{4(d-2)}{D+1}\leq \kappa(\Psi,\Sigma(d))-\hka(\Psi,\Sigma(d))	\leq \frac{2(d-2)(d+1)}{D+d}<\frac{2d^2}{D}. \label{eq:kappaMhkaSigLUB}
	\end{gather}
	If  $d\geq 5$, then the upper bound in \eref{eq:kappaMhkaTLUB}  can be replaced by $d/D$.	
	If  $\kappa(\Psi,\caT)\geq 0$ for all $\caT\in \Sigma(d)$, then 
	\begin{gather}
	\hka(\Psi,\caT)\geq -\frac{d}{D-1},\quad  |\hka(\Psi,\caT)|-\hka(\Psi,\caT)\leq \frac{2d}{D-1},  \label{eq:hkaAbsMhkaLUB}\\
	-\frac{d}{D-1}\leq \kappa(\Psi,\caT)-|\hka(\Psi,\caT)|\leq \frac{d}{D-1}, \label{eq:kappaMhkaAbsLUB}\\ |\hka|(\Psi,\Sigma(d))-\hka(\Psi,\Sigma(d))	\leq \frac{4d^2}{D},  \label{eq:hkaAbsMhkaSigUB}\\
	-\frac{2d^2}{D}\leq \kappa(\Psi,\Sigma(d))-|\hka|(\Psi,\Sigma(d))	\leq \frac{2d^2}{D}. \label{eq:kappaMhkaAbsSigLUB}
	\end{gather}
\end{lemma}

\begin{lemma}\label{lem:QPsiObOb0}
	Suppose $d$ is an odd prime,  $|\Psi\>\in\caH_d^{\otimes n}$, $\Ob\in \caL^\rmH\bigl(\caH_d^{\otimes n}\bigr)$, and $\Ob_0=\Ob-\tr(\Ob)\bbI/D$. Then 
	\begin{align}
	6\bcaQ_{\orb(\Psi)}(\Ob_0)=6\bcaQ_{\orb(\Psi)}(\Ob)-\frac{(D-1)(D+2)(\tr \Ob)^2}{D^2}\bbI-\frac{2(D+2)\tr(\Ob)}{D}\Ob. \label{eq:QPsiObOb0}
	\end{align}
\end{lemma}
\Lref{lem:QPsiObOb0} can be proved using a similar reasoning that leads to \lref{lem:QndObOb0} given that any Clifford orbit forms a 2-design, just like the orbit of  stabilizer states.

\begin{lemma}\label{lem:QPsiProj}
	Suppose $d$ is an odd prime,  $|\Psi\>\in\caH_d^{\otimes n}$, $\Ob$ is a rank-$K$ stabilizer projector on $\caH_d^{\otimes n}$, and $\Ob_0=\Ob-\tr(\Ob)\bbI/D$. Then 
	\begin{gather}
	6\bcaQ_{\orb(\Psi)}(\Ob)=\frac{K+1}{2}\left\{2K\hka(\Psi,\Delta)\bbI+\left[\hka(\Psi,\Sigma(d))-2\hka(\Psi,\Delta)\right]\Ob\right\}, \quad 
	6\bcaQ_{\orb(\Psi)}(\Ob_0)
	=\upsilon_2 \bbI+(\upsilon_1-\upsilon_2)\Ob,
	\label{eq:QPsiProj0}
	\end{gather}	
	where $\upsilon_1$ and $\upsilon_2$ are defined in \eref{eq:QPsiProj0eig}.
\end{lemma}
Note that $\upsilon_1$ and $\upsilon_2$ are the two distinct eigenvalues of the operator $6\bcaQ_{\orb(\Psi)}(\Ob_0)$ when $1\leq K\leq D/d$.

\begin{lemma}\label{lem:QPsiProjEig}
	Suppose $1\leq K\leq D/d$ in \lref{lem:QPsiProj}. 	Then $\upsilon_1,\upsilon_2$ defined in \eref{eq:QPsiProj0eig} satisfy 
	\begin{gather}
	\frac{(D+2)(D-K)^2}{D^2(D+1)}\leq \upsilon_1\leq \frac{(D+2)(D-K)[dD+(dD-D-d)K]}{D^2(D+d)},  \label{eq:QPsiProj0eig1LUB}\\
	\frac{K}{d}\leq \frac{(D+2)K}{dD}\leq \frac{(D+2)[D^2-(dD-D-d)K]K}{D^2(D+d)}\leq \upsilon_2 \leq \frac{(D+2)(D^2+K)K}{D^2(D+1)}\leq \frac{(D+2)K}{D},  \label{eq:QPsiProj0eig2LUB}\\
	\frac{(D-K)K}{D}\leq 	\frac{(D+2)(D-K)K}{D^2}\leq \max\{\upsilon_1, \upsilon_2\}\leq \frac{(D+2)(D-K)[dD+(dD-D-d)K]}{D^2(D+d)}, \label{eq:QPsiProj0eigMaxLUB}\\
	- K\leq-\frac{(D+2)[(D+2)K-D]}{D(D+1)}\leq \upsilon_1-\upsilon_2\leq \frac{(D+2)[dD+(dD-2D-2d)K]}{D(D+2)}\leq d+(d-2)K.  \label{eq:QPsiProj0eigDif}
	\end{gather}
	The upper bound in \eref{eq:QPsiProj0eig1LUB} is saturated iff $\hka(\Psi, \Delta)=(D+2)/(D+d)$, and the same conclusion applies to the upper bound in \eref{eq:QPsiProj0eigMaxLUB}. 
	In addition, the following four inequalities  are equivalent,
	\begin{align} \label{QPsiProj0EigOrder}
	\begin{gathered}
	\upsilon_2\leq \upsilon_1,\quad \hka(\Psi,\Delta)\leq \frac{(D+2)(DK+D-2K)}{D^2(K+1)},\\
	\kappa(\Psi,\scrT_\ns)\geq \frac{2D(d-2D)+2(3dD-2D-2d)K}{D^2(K+1)},\quad 
	\hka(\Psi,\scrT_\ns)\geq -\frac{4(D+2)(D-2K)}{D^2(K+1)}.
	\end{gathered}
	\end{align}
	If $\hka(\Psi,\Delta)\geq1$ or $\hka(\Psi,\scrT_\ns)\leq 0$, then $\max\{\upsilon_1, \upsilon_2\}\leq K+2$.
\end{lemma}

\subsubsection{Proofs of \lsref{lem:hkaDeltaNS}-\ref{lem:kahakRelation}}

\begin{proof}[Proof of \lref{lem:hkaDeltaNS}]
	The first two equalities in \eref{eq:hkaDeltaNS} follow from \pref{pro:kappaSym} and the fact that $\hka(\Delta)\geq0$ by \eref{eq:hkakaTLUB}. The last two equalities hold because
	(see  \pref{pro:hkaka})
	\begin{align}
	\hka(\Sigma(d))=6\hka(\Delta)+\hka(\scrT_\ns),\quad 1=\kappa(\Delta)=\frac{D-1}{D+2}\hka(\Delta)+\frac{\hka(\Sigma(d))}{2(D+2)}.
	\end{align}	
	\Eqsref{eq:hkaSym}{eq:hkaSigNS}   are simple corollaries of \eref{eq:hkaDeltaNS} given that $|\scrT_\sym|=6$ and $\hka(\scrT_\sym)=6\hka(\Delta)$ [cf. \eref{eq:hkakaSig} in  \pref{pro:hkaka}].
	\Eref{eq:hkaNskappaNs} is a simple corollary of  \eref{eq:hkakaSig} given that $\kappa(\Sigma(d))=\kappa(\scrT_\ns)+6$ by \eref{eq:kappaSigNSscrT}. 
	
	By virtue of \eref{eq:hkaSym}  we can further deduce that
	\begin{align}	
	|\hka|(\Sigma(d))&=|\hka|(\scrT_\sym)+|\hka|(\scrT_\ns)=6-\frac{3}{D+2}\hka(\scrT_\ns)+|\hka|(\scrT_\ns), 
	\end{align}
	which implies \eref{eq:hkaSigAbsNS} given that  $|\hka(\scrT_\ns)|\leq |\hka|(\scrT_\ns)$. In addition,
	\begin{equation}
	\begin{aligned}
	&\hka(\scrT)-6=\hka(\scrT_\sym)+\hka(\scrT_\ns\cap\scrT)-6=-\frac{3}{D+2}\hka(\scrT_\ns)+ \hka(\scrT_\ns\cap\scrT)\\
	&=-\frac{3}{D+2}[\hka(\scrT_\ns\cap \scrT)+\hka(\scrT_\ns\setminus \scrT)]+ \hka(\scrT_\ns\cap\scrT)=\frac{1}{D+2}[(D-1)\hka(\scrT_\ns\cap\scrT)-3\hka(\scrT_\ns\setminus \scrT)],\\
	&|\hka(\scrT)-6|\leq \frac{D}{D+2}[\,|\hka|(\scrT_\ns\cap\scrT)+|\hka|(\scrT_\ns\setminus \scrT)]=\frac{D}{D+2}|\hka|(\scrT_\ns),
	\end{aligned}
	\end{equation}
	which confirms the first inequality in \eref{eq:hkaGenAbsNS}. 	The second inequality in \eref{eq:hkaGenAbsNS} follows from \eref{eq:hkaSym} and the fact that $|\hka|(\scrT)\geq |\hka|(\scrT_\sym)=\hka(\scrT_\sym)$; the last inequality  follows from \eref{eq:hkaSigAbsNS} and the inequality $|\hka|(\scrT)\leq 	|\hka|(\Sigma(d))$. 
\end{proof}

\begin{proof}[Proof of \lref{lem:kahkaDef}]
	Thanks to \lref{lem:SemigroupT} and \pref{pro:kappaSym},
	$\kappa(\caT)$ is independent of  $\caT$ in $\scrT_\defe$. 	
	According to \eref{eq:Sigma33DefectNum}, \lref{lem:RTisodefe},  and \tref{tab:RisoRdefe}, 
	if $d=3$, then 
	\begin{align}
	|\scrT_\defe|=2,\quad \hka(\Sigma(d))=\hka(\scrT_\iso)+\hka(\scrT_\defe)=\hka(\scrT_\iso)+2\hka(\caT),\quad \|R(\scrT_\defe)\|=2D;
	\end{align}
	if instead $d=1\mmod 3$, then 
	\begin{align}
	|\scrT_\defe|=4, \quad \hka(\Sigma(d))=\hka(\scrT_\iso)+\hka(\scrT_\defe)=\hka(\scrT_\iso)+4\hka(\caT),\quad \|R(\scrT_\defe)\|=2D+2. 
	\end{align}
	In both cases, we can deduce that
	\begin{align}
	&\hka(\scrT_\iso) +\hka(\caT)\|R(\scrT_\defe)\|=
	\hka(\Sigma(d)) +(2D-2)\hka(\caT)\nonumber\\
	&=\frac{D+2}{D+d}\kappa(\Sigma(d))+2(D+2)\left[\kappa(\caT)-\frac{\kappa(\Sigma(d))}{2D+2d}\right]=2(D+2)\kappa(\caT),
	\end{align}
	which confirms \eref{eq:kahkaDef}. 
	Here the second equality follows from \pref{pro:hkaka}. 
\end{proof}

\begin{proof}[Proof of \lref{lem:kahakRelation}]
	The equivalence of the first three statements in \lref{lem:kahakRelation} is a simple corollary of \lref{lem:hkaDeltaNS}. The equivalence of statements 3 and 4 follows from \pref{pro:hkaka} and the fact that $\kappa(\caT)\leq \kappa(\Delta)=1$ for all $\caT\in \Sigma(d)$ by  \lref{lem:kappaTLUB}. The equivalence of the three statements 2, 6, and 7 follows from the relations 
	$\kappa(\Sigma(d))=\kappa(\scrT_\ns)+6$ and
	$\hka(\Sigma(d))=(D+2)\kappa(\Sigma(d))/(D+d)$ by \pref{pro:hkaka}. If statement 5 holds, then statement~4 holds given that $\kappa(\caT)\leq 1$ 
	for all $\caT\in \Sigma(d)$ by \lref{lem:kappaTLUB}. Conversely, if statement 4 holds, then statement 6 holds, that is, $\kappa(\Sigma(d))\geq 6(D+d)/(D+2)$. Therefore, $\hka(\caT)\leq \kappa(\caT)$ by \pref{pro:hkaka} given that $\kappa(\caT)\leq 1$ for all $\caT\in \Sigma(d)$, which confirms statement 5 and completes the proof of \lref{lem:kahakRelation}. 
\end{proof}

\subsubsection{\label{app:lem:hkaLUBproof}Proofs of \lsref{lem:hkakaNsIneq}-\ref{lem:hkaLUB}}

\begin{proof}[Proof of \lref{lem:hkakaNsIneq}]	
	According to  \pref{pro:hkaka} we have
	\begin{align}
	6\hka(\Delta)+\hka(\scrT_\ns)=\frac{D+2}{D+d}[6+\kappa(\scrT_\ns)], \quad \hka(\Delta)\geq \frac{D+2}{D+d},
	\end{align}	
	which implies the first inequality in \eref{eq:hkakascrT}.	Next, let $k=|\scrT\setminus\scrT_\ns|$; then $0\leq k\leq 6$ and 
	\begin{align}
	\kappa(\scrT)=k+\kappa(\scrT_\ns),\quad \hka(\scrT)=k\hka(\Delta)+\hka(\scrT_\ns). 
	\end{align}	
	So the second inequality in \eref{eq:hkakascrT} follows from the  first inequality.
\end{proof}

\begin{proof}[Proof of  \lref{lem:hkaLB}]
	Let $R_\rmH(\caT)=\bigl[R(\caT)+R(\caT)^\dag\bigr]/2$.  By definition and  \lref{lem:kappaTLUB} we have
\begin{equation}
\kappa(\caT)=\tr\bigl[R(\caT)(|\Psi\>\<\Psi|)^{\otimes 3}\bigr]=\tr\bigl[R_\rmH(\caT)(|\Psi\>\<\Psi|)^{\otimes 3}\bigr].
\end{equation} 
	If $\caT\in \scrT_\iso$, then $R(\caT)$ is a unitary operator and $-\bbI\leq R_\rmH(\caT)\leq \bbI$.  Let $\caV$ be  the support of $R(\scrT_\iso)$ and $\Pi$ the projector onto $\caV$. Then $\caV$ is the common eigenspace of $R(\caT)$ and $R_\rmH(\caT)$  with eigenvalue 1 for $\caT\in \scrT_\iso$ and $\Pi=R(\scrT_\iso)/|O_3(d)|$. 
	Let  $p=\tr\bigl[\Pi(|\Psi\>\<\Psi|)^{\otimes 3}\bigr]=\kappa(\scrT_\iso)/|O_3(d)|$. Then 
	\begin{align}
	\kappa(\caT)&=\tr\bigl[R_\rmH(\caT)(|\Psi\>\<\Psi|)^{\otimes 3}\bigr]\geq \tr\bigl\{[\Pi-(\bbI-\Pi)](|\Psi\>\<\Psi|)^{\otimes 3}\bigr\}
	=  2p-1=\frac{2\kappa(\scrT_\iso)}{|O_3(d)|}-1\nonumber\\
	&\geq \frac{2[\kappa(\Sigma(d))-\kappa(\scrT_\defe)]}{|O_3(d)|}-1\geq \frac{2[\kappa(\Sigma(d))-|\scrT_\defe|]}{|O_3(d)|}-1,
	\end{align}
	where the last inequality holds because $\kappa(\scrT_\defe)\leq |\scrT_\defe|$ by \lref{lem:kappaTLUB}. This equation confirms the first two inequalities in \eref{eq:kappaTSLB}. The last inequality  in \eref{eq:kappaTSLB} holds because [see \lsref{lem:O3Dh}, \ref{lem:kappaTLUB}, and \eref{eq:Sigma33DefectNum}],
	\begin{align}
	\kappa(\Sigma(d))\geq \frac{4(D+d)}{D+1},\quad |O_3(d)|=
	\begin{cases}
	2d   & d=3,\\
	2(d-1)   & d = 1\mmod 3,\\
	2(d+2) & d=2\mmod 3,
	\end{cases}\quad 
	|\scrT_\defe|=
	\begin{cases}
	2   & d=3,\\
	4   & d = 1\mmod 3,\\
	0 & d=2\mmod 3.
	\end{cases}
	\end{align}

	Next, we turn to \eref{eq:hkaLB} with $\caT\in \scrT_\iso$. The first inequality in \eref{eq:hkaLB} follows from \eref{eq:kappaTSLB}. 
	If $d=3$, then $\scrT_\iso=\scrT_\sym$ and $\hka(\caT)\geq (D+2)/(D+d)$ by \eref{eq:hkakaTLUB}, which implies the second inequality in \eref{eq:hkaLB}. If $d=1\mmod 3$, then
	the second inequality in \eref{eq:hkaLB} can be derived  from the  two equations above
	and \pref{pro:hkaka} as follows,
	\begin{align}
	\hka(\caT)
	&\geq \frac{D+2}{D-1}\left[\frac{\kappa(\Sigma(d))-4}{d-1} -\frac{\kappa(\Sigma(d))}{2D+2d}-1\right]
	\geq \frac{D+2}{D-1}\left[\frac{4}{D+1} -\frac{2}{D+1}-1\right]=- \frac{D+2}{D+1}.
	\end{align}
	A similar reasoning also applies to the case $d=2\mmod 3$.

	Finally, we consider \eref{eq:hkaLB} with $\caT\in \scrT_\defe$. Now $\kappa(\caT)\geq 0$ and $\kappa(\Sigma(d))\leq 2(d+1)$ by \lref{lem:kappaTLUB}, so the first inequality in \eref{eq:hkaLB} holds.  In conjunction with \pref{pro:hkaka} we can deduce that
	\begin{align}
	\hka(\caT)
	&=\frac{D+2}{D-1}\left[\kappa(\caT) -\frac{\kappa(\Sigma(d))}{2D+2d}\right]\geq -\frac{(d+1)(D+2)}{(D-1)(D+d)},
	\end{align}
	which implies the second inequality in \eref{eq:hkaLB} if $D\geq 5$. If instead $D=d=3$, then  $\scrT_\iso=\scrT_\sym$ and  $|\scrT_\defe|=2$, which means $\kappa(\Sigma(d))=\kappa(\scrT_\sym)+\kappa(\scrT_\defe)=6+2\kappa(\caT)$ by \lref{lem:SemigroupT} and \pref{pro:kappaSym}. Therefore,
	\begin{align}
	\hka(\caT)
	&=\frac{d+2}{d-1}\left[\kappa(\caT) -\frac{6+2\kappa(\caT)}{4d}\right]\geq  -\frac{3(d+2)}{2d(d-1)}=-\frac{d+2}{d+1},
	\end{align}
	which confirms the second inequality in \eref{eq:hkaLB} and completes the proof of \lref{lem:hkaLB}. 
\end{proof}

\begin{proof}[Proof of \lref{lem:hkaLUB}]
	The first inequality in \eref{eq:hkaLUB} follows from 	\lref{lem:hkaLB}, and the other two inequalities follow from \pref{pro:hkaka}. 	The first inequality in \eref{eq:hkaSigLUB} follows from \lref{lem:kappaTLUB} and \pref{pro:hkaka}; the second and fourth inequalities are trivial; the third inequality follows from \eref{eq:hkaLUB} and the fact that $|\Sigma(d)|=2(d+1)$.

	The equality in  \eref{eq:hkaSigAbsLUB} follows from \pref{pro:hkaka}. The first inequality is trivial. The second inequality	
	can be proved by virtue of \pref{pro:hkaka} as follows, given that $\kappa(\Sigma(d))\geq0$, $|\scrT_\sym|=6$, and $|\scrT_\ns|=2d-4$,
	\begin{align}
	&|\hka|(\Sigma(d))=\sum_{\caT\in \Sigma(d)}|\hka(\caT)|=\sum_{\caT\in \scrT_\sym}\hka(\caT)+\sum_{\caT\in \scrT_\ns}|\hka(\caT)|\nonumber\\
	&=\frac{D+2}{D-1}\left\{\sum_{\caT\in \scrT_\sym}\left[\kappa(\caT) -\frac{\kappa(\Sigma(d))}{2D+2d}\right]+
	\sum_{\caT\in\scrT_\ns}\left|\kappa(\caT) -\frac{\kappa(\Sigma(d))}{2D+2d}\right|\right\}\nonumber\\	
	&\leq \frac{D+2}{D-1}\left\{\sum_{\caT\in \scrT_\sym}\left[|\kappa(\caT)| -\frac{\kappa(\Sigma(d))}{2D+2d}\right]+
	\sum_{\caT\in\scrT_\ns}\left[|\kappa(\caT)| +\frac{\kappa(\Sigma(d))}{2D+2d}\right]\right\}\nonumber\\	
	&= \frac{D+2}{D-1}\left[|\kappa|(\Sigma(d)) +\frac{d-5}{D+d}\kappa(\Sigma(d)) \right]
	\leq \frac{D+2}{D-1}\left(1+\frac{d-5}{D+d} \right)|\kappa|(\Sigma(d))\leq  \frac{D+d-1}{D}|\kappa|(\Sigma(d)).
	\end{align}
	
	The first inequality in \eref{eq:hkaNsLUB} follows from  \eqsref{eq:hkaSigNS}{eq:hkaSigLUB}; the second inequality follows from \lref{lem:hkakaNsIneq}, the third inequality follows from \eref{eq:hkaLUB} given that $|\scrT_\ns|=2(d-2)$;  the  last inequality is  trivial.
\end{proof}

\subsubsection{\label{app:kappaMhkaTLUBproof}Proof of \lref{lem:kappaMhkaTLUB}}

\begin{proof}[Proof of \lref{lem:kappaMhkaTLUB}]
	According to \pref{pro:hkaka}  we have
	\begin{align}
	\hka(\caT) -\kappa(\caT)=\frac{3\kappa(\caT)}{D-1}-\frac{(D+2)\kappa(\Sigma(d))}{2(D-1)(D+d)}\leq \frac{3}{D-1}-\frac{2(D+2)}{(D-1)(D+1)}=\frac{1}{D+1}< \frac{1}{D}, 
	\end{align}
	which implies the lower bounds in \eref{eq:kappaMhkaTLUB}. Here the first inequality holds because  $\kappa(\caT)\leq 1$ for $\caT\in \Sigma(d)$ and $\kappa(\Sigma(d))\geq 4(D+d)/(D+1)$ by \lref{lem:kappaTLUB}. 
	
	To prove the upper bound in \eref{eq:kappaMhkaTLUB}, let $\scrT=S_3 \caT$ and $\scrT'=\Sigma(d)\setminus \scrT$. Then  $|\scrT|\leq 6$, $\kappa(\scrT)=|\scrT|\kappa(\caT)$, 
	and  $\kappa(\scrT')\leq |\scrT'|=2d+2-|\scrT|$ given that  $\kappa(\caT)\leq 1$ for $\caT\in \Sigma(d)$ and $|\Sigma(d)|=2d+2$. Therefore,
	\begin{align}
	\kappa(\caT)-\hka(\caT) =\frac{(D+2)[\kappa(\scrT)+\kappa(\scrT')]}{2(D-1)(D+d)}-\frac{3\kappa(\caT)}{D-1}\leq \frac{(D+2)[|\scrT|\kappa(\caT)+2d+2-|\scrT|]}{2(D-1)(D+d)}-\frac{3\kappa(\caT)}{D-1}. \label{eq:hkamkappaTLUBproof}
	\end{align}
	Note that $\Sigma(d)=\scrT_\iso\sqcup \scrT_\defe$. 
	If $\caT\in \scrT_\iso$, then  $\kappa(\caT)\geq -1$ by \lref{lem:kappaTLUB} and $|\scrT|=|S_3|=6$; therefore,
	\begin{align}
	\kappa(\caT)-\hka(\caT)\leq \frac{3}{D-1}+\frac{(D+2)(d-5)}{(D-1)(D+d)}=\frac{d}{D}-\frac{2D^2+(d^2-6d+10)D-d^2}{D(D-1)(D+d)}\leq \frac{d}{D},
	\end{align}
	which implies the upper bound in \eref{eq:kappaMhkaTLUB}. 
	If instead $\caT\in \scrT_\defe$, then $\kappa(\caT)\geq 0$ by \lref{lem:kappaTLUB} and $|\scrT|\geq 2$ by \lref{lem:SemigroupT}, 
	so \eref{eq:hkamkappaTLUBproof} yields
	\begin{align}
	\kappa(\caT)-\hka(\caT)\leq \frac{(D+2)(2d+2-|\scrT|)}{2(D-1)(D+d)}\leq \frac{d(D+2)}{(D-1)(D+d)}
	\leq \frac{d}{D-1}, \label{eq:kappaMhkaTLUBproof}
	\end{align}
	which confirms the upper bound in \eref{eq:kappaMhkaTLUB}. 
	If in addition $d\geq 5$, then the last upper bound above can be replaced by $d/D$, so the upper bound in \eref{eq:kappaMhkaTLUB} can be replaced by $d/D$ accordingly.

	\Eref{eq:kappaMhkaSigLUB} is a simple corollary of 
	\pref{pro:hkaka} and \eref{eq:kappaSigLUB}.

	Next, suppose  $\kappa(\caT)\geq 0$ for $\caT\in \Sigma(d)$. Then \eref{eq:hkaAbsMhkaLUB} and the upper bound in \eref{eq:kappaMhkaAbsLUB} are simple corollaries of \eref{eq:kappaMhkaTLUB}. If in addition $\hka(\caT)\geq 0$, then   the lower bound in \eref{eq:kappaMhkaAbsLUB} follows from  \eref{eq:kappaMhkaTLUB}. If instead $\hka(\caT)\leq 0$,
	then by virtue of a similar reasoning that leads to \eref{eq:kappaMhkaTLUBproof} we can deduce that
	\begin{align}
	|\hka(\caT)|- \kappa(\caT)&\leq -\hka(\caT)=\frac{(D+2)[|\scrT|\kappa(\caT)+\kappa(\scrT')]}{2(D-1)(D+d)}-\frac{(D+2)\kappa(\caT)}{D-1}
	\nonumber\\
	&\leq \frac{(D+2)(2d+2-|\scrT|)}{2(D-1)(D+d)}\leq \frac{d(D+2)}{(D-1)(D+d)}
	<\frac{d}{D-1},
	\end{align}
	which confirms the lower bound in \eref{eq:kappaMhkaAbsLUB}. 
	
	It remains to prove \eqsref{eq:hkaAbsMhkaSigUB}{eq:kappaMhkaAbsSigLUB}. 
	Note that $\Sigma(d)=\scrT_\sym\sqcup\scrT_\ns$ with $|\scrT_\sym|=6$ and $|\scrT_\ns|=2(d-2)$;  meanwhile, $\hka(\caT)\geq 0$ for $\caT\in \scrT_\sym$ by \eref{eq:hkakaTLUB}. In conjunction with \eref{eq:hkaAbsMhkaLUB} we can deduce that
	\begin{align}
	|\hka|(\Sigma(d))-\hka(\Sigma(d))=|\hka|(\scrT_\ns)-\hka(\scrT_\ns)\leq \frac{4d(d-2)}{D-1}<\frac{4d^2}{D}, 
	\end{align} 
	which confirms \eref{eq:hkaAbsMhkaSigUB}.
	In conjunction with \eqsref{eq:kappaMhkaTLUB}{eq:kappaMhkaAbsLUB} we can deduce that 
	\begin{align}
	|\hka|(\Sigma(d))-\kappa(\Sigma(d))&=|\hka|(\scrT_\sym)-\kappa(\scrT_\sym)+|\hka|(\scrT_\ns)-\kappa(\scrT_\ns)\leq \frac{6}{D+1}+\frac{2d(d-2)}{D-1}\leq \frac{2d^2}{D},
	\end{align}
	which implies the lower bound in \eref{eq:kappaMhkaAbsSigLUB}. Alternatively, this lower bound follows from \eqsref{eq:kappaSigLUB}{eq:hkaSigAbsLUB}. The upper bound in \eref{eq:kappaMhkaAbsSigLUB} follows from \eref{eq:kappaMhkaSigLUB}, which completes the proof of \lref{lem:kappaMhkaTLUB}. 
\end{proof}

\subsubsection{\label{app:lem:QPsiProjProof}Proofs of \lsref{lem:QPsiProj} and \ref{lem:QPsiProjEig}}

\begin{proof}[Proof of \lref{lem:QPsiProj}]By virtue of  \eref{eq:ObShNormPsiDef} we can deduce that
	\begin{align}
	6\bcaQ_{\orb(\Psi)}(\Ob)&=\sum_{\caT\in \Sigma(d)}\hka(\caT)\caR_\caT(\Ob)
	=[K^2 \hka(\Delta)+K\hka(\tau_{23})]\bbI+\left[\hka(\scrT_0\setminus \{\Delta\})+K\hka\left(\scrT_1\setminus \{\caT_{\tau_{23}}\}\right)\right]\Ob
	\nonumber\\
	&=(K^2+K)\hka(\Delta)\bbI+\left[\frac{K+1}{2} \hka(\Sigma(d))-(K+1)\hka(\Delta)\right]\Ob,  \label{eq:QPsiProjProof}
	\end{align}
	which confirms the first equality in \eref{eq:QPsiProj0}.
	Here the second equality follows from \lref{lem:RTObStabProj}, and the third equality holds because $\hka(\tau_{23})=\hka(\Delta)$ and $\hka(\Sigma(d))=2\hka(\scrT_0)=2\hka(\scrT_1)$ by \pref{pro:kappaSym}.

	By virtue of \lref{lem:QPsiObOb0} and  \eref{eq:QPsiProjProof}	we can deduce that
	\begin{align}
	&6\bcaQ_{\orb(\Psi)}(\Ob_0)=6\bcaQ_{\orb(\Psi)}(\Ob)+\frac{(D+1)(D+2)K^2}{D^2}\bbI-\frac{2(D+2)K(K \bbI+\Ob)}{D}\nonumber\\
	&=(K^2+K)\hka(\Delta)\bbI+\left[\frac{K+1}{2} \hka(\Sigma(d))-(K+1)\hka(\Delta)\right]\Ob  +\frac{(D+1)(D+2)K^2}{D^2}\bbI-\frac{2(D+2)K(K \bbI+\Ob)}{D}\nonumber\\
	&
	=\left[(K^2+K)\hka(\Delta)-\frac{(D-1)(D+2)K^2}{D^2}\right] \bbI +\left[\frac{K+1}{2} \hka(\Sigma(d))-(K+1)\hka(\Delta)-\frac{2(D+2)K}{D}\right]\Ob \nonumber\\
	&=\left[(K^2+K)\hka(\Delta)-\frac{(D-1)(D+2)K^2}{D^2}\right] \bbI +\left[-D(K+1)\hka(\Delta)+\frac{(D+2)(DK+D-2K)}{D}
	\right]\Ob\nonumber\\
	&=\upsilon_2 \bbI+(\upsilon_1-\upsilon_2)\Ob,
	\end{align}	
	which confirms  the second equality in \eref{eq:QPsiProj0}. Here the fourth equality follows from \lref{lem:hkaDeltaNS}, and the last equality follows from the definition in \eref{eq:QPsiProj0eig}. 
\end{proof}

\begin{proof}[Proof of \lref{lem:QPsiProjEig}]
	According to  \eref{eq:QPsiProj0eig}, $\upsilon_1$ is monotonically decreasing in $\hka(\Delta)$, while $\upsilon_2$ is monotonically increasing in $\hka(\Delta)$. So \eqsref{eq:QPsiProj0eig1LUB}{eq:QPsiProj0eig2LUB} follow from the assumption $1\leq K\leq D/d$ and  the inequalities below  (see \pref{pro:hkaka}), 
	\begin{align}
	\frac{D+2}{D+d}\leq \hka(\Delta)\leq \frac{D+2}{D+1}.   \label{eq:hkaDeltaLUB}
	\end{align}	
	In addition, the upper bound in \eref{eq:QPsiProj0eig1LUB} is saturated iff $\hka(\Delta)=(D+2)/(D+d)$.
	
	The equivalence of the first two inequalities in \eref{QPsiProj0EigOrder} follows from \eref{eq:QPsiProj0eig}, and the equivalence of the last three inequalities follows from \pref{pro:hkaka} and \lref{lem:hkaDeltaNS}.

	If the inequalities in \eref{QPsiProj0EigOrder} are saturated, then 	
	\begin{align}
	\upsilon_1=\upsilon_2=\frac{(D+2)(D-K)K}{D^2}, 
	\end{align}
	which sets  a lower bound for $\max\{\upsilon_1,\upsilon_2\}$ and implies the first two inequalities in \eref{eq:QPsiProj0eigMaxLUB} given	the monotonicity properties  of $\upsilon_1,\upsilon_2$ mentioned above. The last inequality in \eref{eq:QPsiProj0eigMaxLUB} follows from \eref{eq:QPsiProj0eig1LUB} and the inequalities below,	
	\begin{align}
	\upsilon_2 \leq \frac{(D+2)(D^2+K)K}{D^2(D+1)}<\frac{(D+2)(D-K)[dD+(dD-D-d)K]}{D^2(D+d)}.  \label{eq:QPsiProjEigProof2}
	\end{align}	
	Here the first inequality is reproduced from \eref{eq:QPsiProj0eig2LUB}, and the second inequality holds because	
	\begin{align}
	\frac{(D-K)[dD+(dD-D-d)K]}{K(D+d)}- \frac{D^2+K}{D+1}\geq \frac{D[(d^2-3d+1)D+d^2-2d]}{d(D+1)}>0,
	\end{align}
	given that $1\leq K\leq D/d$ by assumption. Thanks to \eref{eq:QPsiProjEigProof2}, the upper bound in \eref{eq:QPsiProj0eigMaxLUB} is saturated iff the upper bound in \eref{eq:QPsiProj0eig1LUB} is saturated, which is the case iff
	$\hka(\Delta)=(D+2)/(D+d)$ according to the conclusion presented after \eref{eq:hkaDeltaLUB}.

	Next,  by virtue of \eref{eq:QPsiProj0eig} we can deduce that
	\begin{align}
	\upsilon_1-\upsilon_2&=\frac{(D+2)(DK+D-2K)}{D}-D(K+1)\hka(\Delta),
	\end{align}	
	which implies \eref{eq:QPsiProj0eigDif} given \eref{eq:hkaDeltaLUB}.

	Thanks to \eref{eq:QPsiProj0eig2LUB}, we have $\upsilon_2\leq K+2$. 
	If in addition $\hka(\Delta)\geq1$, then from \eref{eq:QPsiProj0eig} we can deduce that
	\begin{align}
	\upsilon_1\leq \frac{(D-K)(DK+2D-2K)}{D^2}\leq K+2,
	\end{align}
	given that  $\upsilon_1$ is monotonically decreasing in $\hka(\Delta)$.  Therefore, $\max\{\upsilon_1, \upsilon_2\}\leq K+2$. If $\hka(\Psi,\scrT_\ns)\leq 0$, then $\hka(\Delta)\geq1$ by \lref{lem:hkaDeltaNS}, so the same conclusion holds, 
	which completes the proof of \lref{lem:QPsiProjEig}.
\end{proof}

\subsection{\label{app:3designPhi3LUB}Proofs of \pref{pro:orbit3designCon}, \thref{thm:Phi3LUB}, and \coref{cor:Phi3UB}}

\begin{proof}[Proof of \pref{pro:orbit3designCon}]
	If $\orb(\Psi)$ is a 3-design, then $Q(\orb(\Psi))=6P_\sym/[D(D+1)(D+2)]$. Therefore,
	\begin{align}
	\kappa(\caT)& =\tr[R(\caT)Q(\orb(\Psi))]=\frac{6\tr[P_\sym R(\caT)]}{D(D+1)(D+2)}
	=\frac{\sum_{O\in S_3} \tr[R(O) R(\caT)]}{D(D+1)(D+2)}=\frac{3}{D+2}\quad \forall \caT\in \scrT_\ns,
	\end{align}	
	which confirms the implication $1\imply 2$. Here the fourth equality follows from \lref{lem:SemigroupscrT}, \pref{pro:RTT1T2tr}, and the fact that $R(O) R(\caT)=R(O\caT)$.
	
	If  $\kappa(\Psi, \caT)=3/(D+2)$  for all $\caT\in \scrT_\ns$, then $\kappa(\Psi,\Sigma(d))=6+\kappa(\Psi,\scrT_\ns)=6(D+d)/(D+2)$	given that $|\scrT_\ns|=2d-4$. In addition,
	$\hka(\caT)=1$ for all $\caT\in \scrT_\sym$ and
	$\hka(\caT)=0$ for all $\caT\in \scrT_\ns$ by \pref{pro:hkaka},  which confirms the implication $2\imply 3$. Meanwhile, 
	\begin{align}
	\bQ(\orb(\Psi))&=\frac{1}{6}\sum_{\caT\in \Sigma(d)}\hka(\caT)R(\caT)=\frac{1}{6}\sum_{\caT\in \scrT_\sym}R(\caT)=\frac{1}{6}\sum_{O\in S_3}R(O)=P_\sym,
	\end{align}
	so $\orb(\Psi)$ is a 3-design, which confirms the implication $2\imply 1$. 
	
	If $\hka(\caT)=0$ for all $\caT\in \scrT_\ns$, then $\kappa(\caT)$ for all $\caT\in \scrT_\ns$ are equal to each other by \pref{pro:hkaka}. In addition,
	\begin{align}
	\kappa(\caT)=\frac{\kappa(\Sigma(d))}{2D+2d}
	=\frac{3+(d-2)\kappa(\caT)}{D+d} \quad \forall \caT\in \scrT_\ns,
	\end{align}
	which implies that  $\kappa(\caT)=3/(D+2)$  for all $\caT\in \scrT_\ns$ and confirms the implication $3\imply 2$.
\end{proof}

\begin{proof}[Proof of \thref{thm:Phi3LUB}]
	By virtue of \eqsref{eq:kappapsiT}{eq:MomentQpsiBar} we can deduce that
	\begin{align}
	\bar{\Phi}_3(\orb(\Psi))
	&=\tr[Q(\orb(\Psi))\bQ(\orb(\Psi))]=\frac{1}{6}\sum_{\caT\in \Sigma(d)}\hka(\caT)\tr[Q(\orb(\Psi))R(\caT)]=\frac{1}{6}\sum_{\caT\in \Sigma(d)}\kappa(\caT)\hka(\caT), \label{eq:Phi3psiProof}
	\end{align}
	which confirms the first equality in \eref{eq:Phi3psi}. 
	The second equality in \eref{eq:Phi3psi} can be proved as follows,
	\begin{align}
	\sum_{\caT\in \Sigma(d)}\kappa(\caT)\hka(\caT)
	=6\hka(\Delta)+\sum_{\caT\in\scrT_\ns}\kappa(\caT)\hka(\caT)=6+\sum_{\caT\in \scrT_\ns}\left[\kappa(\caT)-\frac{3}{D+2}\right]\hka(\caT).
	\end{align}
	Here the first equality holds because $\Sigma(d)=\scrT_\sym\sqcup\scrT_\ns$,  $|\scrT_\sym|=6$, and  $\kappa(\caT)=1$, $\hka(\caT)=\hka(\Delta)$ for all $\caT\in \scrT_\sym$; the second equality is a simple corollary of \eref{eq:hkaDeltaNS}. The third and fourth equalities in \eref{eq:Phi3psi} follow from the first two equalities and \pref{pro:hkaka}; the last equality follows from \lref{lem:hkaDeltaNS}, which means	
	\begin{align}
	\hka(\Sigma(d))&
	=6+\frac{D-1}{D+2}\hka(\scrT_\ns),\\
	\hka(\Sigma(d),2)&	=6\left[1-\frac{\hka(\scrT_\ns)}{2(D+2)}\right]^2+\hka(\scrT_\ns,2)=6-\frac{6\hka(\scrT_\ns)}{D+2}+\frac{3\hka^2(\scrT_\ns)}{2(D+2)^2}+\hka(\scrT_\ns,2).
	\end{align}
	
	The upper bound for $\bar{\Phi}_3(\orb(\Psi))$ in \eref{eq:Phi3psi}  follows from \eref{eq:Phi3psiProof} given that $|\Sigma(d)|=2(d+1)$ and  $|\kappa(\caT)|\leq 1$, $\hka(\caT)\leq (D+2)/(D+1)$ for all $\caT\in \Sigma(d)$ by  \lref{lem:kappaTLUB} and \eref{eq:hkakaTLUB}.		
\end{proof}

\begin{proof}[Proof of \coref{cor:Phi3UB}]
	By assumption and  \lref{lem:kappaTLUB} we have
	$0\leq \kappa(\caT)\leq 1$ for all $\caT\in \Sigma(d)$, which means $\kappa(\Sigma(d))\geq \kappa(\Sigma(d),2)$ and $\kappa(\Sigma(d))\geq \kappa(\scrT_\sym)=6$. In conjunction with   \eref{eq:Phi3psi} we can deduce that
	\begin{align}
	&\bar{\Phi}_3(\orb(\Psi))\leq \frac{D+2}{6(D-1)}\left[1-\frac{3}{D+d}\right]\kappa(\Sigma(d),2)\leq \frac{D^2+3(d-2)}{6D^2}\kappa(\Sigma(d),2)=\frac{D^2+3(d-2)}{6D^2}\sum_{\caT\in \Sigma(d)}\kappa^2(\caT),
	\end{align}
	which confirms \eref{eq:Phi3PosUB1}. 
	If   $\kappa(\Sigma(d))\geq 6(D+d)/(D+2)$, then  $\kappa^2(\Sigma(d))/(D+d)\geq 6\kappa(\Sigma(d),2)/(D+2)$, which implies that $\bar{\Phi}_3(\orb(\Psi))\leq \kappa(\Sigma(d),2)/6$. 
\end{proof}

\subsection{\label{app:thm:MomentNormProof}Proofs of \thsref{thm:MomentNormd2} and \ref{thm:MomentNormd13}}

\begin{proof}[Proof of \thref{thm:MomentNormd2}]
	By virtue of  \eref{eq:MomentQpsiBar} and the relation  $\Sigma(d)=\scrT_\sym\sqcup \scrT_\ns$ we can deduce that 
	\begin{align}
	\bQ(\orb(\Psi))&=\frac{1}{6}\sum_{\caT\in \Sigma(d) }\hka(\caT)R(\caT)=\hka(\Delta)P_\sym +\frac{1}{6}\sum_{\caT\in \scrT_\ns }\hka(\caT)R(\caT), \label{eq:MomentNormProof1}\\
	\bQ(\orb(\Psi))-P_\sym&=[\hka(\Delta)-1]P_\sym+\frac{1}{6}\sum_{\caT\in \scrT_\ns }\hka(\caT)R(\caT)=-\frac{\hka(\scrT_\ns)}{2(D+2)}P_\sym +\frac{1}{6}\sum_{\caT\in \scrT_\ns }\hka(\caT)R(\caT).  \label{eq:MomentNormProof2}
	\end{align}
	The second equality in \eref{eq:MomentNormProof1} holds because  $\sum_{\caT\in \scrT_\sym} R(\caT)=6P_\sym$  and  $\hka(\caT)$ for all $\caT\in \scrT_\sym$ are equal to $\hka(\Delta)$ by \pref{pro:kappaSym}. The second equality in \eref{eq:MomentNormProof2} follows from \eref{eq:hkaDeltaNS} in \lref{lem:hkaDeltaNS}.
	In addition,  $R(\caT)$ is a unitary operator for each $\caT\in \Sigma(d)$,  and the third tensor power of each stabilizer state is a common eigenstate of $R(\caT)$ with eigenvalue 1. Therefore, $\hka(\Sigma(d))/6$ is an eigenvalue of $\bQ(\orb(\Psi))$, and the first two inequalities in  \eref{eq:MomentNormd2a} hold. The last inequality in \eref{eq:MomentNormd2a} follows from \lref{lem:hkaLUB}.  Furthermore, 
	\begin{align}
	\|\bQ(\orb(\Psi))-P_\sym\|&\leq \frac{|\hka(\scrT_\ns)|}{2(D+2)}+\frac{|\hka|(\scrT_\ns)}{6}\leq \frac{D+5}{6(D+2)}|\hka|(\scrT_\ns)< \frac{d+2}{3}, 
	\end{align}	
	which confirms  \eref{eq:MomentNormd2b}. Here the last  inequality also
	follows from \lref{lem:hkaLUB}. 
\end{proof}

\begin{proof}[Proof of \thref{thm:MomentNormd13}]
	Note that $\Sigma(d)=\scrT_\sym\sqcup \scrT_\ns=\scrT_\iso\sqcup \scrT_\defe$ and $\scrT_\iso=\scrT_\sym \sqcup(\scrT_\iso\cap \scrT_\ns)$. According to \eref{eq:MomentQpsiBar},  the third normalized moment operator $\bQ=\bQ(\orb(\Psi))$ can be expressed as follows (cf. the proof of \thref{thm:MomentNormd2}),
	\begin{equation}
	\begin{gathered}
	\bQ=M_\iso+M_\defe,\quad M_\defe:=\frac{1}{6}\sum_{\caT\in \scrT_\defe }\hka(\caT)R(\caT)	=\frac{1}{6}\hka(\caT_\defe)R(\scrT_\defe), \\
	M_\iso:=\frac{1}{6}\sum_{\caT\in \scrT_\sym}\hka(\caT)R(\caT)+\frac{1}{6}\sum_{\caT\in \scrT_\iso\cap\scrT_\ns }\hka(\caT)R(\caT)
	=\hka(\Delta)P_\sym+\frac{1}{6}\sum_{\caT\in \scrT_\iso\cap\scrT_\ns }\hka(\caT)R(\caT).
	\end{gathered}
	\end{equation}
	Here the third equality holds because $\kappa(\caT)$ for all $\caT\in \scrT_\defe$ are equal to each other by \lref{lem:SemigroupT} and \pref{pro:kappaSym}. Note that $R(\caT)$ for each $\caT\in \scrT_\iso$ is a unitary operator. 
	Let $\caV$ be the support of $R(\scrT_\iso)$; then $\caV$ is  the common eigenspace of $R(\caT)$ with eigenvalue 1  for $\caT\in \scrT_\iso$ and is contained in $\Sym_3\bigl(\caH_d^{\otimes n}\bigr)$. In addition, $R(\scrT_\defe)$ is proportional to a projector supported in a proper subspace in $\caV$ by  \lref{lem:RTisodefe}. Let $\caV^\perp$ be the orthogonal complement of $\caV$ within $\Sym_3\bigl(\caH_d^{\otimes n}\bigr)$. Then both $\caV$ and $\caV^\perp$ are invariant subspaces of $M_\iso, M_\defe, \bQ$. Let $\Pi$ and $\Pi^\perp$ be the projectors onto $\caV$ and $\caV^\perp$, respectively, and let
	\begin{align}
	Q_1=\Pi \bQ\Pi, \quad Q_2=\Pi^\perp \bQ\Pi^\perp=\Pi^\perp M_\iso \Pi^\perp. 
	\end{align}
	Then 
	\begin{align}
	\|\bQ\|=\max\{\|Q_1\|, \|Q_2\|\},\quad \|\bQ-P_\sym\|=\max\bigl\{\|Q_1-\Pi\|, \bigl\|Q_2-\Pi^\perp\bigr\|\bigr\}. 
	\end{align}

	Furthermore, the above analysis shows  that $\hka(\scrT_\iso)$ and $\hka(\scrT_\iso) +\hka(\caT_\defe)\|R(\scrT_\defe)\|=2(D+2)\kappa(\caT_\defe)$ are eigenvalues of 
	$6\bQ(\orb(\Psi))$, where the equality follows from \lref{lem:kahkaDef}. These eigenvalues are necessarily nonnegative since $6\bQ(\orb(\Psi))$ is a positive operator. They are also the only eigenvalues of  $Q_1$ if we regard $Q_1$ as an operator acting on $\caV$. Therefore,
	\begin{align}
	\|Q_1\| &=\frac{\max\{\hka(\scrT_\iso), 2(D+2)\kappa(\caT_\defe)\}}{6},\\	
	\|Q_1-\Pi\| &=\frac{\max\{|\hka(\scrT_\iso)-6|, |2(D+2)\kappa(\caT_\defe)-6|\}}{6}	
	\leq \frac{\max\{|\hka|(\scrT_\ns), |2(D+2)\kappa(\caT_\defe)-6|\}}{6},
	\end{align}
	where the last inequality holds because $|\hka(\scrT_\iso)-6|\leq  |\hka|(\scrT_\ns)$ by \eref{eq:hkaGenAbsNS} in \lref{lem:hkaDeltaNS}. 
	In addition, 
	\begin{align}
	\|Q_2\|&\leq \|M_\iso\|\leq \frac{1}{6}|\hka|(\scrT_\iso),\\
	\| Q_2-\Pi^\perp\|&\leq \|M_\iso-P_\sym\|\leq
	|\hka(\Delta)-1|+\frac{|\hka|(\scrT_\iso\cap\scrT_\ns)}{6}\leq  \frac{|\hka(\scrT_\ns)|}{2(D+2)}+\frac{|\hka|(\scrT_\ns)}{6}\leq \frac{D+5}{6(D+2)}|\hka|(\scrT_\ns), \label{eq:MomentNormProof50}
	\end{align}
	where the third inequality in \eref{eq:MomentNormProof50} follows from \eref{eq:hkaDeltaNS}  and the fact that $|\hka|(\scrT_\iso\cap\scrT_\ns)\leq |\hka|(\scrT_\ns)$. The above five equations together imply the first two inequalities in \eref{eq:MomentNormd13a} and the first inequality in \eref{eq:MomentNormd13b}. The third inequality in \eref{eq:MomentNormd13a} holds because $|\hka|(\scrT_\iso)\leq |\hka|(\Sigma(d))\leq 2(d+2)$ by
	\eref{eq:hkaSigLUB}
	and that $0\leq \kappa(\caT_\defe)\leq 1$ by \lref{lem:kappaTLUB}. The second inequality in \eref{eq:MomentNormd13b} follows from \eref{eq:hkaNsLUB}.

	Next, we consider  Eqs.~\eqref{eq:MomentNormd13c}-\eqref{eq:MomentNormd13e}. 	Thanks to \pref{pro:hkaka}, $\hka(\caT_\defe)\leq 0$ iff $\kappa(\caT_\defe)	\leq \kappa(\Psi,\Sigma(d))/[2(D+d)]$. If either condition holds, then $2(D+2)\kappa(\caT_\defe)\leq \hka(\scrT_\iso)\leq |\hka|(\scrT_\iso)$ by \lref{lem:kahkaDef}. So \eref{eq:MomentNormd13c} follows from \eref{eq:MomentNormd13a} given that $|\hka|(\scrT_\iso)\leq |\hka|(\Sigma(d))\leq 2(d+2)$ by \lref{lem:hkaLUB}. In conjunction with \eref{eq:hkaGenAbsNS} we can deduce that 
	\begin{align}
	2(D+2)\kappa(\caT_\defe)-6\leq \hka(\scrT_\iso)-6\leq \frac{D}{D+2}|\hka|(\scrT_\ns). 
	\end{align}
	So \eref{eq:MomentNormd13d} follows from \eref{eq:MomentNormd13b} given  that $|\hka|(\scrT_\ns)\leq 2(d-1)$ by \lref{lem:hkaLUB}.

	Finally, suppose  $\hka(\caT_\defe)= 0$ and $\hka(\caT)\geq 0$ for all $\caT\in \Sigma(d)$. Then 	
	\begin{align}
	|\hka|(\scrT_\ns)=\hka(\scrT_\ns),\quad 2(D+2)\kappa(\caT_\defe)=|\hka|(\scrT_\iso)=\hka(\scrT_\iso)=\hka(\Sigma(d))=6+\frac{D-1}{D+2}\hka(\scrT_\ns)\geq 6 
	\end{align}
	by 	\lref{lem:kahkaDef} and \pref{pro:hkaka}. 
	Together with \eqsref{eq:MomentNormd13c}{eq:MomentNormd13d},	this equation implies
	\eref{eq:MomentNormd13e} and completes the proof of \thref{thm:MomentNormd13}. 
\end{proof}

\subsection{\label{app:thm:StabShNormPsi}Proofs of \thref{thm:StabShNormGen} and \coref{cor:Stab1ShNormGen}}

\begin{proof}[Proof of \thref{thm:StabShNormGen}]
	\Eref{eq:QPsiProjNorm}	follows from  \lref{lem:QPsiProj} and the  inequalities below
	(see \lref{lem:hkaLUB}),
	\begin{gather}
	\frac{D+2}{D+d}\leq \hka(\Delta)\leq \frac{D+2}{D+1},\quad 
	\hka(\Sigma(d))\geq \frac{4(D+2)}{D+1}\geq 4\hka(\Delta). 
	\end{gather}
	By virtue of \lsref{lem:QPsiProj}, \ref{lem:QPsiProjEig} and  the observation $\|\Ob_0\|_2^2=K(D-K)/D$ we can deduce the following results,	
	\begin{gather}
	\frac{(D+2)(D-K)K}{D^2}\leq 6\|\bcaQ_{\orb(\Psi)}(\Ob_0)\|=\max\{\upsilon_1, \upsilon_2\}\leq \frac{(D+2)(D-K)[dD+(dD-D-d)K]}{D^2(D+d)}, \label{eq:StabShNormPsiProof} \\
	\frac{(D+2)}{D}\leq 	\frac{6\|\bcaQ_{\orb(\Psi)}(\Ob_0)\|}{\|\Ob_0\|_2^2}\leq \frac{D+2}{D+d}\left(d-1+\frac{d}{K}-\frac{d}{D}\right),  
	\end{gather}	
	which imply \eqsref{eq:StabShNormGen}{eq:StabShNormRatioGen} given that $\|\Ob_0\|^2_{\orb(\Psi)}= 6(D+1)\|\bcaQ_{\orb(\Psi)}(\Ob_0)\|/(D+2)$ by \eref{eq:ObShNormPsiDef}.   If  $|\Psi\>$ is a stabilizer state, then $\hka( \Delta)=(D+2)/(D+d)$, so  the upper bounds in the above two equations are saturated by \lref{lem:QPsiProjEig}, and the upper bounds in \eqsref{eq:StabShNormGen}{eq:StabShNormRatioGen} are saturated accordingly.

	Thanks to \pref{pro:hkaka} and \lref{lem:hkaDeltaNS}, 
	$\hka(\scrT_\ns)\geq -4(d-2)(D+2)/[D(D+d)]$ iff $\kappa(\scrT_\ns)\geq 2(d-2)/D$.	
	If either condition holds, then $\upsilon_2\leq \upsilon_1$ by \lref{lem:QPsiProjEig}, so we have $6\|\bcaQ_{\orb(\Psi)}(\Ob_0)\|=\upsilon_1$, which means $\|\Ob_0\|^2_{\orb(\Psi)}=(D+1)\upsilon_1/(D+2)$. 
	If instead  $\hka(\scrT_\ns)\leq0$, then $6\|\bcaQ_{\orb(\Psi)}(\Ob_0)\|=\max\{\upsilon_1, \upsilon_2\}\leq K+2$ by \lref{lem:QPsiProjEig}, so $\|\Ob_0\|^2_{\orb(\Psi)}\leq (D+1)(K+2)/(D+2)\leq (K+2)$, which completes the proof of \thref{thm:StabShNormGen}.
\end{proof}

\begin{proof}[Proof of \coref{cor:Stab1ShNormGen}]
	By assumption $\Ob$ is the projector onto a stabilizer state, which corresponds to the case $K=1$ in  \lref{lem:QPsiProj} and \thref{thm:StabShNormGen}. So \eref{eq:QPsiProjNormK1} is a simple corollary of \eref{eq:QPsiProjNorm}.

	Now the eigenvalues $\upsilon_1$ and $\upsilon_2$ in \lref{lem:QPsiProj} and \thref{thm:StabShNormGen} read 
	\begin{gather}
	\upsilon_1= \hka(\Sigma(d))-\frac{(D+2)(3D-1)}{D^2},\quad \upsilon_2=2\hka(\Delta)-\frac{(D-1)(D+2)}{D^2}.   
	\end{gather}
	In addition, they satisfy the following relations by \lref{lem:QPsiProjEig},
	\begin{gather}
	\upsilon_2 \leq \upsilon_1+ \frac{2(D+2)}{D(D+1)},\quad 0\leq \upsilon_1\leq  \max\{\upsilon_1, \upsilon_2\}\leq \upsilon_1+\frac{2(D+2)}{D(D+1)}\leq \hka(\Sigma(d))-\frac{(D+2)(3D^2-1)}{D^2(D+1)}. 
	\end{gather}
	In conjunction with \eref{eq:StabShNormGen} we can deduce that
	\begin{align}
	&\frac{D+1}{D+2}\hka(\Sigma(d))-\frac{(D+1)(3D-1)}{D^2}= \frac{(D+1)\upsilon_1}{D+2}\leq \|\Ob_0\|^2_{\orb(\Psi)}=\frac{(D+1)\max\{\upsilon_1, \upsilon_2\}}{D+2}\nonumber\\
	&\leq \frac{D+1}{D+2}\hka(\Sigma(d))-\frac{3D^2-1}{D^2}\leq \hka(\Sigma(d))-3,  \label{eq:StabShNormPsiK1Proof}
	\end{align}
	which impies \eref{eq:Stab1ShNormGen}. Here the last inequality holds because $\hka(\Sigma(d))\geq 4(D+2)/(D+1)$ by \lref{lem:hkaLUB}.

	To prove \eref{eq:Stab1ShNormRatioGen}, note that $\|\Ob_0\|_2^2=(D-1)/D$. 
	So the above three equations means
	\begin{align}
	\frac{\|\Ob_0\|^2_{\orb(\Psi)}}{\|\Ob_0\|_2^2}\geq \frac{D(D+1)}{(D-1)(D+2)} \left[\hka(\Sigma(d))-\frac{(D+2)(3D-1)}{D^2}\right]\geq \hka(\Sigma(d))-3-\frac{5}{D}, 
	\end{align}
	which confirms the lower bound in \eref{eq:Stab1ShNormRatioGen}. Meanwhile,  from \eref{eq:StabShNormPsiK1Proof} we can deduce  that
	\begin{align}
	\frac{\|\Ob_0\|^2_{\orb(\Psi)}}{\|\Ob_0\|_2^2}&\leq \frac{D(D+1)}{(D-1)(D+2)} \hka(\Sigma(d))-\frac{3D^2-1}{D(D-1)}=\hka(\Sigma(d))+\frac{2\hka(\Sigma(d))}{(D-1)(D+2)} -\frac{3D^2-1}{D(D-1)}\nonumber\\
	&\leq \hka(\Sigma(d))+ \frac{4(d+1)}{(D-1)(D+d)} -\frac{3D^2-1}{D(D-1)}=\hka(\Sigma(d))-3-\frac{3D^2-(d+5)D-d}{D(D-1)(D+d)}\nonumber\\
	&\leq \hka(\Sigma(d))-3, 
	\end{align}
	given that $\hka(\Sigma(d))\leq 2(d+1)(D+2)/(D+d)$ by \pref{pro:hkaka} and \lref{lem:kappaTLUB}. This equation confirms the upper bound in \eref{eq:Stab1ShNormRatioGen}. 
	
	If  $\kappa(\scrT_\ns)\geq -2(D^2-2dD+D+d)/D^2$, then $\upsilon_2\leq \upsilon_1$ by \lref{lem:QPsiProjEig}, so 
	the first inequality in \eref{eq:StabShNormPsiK1Proof} is saturated, and the lower bound in \eref{eq:Stab1ShNormGen} is saturated accordingly.  
\end{proof}

\subsection{\label{app:thm:orbitShadowProof}Proofs of \thref{thm:ShNormGen} and \coref{cor:ShNormGenPos}}

\begin{proof}[Proof of \thref{thm:ShNormGen}]
	Note that  $\Sigma(d)=\scrT_\sym\sqcup \scrT_\ns$ and $\hka(\caT)=\hka(\Delta)$ for $\caT\in \scrT_\sym$ thanks to \pref{pro:kappaSym}. By virtue of \eref{eq:ObShNormPsiDef} we can deduce that
	\begin{align}\label{eq:ObShadowProof}
	\frac{D+2}{D+1}\|\Ob\|^2_{\orb(\Psi)}&= 6\|\bcaQ_{\orb(\Psi)}(\Ob)\|\leq  \hka(\Delta)\sum_{\caT\in \scrT_\sym}\left\|\caR_\caT(\Ob) \right\| 
	+\sum_{\caT\in \scrT_\ns}| \hka(\caT)|\left\|\caR_\caT(\Ob) \right\|. 
	\end{align}
	In addition, thanks to \eref{eq:RTOS3norm} and \lref{lem:RTO}, we have	
	\begin{align}\label{eq:RTOsymNsProof}
	\sum_{\caT\in \scrT_\sym}\left\|\caR_\caT(\Ob) \right\|=\|\Ob\|_2^2 +2\|\Ob\|^2  ,\quad \sum_{\caT\in \scrT_\ns}| \hka(\caT)|\left\|\caR_\caT(\Ob) \right\|\leq |\hka|(\scrT_\ns)\|\Ob\|_2^2,
	\end{align}	
	which  implies the first inequality in \eref{eq:OShNormGen}. The second  and third inequalities in \eref{eq:OShNormGen} hold because 
	\begin{align}\label{eq:hkaDeltaNsAbsUBproof}
	\frac{D+1}{D+d}\leq \frac{D+1}{D+2}\hka(\Delta)\leq 1,\quad  \frac{D+1}{D+2}|\hka|(\scrT_\ns)\leq 2d-4
	\end{align}	
	by \lref{lem:hkaLUB}.
	The last inequality in \eref{eq:OShNormGen}  holds because  $\|\Ob\|\leq \|\Ob\|_2$.

	\Eref{eq:ShNormGen}  follows from 
	\eqsref{eq:OShNormGen}{eq:hkaDeltaNsAbsUBproof}.
	The first inequality in \eref{eq:ShNormGenB} follows from \eref{eq:Stab1ShNormRatioGen}; the second inequality  follows from 
	\eref{eq:ShNormGen} and the fact that 
	$|\hka|(\Sigma(d))=6\hka(\Delta)+|\hka|(\scrT_\ns)$.  The last inequality in \eref{eq:ShNormGenB} can be proved as follows,
	\begin{align}
	&\frac{D+1}{D+2}[|\hka|(\Sigma(d))-3\hka(\Delta)] =|\hka|(\Sigma(d))-\frac{1}{D+2}|\hka|(\Sigma(d))-\frac{3(D+1)}{D+2}\hka(\Delta)
	\nonumber\\
	&\leq |\hka|(\Sigma(d))-\frac{1}{D+2}\hka(\Sigma(d))-\frac{3(D+1)}{D+2}\hka(\Delta)
	\leq |\hka|(\Sigma(d))-3+\frac{d-5}{D+d},
	\end{align}
	where the last inequality follows from \lref{lem:hkaDeltaNS} and the first inequality in \eref{eq:hkaDeltaNsAbsUBproof}.

	If  $\Ob$ is  diagonal in some stabilizer basis, then
	by virtue of \lsref{lem:RTO} and \ref{lem:RTOdiag} we can improve the second inequality in \eref{eq:RTOsymNsProof} as follows,
	\begin{align}
	\sum_{\caT\in \scrT_\ns}| \hka(\caT)|\left\|\caR_\caT(\Ob) \right\|&=\sum_{\caT\in \scrT_\ns\cap\scrT_0}| \hka(\caT)|\left\|\caR_\caT(\Ob) \right\|   
	+\sum_{\caT\in \scrT_\ns\cap\scrT_1}| \hka(\caT)|\left\|\caR_\caT(\Ob) \right\|\nonumber\\
	&\leq |\hka|(\scrT_\ns\cap\scrT_0) \|\Ob\|^2 +|\hka|(\scrT_\ns\cap\scrT_1)\|\Ob\|_2^2= \frac{1}{2}|\hka|(\scrT_\ns)(\|\Ob\|_2^2+\|\Ob\|^2), 
	\end{align}
	where the last equality holds because $|\hka|(\scrT_\ns)=2|\hka|(\scrT_\ns\cap\scrT_0)=2|\hka|(\scrT_\ns\cap\scrT_1)$ by \pref{pro:kappaSym}. Together with \eqsref{eq:ObShadowProof}{eq:RTOsymNsProof}, this equation   implies the first inequality in \eref{eq:DiagOShNormGen}. 
	The second inequality  in \eref{eq:DiagOShNormGen} follows from \eref{eq:hkaDeltaNsAbsUBproof}, which completes the proof of \thref{thm:ShNormGen}. 
\end{proof}

\begin{proof}[Proof of \coref{cor:ShNormGenPos}]
	By assumption $\hka(\caT)\geq 0$ for all $\caT\in \Sigma(d)$,  which means $|\hka|(\scrT_\ns)=\hka(\scrT_\ns)$. In conjunction with \lsref{lem:hkakaNsIneq} and \ref{lem:hkaLUB} we can deduce that
	\begin{align}
	&[|\hka|(\scrT_\ns) +\hka(\Delta) ]\|\Ob\|_2^2+2\hka(\Delta)\|\Ob\|^2=[\hka(\scrT_\ns) +3\hka(\Delta) ]\|\Ob\|_2^2-2\hka(\Delta)(\|\Ob\|_2^2-\|\Ob\|^2)\nonumber\\
	&\leq \frac{D+2}{D+d}[3+\kappa(\scrT_\ns) ]\|\Ob\|_2^2-\frac{2(D+2)}{D+d}(\|\Ob\|_2^2-\|\Ob\|^2)= \frac{D+2}{D+d}\left\{[1+\kappa(\scrT_\ns) ]\|\Ob\|_2^2+2\|\Ob\|^2\right\}.
	\end{align}
	This equation and \eref{eq:OShNormGen} together imply the first inequality in \eref{eq:OShNormGenPos}. The second inequality in \eref{eq:OShNormGenPos} holds because
	$\kappa(\scrT_\ns)\leq 2d-4$ by \lref{lem:kappaTLUB}; the last inequality holds because $\|\Ob\|\leq \|\Ob\|_2$. \Eref{eq:ShNormGenPos} is a simple corollary of  \eref{eq:OShNormGenPos}.
	
	Next, suppose $\Ob$ is  diagonal in some stabilizer basis. Then  \eref{eq:DiagOShNormGenPos} follows from  \thref{thm:ShNormGen} and \lref{lem:hkakaNsIneq}, given that $\kappa(\scrT_\ns)\leq 2d-4$ by \lref{lem:kappaTLUB}. 
\end{proof}

\section{\label{app:BalanceProofs}Proofs of results on Clifford orbits based on balanced ensembles}
In this appendix we prove \lref{lem:kahkaBalance} and \thsref{thm:Moment3Balance}, \ref{thm:ShNormBalance}.

\begin{proof}[Proof of \lref{lem:kahkaBalance}]The first equality in \eref{eq:kaSigBalance} holds because $\kappa(\scrE,\scrT_\ns)=|\scrT_\ns|\kappa(\scrE)$ and  $|\scrT_\ns|=2(d-2)$; the second equality holds because $\kappa(\scrE,\Sigma(d))=6\kappa(\scrE,\Delta)+\kappa(\scrE,\scrT_\ns)$ and $\kappa(\scrE,\Delta)=1$.
	\Eref{eq:hkaBalance} follows from \pref{pro:hkaka} and \eref{eq:kaSigBalance}. 
	\Eref{eq:hkaSigBalance} follows from \eref{eq:hkaBalance} given that
	$\hka(\scrE,\scrT_\ns)=2(d-2)\hka(\scrE)$ and  $\hka(\scrE,\Sigma(d))=6\hka(\scrE,\Delta)+\hka(\scrE,\scrT_\ns)$. The equality in  \eref{eq:hkaSigAbsBalance} holds because $|\hka|(\scrE,\scrT_\ns)=|\hka(\scrE,\scrT_\ns)|$ and $\hka(\scrE,\Delta)\geq0$,  which means
	$|\hka|(\scrE,\Sigma(d))=6\hka(\scrE,\Delta)+2(d-2)|\hka(\scrE)|$; the inequality  holds because
	\begin{align}
	\frac{(d-2)(D+2)(D+2-3s)}{d(D-1)(D+d)}\leq \frac{(d-2)(D+2)(D+5)}{d(D-1)(D+d)}\leq  \frac{10}{9}. 
	\end{align}
	If in addition $d\geq 7$ or $n\geq 2$, then the upper
	bound $10/9$  can be improved to 1. If $\kappa(\scrE)\geq 3/(D+2)$, then $s=1$ and 
	\begin{align}
	\frac{(d-2)(D+2)(D+2-3s)}{d(D-1)(D+d)}\leq 
	\frac{(d-2)(D+2)}{d(D+d)}\leq \frac{d-2}{d}\leq 1.
	\end{align}
	In both cases,  the coefficient $20/9$ in \eref{eq:hkaSigAbsBalance} can be replaced by 2. 
\end{proof}

\begin{proof}[Proof of \thref{thm:Moment3Balance}]
	According to \eref{eq:MomentBal}, we have
	\begin{align}
	\bQ(\orb(\scrE))&=[\hka(\scrE,\Delta)-\hka(\scrE)]P_\sym+ \frac{1}{6}\hka(\scrE) R(\scrT_\iso)+ \frac{1}{6}\hka(\scrE) R(\scrT_\defe).
	\end{align}
	If $d=3$, then $R(\scrT_\iso)=6P_\sym$ and $R(\scrT_\defe)/(2D)$ is a projector of rank $(D+1)/2$ by \lref{lem:RTisodefe}. So $\bQ(\orb(\scrE))$ has two distinct eigenvalues, namely
	\begin{align}
	\lambda_1&=\hka(\scrE,\Delta)+\frac{D}{3}\hka(\scrE)=\frac{D+2}{3}\kappa(\scrE),  \quad \lambda_2=\hka(\scrE,\Delta)=1-\frac{(d-2)(D+2)}{(D-1)(D+d)}\left[\kappa(\scrE)-\frac{3}{D+2}\right],
	\end{align}
	where the second and fourth equalities follow from \lref{lem:kahkaBalance}. 
	If $d=2\mmod 3$, then $R(\scrT_\defe)=0$ and $R(\scrT_\iso)/(2d+2)$ is a projector of rank $D(D+1)(D+d)/(2d+2)$ by \lref{lem:RTisodefe}. So $\bQ(\orb(\scrE))$ has two distinct eigenvalues, namely, 
	\begin{align}
	\!\!\lambda_1&=\hka(\scrE,\Delta)+\frac{d-2}{3}\hka(\scrE)=\frac{(d-2)(D+2)}{3(D+d)}\kappa(\scrE)+\frac{D+2}{D+d},  \;\; \lambda_2=\hka(\scrE,\Delta)-\hka(\scrE)=\frac{D+2}{D-1}[1-\kappa(\scrE)]. 
	\end{align}
	In both cases, $\lambda_1>\lambda_2$ if $\kappa(\scrE)> 3/(D+2)$, while $\lambda_1<\lambda_2$ 
	if $\kappa(\scrE)< 3/(D+2)$. In addition, $|\lambda_1-1|\geq |\lambda_2-1|$ except when $d=5$, in which case the inequality is reversed.

	Next, suppose  $d=1\mmod 3$; then $R(\scrT_\iso)/(2d-2)$ is a projector of rank $D(D+1)(D+d-2)/(2d-2)$, and $R(\scrT_\defe)/(2D+2)$ is a projector of rank $D$ by \lref{lem:RTisodefe}. In addition, the support of 
	$R(\scrT_\defe)$ is contained in the support of $R(\scrT_\iso)$, which in  turn is contained in $\Sym_3\bigl(\caH_d^{\otimes n}\bigr)$. So $\bQ(\orb(\scrE))$ has three distinct eigenvalues, namely, 
	\begin{equation}
	\begin{aligned}
	\lambda_1&=\hka(\scrE,\Delta)+\frac{D+d-3}{3}\hka(\scrE)=\frac{D+2}{3}\kappa(\scrE),\\
	\lambda_2&=\hka(\scrE,\Delta)+\frac{d-4}{3}\hka(\scrE)=\frac{(D+2)(dD-4D-d-2)}{3(D-1)(D+d)} \kappa(\scrE)+  \frac{(D+1)(D+2)}{(D-1)(D+d)},\\
	\lambda_3&=\hka(\scrE,\Delta)-\hka(\scrE)=\frac{D+2}{D-1}[1-\kappa(\scrE)].
	\end{aligned}
	\end{equation}
	Note that $|\lambda_1-1|\geq |\lambda_3-1|$. In addition,
	$\lambda_1>\lambda_2>\lambda_3$ if $\kappa(\scrE)> 3/(D+2)$, while $\lambda_1<\lambda_2<\lambda_3$ 
	if $\kappa(\scrE)< 3/(D+2)$. The above analysis clarifies the eigenvalues of the third normalized moment operator $\bQ(\orb(\scrE))$ as
	summarized in \thref{thm:Moment3Balance} and \tref{tab:Qbalance}. In conjunction with \lref{lem:RTisodefe} it is straightforward to determine the multiplicities of these  eigenvalues.
	It is also  straightforward	to determine the operator norms of $\bQ(\orb(\scrE))$ and $\bQ(\orb(\scrE))-P_\sym$. 
\end{proof}

\begin{proof}[Proof of \thref{thm:ShNormBalance}]Note that \thref{thm:ShNormGen} is still applicable when the pure state $|\Psi\>$ is replaced by the ensemble~$\scrE$. Based on this observation \eref{eq:OShNormBalance} can be proved as follows,
	\begin{align}
	\|\Ob\|^2_{\orb(\scrE)}&\leq 	
	\left[1+\frac{D+1}{D+2}|\hka|(\scrE,\scrT_\ns) \right]\|\Ob\|_2^2+2\|\Ob\|^2
	\nonumber\\
	&\leq\left[1+	
	\frac{2(d-2)(D+1)(D+2)}{(D-1)(D+d)}\left|\kappa(\scrE)-\frac{3}{D+2}\right|\right]
	\|\Ob\|_2^2+
	2\|\Ob\|^2\nonumber\\
	&\leq\left[1+	
	2d\left|\kappa(\scrE)-\frac{3}{D+2}\right|\right]
	\|\Ob\|_2^2+
	2	\|\Ob\|^2. 
	\end{align}
	Here the second inequality follows from \lref{lem:kahkaBalance}, and  the third inequality holds because
	\begin{align}
	\frac{2(d-2)(D+1)(D+2)}{(D-1)(D+d)}\leq 2d.
	\end{align}
	\Eref{eq:ShNormBalance} is a simple corollary of \eref{eq:OShNormBalance}. 
\end{proof}

\section{\label{app:GaussJacobi}Gauss sums and Jacobi sums}
Here we review some basic results on Gauss sums and Jacobi sums that are relevant to the current study. To this end we first need to review a few basic concepts about multiplicative characters of finite fields. Our discussion is based on \rscite{LidlN97book,BernEW98book}.

\subsection{\label{app:MultiChar}Multiplicative characters of $\bbF_d$}
Suppose $d$ is a prime. A \emph{multiplicative character} of the field $\bbF_d$ is a character of the multiplicative group $\bbF_d^{\times}=\bbF_d\setminus\{0\}$. Let $\nu$ be a primitive element in the field $\bbF_d$, then $\nu^{d-1}=1$ and $\bbF_d^{\times}=\{1,\nu, \nu^2, \ldots, \nu^{d-2}\}$.  For each $j=0,1,\ldots, d-1$, a multiplicative character of $\bbF_d$ can be defined as follows,
\begin{align}\label{eq:MultiChar}
\eta(\nu^k,d):=\rme^{2\pi\rmi jk/(d-1)},  \quad k=0,1,\ldots, d-2.
\end{align}
Moreover, any multiplicative character of $\bbF_d$ has this form for some $j$. 
Here $d$ may be omitted if it is clear from the context. The multiplicative character corresponding to $j=0$ is trivial, while other characters are nontrivial.  By convention $\eta(0)=1$ if $\eta$ is trivial and $\eta(0)=0$ if $\eta$ is nontrivial.

Given a multiplicative character $\eta$ of $\bbF_d$, denote by $\eta^*$  the character obtained from $\eta$ by taking complex conjugation, that is, $\eta^*(a):=[\eta(a)]^*$. Given $m\in \bbN$, denote by $\eta^m$  the character defined as follows,
\begin{align}\label{eq:etam}
\eta^m(a):=[\eta(a)]^m\quad \forall a\in \bbF_d^{\times}. 
\end{align}
Note that the equality $\eta^m(a)=[\eta(a)]^m$ is not guaranteed  when $a=0$ because $\eta^m(0)$  is determined by whether $\eta^m$ is trivial. 
If $\eta$ is the character specified by \eref{eq:MultiChar}, then $\eta^m$ is identical to the  character corresponding to $jm \mmod (d-1)$ in \eref{eq:MultiChar}. 
The order of $\eta$ is defined as the smallest positive integer $m$ such that $\eta^m$ is trivial. By definition the order of $\eta$ is a divisor of $d-1$; 
$\eta$ has order 1 iff it is trivial.

An element in $\bbF_d$ is a \emph{quadratic residue} if it is the square of another element and a \emph{quadratic nonresidue} otherwise. 
The \emph{quadratic character} is the unique multiplicative character of order 2 and is denoted by $\eta_2$. By definition we have $\eta_2(\nu^k)=(-1)^k$,
which means $\eta_2(a)=1$ if $a$ is a nonzero quadratic residue and $\eta_2(a)=-1$ if $a$ is a quadratic nonresidue.  The quadratic character $\eta_2$ obeys an important  quadratic reciprocity relation as stated in the following lemma, which reproduces Theorem~5.17 in \rcite{LidlN97book}. 
\begin{lemma}\label{lem:QuadraticReciprocity}
	Suppose $p$ and $s$ are two distinct odd primes. Then $\eta_2(p,s)\eta_2(s,p)=(-1)^{(p-1)(s-1)/4}$.
\end{lemma}

An element in $\bbF_d$ is a \emph{cubic residue} if it is the cube of another element and a \emph{cubic nonresidue} otherwise. When $d\neq 1\mod 3$, every element in $\bbF_d$ is a cubic residue. 
When $d=1\mmod 3$, the situation is different: 0 and one third elements in $\bbF_d^{\times}$ are cubic residues. A multiplicative character  $\eta$ of $\bbF_d$ is a \emph{cubic character} if it has order 3. In that case, $\eta(a)=1$ iff $a$ is a nonzero cubic residue.  In addition,
$\eta$  corresponds to the character in \eref{eq:MultiChar} with
$j=(d-1)/3$ or $j=2(d-1)/3$, and the first choice means $\eta(\nu^k)=\omega_3^k$,
where $\omega_3=\rme^{2\pi\rmi/3}$. 
This cubic character depends on the choice of the primitive element $\nu$, in contrast with $\eta_2$, which is independent of $\nu$.

\subsection{Gauss sums}
Here we review two types of Gauss sums that are crucial to the current study. 
Let $m>1$ be  a positive integer, $a\in \bbF_d^{\times}$, and let $\eta$ be a nontrivial multiplicative character of $\bbF_d$ of order $m$.  The Gauss sums  $g(m,a)$ and $G(\eta,a)$ \cite{LidlN97book,BernEW98book} are defined as 
\begin{align}\label{eq:GaussSumDef}
g(m,a):=\sum_{b\in \bbF_d}\omega_d^{ab^m},\quad 
G(\eta,a):=\sum_{b\in \bbF_d}\eta(b)\omega_d^{ab}=\sum_{b\in \bbF_d^{\times}}\eta(b)\omega_d^{ab},
\end{align}
where $\omega_d=\rme^{2\pi \rmi/d}$ and the last equality holds because $\eta(0)=0$ given that $\eta$ is nontrivial. When $a=1$, $g(m,a)$ and $G(\eta,a)$ are abbreviated as $g(m)$ and $G(\eta)$, respectively.  According to Eq.~(1.1.4) in \rcite{BernEW98book},
the two types of Gauss sums are connected by the following formula,
\begin{align}\label{eq:GaussSumConect}
g(m,a)=\sum_{j=1}^{m-1}G(\eta^j,a), 
\end{align}
which means $G(\eta,a)=g(m,a)$ when $m=2$. 
The complex conjugate and $m$th power of 
$G(\eta,a)$ are denoted by $G^*(\eta,a)$ and $G^m(\eta,a)$, respectively, and similarly for  $g(m,a)$.

Next, we summarize a few basic properties of Gauss sums. 
\pref{pro:GaussSum} below follows from Theorems~1.1.3 and 1.1.4 in \rcite{BernEW98book}; see also Theorems~5.11 and 5.12 in \rcite{LidlN97book}.
\begin{proposition}\label{pro:GaussSum}
	Suppose $d$ is an odd prime, $a\in \bbF_d^{\times}$, and $\eta$ is a nontrivial multiplicative character of $\bbF_d$. Then 
	\begin{align}
	|G(\eta,a)|=\sqrt{d},\quad
	G(\eta,a)=\eta(1/a)G(\eta)=\eta^*(a)G(\eta),\quad 
	G^*(\eta,a)=\eta(-1)G(\eta^*,a).
	\end{align} 
\end{proposition}
Now suppose $m>1$ is a positive integer that divides $d-1$ and $\eta$ is a  multiplicative character of $\bbF_d$ of order $m$. As a simple corollary of 
\eref{eq:GaussSumConect} and \pref{pro:GaussSum}, we can deduce that
\begin{gather}
g(m,a)=\sum_{j=1}^{m-1}\eta^j(1/a)G(\eta^j)=\sum_{j=1}^{m-1}\eta^{j*}(a)G(\eta^j),  \quad 
|g(m,a)|\leq (m-1)\sqrt{d}\quad \forall a\in \bbF_d^{\times}, \label{eq:gmaGUB}\\
\sum_{k=0}^{m-1}g(m,\nu^k)=\sum_{k=0}^{m-1}G(\eta,\nu^k)=0 .\label{eq:GaussSumSum}
\end{gather}
In addition, if $a,b\in \bbF_d^{\times}$ satisfy $\eta(a)=\eta(b)$, then $G(\eta,a)=G(\eta,b)$ and $g(m,a)=g(m,b)$.

\begin{table*}
	\caption{\label{tab:tgd}The value of $\tg(d)^2$   [see \eref{eq:tgd}]  for each odd prime $d$ satisfying $d=1\mmod 3$ and $d\leq 1000$.  }
	\renewcommand{\arraystretch}{1.4}
	\small
	\begin{tabular}{c | ccccc  ccccc   ccccc   ccccc}
		\hline\hline
		$d$  & 7&  13&  19&  31&  37&  43&  61&  67&  73&  79&  97&  103&  109&  127&  139&  151& 	157&  163&  181&  193\\
		
		$\tg(d)^2$ &  1&  25&  49&  16&  121&  64&  1&  25&  49&  289&  361&  169&  4&  400&  529&  361& 	196&  625&  49&  529\\
		
		$d$  & 199&  211&  223&  229&  241&  271& 277& 283& 307& 313& 331& 337& 349&  367& 373& 379& 397& 409& 421& 433\\ 
		
		$\tg(d)^2$ &  121&  169&  784&  484&  289&  841&  676& 1024& 256& 1225& 1& 25& 1369& 1225& 169& 841& 1156& 961& 361& 4\\
		
		$d$  & 439& 457& 463&  487&  499&  523&  541&  547&  571&  577&  601&  607&  613& 	619&  631&  643&  661&  673&  691&  709\\ 
		
		$\tg(d)^2$ &  784&  100&  529&  625&  1024&  1849&  841&  1&  961&  121&  676&  2401&  2209&  289&  1849&  1600&  2401&  1369&  64& 2809\\
		
		$d$  &  727&  733&  739&  751&  757&  769&  787&  811&  823&  829&  853&  859&  877&  883&  907&  919&  937&  967&  991&  997\\ 
		$\tg(d)^2$ &  1936&  2500&  256&  1681&  841&  2401&  961&  3136&  25&  49&  1225&  169&  3481&  2209&  361&  2704&  3721&  1681&  3721&  100\\
		\hline	\hline 
	\end{tabular}
\end{table*}

When $\eta$ is a cubic character, which is the most important case to the current study, we have $\eta^*=\eta^2$ and  $\eta(-1)=1$, so \eref{eq:GaussSumConect} and
\pref{pro:GaussSum} imply that
\begin{gather}
G^*(\eta,a)=G(\eta^*,a)=G(\eta^2,a),\quad 
g(3,a)=2\Re G(\eta,a) =2\Re G(\eta^2,a)=2\Re G(\eta^*,a). \label{eq:gGaussSumCubic}
\end{gather}	
Therefore, $\Re G(\eta,a)$ is independent of the cubic character $\eta$, in contrast with $G(\eta,a)$.  In addition, the function  $\tg(d)=g(3,1)g(3,\nu)g(3,\nu^2)/d$  defined in \eref{eq:tgd} is of special interest to us; note that it is  independent of the choice of  $\nu$. The  value of $\tg(d)^2$ for each odd prime $d$ satisfying $d=1\mmod 3$ and $d\leq 1000$ is shown in \tref{tab:tgd}. The following lemma offers informative bounds for the Gauss sums $g(3,1)$, $g(3,\nu)$,  $g(3,\nu^2)$ and the function $\tg(d)$; it improves over Theorem~4.1.2 in \rcite{BernEW98book}.

\begin{lemma}\label{lem:GaussSum3LUB}
	Suppose $d$ is an odd prime satisfying $d=1\mmod 3$. Then  $g(3,1)$, $g(3,\nu)$,  $g(3,\nu^2)$ are three distinct real numbers, and
	$\tg(d)$ is an integer that  satisfies $\tg(d)=1\mmod 3$.  In addition,
	\begin{gather}
	g(3,1)+g(3,\nu)+g(3,\nu^2)=0, \quad 1\leq |\tg(d)|\leq \sqrt{4d-27},	\quad \frac{1}{3}< |g(3,\nu^j)|<\sqrt{4d-3}\quad \forall j=0,1,2.  \label{eq:GaussSumProdLUB}
	\end{gather}
	Moreover, each of  the following three open intervals 
	\begin{align}\label{eq:GaussInterval}
	\bigl(-\sqrt{4d-3}, -\sqrt{d}\lsp\bigr),\quad \bigl(-\sqrt{d},\sqrt{d}\lsp\bigr),\quad \bigl(\sqrt{d},\sqrt{4d-3}\lsp\bigr)
	\end{align}
	contains exactly one of the three numbers
	$g(3,1)$, $g(3,\nu)$, $g(3,\nu^2)$.
	Each of  the following three open intervals 
	\begin{align}\label{eq:GaussInterval2}
	(0, d),\quad(d,3d),\quad (3d,4d-3)
	\end{align}
	contains exactly one of the three numbers
	$g^2(3,1)$, $g^2(3,\nu)$, and $g^2(3,\nu^2)$.		
\end{lemma}

Note that the first two inequalities in \eref{eq:GaussSumProdLUB}  are saturated simultaneously when $d=7$. Let $\eta$ be any cubic character of the field $\bbF_d$.  By virtue of \eref{eq:gmaGUB} and \lref{lem:GaussSum3LUB} we can deduce that $g(3,a)=g(3,b)$ iff $\eta(a)=\eta(b)$;  the same result still holds if $g(3,a)$ and $g(3,b)$ are replaced by $G(\eta,a)$ and $G(\eta,b)$, respectively. Notably, $g(3,a)$  for each $a\in \bbF_d$ is equal to $g(3,\nu^j)$ for some $j=0,1,2$.

\begin{proof}[Proof of \lref{lem:GaussSum3LUB}]
	According to Theorems 3.1.3 and 4.1.2 in \rcite{BernEW98book},  $g(3,1)$, $g(3,\nu)$,  $g(3,\nu^2)$ are three distinct roots of the following cubic polynomial with real coefficients,
	\begin{align}\label{eq:GaussSumfx}
	f(x)=x^3-3dx-ds,
	\end{align}
	where $s$ is an integer that satisfies the following condition with $j$ being another integer,
	\begin{align}\label{eq:rscondition}
	s^2+3j^2=4d, \quad s=1\mmod 3,\quad j=0\mmod 3.
	\end{align}
	Since the polynomial $f(x)$ cannot have two roots with the same absolute value, $|g(3,1)|$, $|g(3,\nu)|$, $|g(3,\nu^2)|$ are also distinct from each other. Now, the equality in \eref{eq:GaussSumProdLUB} is a simple corollary of the above analysis; alternatively, it follows from \eref{eq:GaussSumSum}.	Furthermore, we have
	$\tg(d)=s$ and $1\leq |s|\leq \sqrt{4d-27}<2\sqrt{d}$, which means $\tg(d)$ is an integer that satisfies $\tg(d)=1\mmod 3$ and  the first two inequalities in \eref{eq:GaussSumProdLUB}.

	Next, direct calculation shows that	
	\begin{align}\label{eq:fxvalues}
	\begin{gathered}
	f\bigl(-\sqrt{4d-3}\lsp\bigr)=-(d-3)\sqrt{4d-3}-ds<0,\quad 
	f\bigl(-\sqrt{d}\lsp\bigr)=d\bigl(2\sqrt{d}-s\bigr)>0,\\
	f\bigl(\sqrt{d}\lsp\bigr)=-d\bigl(2\sqrt{d}+s\bigr)<0,\quad 
	f\bigl(\sqrt{4d-3}\lsp\bigr)=(d-3)\sqrt{4d-3}-ds>0,
	\end{gathered}
	\end{align}	
	where the four inequalities hold because
	\begin{align}
	\bigl[(d-3)\sqrt{4d-3}\lsp\bigr]^2-(ds)^2\geq (d-3)^2(4d-3)-d^2(4d-27)
	=54d-27>0,\quad |s|<2\sqrt{d}.
	\end{align}
	\Eref{eq:fxvalues} confirms that the polynomial $f(x)$ has three distinct real roots and implies that each of the  three open intervals in \eref{eq:GaussInterval} contains exactly one root, that is,  one of the three numbers
	$g(3,1)$, $g(3,\nu)$, $g(3,\nu^2)$. This observation also implies the last inequality  in \eref{eq:GaussSumProdLUB}. 
	
	According to the above analysis, the interval $(0,d)$ contains one of the three numbers $g^2(3,1)$, $g^2(3,\nu)$, $g^2(3,\nu^2)$, while the interval $(d,4d-3)$ contains the other two numbers. Suppose $x$ is a root of $f(x)$ in \eref{eq:GaussSumfx}. Then $x^2>3d$  if $x$ has the same sign as $s$, while  $x^2<3d$  if $x$ has the opposite sign. Therefore, each of the three open intervals in \eref{eq:GaussInterval2}
	contains exactly one of the three numbers
	$g^2(3,1)$, $g^2(3,\nu)$, $g^2(3,\nu^2)$.

	Finally, according to the above analysis we can deduce that $0< g^2(3,\nu^j)\leq 4d-3$ and $g^2(3,\nu^j)\neq 3d$ [cf. \eref{eq:GaussSumfx}], which imply that $	0<|g^2(3,\nu^j)-3d|<3d$.
	As a corollary, 
	\begin{align}
	|g(3,\nu^j)|=\frac{|ds|}{|g^2(3,\nu^j)-3d|}>\frac{|s|}{3} \geq \frac{1}{3},
	\end{align}
	which  confirms the third inequality  in \eref{eq:GaussSumProdLUB} and completes the proof of \lref{lem:GaussSum3LUB}. 
\end{proof}

\subsection{Jacobi sums}
Let $\theta_1, \theta_2, \ldots, \theta_k$ be $k$ multiplicative characters of $\bbF_d$; the Jacobi sum associated with these characters \cite{LidlN97book,BernEW98book} is defined as 
\begin{align}
J(\theta_1, \theta_2, \ldots, \theta_k)=\sum_{b_1 +b_2 +\cdots +b_k=1}\theta_1(b_1)\theta_2(b_2)\cdots \theta_k(b_k),
\end{align}
where the summation is taken over all $k$-tuples $(b_1, b_2, \ldots, b_k)$ of elements in $\bbF_d$ that sum up to 1. Note that $J(\theta)=1$ for any multiplicative character $\theta$ of $\bbF_d$, so the special case $k=1$ is trivial.

\Lsref{lem:JacobiGauss} and \ref{lem:JacobiAbs} below follow from Theorems 5.21 and 5.22 in \rcite{LidlN97book}, respectively. 
\begin{lemma}\label{lem:JacobiGauss}
	Suppose $d$ is a prime, $a\in \bbF_d^{\times}$, and  $\theta_1, \theta_2, \ldots, \theta_k$ are nontrivial multiplicative characters of $\bbF_d$. If the product character $\theta_1 \theta_2 \cdots \theta_k$ is nontrivial, then 
	\begin{align}
	J(\theta_1, \theta_2, \ldots, \theta_k)=\frac{G(\theta_1, a)G(\theta_2, a)\cdots G(\theta_k,a)}{G(\theta_1\theta_2\cdots \theta_k,a)}; 
	\end{align}
	if instead  $\theta_1 \theta_2 \cdots \theta_k$ is trivial, then 
	\begin{align}
	&J(\theta_1, \theta_2, \ldots, \theta_k)=-\theta_k(-1)J(\theta_1, \theta_2,\ldots, \theta_{k-1})=-\frac{1}{d}G(\theta_1, a)G(\theta_2, a)\cdots G(\theta_k,a). 
	\end{align}
\end{lemma}

\begin{lemma}\label{lem:JacobiAbs}
	Suppose $d$ is a prime and $\theta_1, \theta_2, \ldots, \theta_k$ are nontrivial multiplicative characters of $\bbF_d$. If the product character $\theta_1 \theta_2 \cdots \theta_k$ is nontrivial, then $|J(\theta_1, \theta_2, \ldots, \theta_k)|=d^{(k-1)/2}$; if instead $\theta_1 \theta_2 \cdots \theta_k$ is trivial, then $|J(\theta_1, \theta_2, \ldots, \theta_k)|=d^{(k-2)/2}$.
\end{lemma}

Here we are particularly interested in Jacobi sums associated with the quadratic character $\eta_2$ and any given cubic character $\eta_3$ of $\bbF_d$.  
By definition we can deduce the following relations
\begin{align}\label{eq:Jacobi23Conjugate}
J^*(\eta_2, \eta_3)&=J(\eta_2, \eta_3^*)=J(\eta_2, \eta_3^2), \quad  J^*(\eta_3, \eta_3)=J(\eta_3^*, \eta_3^*)=J(\eta_3^2, \eta_3^2).
\end{align}
In addition, by virtue of \pref{pro:GaussSum}, \eref{eq:gGaussSumCubic},  and \lsref{lem:JacobiGauss}, \ref{lem:JacobiAbs} we can deduce that
\begin{align}
J(\eta_3, \eta_3)&=\frac{G^2(\eta_3,a)}{G(\eta_3^2,a)}=\frac{G^3(\eta_3,a)}{d}=\frac{G^3(\eta_3)}{d}\quad \forall a \in \bbF_d^{\times},  \quad 
J(\eta_3, \eta_3^2)=-1,\label{eq:JacobiGaussCubic} 
\end{align}
and the same results still hold if $\eta_3$ and $\eta_3^2$ are exchanged. The relation between $J(\eta_2, \eta_3)$ and $J(\eta_3, \eta_3)$ is clarified in 
\lref{lem:Jacobi2333} below, which  follows from Eq.~(3.1.6) and Table~3.1.2 in \rcite{BernEW98book}.
\begin{lemma}\label{lem:Jacobi2333}
	Suppose $d$ is an odd prime satisfying $d=1\mmod 3$. Then $\eta_3(2)J(\eta_2, \eta_3)=J(\eta_3, \eta_3)$,
	and the same result still holds if $\eta_3$ is replaced by $\eta_3^2$. 
\end{lemma}

In addition, the Jacobi sum $J(\eta_3, \eta_3)$ can be expressed as a sum of cubic characters of quadratic functions as shown in the following lemma.
\begin{lemma}\label{lem:eta3}
	Suppose $d$ is an odd prime satisfying $d=1\mmod 3$. Then 
	\begin{align}\label{eq:eta3}
	\sum_{x\in \bbF_d} \eta_3(x^2+x)=J(\eta_3, \eta_3),\quad   \Biggl|\sum_{x\in \bbF_d} \eta_3(x^2+x)\Biggr|=\sqrt{d}. 
	\end{align}
\end{lemma}

\begin{proof}[Proof of \lref{lem:eta3}]
	According to the definition of $\eta_3$ we can deduce that
	\begin{align}
	\sum_{x\in \bbF_d} \eta_3(x^2+x)=\sum_{x\in \bbF_d} \eta_3((x-1)^2+x-1)=\sum_{x\in \bbF_d} \eta_3(x-x^2)
	=\sum_{x\in \bbF_d} \eta_3(x)\eta_3(1-x)=
	J(\eta_3, \eta_3),
	\end{align}
	which confirms the first equality in \eref{eq:eta3}. Here the second equality holds because $\eta_3(x^2-x)=\eta_3(x-x^2)$ given that $\eta_3(-1)=1$; the last equality follows from the definition of $J(\eta_3,\eta_3)$. The second equality in \eref{eq:eta3} follows from  the first equality and the fact that $|J(\eta_3,\eta_3)|=\sqrt{d}$ by  \lref{lem:JacobiAbs}. 
\end{proof}

\section{\label{sec:CubicEqFF}Cubic equations over finite fields}
Here we clarify the numbers of solutions of several cubic equations over the finite field $\bbF_d$ that are relevant to the current study. Our analysis is based on Theorems 6.33 and 6.34 in \rcite{LidlN97book}. Suppose $f$ is a polynomial over $\bbF_d$ and $c\in \bbF_d$. The number of solutions to the equation $f=c$ is denoted by $N_c(f)$ or $N(f=c)$; when $c=0$, $N_c(f)$ is abbreviated as $N(f)$. In the following discussion, $\eta_2$ denotes the quadratic character of $\bbF_d$, and $\eta_3$ denotes  a cubic character of $\bbF_d$.

\begin{lemma}\label{lem:CubicSolNum}
	Suppose $d$ is an odd prime, $a,b\in \bbF_d^{\times}$, $c\in \bbF_d$, and 
	$f(x,y)=ax^3+by^2$  is a polynomial in $\bbF_d[x,y]$.  If  $c=0$ or   $d\neq 1\mmod 3$, then $N_c(f)=d$. 	
	If instead $c\neq 0$ and $d=1\mmod 3$, then 	
	\begin{gather}
	\!\!	N_c(f)=d+2\Re[\eta_2(c/b)\eta_3(c/(2a))J(\eta_3,\eta_3)]=d+\frac{2}{d}\Re\bigl[\eta_2(c/b)\eta_3(c/(2a))G^3(\eta_3)\bigr],\;\;  |N_c(f)-d|\leq 2\sqrt{d}. \label{eq:CubicSolNum}
	\end{gather}	
\end{lemma}

When  $a,b,c\neq 0$ and $d=1\mmod 3$, the number $N_c(f)$ in \lref{lem:CubicSolNum} only depends on the cubic character of $c/a$ and  quadratic character of $c/b$. \Eref{eq:CubicSolNum} can be simplified when $c/a$ is a cubic residue, which means $\eta_3(c/a)=1$, or if  $c/b$ is a quadratic residue, which means $\eta_2(c/b)=1$. 

\begin{proof}[Proof of \lref{lem:CubicSolNum}]
	If  $d\neq1\mmod 3$ (which means $d=3$ or $d=2\mmod 3$), then every element in $\bbF_d$ is a cubic residue, and  the equation
	$a x^3 +b y^2=c$ has exactly one solution for each $y\in \bbF_d$, so $N_c(f)=d$.	Alternatively, this result follows from Theorems~6.33 and 6.34 in \rcite{LidlN97book}.

	Next, suppose $d=1\mmod 3$. Then $N_0(f)=d$ according to  Theorem~6.33 in \rcite{LidlN97book}. If $c\neq 0$, then	
	\begin{align}
	N_c(f)&=
	d+\sum_{j=1}^2
	\eta_2(c/b)\eta_3^j(c/a)J(\eta_2, \eta_3^j)=d+2\Re[\eta_2(c/b)\eta_3(c/a)J(\eta_2, \eta_3)]	\nonumber\\
	&=d+2\Re[\eta_2(c/b)\eta_3(c/(2a))J(\eta_3, \eta_3)]=d+\frac{2}{d}\Re\bigl[\eta_2(c/b)\eta_3(c/(2a))G^3(\eta_3)\bigr],
	\end{align}
	which confirms the equalities in \eref{eq:CubicSolNum}. 
	Here the first equality follows from 
	Theorem~6.34 in \rcite{LidlN97book}, the second equality follows from \eref{eq:Jacobi23Conjugate} given that $\eta_3^2=\eta_3^*$,  the third equality follows from \lref{lem:Jacobi2333},
	and the fourth equality follows from \eref{eq:JacobiGaussCubic}. The inequality $|N_c(f)-d|\leq 2\sqrt{d}$ follows from the fact  that $|\eta_2(c/b)|=|\eta_3(c/(2a))|=1$ and $|J(\eta_3,\eta_3)|=|G(\eta_3)|=\sqrt{d}$
	(see \pref{pro:GaussSum} and \lref{lem:JacobiAbs}).
\end{proof}

\begin{lemma}\label{lem:Cubic312}
	Suppose $d$ is an odd prime, $a,b\in \bbF_d^{\times}$, $c\in \bbF_d$, $a/b$ is a quadratic residue,	
	and $f(x,y)=ax^3-bxy^2$  is a polynomial in $\bbF_d[x,y]$. Then $N_0(f)=3d-2$.	
	If in addition $c\neq 0$ and  $d\neq 1 \mmod 3$, then $N_c(f)=d-2$.
	If instead $c\neq 0$ and $d=1 \mmod 3$, then
	\begin{gather}	
	N_c(f)=d-2+2\Re[\eta_3(a/(2c))J(\eta_3,\eta_3)]=d-2+\frac{2}{d}\Re\bigl[\eta_3(a/(2c))G^3(\eta_3)\bigr],  \label{eq:Cubic312Ncfd1}\\
	|N_c(f)-(d-2)|\leq 2\sqrt{d}.  \label{eq:Cubic312NcfLUB}
	\end{gather}
\end{lemma}

\begin{proof}[Proof of \lref{lem:Cubic312}]
	Note that $f$ has $d$ zeros with $x=0$ and two zeros for each $x\in \bbF_d^{\times}$ given that $a/b$ is a quadratic residue by assumption. Therefore, $N_0(f)=3d-2$.

	Next, suppose $c\neq 0$;  then the equation $f(x,y)=c$ means $x\neq 0$. Let $u=1/x$ and $v=u y$; then the equation $f(x,y)=c$ amounts to the following equation in $u$ and $v$, 
	\begin{align}
	c u^3+b v^2=a,\quad u\neq 0.
	\end{align} 
	Therefore,
	\begin{align}
	N_c(f)=N(c u^3+b v^2=a,u\neq 0)=N(c u^3+b v^2=a)-2,
	\end{align}
	which completes the proof of \lref{lem:Cubic312}	thanks to \lref{lem:CubicSolNum}. 	
\end{proof}

\begin{lemma}\label{lem:Cubic33SolNum}
	Suppose $d$ is an odd prime,  $c\in \bbF_d$, and  $f(x,y)=x^3-y^3$  is a polynomial in $\bbF_d[x,y]$. If  $d\neq 1\mmod 3$, then $N_c(f)=d$. If instead  $d=1\mmod 3$, then 	
	\begin{gather}\label{eq:Cubic33SolNumd1}
	N_c(f)=\begin{cases}
	3d-2 & \mbox{if} \; c=0, \\
	d-2+
	2\Re \eta_3^2(c)J(\eta_3, \eta_3) & \mbox{if} \; c\neq 0,
	\end{cases}\\
	|N_c(f)-(d-2)|\leq 2\sqrt{d}\quad \forall c\in \bbF_d^{\times}. \label{eq:Cubic33SolNumLUB}
	\end{gather}
\end{lemma}

\begin{proof}[Proof of \lref{lem:Cubic33SolNum}]
	If  $d\neq 1\mmod 3$, then every element in $\bbF_d$ is a cubic residue, and the equation $x^3-y^3=c$ has exactly one solution for each $x\in \bbF_d$, so $N_c(f)=d$. Alternatively,  this result follows from Theorems~6.33 and 6.34  in \rcite{LidlN97book}.
	
	Next, suppose  $d=1\mmod 3$. Then the equation $x^3-y^3=0$ has one solution with $x=0$ and three solutions for each $x\in \bbF_d^{\times}$. Therefore,  $N_0(f)=3d-2$, which confirms \eref{eq:Cubic33SolNumd1} in the case $c=0$. When $c\in \bbF_d^{\times}$, by virtue of Theorem~6.34 in \rcite{LidlN97book} we can deduce that 
	\begin{align}
	N_c(f)&=N(x^3+y^3=c)=d+\sum_{j_1, j_2=1}^2 \eta_3^{j_1}(c)\eta_3^{j_2}(c)J\bigl(\eta_3^{j_1}, \eta_3^{j_2}\bigr)=d+2J\bigl(\eta_3, \eta_3^2\bigr)+
	\sum_{j=1}^2 \eta_3^{2j}(c)J\bigl(\eta_3^{j}, \eta_3^{j}\bigr)\nonumber\\
	&=d-2+
	2\Re \eta_3^2(c)J(\eta_3, \eta_3), 
	\end{align}
	where the last equality follows from \eqsref{eq:Jacobi23Conjugate}{eq:JacobiGaussCubic}. This equation  confirms \eref{eq:Cubic33SolNumd1} 
	and implies
	\eref{eq:Cubic33SolNumLUB} given that $|\eta_3^2(c)J(\eta_3, \eta_3)|=|J(\eta_3, \eta_3)|=\sqrt{d}$ by \lref{lem:JacobiAbs}. 
\end{proof}

\begin{lemma}\label{lem:Cubic213SolNum}
	Suppose $d$ is an odd prime satisfying $d=1\mmod 3$, $c\in\bbF_d^{\times}$, and $f(x,y)=x^2+x-cy^3$  is a polynomial in $\bbF_d[x,y]$. Then 
	\begin{align}\label{eq:Cubic213SolNum}
	N(f)=d+2 \Re \bigl[\eta_3^2(c)J(\eta_3,\eta_3)\bigr]=d+\frac{2}{d} \Re \bigl[\eta_3^2(c)G^3(\eta_3)\bigr].
	\end{align}
\end{lemma}

\begin{proof}[Proof of \lref{lem:Cubic213SolNum}]
	Denote by $1/2$ and $1/4$ the inverses of $2$ and $4$ in $\bbF_d$, then we have	 $x^2+x=(x+1/2)^2-1/4$. Therefore,
	\begin{align}
	N(f)&=N\bigl(cy^3-x^2=-1/4\bigr)=d+\sum_{j=1}^2\eta_2(1/4)\eta_3^j(-1/(4c))J\bigl(\eta_2, \eta_3^j\bigr)
	=d+2\Re[\eta_3(2/c)J(\eta_2, \eta_3)]\nonumber\\
	&=d+2\Re[\eta_3(1/c)\eta_3(2)J(\eta_2, \eta_3)]=d+2\Re\bigl[\eta_3^2(c)J(\eta_3, \eta_3)\bigr],
	\end{align}
	where the second equality follows from Theorem~6.34 in \rcite{LidlN97book}, the third equality holds because 
	$\eta_2(1/4)=1$ and $\eta_3(-1/(4c))=\eta_3(-2/c)=\eta_3(2/c)$,
	and the last equality follows from \lref{lem:Jacobi2333}. The above equation  confirms the first equality in \eref{eq:Cubic213SolNum}; the second equality in \eref{eq:Cubic213SolNum} holds because  $J(\eta_3, \eta_3)=G^3(\eta_3)/d$ by \eref{eq:JacobiGaussCubic}. 
\end{proof}

\section{\label{app:CubicEqSLS}Cubic equations on stochastic Lagrangian subspaces}
To better understand the properties of magic orbits and to prove several key lemmas (including \lsref{lem:kappaMagicOned23} and \ref{lem:kappaMagicOned1}), in this appendix we clarify the numbers of solutions of certain cubic equations on stochastic Lagrangian subspaces. 

\subsection{Main results}
Suppose $d$ is an odd prime and let $f$ be a function from  $\bbF_d$ to $\bbF_d$, that is, $f\in \bbF_d[x]$; then $f$ can be extended to a function from $\bbF_d^3$ to $\bbF_d$ as follows,
\begin{align}\label{eq:fbfx}
f(\bfx):=f(x_1)+f(x_2)+f(x_3), \quad \bfx=(x_1, x_2, x_3)^\top\in \bbF_d^3.
\end{align}
Given $f\in \bbF_d[x]$, $\caT\in \Sigma(d)$, and $\alpha\in \bbF_d$, define $N_\alpha(f,\caT)$ as the number of solutions to the equation $f(\bfy)-f(\bfx)=\alpha$ for $(\bfx;\bfy)\in \caT$, that is,
\begin{align}\label{eq:NalphafcaT}
N_\alpha(f,\caT):=|\{ (\bfx;\bfy)\in \caT \ | \  f(\bfy)-f(\bfx)=\alpha\}|. 
\end{align}
By definition it is easy to verify the following results,
\begin{equation}\label{eq:NfTsymSum}
\begin{gathered}
N_\alpha(f,O_1\caT O_2)=N_\alpha(f,\caT)\quad \forall \caT\in \Sigma(d),\; O_1, O_2 \in S_3,   \\
\sum_{\alpha\in \bbF_d} N_\alpha(f,\caT)=|\caT|=d^3\quad \forall \caT\in \Sigma(d). 
\end{gathered}
\end{equation}

If in addition $\caT\in \scrT_\sym$, then $f(\bfx)=f(\bfy)$ whenever $(\bfx;\bfy)\in \caT$. Therefore,
\begin{align}\label{eq:NfTsym3}
N_\alpha(f,\caT)=\begin{cases}
d^3 & \alpha=0,\\
0 & \alpha \neq 0.
\end{cases}
\end{align}
This equation also holds for all $\caT\in \Sigma(d)$ when $d=3$ and $f$ is a cubic polynomial, which can be verified by virtue of \lref{lem:defectT}. In general, it is not easy to determine $N_\alpha(f,\caT)$ when $\caT\in \scrT_\ns$, where 
\begin{align}\label{eq:Sigma33scrTsym}
\scrT_\ns=\Sigma(d)\setminus \scrT_\sym=(\scrT_1\setminus \scrT_\sym)\sqcup [(\scrT_1\setminus \scrT_\sym)\tau_{12}].
\end{align}
Here our main goal is to determine $N_\alpha(f,\caT)$   when $f$ is a cubic polynomial. For the convenience of the following discussion, the special cubic polynomial $f(x)=x^3$ is referred to as the  the canonical cubic polynomial. When $d\geq 5$, this definition agrees with the definition in \sref{sec:MagicState}; when $d=3$, however, the canonical cubic polynomial is different from the canonical cubic function defined in  \sref{sec:MagicState}. When $f$ is the canonical cubic polynomial, $N_\alpha(f,\caT)$ is abbreviated as $N_\alpha(\caT)$ for simplicity.

\begin{lemma}\label{lem:NfTbasic}
	Suppose $d$ is an odd prime, $\caT\in \Sigma(d)$, and $f\in \bbF_d[x]$ is any cubic polynomial with cubic coefficient $c\neq 0$. Then 
	\begin{align}\label{eq:NfTNT}
	N_{\alpha\beta^3}(f,\caT)=N_{-\alpha}(f,\caT)=N_\alpha(f,\caT)&=N_{\alpha/c}(\caT) \quad  \forall \alpha\in \bbF_d, \; \beta\in \bbF_d^{\times}. 
	\end{align}
	If in addition $d\neq 1\mmod 3$, then 
	\begin{align}\label{eq:NfTNT2}
	N_\alpha(f,\caT)&=
	N_1(\caT)=\frac{d^3-N_0(\caT)}{d-1}\quad \alpha\in \bbF_d^{\times}. 
	\end{align}	
\end{lemma}

\Eref{eq:NfTNT} in particular implies that $N_0(f,\caT)=N_0(\caT)$ whenever $d$ is an odd prime, $\caT\in \Sigma(d)$, and $f$ is a polynomial of degree 3. If in addition $d=1\mmod 3$, then $N_\alpha(f,\caT)$ is completely determined by $\caT$ and $\eta_3(\alpha/c)$, where $\eta_3$ is a given cubic character.

\begin{proof}[Proof of \lref{lem:NfTbasic}]
	By assumption $f(x)$ can be expressed as follows,
	\begin{align}
	f(x)=cx^3+c_2x^2+c_1x+c_0, \quad c_0, c_1, c_2\in \bbF_d. 
	\end{align}
	Let $(\bfx;\bfy)\in \caT$; then $\bfx\cdot \mathbf{1}_3=\bfy\cdot \mathbf{1}_3$ and $\bfx\cdot \bfx=\bfy\cdot\bfy$. Therefore,
	\begin{align}
	f(\bfy)-f(\bfx)&=c\bigl(y_1^3+y_2^3+y_3^3-x_1^3-x_2^3-x_3^3\bigr),\quad 
	f(\beta\bfy)-f(\beta\bfx)=\beta^3[f(\bfy)-f(\bfx)]\quad \forall \beta\in \bbF_d,
	\end{align}
	which implies  \eref{eq:NfTNT}. 
	
	If in addition $d\neq 1\mmod 3$ (which means $d=3$ or $d=2\mmod 3$), then every element in $\bbF_d$ is a cubic residue, so \eref{eq:NfTNT2} is a simple corollary of  \eqsref{eq:NfTsymSum}{eq:NfTNT}. 
\end{proof}

\begin{lemma}\label{lem:NfT}
	Suppose $d\geq 5$ is an odd prime, $\alpha\in \bbF_d^{\times}$, $\caT\in \scrT_\ns$, 
	and $f$ is a cubic polynomial in $\bbF_d[x]$ with cubic coefficient $c\neq0$; let $\bfv=(v_1,v_2,v_3)^\top$ be a characteristic vector of $\caT$ and $a=v_1^3+v_2^3+v_3^3$. Then
	\begin{gather}
	N_0(f,\caT)=3d^2-2d.  \label{eq:N0fT}
	\end{gather}
	If in addition $d=2 \mmod 3$, then 
	\begin{align}
	N_\alpha(f,\caT)=d^2-2d.  \label{eq:NfTd2} 
	\end{align}
	If instead  $d=1 \mmod 3$, then 
	\begin{gather}	
	N_\alpha(f,\caT)=d^2-2d+2d\Re[\eta_3(ac/\alpha)J(\eta_3,\eta_3)]=d^2-2d+2\Re\bigl[\eta_3(ac/\alpha)G^3(\eta_3)\bigr], \label{eq:NfTd1}\\
	|N_\alpha(f,\caT)-(d^2-2d)|\leq 2d\sqrt{d}. \label{eq:NfTd1LUB}
	\end{gather}	
\end{lemma}
Note that $N_\alpha(f,\caT)$ is determined by \eref{eq:NfTsym3} when $d=3$, in which case the situation is very different.

\begin{proof}[Proof of \lref{lem:NfT}]
	Thanks to \eqsref{eq:NfTsymSum}{eq:Sigma33scrTsym},
	to prove \lref{lem:NfT}, we can assume that $\caT\in \scrT_1\cap \scrT_\ns$ without loss of generality. Thanks to \lref{lem:NfTbasic}, we can further assume that $f$ is the canonical cubic polynomial, that is, $f(x)=x^3$ and $c=1$.

	Let $\bfu$ be a vector in $\bfv^\perp\setminus \spa(\mathbf{1}_3)$ and $b=-3(u_1^2v_1+u_2^2v_2+u_3^2v_3)$; then $\bfu\cdot \bfv=0$ and $b\neq 0$ according to \lref{lem:QudraticResidueT} below. In addition, 
	$\mathbf{1}_3, \bfu, \bfv$ form a basis for $\bbF_d^3$, so any vector  in $\bbF_d^3$ can be expressed in the form $x \mathbf{1}_3+ y \bfu+z\bfv$ with $x, y, z \in \bbF_d$ in a unique way. Furthermore,
	\begin{gather}
	\caT=\{(x \mathbf{1}_3+ y \bfu-z\bfv; x \mathbf{1}_3+ y \bfu+z\bfv)\ |\ x, y, z\in \bbF_d\},\\
	f( x \mathbf{1}_3+ y \bfu+z\bfv)-f(x \mathbf{1}_3+ y \bfu-z\bfv)=2a z^3-2b y^2z. \label{eq:fyfxProof}
	\end{gather}
	Note that the final result  in \eref{eq:fyfxProof} is independent of the parameter $x$, so $N_\alpha(\caT)$ is $d$ times the number of solutions to the equation $2a z^3 -2b y^2z=\alpha$,
	that is,
	\begin{align}
	N_\alpha(\caT)=d N(2a z^3 -2b   y^2z=\alpha). 
	\end{align}
	Now \lref{lem:NfT} is a simple corollary of this equation and \lref{lem:Cubic312} given that $a/b$ is a quadratic residue by \lref{lem:QudraticResidueT}. 	
\end{proof}

\begin{lemma}\label{lem:NfTsum}
	Suppose $d$ is an odd prime satisfying $d\geq 5$, $\alpha\in \bbF_d^{\times}$, and  $f$ is a cubic polynomial in $\bbF_d[x]$ with cubic coefficient $c\neq0$. Then
	\begin{align}
	\sum_{\caT\in \scrT_\ns}	N_\alpha(f,\caT)&=2\sum_{\caT\in \scrT_1\cap\scrT_\ns}	N_\alpha(\caT), \label{eq:NfTsumSigmaT1}  \\
	\sum_{\caT\in \scrT_\ns}	N_0(f,\caT)&=2d(d-2)(3d-2). \label{eq:N0fTsum}
	\end{align}
	If in addition $d=2 \mmod 3$, then 
	\begin{align}
	\sum_{\caT\in \scrT_\ns}	N_\alpha(f,\caT)=2d(d-2)^2.  \label{eq:NfTsumd2} 
	\end{align}
	If in addition $d=1 \mmod 3$, then 
	\begin{gather}
	\sum_{\caT\in \scrT_\ns}	N_\alpha(f,\caT)=2d(d-2)^2+4d\Re\bigl[\eta_3(3c/\alpha)J^2(\eta_3,\eta_3)\bigr]=2d(d-2)^2+\frac{4}{d}\Re\bigl[\eta_3(3c/\alpha)G^6(\eta_3)\bigr], \label{eq:NfTsumd1} \\
	\Biggl|\sum_{\caT\in \scrT_\ns}	N_\alpha(f,\caT)-2d(d-2)^2\biggr|\leq 4d.  \label{eq:NfTsumLUB}
	\end{gather}	
\end{lemma}

\begin{proof}[Proof of \lref{lem:NfTsum}]
	\Eref{eq:NfTsumSigmaT1} follows from \eqsref{eq:NfTsymSum}{eq:Sigma33scrTsym}. 
	\Eqsref{eq:N0fTsum}{eq:NfTsumd2} are simple corollaries of \eqsref{eq:N0fT}{eq:NfTd2}, respectively, given that $|\scrT_\ns|=2d-4$.

	\Eref{eq:NfTsumd1} can be derived as follows,
	\begin{align}	
	\sum_{\caT\in \scrT_\ns}	N_\alpha(f,\caT)&=2\sum_{\caT\in \scrT_1\cap \scrT_\ns}	N_\alpha(\caT)=2\sum_{y=1}^{d-2}		N_\alpha\bigl(f,\caT_{(1;y;-1-y)}\bigr)\nonumber\\
	&=2\sum_{y=1}^{d-2} \bigl\{d^2-2d+2d\Re\bigl[\eta_3\bigl(-3c(y^2+y)/\alpha\bigr)J(\eta_3,\eta_3)\bigr]\bigr\}\nonumber\\
	&=2d(d-2)^2+4d\sum_{y\in \bbF_d} \Re\bigl[\eta_3\bigl(y^2+y\bigr)\eta_3(3c/\alpha)J(\eta_3,\eta_3)\bigr]\nonumber\\
	&=
	2d(d-2)^2+4d\Re\bigl[\eta_3(3c/\alpha)J^2(\eta_3,\eta_3)\bigr]=2d(d-2)^2+\frac{4}{d}\Re\bigl[\eta_3(3c/\alpha)G^6(\eta_3)\bigr]. 
	\end{align}
	Here the first equality follows from \eref{eq:NfTsumSigmaT1}; the second equality follows from \eref{eq:TvNS}; the third equality follows from \lref{lem:NfT}; the fourth equality holds because  $\eta_3$ is a nontrivial multiplicative character, which means $\eta_3(0)=0$ and $\eta_3(-1)=1$; the fifth equality follows from \lref{lem:eta3}; the last equality follows from \eref{eq:JacobiGaussCubic}.

	\Eref{eq:NfTsumLUB} is a simple corollary of \eref{eq:NfTsumd1}, given  that $|\eta_3(3c/\alpha)|=1$ and $|J(\eta_3,\eta_3)|=|G(\eta_3)|=\sqrt{d}$ by \pref{pro:GaussSum} and \lref{lem:JacobiAbs}.
\end{proof}

\subsection{Two auxiliary lemmas}
\begin{lemma}\label{lem:abcxyz}
	Suppose $d$ is a prime; $a,b,c\in  \bbF_d^{\times}$ and $x,y,z\in \bbF_d$ satisfy the following conditions,
	\begin{align}\label{eq:abcxyzCondition}
	a+b+c=0=xa+yb+zc=x^2a+y^2b+z^2c=0.
	\end{align}
	Then $x=y=z$. 
\end{lemma}

\begin{proof}[Proof of \lref{lem:abcxyz}]
	Let 
	\begin{align}
	A=\begin{pmatrix}
	1 & 1 &1 \\
	x & y & z\\
	x^2& y^2 &z^2
	\end{pmatrix},\quad \bfv=\begin{pmatrix}
	a \\
	b\\
	c
	\end{pmatrix}.
	\end{align}	
	Then the  assumption in \eref{eq:abcxyzCondition} can be expressed as $A\bfv=0$, 
	which implies that
	\begin{align}
	\det(A)=(x-y)(y-z)(z-x)=0,
	\end{align}
	given that $a,b,c\in \bbF_d^{\times}$. So at least two of the three numbers $x,y,z$ must be equal, say $x=y$. If  $x=y\neq z$, then $\ker(A)=\spa(1;-1;0)$ cannot contain any vector that is composed of three nonzero entries. Similarly, $\ker(A)$ cannot contain any vector that is composed of three nonzero entries if $x=z\neq y$ or $y=z\neq x$. 
	Therefore, $x=y=z$. 
\end{proof}

\begin{lemma}\label{lem:QudraticResidueT}
	Suppose $d\geq 5$ is an odd prime, $\caT\in \scrT_1\cap \scrT_\ns$,  $\bfv=(v_1,v_2,v_3)^\top$ is a characteristic vector of $\caT$, and $\bfu=(u_1,u_2,u_3)^\top\in \bfv^\perp$. Let
	\begin{align}\label{eq:abQR}
	a=v_1^3+v_2^3+v_3^3,\quad  b=-3\bigl(u_1^2v_1+u_2^2v_2+u_3^2v_3\bigr);
	\end{align}	
	then $a\neq 0$ and $b/a$ is a quadratic residue. 	In addition, $b=0$ iff $\bfu\in \spa(\mathbf{1}_3)$. 
\end{lemma}

\begin{proof}[Proof of \lref{lem:QudraticResidueT}] By assumption and \eref{eq:TvDeltaDefect} we have $\mathbf{1}_3, \bfu\in \caT_\tDelta=\bfv^\perp$, which implies that
	\begin{align}\label{eq:uvQRproof}
	\mathbf{1}_3\cdot\bfv=	v_1+v_2+v_3=0, \quad \bfu\cdot\bfv=u_1v_1+u_2v_2+u_3v_3=0.
	\end{align}	
	In addition,
	\begin{align}\label{eq:uvQRproof2}
	a=v_1^3+v_2^3+v_3^3=3v_1v_2v_3\neq 0
	\end{align}	
	according to \lref{lem:Tv}. If $\bfu\in \spa(\mathbf{1}_3)$, then  $b=0$ by \eqsref{eq:abQR}{eq:uvQRproof}, so $b/a$ is a quadratic residue.

	Next, suppose $\bfu\in\bfv^\perp\setminus \spa(\mathbf{1}_3)$. Then \eqsref{eq:uvQRproof}{eq:uvQRproof2}  means $b\neq 0$ according to \lref{lem:abcxyz}.
	Let $h(\bfu,\bfv)=b/a$ with $a,b$ defined in \eref{eq:abQR};
	then it is easy to verify that 
	\begin{align}\label{eq:huv}
	h(c\bfu,\bfv)=c^2 h(\bfu,\bfv), \quad h(\bfu+c\mathbf{1}_3,\bfv)=h(\bfu,\bfv) \quad \forall c\in \bbF_d; \quad h(\bfu,c\bfv)=\frac{h(\bfu,\bfv)}{c^2}\quad \forall c\in \bbF_d^{\times}. 
	\end{align}	
	To prove that $h(\bfu,\bfv)$ is a quadratic residue, we can assume that	$\bfv$ has the form $\bfv=(1;y;-1-y)$ with $y\in \bbF_d$ without loss of generality. In addition, we can assume that $\bfu$  has the form $\bfu=(-y;1;0)$  given that $\bfu\cdot\bfv=0$. Then $a=b=-3(y^2+y)$ and $ h(\bfu,\bfv)=1$,
	so $b/a=h(\bfu,\bfv)$ is a quadratic residue. 
\end{proof}

\section{Proofs of results on magic orbits when $d\neq 1\mmod 3$}

In this appendix we prove \lsref{lem:kappaMagicOned23}, \ref{lem:kahkaMagicd23} and  \thsref{thm:Phi3Magicd23}-\ref{thm:ShNormMagicd23}, which are tied to magic orbits in the case $d\neq 1\mmod 3$.

\subsection{\label{app:lem:Qud23MagicProof}Proofs of \lsref{lem:kappaMagicOned23} and \ref{lem:kahkaMagicd23}}

\begin{proof}[Proof of \lref{lem:kappaMagicOned23}]
	Thanks to \pref{pro:MagicGate},  $|\psi\>$ has the form $|\psi\>=\sum_{u\in \bbF_d}\tomega^{f(u)}|u\>/\sqrt{d}$,
	where $f\in \tscrP_3(\bbF_d)$. Note that the function $f$ can be extended from $\bbF_d$ to $\bbF_d^3$ following the recipe in  \eref{eq:fbfx}. 
	By virtue of  \eqsref{eq:rRT}{eq:kappapsiT} with $n=1$ we can deduce that
	\begin{align}
	\kappa(\psi,\caT)&=\tr\bigl[r(T)(|\psi\>\<\psi|)^{\otimes 3}\bigr]=
	\frac{1}{d^3}\sum_{(\bfx;\bfy)\in \caT} \tomega^{f(\bfy)-f(\bfx)}.
	\end{align}
	
	If 	$d=2\mmod 3$, then $\tomega=\omega$ and $f$ is a cubic polynomial over $\bbF_d$ with nonzero cubic coefficient by \pref{pro:MagicGate}.	Now, \eref{eq:kappaMagicOned23} can be proved as follows,
	\begin{align}
	\kappa(\psi,\caT)=\frac{1}{d^3}\sum_{\alpha\in \bbF_d} N_\alpha(f,\caT)\omega^{\alpha}=\frac{3d-2}{d^2}+
	\frac{d-2}{d^2}\sum_{\alpha\in \bbF_d^{\times}} \omega^{\alpha}=\frac{2d}{d^2}+
	\frac{d-2}{d^2}\sum_{\alpha\in \bbF_d} \omega^{\alpha}=\frac{2}{d}, 
	\end{align}	
	where $N_\alpha(f,\caT)=|\{ (\bfx;\bfy)\in \caT \ | \  f(\bfy)-f(\bfx)=\alpha\}|$ is defined in \eref{eq:NalphafcaT}, and the second equality follows from 	\lref{lem:NfT} in \aref{app:CubicEqSLS} given that $\caT\in \scrT_\ns$ by assumption.

	Next, suppose $d=3$. Then $\tomega=\omega_9$ 
	and $f$ is a function from  $\bbF_3$ to $\bbZ_9$ that has the form in \eref{eq:f39} by \pref{pro:MagicGate}. In addition,
	$\scrT_\ns=\scrT_\defe$ and we can assume that $\caT\in \scrT_\defe\cap\scrT_1$ without loss of generality thanks to \pref{pro:kappaSym}. According to \lref{lem:defectT} and \eref{eq:scrT01}, $\caT$ is equal to $\tcaT_0$ defined in \eref{eq:SLStwo} and thus has the form
	\begin{align}
	\caT=\{(a\bfz+b\mathbf{1}_3; a\bfz+c\mathbf{1}_3)\ |\ a,b,c\in \bbF_d\},
	\end{align}
	where $\bfz=(0,1,2)^\top$.  Based on these observations  we can deduce that
	\begin{align}
	\kappa(\psi, \caT)=\tr[r(\caT)(|\psi\>\<\psi|)^{\otimes 3}]=\frac{1}{27}\sum_{a,b,c\in \bbF_d} \omega_9^{f(a\bfz+c\mathbf{1}_3)-f(a\bfz+b\mathbf{1}_3)}.
	\end{align}
	If $a\neq 0$, then the three entries of $a\bfz+b\mathbf{1}_3$ are a permutation of $0,1,2$, and so are the  three entries of $a\bfz+c\mathbf{1}_3$, which means $f(a\bfz+c\mathbf{1}_3)-f(a\bfz+b\mathbf{1}_3)=0\mmod 9$. If $a=0$, then 
	\begin{equation}
	f(a\bfz+c\mathbf{1}_3)-f(a\bfz+b\mathbf{1}_3)=
	3f(c)-3f(b).
	\end{equation}	
	Therefore,	
	\begin{equation}	
	\kappa(\psi, \caT)=\frac{2}{3}+
	\frac{1}{27}\sum_{b,c\in \bbF_d} \omega_3^{f(c)-f(b)} =\frac{2}{3}=\frac{2}{d}.
	\end{equation}
	Here the second equality follows from \eref{eq:f39}, which implies that
	\begin{align}
	\sum_{c\in \bbF_d} \omega_3^{f(c)}=\sum_{b\in \bbF_d} \omega_3^{-f(b)}=0.
	\end{align}
	This observation completes the proof of \lref{lem:kappaMagicOned23}. 	
\end{proof}

\begin{proof}[Proof of \lref{lem:kahkaMagicd23}]
	The first equality in \eref{eq:kappaSigd23} follows from \lref{lem:kappaMagicOned23} and the product relation in  \eref{eq:kappaPsiProd}. The last two equalities in \eref{eq:kappaSigd23} are simple corollaries of the first equality together with \eqsref{eq:kappascrT}{eq:kappaSigNSscrT} given that $|\scrT_\ns|=2(d-2)$. 
	
	The equality in \eref{eq:hkaAbsUBd23} follows from \eqsref{eq:hkakaTLUB}{eq:kappaSigd23} (see also \lref{lem:kahkaBalance}). Meanwhile,	it is easy to verify that $\hka(\caT)<0$ if $n=k=1$ and $d\geq 5$, while $\hka(\caT)>0$ otherwise. In the later case we have
	\begin{align}\label{eq:hkad23Proof}
	|\hka(\caT)|=\hka(\caT)=\frac{ (D+2) \bigl[\frac{2^k}{d^k}(D+2)-3\bigr]}{(D-1)(D+d)}\leq \frac{ (D+2) \frac{2^k}{d^k}(D+2-3)}{(D-1)(D+d)}=
	\frac{2^k(D+2)}{d^k(D+d)}< \frac{2^k}{d^k},
	\end{align}
	which confirms the  inequalities in \eref{eq:hkaAbsUBd23}.  In addition, $|\hka|(\Sigma(d))=\hka(\Sigma(d))=(D+2)\kappa( \Sigma(d))/(D+d)$ by \pref{pro:hkaka}, which confirms \eref{eq:hkaAbsSigd23}.  In the former case ($n=k=1$ and $d\geq 5$) we have $D=d$  and
	\begin{align}
	|\hka(\caT)|=-	\hka(\caT)=\frac{(d-4)(d+2)}{2d^2(d-1)}\leq \frac{d+2}{2d^2}=\frac{2^k(D+2)}{2d^k(D+d)}< \frac{2^k}{d^k},
	\end{align}
	which confirms the inequalities in \eref{eq:hkaAbsUBd23}. The first inequality in \eref{eq:hkaAbsSigd23} still follows from \pref{pro:hkaka}. 	In addition, by virtue of \lref{lem:hkaDeltaNS} we can deduce that
	\begin{align}
	|\hka|(\scrT_\ns)&=-	\hka(\scrT_\ns)=2(d-2)|\hka(  \caT)|=\frac{(d-2)(d-4)(d+2)}{d^2(d-1)},\\
	|\hka(\Delta)|&=\hka(\Delta)=1-\frac{\hka(\scrT_\ns)}{2(d+2)}=	1+\frac{(d-2)(d-4)}{2d^2(d-1)}, \\
	|\hka|(\Sigma(d))&=6|\hka(\Delta)|+|\hka|(\scrT_\ns)=
	\frac{(d+2)(7d^2-21d+20)}{d^2(d-1)}, 
	\end{align}
	which confirm the second equality in \eref{eq:hkaAbsSigd23} and complete the proof of \lref{lem:kahkaMagicd23}. 
\end{proof}

\subsection{\label{app:thm:Phi3Shadow32Proof}Proofs of \thsref{thm:Phi3Magicd23}-\ref{thm:ShNormMagicd23}}

\begin{proof}[Proof of \thref{thm:Phi3Magicd23}]
	The equality in \eref{eq:PHi332}
	follows from \thref{thm:Phi3LUB} and  \eref{eq:kappaSigd23}; alternatively, it follows from \thref{thm:Phi3Balance}. The first inequality in \eref{eq:PHi332} follows from the inequality below,
	\begin{align}
	\left(\frac{2^k}{d^k}-\frac{3}{D+2}\right)^2\leq \frac{2^{2k}(D-1)^2}{d^{2k}(D+2)^2},
	\end{align}	
	and the second inequality in \eref{eq:PHi332} is trivial. If in addition $k\geq 1$, then $\bar{\Phi}_3(\orb(\Psi))\leq 29/25$ by \eref{eq:PHi332}. 
\end{proof}

\begin{proof}[Proof of \thref{thm:StabShNormMagicd23}]According to \lref{lem:kahkaMagicd23} we have
	\begin{align}
	\kappa(\scrT_\ns)=\frac{2^{k+1}(d-2)}{d^k}\geq \frac{2(d-2)}{D}.
	\end{align}
	In conjunction with \lref{lem:QPsiProj} and \thref{thm:StabShNormGen} we can deduce that
	\begin{align}
	\|\Ob_0\|^2_{\orb(\Psi)}=\frac{D+1}{D+2}\upsilon_1=\frac{(D+1)(D-K)(K+1)}{2(D-1)(D+2)}\hka(\Psi,\Sigma(d))-\frac{(D+1)(D-K)(2DK+D-K)}{D^2(D-1)},
	\end{align}
	where $\upsilon_1$ is defined in \lref{lem:QPsiProj}. 	This equation and  \lref{lem:kahkaMagicd23} together imply the two equalities in \eref{eq:StabShNormMagicd23}. 
	In addition, it is straightforward to verify that $\|\Ob_0\|^2_{\orb(\Psi)}/\|\Ob_0\|_2^2$ decreases monotonically with $k$ and $K$ given that $\|\Ob_0\|_2^2=K(D-K)/D$. 
	
	The first inequality in \eref{eq:StabShNormMagicd23} is equivalent to the following inequality, 
	\begin{align}
	&\frac{2D^2-dD+(D^2-2dD+D+d)K}{D(D-1)}+\frac{2^k(d-2)(K+1)D}{d^k(D-1)}\leq K+2+\frac{2^k(d-2)(K+1)}{d^k}. 
	\end{align}
	To prove this inequality, it suffices to consider the case $k=0$, that is, the following inequality,
	\begin{align}
	&\frac{2D^2-dD+(D^2-2dD+D+d)K}{D(D-1)}+\frac{(d-2)(K+1)D}{(D-1)}\leq K+2+(d-2)(K+1),
	\end{align}
	which amounts to the inequality $d(D-1)K\geq 0$. The second inequality in \eref{eq:StabShNormMagicd23} is obvious given that $K\geq 1$.   This  observation completes the proof of \thref{thm:StabShNormMagicd23}.
\end{proof}

\begin{proof}[Proof of \thref{thm:ShNormMagicd23}]
	According to \lref{lem:kahkaMagicd23} we have 
	\begin{align}
	|\hka|(\scrT_\ns)\leq \kappa(\scrT_\ns)=\frac{2^{k+1}(d-2)}{d^k}. 
	\end{align}
	Now \eqsref{eq:OShNormMagicd23}{eq:ShNormMagicd23} in \thref{thm:ShNormMagicd23} follow from  \thref{thm:ShNormGen} (see also \thref{thm:ShNormBalance}). 
	If  $n\geq 2$, $k=0$, or $d=3$, then $\hka(\caT)>0$ for all $\caT\in \Sigma(d)$ thanks to \eref{eq:hkakaTLUB} and \lref{lem:kahkaMagicd23}. By virtue of \coref{cor:ShNormGenPos} and  \lref{lem:kahkaMagicd23} we can derive slightly better results than \eqsref{eq:OShNormMagicd23}{eq:ShNormMagicd23},
	\begin{gather}
	\frac{D+d}{D+1}	\|\Ob\|^2_{\orb(\Psi)}\leq	
	\left[1+\frac{2^{k+1}(d-2)}{d^k}\right]	\|\Ob\|_2^2+2\|\Ob\|^2\leq \left[3+\frac{2^{k+1}(d-2)}{d^k}\right]	\|\Ob\|_2^2,  \\
	\hka(\Sigma(d))-3-\frac{5}{D}\leq	\|\orb(\Psi)\|_\sh\leq 	\frac{D+1}{D+d}\left[3+\frac{2^{k+1}(d-2)}{d^k}\right]\leq 3+ \frac{2^{k+1}}{d^{k-1}}.
	\end{gather}
\end{proof}

\section{\label{app:MagicOrbitd1Proof}Proofs of results on magic orbits when $d= 1\mmod 3$}

In this appendix we prove \lsref{lem:kappaMagicOned1}-\ref{lem:BalMagicExactd1} and \thsref{thm:Phi3MagicUBd1}-\ref{thm:StabShNormMagicd1}, \ref{thm:AccurateDesignMagic}  presented in \sref{sec:MagicOrbitd1}, assuming that  $d\geq 7$ is an odd prime satisfying $d=1\mmod 3$, $n\in \bbN$, $k\in \bbN_0$, and $k\leq n$. To this end, some auxiliary results (\lsref{lem:CubicStateSum}-\ref{lem:hkaAuxMagicId}) are established in \aref{app:MagicOrbitd1Aux}.

\subsection{Proofs of \lsref{lem:kappaMagicOned1}-\ref{lem:kappaMagicLUBd1}}

\begin{proof}[Proof of \lref{lem:kappaMagicOned1}]
	By virtue  of \lref{lem:NfT} in \aref{app:CubicEqSLS} 
	we can deduce that (cf. the proof of \lref{lem:kappaMagicOned23})

	\begin{align}
	\kappa(f,\caT)&=\kappa(\psi_f,\caT)=\frac{1}{d^3}\sum_{\alpha\in \bbF_d} N_\alpha(f,\caT)\omega^{\alpha}=\frac{3d-2}{d^2}+
	\frac{1}{d^3}\sum_{\alpha \in \bbF_d^{\times}} \bigl\{d^2-2d+2\Re\bigl[\eta_3(ac/\alpha )G^3(\eta_3)\bigr]\bigr\}\omega^\alpha 	\nonumber\\	&=\frac{2d+2\Re\bigl[\eta_3(ac)G^2(\eta_3)\bigr]	}{d^2}
	=\frac{2d+2\Re\bigl[G^2(\eta_3,ac)\bigr]}{d^2}=\frac{g^2(3,ac)}{d^2}, 
	\end{align}	
	which confirms the equalities in \eref{eq:kappaMagicOned1}.
	Here  the last three  equalities can be derived as follows,
	\begin{align}
	&\frac{3d-2}{d^2}+
	\frac{d-2}{d^2}\sum_{\alpha \in \bbF_d^{\times}} \omega^{\alpha }=\frac{2d}{d^2}+
	\frac{d-2}{d^2}\sum_{\alpha \in \bbF_d} \omega^{\alpha }=\frac{2}{d}, \\
	&2\sum_{\alpha \in \bbF_d^{\times}} \omega^\alpha \Re\bigl[\eta_3(ac/\alpha )G^3(\eta_3)\bigr]=
	\sum_{\alpha \in \bbF_d^{\times}} \bigl[\eta_3(ac)\omega^\alpha  \eta_3^2(\alpha )G^3(\eta_3)+\eta_3^*(ac)\omega^\alpha  \eta_3(\alpha )G^{3*}(\eta_3)\bigr]\nonumber\\
	&=\eta_3(ac)G(\eta_3^2)G^3(\eta_3)+\eta_3^*(ac)G(\eta_3)G^{3*}(\eta_3)=2d\Re\bigl[\eta_3(ac)G^2(\eta_3)\bigr]\nonumber\\
	&=2d\Re\bigl[G^2(\eta_3,ac)\bigr]=d[g^2(3,ac)-2d].  \label{eq:kappafTproof2}
	\end{align}	
	The second equality in \eref{eq:kappafTproof2} follows from the definitions of Gauss sums in \aref{app:GaussJacobi}, and the last three equalities follow from \pref{pro:GaussSum}  and \eref{eq:gGaussSumCubic} in \aref{app:GaussJacobi}.

	The  inequalities in \eref{eq:kappaMagicOned1} follow from \lref{lem:GaussSum3LUB} in \aref{app:GaussJacobi}. 
\end{proof}

\begin{proof}[Proof of \lref{lem:kappaOrder}]
	The first two equalities in \eref{eq:kappaSumProd} follow from \lref{lem:kappaMagicOned1} and the fact that $|G^2(\eta_3)|=d$; the third equality  follows from \lref{lem:kappaMagicOned1} and the definition of $\tg(d)$ in \eref{eq:tgd}; 
	the two inequalities  follow from \lref{lem:GaussSum3LUB} in \aref{app:GaussJacobi}. \Lsref{lem:kappaMagicOned1}
	and \ref{lem:GaussSum3LUB}	also imply the first six inequalities in \eref{eq:kappaOrder}. In conjunction with \eref{eq:kappaSumProd} we can deduce that $d\kappa_0^\uparrow+d\kappa_1^\uparrow>2$, which implies the last inequality in \eref{eq:kappaOrder} and completes the proof of \lref{lem:kappaOrder}.
\end{proof}

\begin{proof}[Proof of \lref{lem:kappaMagicLUBd1}]
	\Eref{eq:kappaMagicLUBd1} is a simple corollary of \lref{lem:kappaMagicOned1}, \eref{eq:kappaMagicDecomd1}, and \pref{pro:kappaSym}. 
	\Lref{lem:kappaMagicOned1} also implies that $\kappa(\caT')=\kappa(\caT)$ whenever $\caT,\caT'\in \scrT_\ns$ have the same cubic character.

	The equality in \eref{eq:kaNsSigMagicd1} follows from \eref{eq:kappaSigNSscrT}, and 
	the two inequalities  follow from \eref{eq:kappaMagicLUBd1} and the fact that $|\scrT_\ns|=2(d-2)$.

	The first inequality in \eref{eq:hkaAbsNsSigMagicd1} follows from \eqsref{eq:hkakaSig}{eq:kaNsSigMagicd1}, and the second inequality is trivial, so it remains to prove the last two inequalities. 
	If $k=0$, then $|\Psi\>$ is a stabilizer state, $|\hka|(\Sigma(d))=\hka(\Sigma(d))\leq 2(d+1)$, and 
	$|\hka|(\scrT_\ns)=\hka(\scrT_\ns)\leq 2(d-2)$,  so the  last two inequalities in \eref{eq:hkaAbsNsSigMagicd1} hold.  If $k\geq1$, then $|\kappa|(\Sigma(d))= \kappa(\Sigma(d))\leq 14$ by \eref{eq:kaNsSigMagicd1}, which implies  the third inequality in \eref{eq:hkaAbsNsSigMagicd1} thanks to
	\eref{eq:hkaSigAbsLUB}. In addition, by virtue of \pref{pro:hkaka} we can deduce that	
	\begin{align}
	|\hka|(\scrT_\ns)&=\frac{D+2}{D-1}
	\sum_{\caT\in\scrT_\ns}\left|\kappa(\caT) -\frac{1}{2D+2d}\kappa(\Sigma(d))\right|\nonumber\leq \frac{D+2}{D-1}\kappa(\scrT_\ns)+\frac{(d-2)(D+2)}{(D-1)(D+d)}\kappa(\Sigma(d))\nonumber\\
	&\leq \frac{2^{2k+1}(d-2)(D+2)}{d^k(D-1)} +\frac{14(d-2)(D+2)}{(D-1)(D+d)}\leq \frac{2^{2k+1}}{d^{k-1}}+\frac{14d}{D},
	\end{align}
	which implies the last inequality in \eref{eq:hkaAbsNsSigMagicd1} and completes the proof of \lref{lem:kappaMagicLUBd1}. 
\end{proof}

\subsection{\label{app:MagicOrbitd1Aux}Auxiliary results on magic orbits  when $d= 1\mmod 3$}

Here we establish some auxiliary results (\lsref{lem:CubicStateSum}-\ref{lem:hkaAuxMagicId}) on magic orbits in the case  $d= 1\mmod 3$. Most of these auxiliary results are tied to the proofs of \lsref{lem:hkaAuxMagicd1} and \ref{lem:kahkaNsSigMagicSLUBd1}.

\subsubsection{Auxiliary results \lsref{lem:CubicStateSum}-\ref{lem:hkaAuxMagicId}}

The following lemma complements \pref{pro:kahkaMagicd1} in the main text.
\begin{lemma}\label{lem:CubicStateSum}
	Suppose $d$ is an odd prime, $d=1\mmod 3$, and $f\in \tscrP_3(\bbF_d)$ has cubic coefficient $c$. Then 	
	\begin{align} 
	\kappa(f,\scrT_\ns)&=\frac{4(d-2)}{d}+\frac{4}{d^3}
	\Re\bigl[\eta_3(3c)G^5(\eta_3)\bigr] \label{eq:kappafTsum},\\
	\kappa(f,\scrT_\ns,2)&=\frac{12(d-2)}{d^2}+\frac{16}{d^4}
	\Re\bigl[\eta_3(3c)G^5(\eta_3)\bigr]+\frac{4}{d^2}
	\Re\bigl[\eta_3^2(3c)G(\eta_3)\bigr]
	\label{eq:kappaf2Tsum}, \\
	\kappa(f,\scrT_\ns,3)&=\frac{40(d-2)}{d^3}+\frac{60}{d^5}
	\Re\bigl[\eta_3(3c)G^5(\eta_3)\bigr]+\frac{24}{d^3}
	\Re\bigl[\eta_3^2(3c)G(\eta_3)\bigr]+\frac{4(d-2)}{d^6}\Re[G^6(\eta_3)]
	\label{eq:kappaf3Tsum},\\
	\frac{9}{5}\leq \kappa(f,\scrT_\ns)&\leq \frac{9}{2}, \qquad 
	\frac{5}{d}\leq \kappa(f,\scrT_\ns,2)\leq \min\left\{2,\frac{15}{d}\right\}, \qquad 
	\frac{18}{d^2}\leq \kappa(f,\scrT_\ns,3)\leq \min\left\{1, \frac{52}{d^2}\right\}.  \label{eq:kappafTsumLUB1}
	\end{align}
	If in addition $d\geq 100$, then 
	\begin{equation}\label{eq:kappafTsumLUB2}
	\frac{7}{2}\leq \kappa(f,\scrT_\ns)\leq \frac{22}{5},\qquad 
	\frac{39}{4d}\leq \kappa(f,\scrT_\ns,2)\leq \frac{14}{d},\qquad 
	\frac{26}{d^2}\leq \kappa(f,\scrT_\ns,3)\leq \frac{52}{d^2}.
	\end{equation}
\end{lemma}

The upper bounds in  \eqsref{eq:kappafTsumLUB1}{eq:kappafTsumLUB2} are still applicable if  $\kappa(f,\scrT_\ns,m)$ is replaced by $\kappa(\Psi,\scrT_\ns)$ with $|\Psi\>\in \scrM_{n,k}^\id(d)$ and $k\geq m$  thanks to \pref{pro:kappaSym} and \lref{lem:kappaMagicOned1}, which means
\begin{align}\label{eq:kappaPsiProd2}
\kappa(\Psi,\caT)\leq \min\{\kappa(\Psi_1,\caT),\; \kappa(\Psi_2,\caT)\}\;\; \forall  \caT\in \Sigma(d),\quad   \kappa(\Psi,\scrT)\leq \min\{\kappa(\Psi_1,\scrT),\; \kappa(\Psi_2,\scrT)\}\;\; \forall \scrT\subseteq \Sigma(d)
\end{align}
whenever $|\Psi\>=|\Psi_1\>\otimes |\Psi_2\>$.

\begin{lemma}\label{lem:CubicState12Sum}
	Suppose $d$ is an odd prime, $d=1\mmod 3$,  $f_1, f_2\in \tscrP_3(\bbF_d)$ have cubic coefficients $c_1$ and $c_2$, respectively, and $|\Psi\>\in \scrM_{n,2}(d)$ has the form $|\Psi\>=|\psi_{f_1}\>\otimes|\psi_{f_2}\>\otimes |0\>^{\otimes (n-2)}$. Then 
	\begin{gather} 
	\kappa(\Psi,\scrT_\ns)=\frac{6(d-2)}{d^2}+\frac{8}{d^4}\left\{
	\Re\bigl[\eta_3(3c_1)G^5(\eta_3)\bigr]+\Re\bigl[\eta_3(3c_2)G^5(\eta_3)\bigr]\right\}+\frac{4}{d^2}
	\Re\bigl[\eta_3(9c_1c_2)G(\eta_3)\bigr]
	\label{eq:kappaf12Tsum}. 
	\end{gather}
	If in addition $|\Psi\>\in \scrM_{n,2}(d)\setminus \scrM_{n,2}^\id(d)$, which means $\eta_3(c_1)\neq \eta_3(c_2)$,	
	then 
	\begin{gather}
	\frac{6(d-2)}{d^2}-\frac{12\sqrt{d}}{d^2}\leq 
	\kappa(\Psi,\scrT_\ns)\leq 
	\frac{6(d-2)}{d^2}+\frac{12\sqrt{d}}{d^2}
	,\quad 
	\kappa(\Psi,\scrT_\ns)\leq 
	\frac{23}{3d}.  \label{eq:kappaf12TsumLUB}
	\end{gather}
\end{lemma}

\begin{lemma}\label{lem:kappaSigUBmagic}
	Suppose $d$ is an odd prime, $d=1\mmod 3$, and $|\Psi\>\in \scrM_{n,k}(d)$.  Then 
	\begin{align}\label{eq:kappaSigUBmagic}
	\frac{d^k\kappa(\Psi,\Sigma(d))}{d^k+d}\leq \begin{cases}
	2  & k=0, \\
	21/4  & k=1, \\
	80/13 & k\geq 2 \mbox{ and }  d\geq 59,\\
	7 & \mbox{otherwise}.
	\end{cases} 
	\end{align}
\end{lemma}

\begin{lemma}\label{lem:kappaNsSigLBd1}
	Suppose $d$ is an odd prime, $d=1\mmod 3$, and $|\Psi\>\in \scrM^\id_{n,k}(d)$; let $j=3$ if $n=k=1$ and $d=7$ and let $j=2$ otherwise. Then  
	\begin{align}\label{eq:kappaNsSigLBAux}
	\kappa^2(\Psi,\Sigma(d))\geq \frac{6(D+d)}{(D+j)}\kappa(\Psi,\Sigma(d),2),\quad 	d^k\kappa(\Psi,\scrT_\ns)\geq \begin{cases}
	2(d-2) & k=0,1,\\
	3\times 2^{k-1}(d-2) &k\geq 2.
	\end{cases}
	\end{align}	 
\end{lemma}

Next, we consider $\hka(\Psi,\caT)$ with $|\Psi\>\in \scrM_{n,k}(d)$ and related functions, which may depend on $n$ and $D$ in addition to $d$ and $k$. 
\begin{lemma}\label{lem:hkaTLUBd1}
	Suppose $d$ is an odd prime, $d=1\mmod 3$, and	$|\Psi\>\in \scrM_{n,k}(d)$. Then
	\begin{gather}
	\hka(\Psi,\caT)\geq \frac{D+2}{D-1} \left[\frac{1}{9^kd^{2k}}-\frac{\kappa(\Psi,\Sigma(d))}{2(D+d)}\right]\quad \forall \caT\in \Sigma(d). \label{eq:hkaTLBd1}
	\end{gather}
	If in addition $n\geq [2+(\log_9d)^{-1}]k+(\log_6d)^{-1}$, then $0\leq \hka(\Psi,\caT)\leq \kappa(\Psi,\caT)$ for all $\caT\in \Sigma(d)$ and
	\begin{align}\label{eq:hkad1}
	|\hka|(\Psi,\Sigma(d))= \hka(\Psi,\Sigma(d))=\frac{D+2}{D+d}\kappa(\Psi,\Sigma(d)), \quad |\hka|(\Psi,\scrT_\ns)=\hka(\Psi,\scrT_\ns)\leq \frac{D+2}{D+d}\kappa(\scrT_\ns).
	\end{align} 
\end{lemma}

\begin{lemma}\label{lem:hkaAuxAux}
	Suppose $d$ is an odd prime, $d=1\mmod 3$, and $|\Psi\>\in \scrM_{n,k}^\id(d)$. Then 
	\begin{gather}\label{eq:hkaNsAbsUBaux}  \frac{D-1}{D+2}|\hka|(\Psi,\scrT_\ns)\leq \kappa(\Psi,\scrT_\ns) -\frac{h_d\kappa(\Psi,\Sigma(d))}{D+d},\quad \frac{D-1}{D+2}|\hka|(\Psi,\Sigma(d))\leq \kappa(\Psi,\Sigma(d)) -\frac{(h_d+3)\kappa(\Psi,\Sigma(d))}{D+d},	
	\end{gather}	
	where
	\begin{gather}\label{eq:cd-hkaAuxAux}
	h_d=\begin{cases}
	-1 & d=7, 13,\\
	1 & d=19, \\
	2 & d\geq 37,\\
	3 & d= 31  \mbox{ or }  d\geq 43.  
	\end{cases}
	\end{gather}
\end{lemma}

\begin{lemma}\label{lem:hkaAuxMagicId}
	Suppose $d$ is an odd prime, $d=1\mmod 3$, and $|\Psi\>\in \scrM_{n,k}^\id(d)$. Then 
	\begin{align}\label{eq:hkaAbsd1UB1}
	|\hka|(\Psi,\Sigma(d))\leq \kappa(\Psi,\Sigma(d)) +\frac{16}{D},\quad |\hka|(\Psi,\scrT_\ns)\leq \kappa(\Psi,\scrT_\ns) +\frac{16}{D}.
	\end{align}
	If $d\geq 31$ or  if $d\geq 19$ and $n\geq 2$, then 
	\begin{align}
	|\hka|(\Psi,\Sigma(d))&\leq \kappa(\Psi,\Sigma(d)),\quad |\hka|(\Psi,\scrT_\ns)\leq \kappa(\Psi,\scrT_\ns). \label{eq:hkaAbsd1UB2}
	\end{align}
\end{lemma}

\subsubsection{Proofs of \lsref{lem:CubicStateSum} and \ref{lem:CubicState12Sum}}

\begin{proof}[Proof of \lref{lem:CubicStateSum}]
	Thanks to \pref{pro:kahkaMagicd1} we have
	\begin{align}
	\kappa(f,\scrT_\ns,m)=2\sum_{j=0}^2
	\frac{\left\{d(d-2)+2\Re\bigl(\alpha\omega_3^{2j}\bigr)\right\} \left\{2d+2\Re\bigl(\beta \omega_3^j\bigr)\right\}^m}{3d^{2m+1}}\quad \forall m\in \bbN,
	\end{align}
	where $\alpha=G^3(\eta_3)$ and $\beta=\eta_3(3c)G^2(\eta_3)$.
	Note that $|\eta_3(3c)|=1$ and $|G(\eta_3)|=\sqrt{d}$ by \pref{pro:GaussSum} in \aref{app:GaussJacobi}, which implies the following results,
	\begin{align}
	|\alpha|=d^{3/2},\quad |\beta|=d, \quad \alpha\beta=\eta_3(3c)G^5(\eta_3),\quad  \alpha^*\beta^2=d^3\eta_3^2(3c)G(\eta_3),\quad \beta^3=G^6(\eta^3).
	\end{align}
	In addition, simple calculation yields
	\begin{equation}
	\begin{aligned}
	\sum_{j=0}^2 \bigl[2\Re\bigl(\beta\omega_3^j\bigr)\bigr]^m&=3\sum_{l\in \scrA_m}\binom{m}{l}\beta^l\beta^{*(m-l)}=\begin{cases}
	0 & m=1,\\
	6|\beta|^2 & m=2,\\
	3\beta^3+3\beta^{*3} & m=3,
	\end{cases}\\
	\sum_{j=0}^2 \left\{2\Re\bigl(\alpha\omega_3^{2j}\bigr)\bigl[2\Re\bigl(\beta\omega_3^j\bigr)\bigr]^m\right\}&=6\!\sum_{l\in \scrA_m'}\!\!\binom{m}{l}\Re\bigl(\alpha\beta^l\beta^{*(m-l)}\bigl)
	=\begin{cases}
	6\Re(\alpha\beta) & m=1,\\
	6\Re(\alpha^*\beta^2) & m=2,\\
	18|\beta|^2\Re(\alpha\beta)& m=3,
	\end{cases}
	\end{aligned}
	\end{equation}
	where $\scrA_m=\{l: 0\leq l\leq m,  3|(2l-m)\}$ and $\scrA_m'=\{l: 0\leq l\leq m,  3|(2l-m+2)\}$. 	
	Now Eqs.~\eqref{eq:kappafTsum}-\eqref{eq:kappaf3Tsum} can be derived by combining the above results,
	\begin{align}
	\kappa(f,\scrT_\ns)&=\frac{4d^2(d-2)+4\Re(\alpha\beta)}{d^3}=\frac{4(d-2)}{d}+\frac{4}{d^3}
	\Re\bigl[\eta_3(3c)G^5(\eta_3)\bigr],\\
	\kappa(f,\scrT_\ns,2)&=\frac{12d^3(d-2)+16d\Re(\alpha\beta)+4\Re(\alpha^*\beta^2)}{d^5}\nonumber\\
	&=\frac{12(d-2)}{d^2}+\frac{16}{d^4}
	\Re\bigl[\eta_3(3c)G^5(\eta_3)\bigr]+\frac{4}{d^2}
	\Re\bigl[\eta_3^2(3c)G(\eta_3)\bigr],\\
	\kappa(f,\scrT_\ns,3)&=\frac{8d^2(d-2)(2d^2+3|\beta|^2) +12(4d^2+|\beta|^2)\Re(\alpha\beta)+24d\Re(\alpha^*\beta^2) +4d(d-2)\Re(\beta^3)}{d^7}\nonumber\\
	&=\frac{40(d-2)}{d^3}+\frac{60}{d^5}
	\Re\bigl[\eta_3(3c)G^5(\eta_3)\bigr]+\frac{24}{d^3}
	\Re\bigl[\eta_3^2(3c)G(\eta_3)\bigr]+\frac{4(d-2)}{d^6}\Re\bigl[G^6(\eta_3)\bigr].
	\end{align}

	Next, by virtue of Eqs.~\eqref{eq:kappafTsum}-\eqref{eq:kappaf3Tsum} 
	and the facts  that $|\eta_3(3c)|=1$ and $|G(\eta_3)|=\sqrt{d}$ 	we can derive the following inequalities,
	\begin{equation}
	\begin{gathered} 
	\frac{4(d-2)}{d}-\frac{4\sqrt{d}}{d}
	\leq \kappa(f,\scrT_\ns)\leq \frac{4(d-2)}{d}+\frac{4\sqrt{d}}{d},\\
	\frac{12(d-2)}{d^2}+\frac{20\sqrt{d}}{d^2}\leq 
	\kappa(f,\scrT_\ns,2)\leq 
	\frac{12(d-2)}{d^2}+\frac{20\sqrt{d}}{d^2},\\
	\frac{36(d-2)}{d^3}-\frac{84\sqrt{d}}{d^3}\leq 
	\kappa(f,\scrT_\ns,3)\leq \frac{44(d-2)}{d^3}+\frac{84\sqrt{d}}{d^3}. 
	\end{gathered}
	\end{equation}
	If in addition $d\geq 100$, then these inequalities imply \eref{eq:kappafTsumLUB2}, which in turn implies \eref{eq:kappafTsumLUB1}. If $d<100$, then \eref{eq:kappafTsumLUB1} can be verified by direct calculation. 
\end{proof}

\begin{proof}[Alternative proof of Eqs.~\eqref{eq:kappafTsum}-\eqref{eq:kappaf3Tsum} in \lref{lem:CubicStateSum}]
	By virtue of  \pref{pro:kappaSym}, \eref{eq:TvNS}, and \lref{lem:kappaMagicOned1} we can deduce that
	\begin{align}
	\kappa(f,\scrT_\ns,m)&=2\kappa(f,\scrT_1\cap \scrT_\ns,m)
	=2\sum_{y=1}^{d-2}\kappa^m(f,\caT_{\bfv_y})
	=2\sum_{y=1}^{d-2}
	\Bigl\{\frac{2}{d}+\frac{2}{d^2}\Re\bigl[\eta_3\bigl(-3c\bigl(y+y^2\bigr)\bigr)G^2(\eta_3)\bigr]\Bigr\}^m\nonumber\\
	&=2\sum_{y=1}^{d-2}
	\Bigl\{\frac{2}{d}+\frac{2}{d^2}\Re\bigl[\eta_3(3c)\eta_3\bigl(y+y^2\bigr)G^2(\eta_3)\bigr]
	\Bigr\}^m \quad \forall m\in \bbN.  \label{eq:kappafTsumProof}
	\end{align}
	In addition, thanks to \lref{lem:eta3} in \aref{app:GaussJacobi} and the  identities $G^*(\eta_3)=d/G(\eta_3)$ and  $J(\eta_3, \eta_3)=G^3(\eta_3)/d$ by \pref{pro:GaussSum} and \eref{eq:JacobiGaussCubic},
	we have
	\begin{align}
	\begin{aligned}
	\sum_{y=1}^{d-2}\eta_3\bigl(y+y^2\bigr)G^j(\eta_3)&=\sum_{y\in \bbF_d}\eta_3\bigl(y+y^2\bigr)G^j(\eta_3)=J(\eta_3, \eta_3)G^j(\eta_3)=\frac{G^{j+3}(\eta_3)}{d},\\
	\sum_{y=1}^{d-2}\eta_3^2\bigl(y+y^2\bigr)G^j(\eta_3)&=J^*(\eta_3, \eta_3)G^j(\eta_3)=d^2G^{j-3}(\eta_3).
	\end{aligned}
	\end{align}	
	The above equations together imply Eqs.~\eqref{eq:kappafTsum}-\eqref{eq:kappaf3Tsum} given that $\eta_3$ is a cubic character and $|G(\eta_3)|=\sqrt{d}$.

	Alternatively, \eref{eq:kappafTsum} can  be proved by virtue of  \lref{lem:NfTsum} in \aref{app:CubicEqSLS} as follows  (cf. the proofs of \lsref{lem:kappaMagicOned23} and \ref{lem:kappaMagicOned1}),
	\begin{align}
	\kappa(f,\scrT_\ns)&=\sum_{\caT\in \scrT_\ns}\kappa(f,\caT)=
	\frac{1}{d^3}\sum_{\alpha\in \bbF_d}\omega^{\alpha}\sum_{\caT\in \scrT_\ns} N_\alpha(f,\caT)\nonumber\\
	&=\frac{2(d-2)(3d-2)}{d^2}+
	\frac{1}{d^2}\sum_{\alpha\in \bbF_d^{\times}} \bigl\{2(d-2)^2+4\Re\bigl[\eta_3(3c/\alpha)J^2(\eta_3,\eta_3)\bigr]\bigr\}\omega^\alpha&\nonumber\\
	&=\frac{4(d-2)}{d}+\frac{4}{d^3}\Re\bigl[\eta_3(3c)G^5(\eta_3)\bigr].
	\end{align}	 
	Here the last equality holds because [cf. \eref{eq:kappafTproof2}]
	\begin{align}
	\sum_{\alpha\in \bbF_d^{\times}}\Re\bigl[\eta_3(3c/\alpha)J^2(\eta_3,\eta_3)\bigr]
	&=\Re\bigl[\eta_3(3c)G(\eta_3^2)J^2(\eta_3,\eta_3)\bigr]=\frac{1}{d}\Re\bigl[\eta_3(3c)G^5(\eta_3)\bigr]. 
	\end{align}	
\end{proof}

\begin{proof}[Proof of \lref{lem:CubicState12Sum}]
	Thanks to  \pref{pro:kappaSym} and \eref{eq:TvNS}, we can deduce the following result	[cf. \eref{eq:kappafTsumProof}],
	\begin{align}
	\kappa(\Psi,\scrT_\ns)
	&=2	\kappa(\Psi,\scrT_1\cap \scrT_\ns)=
	2\sum_{y=1}^{d-2}\kappa(\Psi,\caT_{\bfv_y})=2\sum_{y=1}^{d-2}\kappa(f_1,\caT_{\bfv_y})\kappa(f_2,\caT_{\bfv_y}),
	\end{align}	
	where $\bfv_y$ is specified in \eref{eq:vy}, and  $\caT_{\bfv_y}$ is the stochastic Lagrangian subspace determined by the characteristic vector $\bfv_y$ as defined in  \eref{eq:Tv}. 
	In conjunction with \lref{lem:kappaMagicOned1} 	we can deduce that
	\begin{align}	
	\kappa(\Psi,\scrT_\ns)&=2\sum_{y=1}^{d-2}
	\Bigl\{\frac{2}{d}+\frac{2}{d^2}\Re\bigl[\eta_3(3c_1)\eta_3\bigl(y+y^2\bigr)G^2(\eta_3)\bigr]
	\Bigr\} \Bigl\{\frac{2}{d}+\frac{2}{d^2}\Re\bigl[\eta_3(3c_2)\eta_3\bigl(y+y^2\bigr)G^2(\eta_3)\bigr]
	\Bigr\}\nonumber  \\
	&=\frac{6(d-2)}{d^2}+\frac{8}{d^4}\left\{
	\Re\bigl[\eta_3(3c_1)G^5(\eta_3)\bigr]+\Re\bigl[\eta_3(3c_2)G^5(\eta_3)\bigr]\right\}+\frac{4}{d^2}
	\Re\bigl[\eta_3(9c_1c_2)G(\eta_3)\bigr],
	\end{align}	
	which confirms \eref{eq:kappaf12Tsum}. 
	
	Next, suppose $\eta_3(c_1)\neq \eta_3(c_2)$, then the first two inequalities in 	\eref{eq:kappaf12TsumLUB} are a simple corollary of \eref{eq:kappaf12Tsum} given that $\eta_3$ is a cubic character and $|G(\eta_3)|=\sqrt{d}$. If in addition $d\geq 100$, then the third inequality follows from the second  inequality;  if instead $d< 100$, then this inequality can be verified by direct calculation. 	
\end{proof}

\subsubsection{Proofs of \lsref{lem:kappaSigUBmagic} and \ref{lem:kappaNsSigLBd1}}

\begin{proof}[Proof of \lref{lem:kappaSigUBmagic}]
	If $k=0$, then $|\Psi\>$ is a stabilizer state and $\kappa(\Sigma(d))=2(d+1)$, so \eref{eq:kappaSigUBmagic} holds. If $k=1$, then
	\eref{eq:kappaSigUBmagic} follows from the fact that $\kappa(\Sigma(d))=6+\kappa(\scrT_\ns)\leq 21/2$ by  \lref{lem:CubicStateSum}.

	If  $k= 2$, then $\kappa(\Sigma(d))=6+\kappa(\scrT_\ns)\leq 6+15/d$  by \lsref{lem:CubicStateSum} and \ref{lem:CubicState12Sum}, which implies \eref{eq:kappaSigUBmagic}  when $d\geq 8$. When $k=2$ and $d=7$, \eref{eq:kappaSigUBmagic} can be verified by direct calculation.

	If $k\geq 3$, then  $\kappa(\Sigma(d))=6+\kappa(\scrT_\ns)\leq 6+52\times 4^{k-3}/d^{k-1}$ by \lsref{lem:CubicStateSum} and \ref{lem:CubicState12Sum}, which implies \eref{eq:kappaSigUBmagic}.
\end{proof}

\begin{proof}[Proof of \lref{lem:kappaNsSigLBd1}]
	If $k=0$, then $|\Psi\>$ is a stabilizer state, which means	$\kappa(\Psi,\Sigma(d),2)=\kappa(\Psi,\Sigma(d))=2(d+1)$ and $\kappa(\Psi,\scrT_\ns)=2(d-2)$, so \eref{eq:kappaNsSigLBAux} holds. In the rest of this proof we assume that $k\geq 1$.

	Without loss of generality we can assume that $|\Psi\>$ has the form $|\Psi\>=|\psi_f\>^{\otimes k}\otimes |0\>^{\otimes (n-k)}$, where $f\in \tscrP_3(\bbF_d)$; then $\kappa(\Psi,\caT)=\kappa^k(f, \caT)$. 
	If in addition $k\geq 3$, then  
	\begin{align}
	\!\!\kappa(\Psi,\scrT_\ns)=\sum_{\caT\in \scrT_\ns} \kappa^k(f,\caT)\geq 2(d-2)\left[\frac{\sum_{\caT\in \scrT_\ns} \kappa^3(f,\caT)}{2(d-2)}\right]^{k/3}\geq 2(d-2)\left[\frac{18}{2d^2(d-2)}\right]^{k/3}\geq \frac{2\times 9^{k/3}}{d^{k-1}}, 
	\end{align}
	which implies the second inequality in \eref{eq:kappaNsSigLBAux}. Here the first two inequalities hold because $|\scrT_\ns|=2(d-2)$ and $\kappa(f,\scrT_\ns, 3)\geq 18/d^2$ by \lref{lem:CubicStateSum}. If $k=1,2$ and $d\geq 100$, then the second inequality in \eref{eq:kappaNsSigLBAux} holds because $\kappa(f,\scrT_\ns)\geq 7/2$ when $k=1$ and
	$\kappa(f,\scrT_\ns, 2)\geq 39/(4d)$ when $k=2$ by \lref{lem:CubicStateSum}.
	If $k=1,2$ and $d< 100$, then this inequality can be verified by direct calculation.

	If $k\geq 2$ or $n\geq k+1$, then the first  inequality in  \eref{eq:kappaNsSigLBAux}
	holds because $\kappa(\Psi,\Sigma(d))\geq \kappa(\Psi,\Sigma(d),2)$ and
	$\kappa(\Psi,\Sigma(d))\geq 6(D+d)/(D+2)$ thanks to the second inequality proved above.   If $n=k=1$ and $d\geq 100$, then the first inequality in  \eref{eq:kappaNsSigLBAux} can be proved  by virtue of \lref{lem:CubicStateSum} as follows,
	\begin{align}
	\frac{6(D+d)}{D+2}\kappa(\Psi,\Sigma(d),2)\leq \frac{12d}{d+2}\kappa(f,\Sigma(d),2)\leq \frac{12d}{d+2}\left(6+\frac{15}{d}\right)\leq 74\leq \kappa^2(\Psi,\Sigma(d)). 
	\end{align} 
	If $n=k=1$ and $7\leq d< 100$, then the first  inequality in  \eref{eq:kappaNsSigLBAux} holds by direct calculation.
\end{proof}

\subsubsection{Proofs of \lsref{lem:hkaTLUBd1}-\ref{lem:hkaAuxMagicId}}

\begin{proof}[Proof of \lref{lem:hkaTLUBd1}]
	\Eref{eq:hkaTLBd1} is a simple corollary of \pref{pro:hkaka} and \eref{eq:kappaMagicLUBd1}. 
	
	Next, suppose  $n\geq [2+(\log_9d)^{-1}]k+(\log_6d)^{-1}$. 
	If $k=0$, then $|\Psi\>$ is a stabilizer state, which means  $\kappa(\caT)=1$ and $0\leq \hka(\caT)=(D+2)/(D+d)\leq \kappa(\caT)$ for all $\caT\in \Sigma(d)$, so \eref{eq:hkad1} holds. 
	If $k\geq 1$, then $\kappa(\Sigma(d))\leq 11$ by \lref{lem:CubicStateSum}, which implies that
	\begin{align}
	\frac{9^kd^{2k}\kappa(\Sigma(d))}{2(D+d)}\leq \frac{6\times 9^kd^{2k}}{d^n}\leq  1.
	\end{align}
	Therefore, $\hka(\caT)\geq 0$ for all $\caT\in \Sigma(d)$ thanks to \eref{eq:hkaTLBd1}. In conjunction with  \lref{lem:kahakRelation} we can further deduce that $0\leq \hka(\caT)\leq \kappa(\caT)$ for all $\caT\in \Sigma(d)$, which implies \eref{eq:hkad1} thanks to \lref{lem:hkaLUB} in \aref{app:ThirdMomentAux}. 
\end{proof}

\begin{proof}[Proof of \lref{lem:hkaAuxAux}]
	According to \psref{pro:kappaSym} and \ref{pro:hkaka}  we have
	\begin{align}
	\kappa(\Psi,\Sigma(d))=6+\kappa(\Psi,\scrT_\ns),\quad  |\hka|(\Psi,\Sigma(d))=\frac{D+2}{D-1}\left[6-\frac{3\kappa(\Psi,\Sigma(d))}{D+d}\right] +|\hka|(\Psi,\scrT_\ns).
	\end{align}
	So the two inequalities in \eref{eq:hkaNsAbsUBaux} are equivalent, and we can focus on the first inequality henceforth. If $k=0$, then $\kappa(\Psi,\caT)=1$ and $\hka(\Psi,\caT)=(D+2)/(D+d)$ for all $\caT\in \Sigma(d)$, which  implies \eref{eq:hkaNsAbsUBaux}  given that $|\scrT_\ns|=2(d-2)$ and $\Sigma(d)=2(d+1)$.

	In the rest of this proof we assume that $n\geq k\geq 1$. Without loss of generality, we can further assume that  $|\Psi\>$ has the form $|\Psi\>=|\psi_f\>^{\otimes k}\otimes |0\>^{\otimes (n-k)}$ with $f\in \tscrP_3(\bbF_d)$. 	
	Let $x=\kappa(\Psi,\Sigma(d))/(2D+2d)$. Let $\mu_j$ and $\kappa_j$ be shorthands for  $\mu_j(d)$ and $\kappa_j(f,d)$ for $j=0,1,2$ as defined in \eqsref{eq:muj}{eq:kappaMagid1j}; let $(\kappa_0^\uparrow, \kappa_1^\uparrow, \kappa_2^\uparrow)$ be the vector obtained from  $(\kappa_0, \kappa_1, \kappa_2)$ by arranging its entries in nondecreasing order. 
	Then 
	\begin{align}
	\hka(\Psi,\scrT_\ns)
	&=\sum_{j=0}^2 \mu_j \kappa_j^k,\quad 
	\frac{D-1}{D+2}|\hka|(\Psi,\scrT_\ns)
	=\sum_{j=0}^2 \mu_j \left|\kappa_j^k -x\right|, \quad 
	d\kappa_1^\uparrow\geq 1,
	\quad d^kx\leq \frac{7(d^k+d)}{2D+2d}\leq \frac{7}{2}, 
	\end{align}
	where the first inequality follows from  \lref{lem:kappaOrder}, and the last two inequalities follow from \lref{lem:kappaSigUBmagic}.

	If in addition $n\geq k+1$, then $(d\kappa_1^\uparrow\bigr)^k\geq 1\geq 7/d \geq  d^kx$, and \eref{eq:hkaNsAbsUBaux}  can be proved as follows,
	\begin{align}
	\frac{D-1}{D+2}|\hka|(\Psi,\scrT_\ns)
	& =\sum_{j=0}^2 \mu_j \left|\kappa_j^k -x\right|\leq \sum_{j=0}^2 \mu_j \kappa_j^k-(\mu_0+\mu_1+\mu_2-2\mu_{\max})x=\kappa(\Psi,\scrT_\ns)-2(d-2-\mu_{\max})x \nonumber\\
	&\leq \kappa(\Psi,\scrT_\ns) -\frac{h_d\kappa(\Psi,\Sigma(d))}{D+d},
	\label{eq:kappaNsProof}
	\end{align}
	where the last  equality and last inequality follow from  \lref{lem:muj}.

	Next, suppose $n=k\geq 1$. If $7\leq d\leq 200$ and $1\leq k\leq 5$, then \eref{eq:hkaNsAbsUBaux} can be verified by direct calculation. If  $7\leq d\leq 200$ and $k\geq 6$, then direct calculation shows that $\bigl(d\kappa_1^\uparrow\bigr)^k\geq \bigl(d\kappa_1^\uparrow\bigr)^6 \geq 7/2\geq d^k x$. Therefore,
	\eref{eq:kappaNsProof} still holds,  and \eref{eq:hkaNsAbsUBaux} holds accordingly.

	It remains to consider the case with $n=k\geq 1$ and $d> 200$. If 	$(\kappa_1^\uparrow\bigr)^k\geq  x$,  then \eqsref{eq:kappaNsProof}{eq:hkaNsAbsUBaux} hold as before. Otherwise,
	let $\{i_0,i_1, i_2\}$ be a permutation of $\{0,1,2\}$ such that $\kappa_{i_0}\geq \kappa_{i_1}\geq \kappa_{i_2}$; let $\tka_j=\kappa_{i_j}$ and $\tmu_j =\mu_{i_j}$ for $j=0,1,2$.
	Then 
	\begin{align}
	\frac{D-1}{D+2}|\hka|(\Psi,\scrT_\ns)
	&=\sum_{j=0}^2 \mu_j \left|\kappa_j^k -x\right|=
	\tmu_0(\tka_0^k-x)+\tmu_1(x-\tka_1^k)+\tmu_2(x-\tka_2^k)
	\nonumber\\
	&=\tmu_0 \tka_0^k+\tmu_1 \tka_1^k+\tmu_2 \tka_2^k +(\tmu_1+\tmu_2-\tmu_0)x-2\tmu_1\tka_1^k-2\tmu_2 \tka_2^k
	\nonumber\\
	&=\kappa(\Psi,\scrT_\ns)+2\left(d-2-\tmu_0-\frac{\tmu_1\tka_1^k+\tmu_2 \tka_2^k}{x}\right)x,
	\end{align}
	where the last equality holds because $\tmu_0+\tmu_1+\tmu_2=\mu_0+\mu_1+\mu_2=2(d-2)$ by \lref{lem:muj}. In addition,
	\begin{gather}
	\tmu_0+\frac{\tmu_1\tka_1^k+\tmu_2 \tka_2^k}{x}\geq \tmu_0+\frac{2\tmu_{12}}{d^kx} 
	\geq \tmu_0+\frac{13}{20}\tmu_{12}\geq \frac{1}{30}\bigl(33d-2\sqrt{309d}-66\bigr)\geq d+1,
	\end{gather}	
	where $\tmu_{12}:=\min\{\tmu_1,\tmu_2\}$. 
	Here the first inequality holds because $d^k(\tka_1^k+\tka_2^k)\geq 2$ by   \lref{lem:kappaOrder}; the second inequality holds because $d^kx=d^k\kappa(\Psi,\Sigma(d))/(2D+2d)\leq 40/13$ by \lref{lem:kappaSigUBmagic} and $d>200$ by  assumption;  the third inequality follows from   \lref{lem:muj}; the last inequality is easy to verify. The above two equations together confirm \eref{eq:hkaNsAbsUBaux} and complete the proof of \lref{lem:hkaAuxAux}. 
\end{proof}

\begin{proof}[Proof of \lref{lem:hkaAuxMagicId}]
	If $k=0$, then $\kappa(\caT)=1$ and $\hka(\caT)=(D+2)/(D+d)$ for all $\caT\in \Sigma(d)$ given that $|\Psi\>$ is a stabilizer state, so \eqsref{eq:hkaAbsd1UB1}{eq:hkaAbsd1UB2} hold.

	Next, suppose $k\geq 1$.  According to \eref{eq:kappaSigNSscrT} and \lsref{lem:CubicStateSum}, \ref{lem:CubicState12Sum} we have
	\begin{align}\label{eq:hkaAuxProof1} 
	\kappa(\Sigma(d))= 6+\kappa(\scrT_\ns),\quad \kappa(\scrT_\ns)\leq \begin{cases}
	9/2 & k=1,\\
	2 & k\geq 2, 
	\end{cases}
	\end{align}
	which means $\kappa(\Sigma(d))\geq 2 \kappa(\scrT_\ns)$ for $k\geq 1$ and  $\kappa(\Sigma(d))\geq 4 \kappa(\scrT_\ns)$ for $k\geq 2$.
	In addition, \lref{lem:hkaAuxAux} means
	\begin{gather}\label{eq:hkaAuxProof2}  \frac{D-1}{D+2}|\hka|(\scrT_\ns)\leq \kappa(\scrT_\ns) -\frac{h_d\kappa(\Sigma(d))}{D+d},\quad \frac{D-1}{D+2}|\hka|(\Sigma(d))\leq \kappa(\Sigma(d)) -\frac{(h_d+3)\kappa(\Sigma(d))}{D+d},		
	\end{gather}
	where $h_d$ is defined in \eref{eq:cd-hkaAuxAux}. 
	If $d=31$ or  $d\geq 43$, then $h_d=3$, so 
	\eref{eq:hkaAbsd1UB2} holds. If $d=37$, then $h_d=2$, so \eref{eq:hkaAbsd1UB2} holds as long as $n\geq 2$. If  $d=37$ and $n=1$,  then \eref{eq:hkaAbsd1UB2} can be verified by direct calculation.

	Next, suppose $d=19$ and $k\geq 1$; then $h_d=1$ by \eref{eq:cd-hkaAuxAux}.  If in addition $k\geq 2$, then $\kappa(\Sigma(d))\geq 4 \kappa(\scrT_\ns)$, so \eref{eq:hkaAbsd1UB2} holds.  If $k=1$ and $n\geq 4$, then \eref{eq:hkaAbsd1UB2} follows from \lref{lem:hkaTLUBd1}. If $k=1$ and $n=2,3$,  then 
	\eref{eq:hkaAbsd1UB2} can be verified by direct calculation.  If  $n=k=1$,  then 
	\eref{eq:hkaAbsd1UB1} can be verified by direct calculation.

	Now suppose $d=7,13$, which means  $h_d=-1$ by \eref{eq:cd-hkaAuxAux}. If $k\geq 2$, then  \eqsref{eq:hkaAuxProof1}{eq:hkaAuxProof2} yield
	\begin{equation}
	\begin{aligned}
	|\hka|(\scrT_\ns) &\leq  \frac{D+2}{D-1}\left[\kappa(\scrT_\ns) +\frac{\kappa(\Sigma(d)}{D+d}\right]\leq \kappa(\scrT_\ns)+\frac{3}{D-1}\kappa(\scrT_\ns) +\frac{\kappa(\Sigma(d))}{D}\leq \kappa(\scrT_\ns) +\frac{16}{D}, \\
	|\hka|(\Sigma(d)) &\leq  \frac{D+2}{D-1}\left[\kappa(\Sigma(d)) -\frac{2\kappa(\Sigma(d)}{D+d}\right]\leq \kappa(\Sigma(d)) +\frac{2\kappa(\Sigma(d))}{D}\leq \kappa(\Sigma(d)) +\frac{16}{D},
	\end{aligned}
	\end{equation}
	which confirms \eref{eq:hkaAbsd1UB1}.  If $k=1$ and  $n\geq 5$, then \eref{eq:hkaAbsd1UB1} follows from \lref{lem:hkaTLUBd1}. If $k=1$ and   $n\leq 4$, then \eref{eq:hkaAbsd1UB1} can be verified by direct calculation, which completes the proof of \lref{lem:hkaAuxMagicId}. 	
\end{proof}

\subsection{\label{app:kahkaUBd1}Proofs of \lsref{lem:hkaAuxMagicd1} and \ref{lem:kahkaNsSigMagicSLUBd1}}

\begin{proof}[Proof of \lref{lem:hkaAuxMagicd1}]
	If $|\Psi\>\in \scrM_{n,k}^\id(d)$, then \eref{eq:hkaSigNsAbsd1A} follows from \lref{lem:hkaAuxMagicId}.

	Next, suppose $|\Psi\>\in \scrM_{n,k}(d)\setminus\scrM_{n,k}^\id(d)$, which means   $n\geq k\geq 2$. Then  $\kappa(\scrT_\ns)\leq 8/d$ according to \lref{lem:CubicState12Sum}.
	By virtue of \pref{pro:hkaka} we can deduce that
	\begin{align}
	\frac{D-1}{D+2}|\hka|(\scrT_\ns)&=\sum_{\caT\in \scrT_\ns}\left|\kappa(\caT) -\frac{1}{2(D+d)}\kappa(\Sigma(d))\right|\leq  \kappa(\scrT_\ns) +\frac{(d-2)[\kappa(\scrT_\ns)+6]}{D+d}\nonumber\\
	& \leq \frac{(D+\frac{d}{2}-2)\kappa(\scrT_\ns)+6d}{D+d}\leq \frac{D-1}{D+2}\left[\kappa(\scrT_\ns)+\frac{6d}{D}\right],
	\end{align}	
	which implies the second inequality in \eref{eq:hkaSigNsAbsd1A}. If $\kappa(\Sigma(d))\geq 6(D+d)/(D+2)$, then $|\hka|(\scrT_\sym)=6\hka(\Delta)\leq 6=\kappa(\scrT_\sym)$ by \lref{lem:kahakRelation}, 
	which implies the first inequality in \eref{eq:hkaSigNsAbsd1A} given the second inequality proved above. If instead $\kappa(\Sigma(d))\leq 6(D+d)/(D+2)$, then by  virtue of \lref{lem:hkaLUB} in \aref{app:ThirdMomentAux} we can deduce 
	\begin{align}
	|\hka|(\Sigma(d))\leq \frac{D+d-1}{D}\kappa(\Sigma(d))\leq
	\kappa(\Sigma(d))+ \frac{d-1}{D}\times \frac{6(D+d)}{D+2} \leq\kappa(\Sigma(d))+\frac{6d}{D},
	\end{align}
	which confirms the first inequality in \eref{eq:hkaSigNsAbsd1A}. Here the last inequality holds because  $n\geq k\geq 2$.

	If $n\geq [2+(\log_9d)^{-1}]k+(\log_6d)^{-1}$,  then  
	$0\leq \hka(\Psi,\caT)\leq \kappa(\Psi,\caT)$ for all $\caT\in \Sigma(d)$	and
	\eref{eq:hkaSigNsAbsd1B} holds  by \lref{lem:hkaTLUBd1}.  
\end{proof}

\begin{proof}[Proof of \lref{lem:kahkaNsSigMagicSLUBd1}]
	The equality in \eref{eq:kaNsSigUBd1} follows from \eref{eq:kappaSigNSscrT}, and the second inequality follows from the definition of $\gamma_{d,k}$ in \eref{eq:gammadk}. So we can focus on  the first inequality in \eref{eq:kaNsSigUBd1} in the following discussion.	
	
	If  $k=1$, then $|\Psi\>\in \scrM_{n,k}^\id(d)$ and $\kappa(\Psi,\scrT_\ns)\leq 9/2$ by \eref{eq:kappafTsumLUB1},   which confirms \eref{eq:kaNsSigUBd1}.

	Next, suppose $n\geq k\geq 2$ and
	$|\Psi\>\in \scrM_{n,k}^\id(d)$. Let $|\varphi\>$ be the tensor product of two single-qudit magic states in the tensor decomposition of $|\Psi\>$. Then 
	\begin{align}\label{eq:kappaNsd1Proof1}
	\kappa(\Psi,\scrT_\ns)&\leq \frac{4^{k-2}}{d^{k-2}}\kappa(\varphi,\scrT_\ns)\leq  \frac{4^{k-2}}{d^{k-2}} \times \frac{15}{d}=\frac{15}{16}\times\frac{4^k}{d^{k-1}},
	\end{align}
	where the two inequalities follow from \lref{lem:kappaMagicLUBd1}  and the inequality $\kappa(\varphi,\scrT_\ns)\leq 15/d$ by  \eref{eq:kappafTsumLUB1}. If in addition $k\geq 3$, then a similar reasoning yields $\kappa(\Psi,\scrT_\ns)\leq (13/16)\times 4^k/d^{k-1}$. These results confirm  \eref{eq:kaNsSigUBd1}.

	Next, suppose $k\geq 2$ and $|\Psi\>\in \scrM_{n,k}(d)\setminus \scrM_{n,k}^\id(d)$. Let $|\varphi\>$ be a tensor factor of $|\Psi\>$ that belongs to $\scrM_{2,2}(d)\setminus \scrM_{2,2}^\id(d)$.	
	Then \eref{eq:kaNsSigUBd1} can be proved as follows,
	\begin{align}\label{eq:kappaNsd1Proof2}
	\kappa(\Psi,\scrT_\ns)&\leq \frac{4^{k-2}}{d^{k-2}}\kappa(\varphi,\scrT_\ns,2)\leq  \frac{4^{k-2}}{d^{k-2}} \times \frac{8}{d}=\frac{4^k}{2d^{k-1}}\leq \gamma_{d,k},
	\end{align}
	where the first two inequalities follow from \lref{lem:kappaMagicLUBd1}  and the inequality $\kappa(\varphi,\scrT_\ns)\leq 8/d$ by
	\lref{lem:CubicState12Sum}.

	Now we are ready to  prove the following inequality,
	\begin{equation}\label{eq:hkaNsSigAbsd1Proof}
	\max\{|\hka|(\Psi,\scrT_\ns), |\hka|(\Psi,\Sigma(d))-6\}\leq \gamma_{d,k},
	\end{equation}
	which means the bounds in \eref{eq:kaNsSigUBd1} also apply to $|\hka|(\Psi,\scrT_\ns)$ and $|\hka|(\Psi,\Sigma(d))-6$. If $|\Psi\>\in \scrM_{n,k}(d)\setminus \scrM_{n,k}^\id(d)$, which means $k\geq 2$, then by virtue of  \lref{lem:hkaAuxMagicd1} and \eref{eq:kappaNsd1Proof2} we can deduce that
	\begin{align}
	\max\{|\hka|(\Psi,\scrT_\ns), |\hka|(\Psi,\Sigma(d))-6\}&\leq 	\kappa(\Psi,\scrT_\ns)+\frac{6d}{D}\leq
	\frac{4^k}{2d^{k-1}}+\frac{6}{d^{n-1}}\leq  \left(\frac{1}{2}+\frac{6d^k}{4^kd^n}\right)\frac{4^k}{d^{k-1}}\leq \gamma_{d,k},
	\end{align}	
	which confirms   \eref{eq:hkaNsSigAbsd1Proof}.
	
	Next, suppose $|\Psi\>\in  \scrM_{n,k}^\id(d)$. 	If   $d\geq 31$ or  if $d\geq 19$ and $n\geq 2$, then \eref{eq:hkaNsSigAbsd1Proof} follows from \lref{lem:hkaAuxMagicId} and \eref{eq:kaNsSigUBd1}. If $d=19$ and $n=1$, then \eref{eq:hkaNsSigAbsd1Proof} can be verified by direct calculation. 
	If $d=7,13$ and  $n\geq [2+(\log_9d)^{-1}]k+(\log_6d)^{-1}$, then \eref{eq:hkaNsSigAbsd1Proof} follows from  \lref{lem:hkaTLUBd1} and \eref{eq:kaNsSigUBd1}.
	It remains to consider the case with  $d=7,13$  and  $n< [2+(\log_9d)^{-1}]k+(\log_6d)^{-1}$. 	If  in addition $k\leq 3$, then \eref{eq:hkaNsSigAbsd1Proof} can be verified by direct calculation. 
	If instead $k\geq 4$, then $\kappa(\Psi,\scrT_\ns)\leq (3/4)\times 4^k/d^{k-1}$ by  direct calculation and similar reasoning that leads to \eref{eq:kappaNsd1Proof1}. In conjunction with \lref{lem:hkaAuxMagicId} we can deduce that
	\begin{align}
	\max\{|\hka|(\Psi,\scrT_\ns), |\hka|(\Psi,\Sigma(d))-6\}\leq \kappa(\Psi,\scrT_\ns)+\frac{16}{D}\leq 
	\frac{3}{4}\times\frac{4^k}{d^{k-1}}
	+\frac{16}{D}\leq \frac{13}{16}\times\frac{4^k}{d^{k-1}}\leq \gamma_{d,k},
	\end{align}
	which confirms \eref{eq:hkaNsSigAbsd1Proof}.
	
	Given the assumption  $|\Psi\>\in \scrM^\id_{n,k}(d)$, 
	\eref{eq:kappaNsSigLB} follows from \lref{lem:kappaNsSigLBd1}, which completes the proof of \lref{lem:kahkaNsSigMagicSLUBd1}. 
\end{proof}

\subsection{\label{app:MagicKeyd1Proofs}Proofs of \thsref{thm:Phi3MagicUBd1}-\ref{thm:StabShNormMagicd1}}

\begin{proof}[Proof of \thref{thm:Phi3MagicUBd1}]
	Let $|\Psi_2\>=|\Psi\>^{\otimes 2}$; then $|\Psi_2\>\in \scrM_{2n,2k}(d)$ and 
	\begin{align}
	\kappa(\Psi,\Sigma(d),2)=6+\kappa(\Psi,\scrT_\ns,2)
	=6+\kappa(\Psi_2,\scrT_\ns)=\kappa(\Psi_2,\Sigma(d)).
	\end{align}	
	Since  $\kappa(\Psi,\caT)\geq 0$ for all $\caT\in \Sigma(d)$ according to  \lref{lem:kappaMagicLUBd1}, by virtue of  \thref{thm:Phi3LUB} and \coref{cor:Phi3UB} we can deduce that	
	\begin{align}
	\bar{\Phi}_3(\orb(\Psi))
	&=\frac{D+2}{6(D-1)}\left[\kappa(\Psi,\Sigma(d),2) -\frac{\kappa^2(\Psi,\Sigma(d))}{2(D+d)}\right]\leq \frac{1}{6}\left[1+\frac{3(d-2)}{D^2}\right][6+\kappa(\Psi_2,\scrT_\ns)].  \label{eq:Phi3d1UBproof}
	\end{align}
	
	Now, suppose
	$|\Psi\>\in \scrM_{n,k}^\id(d)$; then 	$|\Psi_2\>\in \scrM_{2n,2k}^\id(d)$. 
	If in addition $n=k=1$	and $d=7,13$, then \eref{eq:Phi3MagicUBd1} can be verified by direct calculation. If instead $d\geq 19$ or $n\geq k\geq 2$, then by virtue of  \eref{eq:Phi3d1UBproof} and \lref{lem:kahkaNsSigMagicSLUBd1}   we can deduce that
	\begin{align}
	\bar{\Phi}_3(\orb(\Psi))\leq \frac{\kappa(\Psi,\Sigma(d),2)}{6}=1+ \frac{\kappa(\Psi_2,\scrT_\ns)}{6}\leq 1+ \frac{\gamma_{d,2k}}{6}\leq 1+\frac{5}{32}\times\frac{16^k}{d^{2k-1}}\leq \frac{13}{11},
	\end{align}	
	which confirms \eref{eq:Phi3MagicUBd1}. 
	
	Next, suppose $|\Psi\>\in \scrM_{n,k}(d)\setminus\scrM_{n,k}^\id(d)$, which means $n\geq k\geq 2$ and $|\Psi_2\>\in \scrM_{2n,2k}(d)\setminus\scrM_{2n,2k}^\id(d)$. Then $\kappa(\Psi_2,\scrT_\ns)=\kappa(\Psi,\scrT_\ns,2)\leq 16^k/(2d^{2k-1})$ by \lref{lem:CubicState12Sum}. 	
	In conjunction with \eqsref{eq:Phi3d1UBproof}{eq:gammadk} we can deduce that
	\begin{equation}
	\begin{aligned}
	\bar{\Phi}_3(\orb(\Psi))&\leq  1+\frac{\kappa(\Psi_2,\scrT_\ns)}{6}+\frac{3(d-2)}{D^2}\left(1+\frac{16^k}{12d^{2k-1}}\right)\leq 1+\frac{\kappa(\Psi_2,\scrT_\ns)}{6}+\frac{3d}{D^2}	, \\
	\bar{\Phi}_3(\orb(\Psi))&\leq \left[1+\frac{3(d-2)}{D^2}\right]\left(1+\frac{16^k}{12d^{2k-1}}\right)
	\leq  1+\frac{16^k}{10d^{2k-1}}\leq 1+ \frac{\gamma_{d,2k}}{6}\leq 1+\frac{5}{32}\times\frac{16^k}{d^{2k-1}}\leq \frac{13}{11},
	\end{aligned}
	\end{equation}
	given that $D^2=d^{2n}\geq d^{2k}$. This result  confirms \eref{eq:Phi3MagicUBd1} and completes the proof of \thref{thm:Phi3MagicUBd1}. 		
\end{proof}

\begin{proof}[Proof of \thref{thm:Moment3MagicLUBd1}]
	The first inequality in \eref{eq:Moment3MagicLUBd1} is obvious given that $D=d^n$. The second inequality  follows from \thref{thm:MomentNormd13} and \lref{lem:kappaMagicLUBd1}. The last inequality  can be proved as follows,
	\begin{align}
	\|\bQ(\orb(\Psi))\| \leq \frac{\max\{|\hka|(\scrT_\iso), 2(D+2)\kappa(\caT_\defe)\}}{6}\leq \max\biggl\{\frac{7}{4}, \frac{D+2}{3}\frac{(4d-3)^k}{d^{2k}}\biggr\}\leq \max\biggl\{\frac{7}{4},\frac{4^kd^n}{3d^k}\biggr\}.  \label{eq:MomentNormd1Proof}
	\end{align}				
	Here the first inequality follows from \thref{thm:MomentNormd13}, 	 the second inequality follows from  \lref{lem:kappaMagicLUBd1} and the fact that $|\hka|(\scrT_\iso)\leq |\hka|(\Sigma(d))\leq 21/2$ by \lref{lem:kahkaNsSigMagicSLUBd1}, and the last inequality can be proved as follows,
	\begin{equation}
	\begin{aligned}
	\frac{D+2}{3}\left(\frac{4d-3}{d^2}\right)^k&\leq \frac{d+2}{3}\times \frac{4}{d}\leq \frac{12}{7}\leq \frac{7}{4},\quad n=k=1,\\
	\frac{D+2}{3}\left(\frac{4d-3}{d^2}\right)^k&\leq
	\frac{d^n(d^2+2)}{3d^2}\times\frac{4d-3}{d^2}\times \frac{4^{k-1}}{d^{k-1}}\leq \frac{4^kd^{n-k}}{3}, \quad n\geq 2.
	\end{aligned}
	\end{equation}
	
	Next, suppose $\hka(\caT_\defe)\leq 0$ for some stochastic Lagrangian subspace $\caT_\defe$  in $\scrT_\defe$. Note that $\hka(\caT_\defe)\leq 0$ iff  $\kappa(\caT_\defe)	\leq \kappa(\Sigma(d))/[2(D+d)]$ by \pref{pro:hkaka}. 
	The lower bound in \eref{eq:Moment3MagicDefNegLUBd1} follows from \eref{eq:Moment3MagicLUBd1}. 
	By virtue of \eref{eq:MomentNormd13c} in \thref{thm:MomentNormd13} and \lref{lem:kahkaNsSigMagicSLUBd1} we can deduce that
	\begin{align}
	\|\bQ(\orb(\Psi))\|\leq \frac{1}{6}|\hka|(\scrT_\iso)\leq \frac{1}{6}|\hka|(\Sigma(d))
	\leq  1+\frac{\gamma_{d,k}}{6}\leq 1+\frac{3\times 4^{k-2}}{d^{k-1}},
	\end{align}
	which confirms the upper bounds in \eref{eq:Moment3MagicDefNegLUBd1} and completes the proof of \thref{thm:Moment3MagicLUBd1}.
\end{proof}

\begin{proof}[Proof of \thref{thm:StabShNormMagicd1}]
	If $n\geq 2$, then  $\kappa(\scrT_\ns)\geq0\geq -2(D^2-2dD+D+d)/D^2$ thanks to \lref{lem:kappaMagicLUBd1}. If $n=1$, then  $|\Psi\>\in \scrM_{n,k}^\id(d)$ and $\kappa(\scrT_\ns)\geq 2(d-2)/D=-2(D^2-2dD+D+d)/D^2$ thanks to \lref{lem:kahkaNsSigMagicSLUBd1}. In both cases, 
	the first equality in \eref{eq:StabShNormMagicd1} holds by \coref{cor:Stab1ShNormGen}. The second equality in \eref{eq:StabShNormMagicd1} holds because $\hka(\Sigma(d))=(D+2)\kappa(\Sigma(d))/(D+d)$ by \pref{pro:hkaka} and  $\kappa(\Sigma(d))=6+\kappa(\scrT_\ns)$ by \eref{eq:kappaSigNSscrT}.

	The first two inequalities in \eref{eq:StabShNormMagicd1UB} follow from \eref{eq:StabShNormMagicd1} and the fact that $\|\Ob_0\|_2^2=(D-1)/D$, and the last  two inequalities  follow from \lref{lem:kahkaNsSigMagicSLUBd1}.
\end{proof}

\subsection{\label{app:thm:AccurateDesignMagicProof}Proof of \thref{thm:AccurateDesignMagic}}

\begin{proof}[Proof of \thref{thm:AccurateDesignMagic}]
	By assumption $\scrE$ is constructed from $\scrM_{n,k}(d)$ with $k\geq 1$ and satisfies the conditions in   \eref{eq:DesignFiducialCon}. So  $\kappa_{\min}(\scrE)\geq 0$ 
	by \lref{lem:kappaMagicLUBd1}, which means 
	\begin{align}\label{eq:AccurateDesignMagicProof}
	\kappa(\scrE,\Sigma(d))\geq 6, \quad 
	-\frac{(D+2)\kappa(\scrE,\Sigma(d))}{2(D-1)(D+d)}\leq \hka(\scrE, \caT_\defe)\leq 0, \quad \kappa'(\scrE,\scrT_\ns)\leq \kappa(\scrE,\scrT_\ns),
	\end{align}
	where $\kappa'(\scrE,\scrT_\ns)$ is defined in \eref{eq:kappaNsP}. 
	In conjunction with the condition in \eref{eq:designProb} we can deduce that 
	\begin{align}
	p&=\frac{-\hka(\scrE, \caT_\defe)}{\frac{D+2}{D+d}-\hka(\scrE, \caT_\defe)}
	\leq \frac{\frac{(D+2)\kappa(\scrE,\Sigma(d))}{2(D-1)(D+d)}}{\frac{D+2}{D+d}+\frac{(D+2)\kappa(\scrE,\Sigma(d))}{2(D-1)(D+d)}}\leq\frac{\frac{\kappa(\scrE,\Sigma(d))}{2(D-1)(D+d)}}{\frac{1}{D+d}+\frac{3}{(D-1)(D+d)}}=\frac{\kappa(\scrE,\Sigma(d))}{2(D+2)}\leq \frac{21}{4D},\label{eq:DesignProbUB}
	\end{align}
	where the last inequality  follows from the inequality $\kappa(\scrE,\Sigma(d)) \leq 21/2$ by \lref{lem:kahkaNsSigMagicSLUBd1},  which  is still applicable if  the  state $|\Psi\>\in \scrM_{n,k}(d)$ is replaced by the ensemble $\scrE$ constructed from $\scrM_{n,k}(d)$.

	\Eqsref{eq:QscrEpMagic1}{eq:QscrEpMagic2}  follow from  \pref{pro:AccurateDesign}, \eref{eq:AccurateDesignMagicProof}, and 	\lref{lem:kahkaNsSigMagicSLUBd1}.
\end{proof}

\subsection{\label{app:BalanceMagicProof}Proofs of \lsref{lem:kahkaBalMagicd1} and \ref{lem:kahkaQBalMagicd1}}

\begin{proof}[Proof of \lref{lem:kahkaBalMagicd1}]
	Suppose $|\Psi\>\in \scrM_{n,k}(d)$; then $\kappa(\Psi)$ is completely determined by $\scrC_l(\Psi)$ for $l=0,1,2$. To prove \lref{lem:kahkaBalMagicd1}, therefore, it suffices to consider any given $k$-balanced ensemble, say the ensemble defined in \eref{eq:kbalanceEx}. If $k=0\mmod 3$, then 
	\begin{align}
	\kappa(\scrE,\caT)=\kappa(\tPsi,\caT)=[\kappa(f_0, \caT)\kappa(f_1,\caT)\kappa(f_2,\caT)]^{\lfloor k/3\rfloor}=\frac{\tg(d)^{2\lfloor k/3\rfloor}}{d^{k+\lfloor k/3\rfloor}},
	\end{align}
	where the last equality follows from \lref{lem:kappaOrder} given that $\{\kappa(f_j,\caT)\}_{j=0,1,2}$ is identical to $\{\kappa_j^\uparrow\}_{j=0,1,2}$
	up to a permutation. If $k=1\mmod 3$, then 
	\begin{align}
	\kappa(\scrE,\caT)=\frac{\kappa(\Psi_0,\caT)+\kappa(\Psi_1,\caT)+\kappa(\Psi_2,\caT)}{3}=\frac{\kappa(\tPsi,\caT)[\kappa(f_0, \caT)+\kappa(f_1,\caT)+\kappa(f_2,\caT)]}{3}=\frac{2\tg(d)^{2\lfloor k/3\rfloor}}{d^{k+\lfloor k/3\rfloor}},
	\end{align}
	given that $\kappa(f_0, \caT)+\kappa(f_1,\caT)+\kappa(f_2,\caT)=6/d$ by  \lref{lem:kappaOrder}.  If $k=2\mmod 3$, then 
	\begin{align}
	\kappa(\scrE,\caT)&
	=\frac{\kappa(\tPsi,\caT)[\kappa(f_0, \caT)\kappa(f_1, \caT)+\kappa(f_1,\caT)\kappa(f_2, \caT)+\kappa(f_2,\caT)\kappa(f_0, \caT)]}{3}
	=\frac{3\tg(d)^{2\lfloor k/3\rfloor}}{d^{k+\lfloor k/3\rfloor}}, 
	\end{align}
	given that $\kappa(f_0, \caT)\kappa(f_1, \caT)+\kappa(f_1,\caT)\kappa(f_2, \caT)+\kappa(f_2,\caT)\kappa(f_0, \caT)=9/d^2$ by  \lref{lem:kappaOrder}. The above three equations confirm the equality in \eref{eq:kaBalance} in all three cases. The first two inequalities in  \eref{eq:kaBalance} follow from the fact that $1\leq \tg(d)^2<4d-27$ by \eref{eq:tgdsq} (see also \lref{lem:GaussSum3LUB} in \aref{app:GaussJacobi}), and the third inequality is easy to verify.

	The equality in \eref{eq:hkaBalMagic} follows from \lref{lem:kahkaBalance}. Since $0\leq \kappa(\scrE,\caT)\leq 1$, this equality means
	\begin{align}\label{eq:hkaBalUBd1Proof1}
	\frac{D+d}{D+2}|\hka(\scrE,\caT)|\leq \begin{cases}
	\kappa(\scrE,\caT) &\kappa(\scrE,\caT) \geq 3/(D+2),\\
	\frac{3}{D-1} & \kappa(\scrE,\caT)\leq  3/(D+2). 
	\end{cases} 
	\end{align}
	Note that $\kappa(\scrE,\caT)=2/d$ when $k=1$ and $\kappa(\scrE,\caT)=3/d^2$ when $k=2$.  If $k=1,2$ or $\kappa(\scrE,\caT)\geq  3/(D+2)$, then 
	\begin{align}\label{eq:hkaBalUBd1Proof2}
	|\hka|(\scrE,\caT)\leq \frac{(D+2)\kappa(\scrE,\caT)}{(D+d)}\leq \frac{4^{(2k+1)/6}(D+2)}{d^k(D+d)}, 
	\end{align}	
	which confirms the inequality in  \eref{eq:hkaBalMagic}. 	
	If $k\geq 3$ and $\kappa(\scrE,\caT)\leq  3/(D+2)$, then 
	\begin{align}
	\frac{D+d}{D+2}|\hka(\scrE,\caT)|\leq \frac{3}{D-1}\leq \frac{4^{(2k+1)/6}}{d^k},
	\end{align}
	which confirms the inequality in  \eref{eq:hkaBalMagic} again.	\Eref{eq:kahkaNsAbsBalMagic} is a simple corollary of Eqs.~\eqref{eq:kappatgammadk}-\eqref{eq:hkaBalMagic}
	given that $|\scrT_\ns|=2(d-2)$.

	If $k=1,2$, then 	\eref{eq:hkaNsAbsBalMagicUB} holds by \eref{eq:hkaBalUBd1Proof2}. 
	If $k<3n/4$, then \eref{eq:hkaNsAbsBalMagicUB} follows from \eref{eq:hkaBalUBd1Proof1} and the equation below
	\begin{gather}
	\kappa(\scrE,\caT)=\kappa_{d,k}=\frac{(s_k+1)\tg(d)^{2\lfloor k/3\rfloor}}{d^{k+\lfloor k/3\rfloor}}\geq \frac{1}{d^{n-1}}=\frac{d}{D}>\frac{3}{D+2}, 
	\end{gather}
	given that  $0\leq s_k\leq 2$ and $1\leq \tg(d)^2<4d$ by \lref{lem:kappaOrder}. This observation completes the proof of \lref{lem:kahkaBalMagicd1}.
\end{proof}

\begin{proof}[Proof of \lref{lem:kahkaQBalMagicd1}]
	\Eref{eq:kahkaNsAbsQBalMagicA} is a simple corollary of \lsref{lem:hkaAuxMagicd1}	and \ref{lem:kahkaBalMagicd1}.

	Next, we consider \eref{eq:kahkaNsAbsQBalMagicB}, assuming that  $n\geq(4k+8)/3$. If $k=0\mmod 3$, then $|\Psi\>$ is $k$-balanced, so \eref{eq:kahkaNsAbsQBalMagicB} holds by \lref{lem:kahkaBalMagicd1}. 
	If $k=1\mmod 3$, then $|\Psi\>$ can be expressed as a tensor product $|\Psi\>=|\psi_f\>\otimes |\tPsi\>$, where $f\in \tscrP_3(\bbF_d)$, $|\psi_f\>$   is a single-qudit magic state, and $|\tPsi\>\in \scrM_{n-1,k-1}(d)$ is $(k-1)$-balanced. Therefore,
	\begin{align}
	\kappa(\Psi,\caT)=\kappa(f,\caT)\kappa(\tPsi,\caT)\geq \frac{1}{9d^2}\times \frac{\tg(d)^{2(k-1)/3}}{d^{4(k-1)/3}}\geq \frac{1}{9d^{(4k+2)/3}}\geq \frac{21}{4d^n}\geq \frac{\kappa(\Psi,\Sigma(d))}{2D+2d}\quad \forall \caT\in \Sigma(d),
	\end{align} 
	where the first two inequalities follow from \lsref{lem:kappaMagicOned1} and \ref{lem:kahkaBalMagicd1} given that  $\Sigma(d)=\scrT_\sym\sqcup \scrT_\ns$, the third inequality follows from the assumption that $d\geq 7$ and $n\geq (4k+8)/3$, and the last inequality follows from the fact that $D=d^n$ and $\kappa(\Psi,\Sigma(d))\leq 21/2$ by \lref{lem:kahkaNsSigMagicSLUBd1}. In conjunction with \pref{pro:hkaka} the above equation means $\hka(\Psi,\caT)\geq 0$ for all $\caT\in \scrT_\ns$, which implies the equality in  \eref{eq:kahkaNsAbsQBalMagicB}.  The inequalities in \eref{eq:kahkaNsAbsQBalMagicB}  follow from \lsref{lem:hkaLUB}  and \ref{lem:kahkaBalMagicd1}. 
	
	If $k=2\mmod 3$, then $|\Psi\>$ can be expressed as a tensor product of the form $|\Psi\>=|\psi_{f_1}\>\otimes |\psi_{f_2}\>\otimes |\tPsi\>$, where $|\tPsi\>\in \scrM_{n-2,k-2}(d)$ is $(k-2)$-balanced, and $|\psi_{f_1}\>, |\psi_{f_2}\>$ are two single-qudit magic states associated with  cubic polynomials $f_1, f_2$ that have different cubic characters, which means $\kappa(f_1,\caT)\kappa(f_2,\caT)\geq 1/(4d^3)$ for $\caT\in \Sigma(d)$ thanks to \pref{pro:kappaSym} and \lref{lem:kappaOrder}. By virtue of a similar reasoning employed above we can deduce that
	\begin{align}
	\kappa(\Psi,\caT)=\kappa(f_1,\caT)\kappa(f_2,\caT)\kappa(\tPsi,\caT)\geq \frac{1}{4d^3}\times \frac{\tg(d)^{2(k-2)/3}}{d^{4(k-2)/3}}\geq \frac{1}{4d^{(4k+1)/3}}\geq  \frac{\kappa(\Psi,\Sigma(d))}{2D+2d}\quad \forall \caT\in \Sigma(d).
	\end{align} 
	So \eref{eq:kahkaNsAbsQBalMagicB} holds as before. 	
\end{proof}

\subsection{\label{app:BalMagicExactd1Proof}Proof of \lref{lem:BalMagicExactd1}}

\begin{proof}[Proof of \lref{lem:BalMagicExactd1}]
	According to the definition in  \eref{eq:kappatgammadk}, if $k<3n/4$, then 
	\begin{gather}
	\kappa_{d,k}=\frac{(s_k+1)\tg(d)^{2\lfloor k/3\rfloor}}{d^{k+\lfloor k/3\rfloor}}\geq \frac{1}{d^{n-1}}=\frac{d}{D}>\frac{3}{D+2}, 
	\end{gather}
	given that  $0\leq s_k\leq 2$ and $1\leq \tg(d)^2<4d$ by \eref{eq:tgdsq} (see also  \lref{lem:GaussSum3LUB} in \aref{app:GaussJacobi}). Therefore, $k_*(d,n)\geq 3n/4$ according to  the definition in \eref{eq:kstardn}, which confirms the first inequality in \eref{eq:kstardnLimLB}. The equality in \eref{eq:kstardnLimLB} follows from \eqsref{eq:kappatgammadk}{eq:kstardn}. The second and third inequalities  in \eref{eq:kstardnLimLB}  follow from the fact that  $1\leq \tg(d)^2<4d$; the second inequality also follows from the first inequality.

	If  $\tg(d)^2\geq 3d$ and $n= 2$, then $k_*(d,n)= n+1$ by \eref{eq:kstardn123}.  If  $\tg(d)^2\geq 3d$, $n\geq 3$, and $k\leq n$, then  
	\begin{gather}
	\kappa_{d,k}=\frac{(s_k+1)\tg(d)^{2\lfloor k/3\rfloor}}{d^{k+\lfloor k/3\rfloor}}\geq \frac{3(s_k+1)}{d^k}\geq \frac{3}{D}>\frac{3}{D+2}, 
	\end{gather}
	which implies that $k_*(d,n)\geq  n+1$.

	To proceed, note that $d=1 \mmod 3$ by assumption and $\tg(d)=1\mmod 3$ by \lref{lem:GaussSum3LUB}, which means 	
	\begin{align}\label{eq:tgdsqMod}
	d-\tg(d)^2=0 \mmod 3,\quad  |d-\tg(d)^2|\geq 3, \quad 3d-\tg(d)^2=2 \mmod 3.
	\end{align}	
	Now, suppose $\tg(d)^2\leq 3d$;  then $\tg(d)^2\leq 3d-2\leq 3d^4/(d^3+2)$, which means  $k_*(d,3)= 3$ by \eref{eq:kstardn123}.

	Next, suppose $\tg(d)^2\geq d$;  then $\tg(d)^2\geq d+3$ by \eref{eq:tgdsqMod} and 
	\begin{gather}
	\kappa_{d,k}=\frac{(s_k+1)\tg(d)^{2\lfloor k/3\rfloor}}{d^{k+\lfloor k/3\rfloor}}\geq \frac{(s_k+1)(d+3)^{\lfloor k/3\rfloor}}{d^{k+\lfloor k/3\rfloor}}. 
	\end{gather}
	If  $k\leq n-1$, then $d^k\leq D/d$ and 
	\begin{gather}
	\kappa_{d,k}\geq \frac{d(s_k+1)}{D}\geq \frac{7}{D}> \frac{3}{D+2},
	\end{gather}
	which means $k_*(d,n)\geq n$. 
	If $n=2\mmod 3$, then $\kappa_{d,n}\geq 3/D>3/(D+2)$,
	which means $k_*(d,n)\geq n+1$. 
	If $n\geq d\ln 3+5$, then
	\begin{gather}
	D\kappa_{d,n}\geq \frac{(d+3)^{\lfloor n/3\rfloor}}{d^{\lfloor n/3\rfloor}} \geq
	\left(1+\frac{3}{d}\right)^{(d\ln 3+2)/3}\geq 3, \label{eq:kdnLargen}
	\end{gather} 
	which means $\kappa_{d,n}>3/(D+2)$ and $k_*(d,n)\geq n+1$. The last inequality in \eref{eq:kdnLargen} holds because the function $[1+(3/d)]^{(d\ln 3+2)/3}$ is monotonically decreasing in $d$ and
	\begin{align}
	\lim_{d\to \infty}\left(1+\frac{3}{d}\right)^{(d\ln 3+2)/3}=\rme^{\ln 3}=3.
	\end{align}

	Finally, suppose $\tg(d)^2\leq d$; then $\tg(d)^2\leq d-3$ given that $|\tg(d)^2-d|\geq 3$ by \eref{eq:tgdsqMod}. If  $n= 2$, then $k_*(d,n)= n+1$ by \eref{eq:kstardn123}. If $n\neq 2$,
	then 
	\begin{align}
	\kappa_{d,n}=\frac{(s_n+1)\tg(d)^{2\lfloor n/3\rfloor}}{d^{n+\lfloor n/3\rfloor}}\leq
	\frac{(s_n+1)(d-3))^{\lfloor n/3\rfloor}}{d^{\lfloor n/3\rfloor}D}< \frac{3}{D+2},
	\end{align}
	which means $k_*(d,n)\leq n$. If in addition $\tg(d)^2=1$, then $\kappa_{d,k}=(s_k+1)/d^{k+\lfloor k/3\rfloor}$ and \eref{eq:kstartgd1} can be verified by direct calculation based on the definition in \eref{eq:kstardn}, which  completes the proof of \lref{lem:BalMagicExactd1}. 
\end{proof}

\bibliography{all_references}

\begin{thebibliography}{96}%
\makeatletter
\providecommand \@ifxundefined [1]{%
 \@ifx{#1\undefined}
}%
\providecommand \@ifnum [1]{%
 \ifnum #1\expandafter \@firstoftwo
 \else \expandafter \@secondoftwo
 \fi
}%
\providecommand \@ifx [1]{%
 \ifx #1\expandafter \@firstoftwo
 \else \expandafter \@secondoftwo
 \fi
}%
\providecommand \natexlab [1]{#1}%
\providecommand \enquote  [1]{``#1''}%
\providecommand \bibnamefont  [1]{#1}%
\providecommand \bibfnamefont [1]{#1}%
\providecommand \citenamefont [1]{#1}%
\providecommand \href@noop [0]{\@secondoftwo}%
\providecommand \href [0]{\begingroup \@sanitize@url \@href}%
\providecommand \@href[1]{\@@startlink{#1}\@@href}%
\providecommand \@@href[1]{\endgroup#1\@@endlink}%
\providecommand \@sanitize@url [0]{\catcode `\\12\catcode `\$12\catcode
  `\&12\catcode `\#12\catcode `\^12\catcode `\_12\catcode `\%12\relax}%
\providecommand \@@startlink[1]{}%
\providecommand \@@endlink[0]{}%
\providecommand \url  [0]{\begingroup\@sanitize@url \@url }%
\providecommand \@url [1]{\endgroup\@href {#1}{\urlprefix }}%
\providecommand \urlprefix  [0]{URL }%
\providecommand \Eprint [0]{\href }%
\providecommand \doibase [0]{https://doi.org/}%
\providecommand \selectlanguage [0]{\@gobble}%
\providecommand \bibinfo  [0]{\@secondoftwo}%
\providecommand \bibfield  [0]{\@secondoftwo}%
\providecommand \translation [1]{[#1]}%
\providecommand \BibitemOpen [0]{}%
\providecommand \bibitemStop [0]{}%
\providecommand \bibitemNoStop [0]{.\EOS\space}%
\providecommand \EOS [0]{\spacefactor3000\relax}%
\providecommand \BibitemShut  [1]{\csname bibitem#1\endcsname}%
\let\auto@bib@innerbib\@empty
\bibitem [{\citenamefont {Delsarte}\ \emph {et~al.}(1977)\citenamefont
  {Delsarte}, \citenamefont {Goethals},\ and\ \citenamefont
  {Seidel}}]{DelsGS77}%
  \BibitemOpen
  \bibfield  {author} {\bibinfo {author} {\bibfnamefont {P.}~\bibnamefont
  {Delsarte}}, \bibinfo {author} {\bibfnamefont {J.~M.}\ \bibnamefont
  {Goethals}},\ and\ \bibinfo {author} {\bibfnamefont {J.~J.}\ \bibnamefont
  {Seidel}},\ }\bibfield  {title} {\bibinfo {title} {Spherical codes and
  designs},\ }\href@noop {} {\bibfield  {journal} {\bibinfo  {journal} {Geom.
  Dedicata}\ }\textbf {\bibinfo {volume} {6}},\ \bibinfo {pages} {363}
  (\bibinfo {year} {1977})}\BibitemShut {NoStop}%
\bibitem [{\citenamefont {Hoggar}(1982)}]{Hogg82}%
  \BibitemOpen
  \bibfield  {author} {\bibinfo {author} {\bibfnamefont {S.~G.}\ \bibnamefont
  {Hoggar}},\ }\bibfield  {title} {\bibinfo {title} {$t$-designs in projective
  spaces},\ }\href@noop {} {\bibfield  {journal} {\bibinfo  {journal} {Eur. J.
  Combinator.}\ }\textbf {\bibinfo {volume} {3}},\ \bibinfo {pages} {233}
  (\bibinfo {year} {1982})}\BibitemShut {NoStop}%
\bibitem [{\citenamefont {Zauner}(2011)}]{Zaun11}%
  \BibitemOpen
  \bibfield  {author} {\bibinfo {author} {\bibfnamefont {G.}~\bibnamefont
  {Zauner}},\ }\bibfield  {title} {\bibinfo {title} {Quantum designs:
  Foundations of a noncommutative design theory},\ }\href@noop {} {\bibfield
  {journal} {\bibinfo  {journal} {Int. J. Quantum Inf.}\ }\textbf {\bibinfo
  {volume} {09}},\ \bibinfo {pages} {445} (\bibinfo {year} {2011})}\BibitemShut
  {NoStop}%
\bibitem [{\citenamefont {Renes}\ \emph {et~al.}(2004)\citenamefont {Renes},
  \citenamefont {Blume-Kohout}, \citenamefont {Scott},\ and\ \citenamefont
  {Caves}}]{ReneBSC04}%
  \BibitemOpen
  \bibfield  {author} {\bibinfo {author} {\bibfnamefont {J.~M.}\ \bibnamefont
  {Renes}}, \bibinfo {author} {\bibfnamefont {R.}~\bibnamefont {Blume-Kohout}},
  \bibinfo {author} {\bibfnamefont {A.~J.}\ \bibnamefont {Scott}},\ and\
  \bibinfo {author} {\bibfnamefont {C.~M.}\ \bibnamefont {Caves}},\ }\bibfield
  {title} {\bibinfo {title} {Symmetric informationally complete quantum
  measurements},\ }\href@noop {} {\bibfield  {journal} {\bibinfo  {journal} {J.
  Math. Phys.}\ }\textbf {\bibinfo {volume} {45}},\ \bibinfo {pages} {2171}
  (\bibinfo {year} {2004})}\BibitemShut {NoStop}%
\bibitem [{\citenamefont {Scott}(2006)}]{Scot06}%
  \BibitemOpen
  \bibfield  {author} {\bibinfo {author} {\bibfnamefont {A.~J.}\ \bibnamefont
  {Scott}},\ }\bibfield  {title} {\bibinfo {title} {Tight informationally
  complete quantum measurements},\ }\href@noop {} {\bibfield  {journal}
  {\bibinfo  {journal} {J. Phys. A: Math. Gen.}\ }\textbf {\bibinfo {volume}
  {39}},\ \bibinfo {pages} {13507} (\bibinfo {year} {2006})}\BibitemShut
  {NoStop}%
\bibitem [{\citenamefont {Ambainis}\ and\ \citenamefont
  {Emerson}(2007)}]{AmbaE07}%
  \BibitemOpen
  \bibfield  {author} {\bibinfo {author} {\bibfnamefont {A.}~\bibnamefont
  {Ambainis}}\ and\ \bibinfo {author} {\bibfnamefont {J.}~\bibnamefont
  {Emerson}},\ }\bibfield  {title} {\bibinfo {title} {Quantum $t$-designs:
  $t$-wise independence in the quantum world},\ }in\ \href@noop {} {\emph
  {\bibinfo {booktitle} {Twenty-Second Annual IEEE Conference on Computational
  Complexity (CCC'07)}}}\ (\bibinfo {year} {2007})\ pp.\ \bibinfo {pages}
  {129--140}\BibitemShut {NoStop}%
\bibitem [{\citenamefont {Zhu}\ \emph {et~al.}(2016)\citenamefont {Zhu},
  \citenamefont {Kueng}, \citenamefont {Grassl},\ and\ \citenamefont
  {Gross}}]{ZhuKGG16}%
  \BibitemOpen
  \bibfield  {author} {\bibinfo {author} {\bibfnamefont {H.}~\bibnamefont
  {Zhu}}, \bibinfo {author} {\bibfnamefont {R.}~\bibnamefont {Kueng}}, \bibinfo
  {author} {\bibfnamefont {M.}~\bibnamefont {Grassl}},\ and\ \bibinfo {author}
  {\bibfnamefont {D.}~\bibnamefont {Gross}},\ }\href@noop {} {\bibinfo {title}
  {{The Clifford group fails gracefully to be a unitary 4-design}}} (\bibinfo
  {year} {2016}),\ \Eprint {https://arxiv.org/abs/1609.08172}
  {arXiv:1609.08172} \BibitemShut {NoStop}%
\bibitem [{\citenamefont {Dankert}(2005)}]{Dank05the}%
  \BibitemOpen
  \bibfield  {author} {\bibinfo {author} {\bibfnamefont {C.}~\bibnamefont
  {Dankert}},\ }\emph {\bibinfo {title} {Efficient Simulation of Random Quantum
  States and Operators}},\ \href@noop {} {\bibinfo {type} {Master thesis}},\
  \bibinfo  {school} {University of Waterloo} (\bibinfo {year} {2005}),\
  \bibinfo {note} {available online at
  \url{http://arxiv.org/abs/quant-ph/0512217}}\BibitemShut {NoStop}%
\bibitem [{\citenamefont {Gross}\ \emph {et~al.}(2007)\citenamefont {Gross},
  \citenamefont {Audenaert},\ and\ \citenamefont {Eisert}}]{GrosAE07}%
  \BibitemOpen
  \bibfield  {author} {\bibinfo {author} {\bibfnamefont {D.}~\bibnamefont
  {Gross}}, \bibinfo {author} {\bibfnamefont {K.}~\bibnamefont {Audenaert}},\
  and\ \bibinfo {author} {\bibfnamefont {J.}~\bibnamefont {Eisert}},\
  }\bibfield  {title} {\bibinfo {title} {Evenly distributed unitaries: On the
  structure of unitary designs},\ }\href@noop {} {\bibfield  {journal}
  {\bibinfo  {journal} {J. Math. Phys.}\ }\textbf {\bibinfo {volume} {48}},\
  \bibinfo {pages} {052104} (\bibinfo {year} {2007})}\BibitemShut {NoStop}%
\bibitem [{\citenamefont {Dankert}\ \emph {et~al.}(2009)\citenamefont
  {Dankert}, \citenamefont {Cleve}, \citenamefont {Emerson},\ and\
  \citenamefont {Livine}}]{DankCEL09}%
  \BibitemOpen
  \bibfield  {author} {\bibinfo {author} {\bibfnamefont {C.}~\bibnamefont
  {Dankert}}, \bibinfo {author} {\bibfnamefont {R.}~\bibnamefont {Cleve}},
  \bibinfo {author} {\bibfnamefont {J.}~\bibnamefont {Emerson}},\ and\ \bibinfo
  {author} {\bibfnamefont {E.}~\bibnamefont {Livine}},\ }\bibfield  {title}
  {\bibinfo {title} {Exact and approximate unitary 2-designs and their
  application to fidelity estimation},\ }\href@noop {} {\bibfield  {journal}
  {\bibinfo  {journal} {Phys. Rev. A}\ }\textbf {\bibinfo {volume} {80}},\
  \bibinfo {pages} {012304} (\bibinfo {year} {2009})}\BibitemShut {NoStop}%
\bibitem [{\citenamefont {Roy}\ and\ \citenamefont {Scott}(2009)}]{RoyS09}%
  \BibitemOpen
  \bibfield  {author} {\bibinfo {author} {\bibfnamefont {A.}~\bibnamefont
  {Roy}}\ and\ \bibinfo {author} {\bibfnamefont {A.~J.}\ \bibnamefont
  {Scott}},\ }\bibfield  {title} {\bibinfo {title} {Unitary designs and
  codes},\ }\href@noop {} {\bibfield  {journal} {\bibinfo  {journal} {Des.
  Codes Cryptogr.}\ }\textbf {\bibinfo {volume} {53}},\ \bibinfo {pages} {13}
  (\bibinfo {year} {2009})}\BibitemShut {NoStop}%
\bibitem [{\citenamefont {Knill}\ \emph {et~al.}(2008)\citenamefont {Knill},
  \citenamefont {Leibfried}, \citenamefont {Reichle}, \citenamefont {Britton},
  \citenamefont {Blakestad}, \citenamefont {Jost}, \citenamefont {Langer},
  \citenamefont {Ozeri}, \citenamefont {Seidelin},\ and\ \citenamefont
  {Wineland}}]{KnilLRB08}%
  \BibitemOpen
  \bibfield  {author} {\bibinfo {author} {\bibfnamefont {E.}~\bibnamefont
  {Knill}}, \bibinfo {author} {\bibfnamefont {D.}~\bibnamefont {Leibfried}},
  \bibinfo {author} {\bibfnamefont {R.}~\bibnamefont {Reichle}}, \bibinfo
  {author} {\bibfnamefont {J.}~\bibnamefont {Britton}}, \bibinfo {author}
  {\bibfnamefont {R.~B.}\ \bibnamefont {Blakestad}}, \bibinfo {author}
  {\bibfnamefont {J.~D.}\ \bibnamefont {Jost}}, \bibinfo {author}
  {\bibfnamefont {C.}~\bibnamefont {Langer}}, \bibinfo {author} {\bibfnamefont
  {R.}~\bibnamefont {Ozeri}}, \bibinfo {author} {\bibfnamefont
  {S.}~\bibnamefont {Seidelin}},\ and\ \bibinfo {author} {\bibfnamefont
  {D.~J.}\ \bibnamefont {Wineland}},\ }\bibfield  {title} {\bibinfo {title}
  {Randomized benchmarking of quantum gates},\ }\href@noop {} {\bibfield
  {journal} {\bibinfo  {journal} {Phys. Rev. A}\ }\textbf {\bibinfo {volume}
  {77}},\ \bibinfo {pages} {012307} (\bibinfo {year} {2008})}\BibitemShut
  {NoStop}%
\bibitem [{\citenamefont {Magesan}\ \emph {et~al.}(2011)\citenamefont
  {Magesan}, \citenamefont {Gambetta},\ and\ \citenamefont
  {Emerson}}]{MageGE11}%
  \BibitemOpen
  \bibfield  {author} {\bibinfo {author} {\bibfnamefont {E.}~\bibnamefont
  {Magesan}}, \bibinfo {author} {\bibfnamefont {J.~M.}\ \bibnamefont
  {Gambetta}},\ and\ \bibinfo {author} {\bibfnamefont {J.}~\bibnamefont
  {Emerson}},\ }\bibfield  {title} {\bibinfo {title} {Scalable and robust
  randomized benchmarking of quantum processes},\ }\href@noop {} {\bibfield
  {journal} {\bibinfo  {journal} {Phys. Rev. Lett.}\ }\textbf {\bibinfo
  {volume} {106}},\ \bibinfo {pages} {180504} (\bibinfo {year}
  {2011})}\BibitemShut {NoStop}%
\bibitem [{\citenamefont {Hayashi}\ \emph {et~al.}(2005)\citenamefont
  {Hayashi}, \citenamefont {Hashimoto},\ and\ \citenamefont
  {Horibe}}]{HayaHH05}%
  \BibitemOpen
  \bibfield  {author} {\bibinfo {author} {\bibfnamefont {A.}~\bibnamefont
  {Hayashi}}, \bibinfo {author} {\bibfnamefont {T.}~\bibnamefont {Hashimoto}},\
  and\ \bibinfo {author} {\bibfnamefont {M.}~\bibnamefont {Horibe}},\
  }\bibfield  {title} {\bibinfo {title} {Reexamination of optimal quantum state
  estimation of pure states},\ }\href@noop {} {\bibfield  {journal} {\bibinfo
  {journal} {Phys. Rev. A}\ }\textbf {\bibinfo {volume} {72}},\ \bibinfo
  {pages} {032325} (\bibinfo {year} {2005})}\BibitemShut {NoStop}%
\bibitem [{\citenamefont {Scott}(2008)}]{Scot08}%
  \BibitemOpen
  \bibfield  {author} {\bibinfo {author} {\bibfnamefont {A.~J.}\ \bibnamefont
  {Scott}},\ }\bibfield  {title} {\bibinfo {title} {Optimizing quantum process
  tomography with unitary 2-designs},\ }\href@noop {} {\bibfield  {journal}
  {\bibinfo  {journal} {J. Phys. A: Math. Theor.}\ }\textbf {\bibinfo {volume}
  {41}},\ \bibinfo {pages} {055308} (\bibinfo {year} {2008})}\BibitemShut
  {NoStop}%
\bibitem [{\citenamefont {Zhu}\ and\ \citenamefont {Englert}(2011)}]{ZhuE11}%
  \BibitemOpen
  \bibfield  {author} {\bibinfo {author} {\bibfnamefont {H.}~\bibnamefont
  {Zhu}}\ and\ \bibinfo {author} {\bibfnamefont {B.-G.}\ \bibnamefont
  {Englert}},\ }\bibfield  {title} {\bibinfo {title} {Quantum state tomography
  with fully symmetric measurements and product measurements},\ }\href@noop {}
  {\bibfield  {journal} {\bibinfo  {journal} {Phys. Rev. A}\ }\textbf {\bibinfo
  {volume} {84}},\ \bibinfo {pages} {022327} (\bibinfo {year}
  {2011})}\BibitemShut {NoStop}%
\bibitem [{\citenamefont {Zhu}(2012)}]{Zhu12the}%
  \BibitemOpen
  \bibfield  {author} {\bibinfo {author} {\bibfnamefont {H.}~\bibnamefont
  {Zhu}},\ }\emph {\bibinfo {title} {Quantum State Estimation and Symmetric
  Informationally Complete {POM}s}},\ \href
  {http://scholarbank.nus.edu.sg/bitstream/10635/35247/1/ZhuHJthesis.pdf}
  {Ph.D. thesis},\ \bibinfo  {school} {National University of Singapore}
  (\bibinfo {year} {2012})\BibitemShut {NoStop}%
\bibitem [{\citenamefont {Zhu}(2014)}]{Zhu14IOC}%
  \BibitemOpen
  \bibfield  {author} {\bibinfo {author} {\bibfnamefont {H.}~\bibnamefont
  {Zhu}},\ }\bibfield  {title} {\bibinfo {title} {Quantum state estimation with
  informationally overcomplete measurements},\ }\href@noop {} {\bibfield
  {journal} {\bibinfo  {journal} {Phys. Rev. A}\ }\textbf {\bibinfo {volume}
  {90}},\ \bibinfo {pages} {012115} (\bibinfo {year} {2014})}\BibitemShut
  {NoStop}%
\bibitem [{\citenamefont {Zhu}\ and\ \citenamefont {Hayashi}(2019)}]{ZhuH19O}%
  \BibitemOpen
  \bibfield  {author} {\bibinfo {author} {\bibfnamefont {H.}~\bibnamefont
  {Zhu}}\ and\ \bibinfo {author} {\bibfnamefont {M.}~\bibnamefont {Hayashi}},\
  }\bibfield  {title} {\bibinfo {title} {Optimal verification and fidelity
  estimation of maximally entangled states},\ }\href@noop {} {\bibfield
  {journal} {\bibinfo  {journal} {Phys. Rev. A}\ }\textbf {\bibinfo {volume}
  {99}},\ \bibinfo {pages} {052346} (\bibinfo {year} {2019})}\BibitemShut
  {NoStop}%
\bibitem [{\citenamefont {Li}\ \emph {et~al.}(2019)\citenamefont {Li},
  \citenamefont {Han},\ and\ \citenamefont {Zhu}}]{LiHZ19}%
  \BibitemOpen
  \bibfield  {author} {\bibinfo {author} {\bibfnamefont {Z.}~\bibnamefont
  {Li}}, \bibinfo {author} {\bibfnamefont {Y.-G.}\ \bibnamefont {Han}},\ and\
  \bibinfo {author} {\bibfnamefont {H.}~\bibnamefont {Zhu}},\ }\bibfield
  {title} {\bibinfo {title} {Efficient verification of bipartite pure states},\
  }\href@noop {} {\bibfield  {journal} {\bibinfo  {journal} {Phys. Rev. A}\
  }\textbf {\bibinfo {volume} {100}},\ \bibinfo {pages} {032316} (\bibinfo
  {year} {2019})}\BibitemShut {NoStop}%
\bibitem [{\citenamefont {Li}\ \emph {et~al.}(2020)\citenamefont {Li},
  \citenamefont {Han},\ and\ \citenamefont {Zhu}}]{LiHZ20}%
  \BibitemOpen
  \bibfield  {author} {\bibinfo {author} {\bibfnamefont {Z.}~\bibnamefont
  {Li}}, \bibinfo {author} {\bibfnamefont {Y.-G.}\ \bibnamefont {Han}},\ and\
  \bibinfo {author} {\bibfnamefont {H.}~\bibnamefont {Zhu}},\ }\bibfield
  {title} {\bibinfo {title} {Optimal verification of
  {Greenberger-Horne-Zeilinger} states},\ }\href@noop {} {\bibfield  {journal}
  {\bibinfo  {journal} {Phys. Rev. Appl.}\ }\textbf {\bibinfo {volume} {13}},\
  \bibinfo {pages} {054002} (\bibinfo {year} {2020})}\BibitemShut {NoStop}%
\bibitem [{\citenamefont {Zhu}\ and\ \citenamefont {Zhang}(2020)}]{ZhuZ20}%
  \BibitemOpen
  \bibfield  {author} {\bibinfo {author} {\bibfnamefont {H.}~\bibnamefont
  {Zhu}}\ and\ \bibinfo {author} {\bibfnamefont {H.}~\bibnamefont {Zhang}},\
  }\bibfield  {title} {\bibinfo {title} {Efficient verification of quantum
  gates with local operations},\ }\href@noop {} {\bibfield  {journal} {\bibinfo
   {journal} {Phys. Rev. A}\ }\textbf {\bibinfo {volume} {101}},\ \bibinfo
  {pages} {042316} (\bibinfo {year} {2020})}\BibitemShut {NoStop}%
\bibitem [{\citenamefont {DiVincenzo}\ \emph {et~al.}(2002)\citenamefont
  {DiVincenzo}, \citenamefont {Leung},\ and\ \citenamefont
  {Terhal}}]{DiViLT02}%
  \BibitemOpen
  \bibfield  {author} {\bibinfo {author} {\bibfnamefont {D.~P.}\ \bibnamefont
  {DiVincenzo}}, \bibinfo {author} {\bibfnamefont {D.}~\bibnamefont {Leung}},\
  and\ \bibinfo {author} {\bibfnamefont {B.}~\bibnamefont {Terhal}},\
  }\bibfield  {title} {\bibinfo {title} {Quantum data hiding},\ }\href@noop {}
  {\bibfield  {journal} {\bibinfo  {journal} {IEEE Trans. Inf. Theory}\
  }\textbf {\bibinfo {volume} {48}},\ \bibinfo {pages} {580} (\bibinfo {year}
  {2002})}\BibitemShut {NoStop}%
\bibitem [{\citenamefont {Ambainis}\ \emph {et~al.}(2009)\citenamefont
  {Ambainis}, \citenamefont {Bouda},\ and\ \citenamefont {Winter}}]{AmbaBW09}%
  \BibitemOpen
  \bibfield  {author} {\bibinfo {author} {\bibfnamefont {A.}~\bibnamefont
  {Ambainis}}, \bibinfo {author} {\bibfnamefont {J.}~\bibnamefont {Bouda}},\
  and\ \bibinfo {author} {\bibfnamefont {A.}~\bibnamefont {Winter}},\
  }\bibfield  {title} {\bibinfo {title} {Nonmalleable encryption of quantum
  information},\ }\href
  {http://scitation.aip.org/content/aip/journal/jmp/50/4/10.1063/1.3094756}
  {\bibfield  {journal} {\bibinfo  {journal} {J. Math. Phys.}\ }\textbf
  {\bibinfo {volume} {50}},\ \bibinfo {eid} {042106} (\bibinfo {year}
  {2009})}\BibitemShut {NoStop}%
\bibitem [{\citenamefont {Matthews}\ \emph {et~al.}(2009)\citenamefont
  {Matthews}, \citenamefont {Wehner},\ and\ \citenamefont {Winter}}]{MattWW09}%
  \BibitemOpen
  \bibfield  {author} {\bibinfo {author} {\bibfnamefont {W.}~\bibnamefont
  {Matthews}}, \bibinfo {author} {\bibfnamefont {S.}~\bibnamefont {Wehner}},\
  and\ \bibinfo {author} {\bibfnamefont {A.}~\bibnamefont {Winter}},\
  }\bibfield  {title} {\bibinfo {title} {Distinguishability of quantum states
  under restricted families of measurements with an application to quantum data
  hiding},\ }\href@noop {} {\bibfield  {journal} {\bibinfo  {journal} {Commun.
  Math. Phys.}\ }\textbf {\bibinfo {volume} {291}},\ \bibinfo {pages} {813}
  (\bibinfo {year} {2009})}\BibitemShut {NoStop}%
\bibitem [{\citenamefont {Abeyesinghe}\ \emph {et~al.}(2009)\citenamefont
  {Abeyesinghe}, \citenamefont {Devetak}, \citenamefont {Hayden},\ and\
  \citenamefont {Winter}}]{AbeyDHW09}%
  \BibitemOpen
  \bibfield  {author} {\bibinfo {author} {\bibfnamefont {A.}~\bibnamefont
  {Abeyesinghe}}, \bibinfo {author} {\bibfnamefont {I.}~\bibnamefont
  {Devetak}}, \bibinfo {author} {\bibfnamefont {P.}~\bibnamefont {Hayden}},\
  and\ \bibinfo {author} {\bibfnamefont {A.}~\bibnamefont {Winter}},\
  }\bibfield  {title} {\bibinfo {title} {The mother of all protocols:
  restructuring quantum information’s family tree},\ }\href
  {https://doi.org/10.1098/rspa.2009.0202} {\bibfield  {journal} {\bibinfo
  {journal} {Proc. R. Soc. A}\ }\textbf {\bibinfo {volume} {465}},\ \bibinfo
  {pages} {2537} (\bibinfo {year} {2009})}\BibitemShut {NoStop}%
\bibitem [{\citenamefont {Szehr}\ \emph {et~al.}(2013)\citenamefont {Szehr},
  \citenamefont {Dupuis}, \citenamefont {Tomamichel},\ and\ \citenamefont
  {Renner}}]{SzehDTR13}%
  \BibitemOpen
  \bibfield  {author} {\bibinfo {author} {\bibfnamefont {O.}~\bibnamefont
  {Szehr}}, \bibinfo {author} {\bibfnamefont {F.}~\bibnamefont {Dupuis}},
  \bibinfo {author} {\bibfnamefont {M.}~\bibnamefont {Tomamichel}},\ and\
  \bibinfo {author} {\bibfnamefont {R.}~\bibnamefont {Renner}},\ }\bibfield
  {title} {\bibinfo {title} {Decoupling with unitary approximate two-designs},\
  }\href {http://stacks.iop.org/1367-2630/15/i=5/a=053022} {\bibfield
  {journal} {\bibinfo  {journal} {New J. Phys.}\ }\textbf {\bibinfo {volume}
  {15}},\ \bibinfo {pages} {053022} (\bibinfo {year} {2013})}\BibitemShut
  {NoStop}%
\bibitem [{\citenamefont {Harrow}\ and\ \citenamefont {Low}(2009)}]{HarrL09}%
  \BibitemOpen
  \bibfield  {author} {\bibinfo {author} {\bibfnamefont {A.~W.}\ \bibnamefont
  {Harrow}}\ and\ \bibinfo {author} {\bibfnamefont {R.~A.}\ \bibnamefont
  {Low}},\ }\bibfield  {title} {\bibinfo {title} {Random quantum circuits are
  approximate 2-designs},\ }\href@noop {} {\bibfield  {journal} {\bibinfo
  {journal} {Commun. Math. Phys.}\ }\textbf {\bibinfo {volume} {291}},\
  \bibinfo {pages} {257} (\bibinfo {year} {2009})}\BibitemShut {NoStop}%
\bibitem [{\citenamefont {Brand\~ao}\ and\ \citenamefont
  {Horodecki}(2013)}]{BranH13}%
  \BibitemOpen
  \bibfield  {author} {\bibinfo {author} {\bibfnamefont {F.~G. S.~L.}\
  \bibnamefont {Brand\~ao}}\ and\ \bibinfo {author} {\bibfnamefont
  {M.}~\bibnamefont {Horodecki}},\ }\bibfield  {title} {\bibinfo {title}
  {Exponential quantum speed-ups are generic},\ }\href@noop {} {\bibfield
  {journal} {\bibinfo  {journal} {Quant. Inf. Comp.}\ }\textbf {\bibinfo
  {volume} {13}},\ \bibinfo {pages} {0901} (\bibinfo {year}
  {2013})}\BibitemShut {NoStop}%
\bibitem [{\citenamefont {Brand\~ao}\ \emph {et~al.}(2016)\citenamefont
  {Brand\~ao}, \citenamefont {Harrow},\ and\ \citenamefont
  {Horodecki}}]{BranHH16L}%
  \BibitemOpen
  \bibfield  {author} {\bibinfo {author} {\bibfnamefont {F.~G. S.~L.}\
  \bibnamefont {Brand\~ao}}, \bibinfo {author} {\bibfnamefont {A.~W.}\
  \bibnamefont {Harrow}},\ and\ \bibinfo {author} {\bibfnamefont
  {M.}~\bibnamefont {Horodecki}},\ }\bibfield  {title} {\bibinfo {title} {Local
  random quantum circuits are approximate polynomial-designs},\ }\href@noop {}
  {\bibfield  {journal} {\bibinfo  {journal} {Commun. Math. Phys.}\ }\textbf
  {\bibinfo {volume} {346}},\ \bibinfo {pages} {397–434} (\bibinfo {year}
  {2016})}\BibitemShut {NoStop}%
\bibitem [{\citenamefont {Haferkamp}\ \emph {et~al.}(2022)\citenamefont
  {Haferkamp}, \citenamefont {Faist}, \citenamefont {Kothakonda}, \citenamefont
  {Eisert},\ and\ \citenamefont {Halpern}}]{HafeFKE22}%
  \BibitemOpen
  \bibfield  {author} {\bibinfo {author} {\bibfnamefont {J.}~\bibnamefont
  {Haferkamp}}, \bibinfo {author} {\bibfnamefont {P.}~\bibnamefont {Faist}},
  \bibinfo {author} {\bibfnamefont {N.~B.~T.}\ \bibnamefont {Kothakonda}},
  \bibinfo {author} {\bibfnamefont {J.}~\bibnamefont {Eisert}},\ and\ \bibinfo
  {author} {\bibfnamefont {N.~Y.}\ \bibnamefont {Halpern}},\ }\bibfield
  {title} {\bibinfo {title} {Linear growth of quantum circuit complexity},\
  }\href@noop {} {\bibfield  {journal} {\bibinfo  {journal} {Nat. Phys.}\
  }\textbf {\bibinfo {volume} {18}},\ \bibinfo {pages} {528–532} (\bibinfo
  {year} {2022})}\BibitemShut {NoStop}%
\bibitem [{\citenamefont {Hayden}\ and\ \citenamefont
  {Preskill}(2007)}]{HaydP07}%
  \BibitemOpen
  \bibfield  {author} {\bibinfo {author} {\bibfnamefont {P.}~\bibnamefont
  {Hayden}}\ and\ \bibinfo {author} {\bibfnamefont {J.}~\bibnamefont
  {Preskill}},\ }\bibfield  {title} {\bibinfo {title} {Black holes as mirrors:
  quantum information in random subsystems},\ }\href@noop {} {\bibfield
  {journal} {\bibinfo  {journal} {J. High Energy Phys.}\ }\textbf {\bibinfo
  {volume} {2007}}\bibinfo  {number} { (09)},\ \bibinfo {pages}
  {120}}\BibitemShut {NoStop}%
\bibitem [{\citenamefont {Roberts}\ and\ \citenamefont
  {Yoshida}(2017)}]{RobeY17}%
  \BibitemOpen
\bibfield  {number} {  }\bibfield  {author} {\bibinfo {author} {\bibfnamefont
  {D.~A.}\ \bibnamefont {Roberts}}\ and\ \bibinfo {author} {\bibfnamefont
  {B.}~\bibnamefont {Yoshida}},\ }\bibfield  {title} {\bibinfo {title} {Chaos
  and complexity by design},\ }\href@noop {} {\bibfield  {journal} {\bibinfo
  {journal} {J. High Energy Phys.}\ }\textbf {\bibinfo {volume} {2017}}\bibinfo
   {number} { (4)},\ \bibinfo {pages} {121}}\BibitemShut {NoStop}%
\bibitem [{\citenamefont {Liu}\ \emph {et~al.}(2018)\citenamefont {Liu},
  \citenamefont {Lloyd}, \citenamefont {Zhu},\ and\ \citenamefont
  {Zhu}}]{LiuLZZ18L}%
  \BibitemOpen
\bibfield  {number} {  }\bibfield  {author} {\bibinfo {author} {\bibfnamefont
  {Z.-W.}\ \bibnamefont {Liu}}, \bibinfo {author} {\bibfnamefont
  {S.}~\bibnamefont {Lloyd}}, \bibinfo {author} {\bibfnamefont
  {E.}~\bibnamefont {Zhu}},\ and\ \bibinfo {author} {\bibfnamefont
  {H.}~\bibnamefont {Zhu}},\ }\bibfield  {title} {\bibinfo {title}
  {Entanglement, quantum randomness, and complexity beyond scrambling},\
  }\href@noop {} {\bibfield  {journal} {\bibinfo  {journal} {J. High Energy
  Phys.}\ }\textbf {\bibinfo {volume} {2018}},\ \bibinfo {pages}
  {41}}\BibitemShut {NoStop}%
\bibitem [{\citenamefont {Huang}\ \emph {et~al.}(2020)\citenamefont {Huang},
  \citenamefont {Kueng},\ and\ \citenamefont {Preskill}}]{HuanKP20}%
  \BibitemOpen
  \bibfield  {author} {\bibinfo {author} {\bibfnamefont {H.-Y.}\ \bibnamefont
  {Huang}}, \bibinfo {author} {\bibfnamefont {R.}~\bibnamefont {Kueng}},\ and\
  \bibinfo {author} {\bibfnamefont {J.}~\bibnamefont {Preskill}},\ }\bibfield
  {title} {\bibinfo {title} {Predicting many properties of a quantum system
  from very few measurements},\ }\href@noop {} {\bibfield  {journal} {\bibinfo
  {journal} {Nat. Phys.}\ }\textbf {\bibinfo {volume} {16}},\ \bibinfo {pages}
  {1050} (\bibinfo {year} {2020})}\BibitemShut {NoStop}%
\bibitem [{\citenamefont {Eisert}\ \emph {et~al.}(2020)\citenamefont {Eisert},
  \citenamefont {Hangleiter}, \citenamefont {Walk}, \citenamefont {Roth},
  \citenamefont {Markham}, \citenamefont {Parekh}, \citenamefont {Chabaud},\
  and\ \citenamefont {Kashefi}}]{EiseHWR20}%
  \BibitemOpen
  \bibfield  {author} {\bibinfo {author} {\bibfnamefont {J.}~\bibnamefont
  {Eisert}}, \bibinfo {author} {\bibfnamefont {D.}~\bibnamefont {Hangleiter}},
  \bibinfo {author} {\bibfnamefont {N.}~\bibnamefont {Walk}}, \bibinfo {author}
  {\bibfnamefont {I.}~\bibnamefont {Roth}}, \bibinfo {author} {\bibfnamefont
  {D.}~\bibnamefont {Markham}}, \bibinfo {author} {\bibfnamefont
  {R.}~\bibnamefont {Parekh}}, \bibinfo {author} {\bibfnamefont
  {U.}~\bibnamefont {Chabaud}},\ and\ \bibinfo {author} {\bibfnamefont
  {E.}~\bibnamefont {Kashefi}},\ }\bibfield  {title} {\bibinfo {title} {Quantum
  certification and benchmarking},\ }\href@noop {} {\bibfield  {journal}
  {\bibinfo  {journal} {Nat. Rev. Phys.}\ }\textbf {\bibinfo {volume} {2}},\
  \bibinfo {pages} {382–390} (\bibinfo {year} {2020})}\BibitemShut {NoStop}%
\bibitem [{\citenamefont {Carrasco}\ \emph {et~al.}(2021)\citenamefont
  {Carrasco}, \citenamefont {Elben}, \citenamefont {Kokail}, \citenamefont
  {Kraus},\ and\ \citenamefont {Zoller}}]{CarrEKK21}%
  \BibitemOpen
  \bibfield  {author} {\bibinfo {author} {\bibfnamefont {J.}~\bibnamefont
  {Carrasco}}, \bibinfo {author} {\bibfnamefont {A.}~\bibnamefont {Elben}},
  \bibinfo {author} {\bibfnamefont {C.}~\bibnamefont {Kokail}}, \bibinfo
  {author} {\bibfnamefont {B.}~\bibnamefont {Kraus}},\ and\ \bibinfo {author}
  {\bibfnamefont {P.}~\bibnamefont {Zoller}},\ }\bibfield  {title} {\bibinfo
  {title} {Theoretical and experimental perspectives of quantum verification},\
  }\href@noop {} {\bibfield  {journal} {\bibinfo  {journal} {PRX Quantum}\
  }\textbf {\bibinfo {volume} {2}},\ \bibinfo {pages} {010102} (\bibinfo {year}
  {2021})}\BibitemShut {NoStop}%
\bibitem [{\citenamefont {Kliesch}\ and\ \citenamefont {Roth}(2021)}]{KlieR21}%
  \BibitemOpen
  \bibfield  {author} {\bibinfo {author} {\bibfnamefont {M.}~\bibnamefont
  {Kliesch}}\ and\ \bibinfo {author} {\bibfnamefont {I.}~\bibnamefont {Roth}},\
  }\bibfield  {title} {\bibinfo {title} {Theory of quantum system
  certification},\ }\href@noop {} {\bibfield  {journal} {\bibinfo  {journal}
  {PRX Quantum}\ }\textbf {\bibinfo {volume} {2}},\ \bibinfo {pages} {010201}
  (\bibinfo {year} {2021})}\BibitemShut {NoStop}%
\bibitem [{\citenamefont {Yu}\ \emph {et~al.}(2022)\citenamefont {Yu},
  \citenamefont {Shang},\ and\ \citenamefont {G\"uhne}}]{YuSG22}%
  \BibitemOpen
  \bibfield  {author} {\bibinfo {author} {\bibfnamefont {X.-D.}\ \bibnamefont
  {Yu}}, \bibinfo {author} {\bibfnamefont {J.}~\bibnamefont {Shang}},\ and\
  \bibinfo {author} {\bibfnamefont {O.}~\bibnamefont {G\"uhne}},\ }\bibfield
  {title} {\bibinfo {title} {Statistical methods for quantum state verification
  and fidelity estimation},\ }\href@noop {} {\bibfield  {journal} {\bibinfo
  {journal} {Adv. Quantum Technol.}\ }\textbf {\bibinfo {volume} {5}},\
  \bibinfo {pages} {2100126} (\bibinfo {year} {2022})}\BibitemShut {NoStop}%
\bibitem [{\citenamefont {Morris}\ \emph {et~al.}(2022)\citenamefont {Morris},
  \citenamefont {Saggio}, \citenamefont {Gočanin},\ and\ \citenamefont
  {Dakić}}]{MorrSGD22}%
  \BibitemOpen
  \bibfield  {author} {\bibinfo {author} {\bibfnamefont {J.}~\bibnamefont
  {Morris}}, \bibinfo {author} {\bibfnamefont {V.}~\bibnamefont {Saggio}},
  \bibinfo {author} {\bibfnamefont {A.}~\bibnamefont {Gočanin}},\ and\
  \bibinfo {author} {\bibfnamefont {B.}~\bibnamefont {Dakić}},\ }\bibfield
  {title} {\bibinfo {title} {Quantum verification and estimation with few
  copies},\ }\href@noop {} {\bibfield  {journal} {\bibinfo  {journal} {Adv.
  Quantum Technol.}\ }\textbf {\bibinfo {volume} {5}},\ \bibinfo {pages}
  {2100118} (\bibinfo {year} {2022})}\BibitemShut {NoStop}%
\bibitem [{\citenamefont {Elben}\ \emph {et~al.}(2023)\citenamefont {Elben},
  \citenamefont {Flammia}, \citenamefont {Huang}, \citenamefont {Kueng},
  \citenamefont {Preskill}, \citenamefont {Vermersch},\ and\ \citenamefont
  {Zoller}}]{ElbeFHK23}%
  \BibitemOpen
  \bibfield  {author} {\bibinfo {author} {\bibfnamefont {A.}~\bibnamefont
  {Elben}}, \bibinfo {author} {\bibfnamefont {S.~T.}\ \bibnamefont {Flammia}},
  \bibinfo {author} {\bibfnamefont {H.-Y.}\ \bibnamefont {Huang}}, \bibinfo
  {author} {\bibfnamefont {R.}~\bibnamefont {Kueng}}, \bibinfo {author}
  {\bibfnamefont {J.}~\bibnamefont {Preskill}}, \bibinfo {author}
  {\bibfnamefont {B.}~\bibnamefont {Vermersch}},\ and\ \bibinfo {author}
  {\bibfnamefont {P.}~\bibnamefont {Zoller}},\ }\bibfield  {title} {\bibinfo
  {title} {The randomized measurement toolbox},\ }\href@noop {} {\bibfield
  {journal} {\bibinfo  {journal} {Nat. Rev. Phys.}\ }\textbf {\bibinfo {volume}
  {5}},\ \bibinfo {pages} {9–24} (\bibinfo {year} {2023})}\BibitemShut
  {NoStop}%
\bibitem [{\citenamefont {Bolt}\ \emph
  {et~al.}(1961{\natexlab{a}})\citenamefont {Bolt}, \citenamefont {Room},\ and\
  \citenamefont {Wall}}]{BoltRW61I}%
  \BibitemOpen
  \bibfield  {author} {\bibinfo {author} {\bibfnamefont {B.}~\bibnamefont
  {Bolt}}, \bibinfo {author} {\bibfnamefont {T.~G.}\ \bibnamefont {Room}},\
  and\ \bibinfo {author} {\bibfnamefont {G.~E.}\ \bibnamefont {Wall}},\
  }\bibfield  {title} {\bibinfo {title} {On the {Clifford} collineation,
  transform and similarity groups.~{I}.},\ }\href@noop {} {\bibfield  {journal}
  {\bibinfo  {journal} {J. Austral. Math. Soc.}\ }\textbf {\bibinfo {volume}
  {2}},\ \bibinfo {pages} {60} (\bibinfo {year}
  {1961}{\natexlab{a}})}\BibitemShut {NoStop}%
\bibitem [{\citenamefont {Bolt}\ \emph
  {et~al.}(1961{\natexlab{b}})\citenamefont {Bolt}, \citenamefont {Room},\ and\
  \citenamefont {Wall}}]{BoltRW61II}%
  \BibitemOpen
  \bibfield  {author} {\bibinfo {author} {\bibfnamefont {B.}~\bibnamefont
  {Bolt}}, \bibinfo {author} {\bibfnamefont {T.~G.}\ \bibnamefont {Room}},\
  and\ \bibinfo {author} {\bibfnamefont {G.~E.}\ \bibnamefont {Wall}},\
  }\bibfield  {title} {\bibinfo {title} {On the {Clifford} collineation,
  transform and similarity groups.~{II}.},\ }\href@noop {} {\bibfield
  {journal} {\bibinfo  {journal} {J. Austral. Math. Soc.}\ }\textbf {\bibinfo
  {volume} {2}},\ \bibinfo {pages} {80} (\bibinfo {year}
  {1961}{\natexlab{b}})}\BibitemShut {NoStop}%
\bibitem [{\citenamefont {Gottesman}(1997)}]{Gott97the}%
  \BibitemOpen
  \bibfield  {author} {\bibinfo {author} {\bibfnamefont {D.}~\bibnamefont
  {Gottesman}},\ }\emph {\bibinfo {title} {Stabilizer Codes and Quantum Error
  Correction}},\ \href@noop {} {Ph.D. thesis},\ \bibinfo  {school} {California
  Institute of Technology} (\bibinfo {year} {1997}),\ \bibinfo {note}
  {available at \url{http://arxiv.org/abs/quant-ph/9705052}}\BibitemShut
  {NoStop}%
\bibitem [{\citenamefont {Bengtsson}\ and\ \citenamefont
  {{\.{Z}}yczkowski}(2017)}]{BengZ17book}%
  \BibitemOpen
  \bibfield  {author} {\bibinfo {author} {\bibfnamefont {I.}~\bibnamefont
  {Bengtsson}}\ and\ \bibinfo {author} {\bibfnamefont {K.}~\bibnamefont
  {{\.{Z}}yczkowski}},\ }\href@noop {} {\emph {\bibinfo {title} {Geometry of
  Quantum States: An Introduction to Quantum Entanglement}}},\ \bibinfo
  {edition} {2nd}\ ed.\ (\bibinfo  {publisher} {Cambridge University Press},\
  \bibinfo {address} {Cambridge, UK},\ \bibinfo {year} {2017})\BibitemShut
  {NoStop}%
\bibitem [{\citenamefont {Gottesman}(1999)}]{Gott99}%
  \BibitemOpen
  \bibfield  {author} {\bibinfo {author} {\bibfnamefont {D.}~\bibnamefont
  {Gottesman}},\ }\bibfield  {title} {\bibinfo {title} {Fault-tolerant quantum
  computation with higher-dimensional systems},\ }in\ \href@noop {} {\emph
  {\bibinfo {booktitle} {Quantum Computing and Quantum Communications}}},\
  \bibinfo {editor} {edited by\ \bibinfo {editor} {\bibfnamefont {C.~P.}\
  \bibnamefont {Williams}}}\ (\bibinfo  {publisher} {Springer Berlin
  Heidelberg},\ \bibinfo {year} {1999})\ pp.\ \bibinfo {pages}
  {302--313}\BibitemShut {NoStop}%
\bibitem [{\citenamefont {Gottesman}\ and\ \citenamefont
  {Chuang}(1999)}]{GottC99}%
  \BibitemOpen
  \bibfield  {author} {\bibinfo {author} {\bibfnamefont {D.}~\bibnamefont
  {Gottesman}}\ and\ \bibinfo {author} {\bibfnamefont {I.~L.}\ \bibnamefont
  {Chuang}},\ }\bibfield  {title} {\bibinfo {title} {Demonstrating the
  viability of universal quantum computation using teleportation and
  single-qubit operations},\ }\href@noop {} {\bibfield  {journal} {\bibinfo
  {journal} {Nature}\ }\textbf {\bibinfo {volume} {402}},\ \bibinfo {pages}
  {390} (\bibinfo {year} {1999})}\BibitemShut {NoStop}%
\bibitem [{\citenamefont {Bravyi}\ and\ \citenamefont
  {Kitaev}(2005)}]{BravK05}%
  \BibitemOpen
  \bibfield  {author} {\bibinfo {author} {\bibfnamefont {S.}~\bibnamefont
  {Bravyi}}\ and\ \bibinfo {author} {\bibfnamefont {A.}~\bibnamefont
  {Kitaev}},\ }\bibfield  {title} {\bibinfo {title} {Universal quantum
  computation with ideal {C}lifford gates and noisy ancillas},\ }\href@noop {}
  {\bibfield  {journal} {\bibinfo  {journal} {Phys. Rev. A}\ }\textbf {\bibinfo
  {volume} {71}},\ \bibinfo {pages} {022316} (\bibinfo {year}
  {2005})}\BibitemShut {NoStop}%
\bibitem [{\citenamefont {Nielsen}\ and\ \citenamefont
  {Chuang}(2010)}]{NielC10book}%
  \BibitemOpen
  \bibfield  {author} {\bibinfo {author} {\bibfnamefont {M.~A.}\ \bibnamefont
  {Nielsen}}\ and\ \bibinfo {author} {\bibfnamefont {I.~L.}\ \bibnamefont
  {Chuang}},\ }\href@noop {} {\emph {\bibinfo {title} {Quantum Computation and
  Quantum Information}}},\ \bibinfo {edition} {2nd}\ ed.\ (\bibinfo
  {publisher} {Cambridge University Press},\ \bibinfo {address} {Cambridge,
  UK},\ \bibinfo {year} {2010})\BibitemShut {NoStop}%
\bibitem [{\citenamefont {Terhal}(2015)}]{Terh15}%
  \BibitemOpen
  \bibfield  {author} {\bibinfo {author} {\bibfnamefont {B.~M.}\ \bibnamefont
  {Terhal}},\ }\bibfield  {title} {\bibinfo {title} {Quantum error correction
  for quantum memories},\ }\href {https://doi.org/10.1103/RevModPhys.87.307}
  {\bibfield  {journal} {\bibinfo  {journal} {Rev. Mod. Phys.}\ }\textbf
  {\bibinfo {volume} {87}},\ \bibinfo {pages} {307} (\bibinfo {year}
  {2015})}\BibitemShut {NoStop}%
\bibitem [{\citenamefont {Bravyi}\ and\ \citenamefont
  {Gosset}(2016)}]{BravG16}%
  \BibitemOpen
  \bibfield  {author} {\bibinfo {author} {\bibfnamefont {S.}~\bibnamefont
  {Bravyi}}\ and\ \bibinfo {author} {\bibfnamefont {D.}~\bibnamefont
  {Gosset}},\ }\bibfield  {title} {\bibinfo {title} {Improved classical
  simulation of quantum circuits dominated by {Clifford} gates},\ }\href
  {https://doi.org/10.1103/PhysRevLett.116.250501} {\bibfield  {journal}
  {\bibinfo  {journal} {Phys. Rev. Lett.}\ }\textbf {\bibinfo {volume} {116}},\
  \bibinfo {pages} {250501} (\bibinfo {year} {2016})}\BibitemShut {NoStop}%
\bibitem [{\citenamefont {Bravyi}\ \emph {et~al.}(2019)\citenamefont {Bravyi},
  \citenamefont {Browne}, \citenamefont {Calpin}, \citenamefont {Campbell},
  \citenamefont {Gosset},\ and\ \citenamefont {Howard}}]{BravBCC19}%
  \BibitemOpen
  \bibfield  {author} {\bibinfo {author} {\bibfnamefont {S.}~\bibnamefont
  {Bravyi}}, \bibinfo {author} {\bibfnamefont {D.}~\bibnamefont {Browne}},
  \bibinfo {author} {\bibfnamefont {P.}~\bibnamefont {Calpin}}, \bibinfo
  {author} {\bibfnamefont {E.}~\bibnamefont {Campbell}}, \bibinfo {author}
  {\bibfnamefont {D.}~\bibnamefont {Gosset}},\ and\ \bibinfo {author}
  {\bibfnamefont {M.}~\bibnamefont {Howard}},\ }\bibfield  {title} {\bibinfo
  {title} {Simulation of quantum circuits by low-rank stabilizer
  decompositions},\ }\href {https://doi.org/10.22331/q-2019-09-02-181}
  {\bibfield  {journal} {\bibinfo  {journal} {{Quantum}}\ }\textbf {\bibinfo
  {volume} {3}},\ \bibinfo {pages} {181} (\bibinfo {year} {2019})}\BibitemShut
  {NoStop}%
\bibitem [{\citenamefont {Heinrich}(2021)}]{Hein21the}%
  \BibitemOpen
  \bibfield  {author} {\bibinfo {author} {\bibfnamefont {M.}~\bibnamefont
  {Heinrich}},\ }\emph {\bibinfo {title} {On Stabiliser Techniques and Their
  Application to Simulation and Certiﬁcation of Quantum Devices}},\
  \href@noop {} {Ph.D. thesis},\ \bibinfo  {school} {University of Cologne}
  (\bibinfo {year} {2021})\BibitemShut {NoStop}%
\bibitem [{\citenamefont {Pashayan}\ \emph {et~al.}(2022)\citenamefont
  {Pashayan}, \citenamefont {Reardon-Smith}, \citenamefont {Korzekwa},\ and\
  \citenamefont {Bartlett}}]{PashRKB22}%
  \BibitemOpen
  \bibfield  {author} {\bibinfo {author} {\bibfnamefont {H.}~\bibnamefont
  {Pashayan}}, \bibinfo {author} {\bibfnamefont {O.}~\bibnamefont
  {Reardon-Smith}}, \bibinfo {author} {\bibfnamefont {K.}~\bibnamefont
  {Korzekwa}},\ and\ \bibinfo {author} {\bibfnamefont {S.~D.}\ \bibnamefont
  {Bartlett}},\ }\bibfield  {title} {\bibinfo {title} {Fast estimation of
  outcome probabilities for quantum circuits},\ }\href
  {https://doi.org/10.1103/PRXQuantum.3.020361} {\bibfield  {journal} {\bibinfo
   {journal} {PRX Quantum}\ }\textbf {\bibinfo {volume} {3}},\ \bibinfo {pages}
  {020361} (\bibinfo {year} {2022})}\BibitemShut {NoStop}%
\bibitem [{\citenamefont {Jafarzadeh}\ \emph {et~al.}(2020)\citenamefont
  {Jafarzadeh}, \citenamefont {Wu}, \citenamefont {Sanders},\ and\
  \citenamefont {Sanders}}]{JafaWSS20}%
  \BibitemOpen
  \bibfield  {author} {\bibinfo {author} {\bibfnamefont {M.}~\bibnamefont
  {Jafarzadeh}}, \bibinfo {author} {\bibfnamefont {Y.-D.}\ \bibnamefont {Wu}},
  \bibinfo {author} {\bibfnamefont {Y.~R.}\ \bibnamefont {Sanders}},\ and\
  \bibinfo {author} {\bibfnamefont {B.~C.}\ \bibnamefont {Sanders}},\
  }\bibfield  {title} {\bibinfo {title} {Randomized benchmarking for qudit
  {Clifford} gates},\ }\href {https://doi.org/10.1088/1367-2630/ab8ab1}
  {\bibfield  {journal} {\bibinfo  {journal} {New J. Phys.}\ }\textbf {\bibinfo
  {volume} {22}},\ \bibinfo {pages} {063014} (\bibinfo {year}
  {2020})}\BibitemShut {NoStop}%
\bibitem [{\citenamefont {Helsen}\ and\ \citenamefont
  {Walter}(2023)}]{HelsW23}%
  \BibitemOpen
  \bibfield  {author} {\bibinfo {author} {\bibfnamefont {J.}~\bibnamefont
  {Helsen}}\ and\ \bibinfo {author} {\bibfnamefont {M.}~\bibnamefont
  {Walter}},\ }\bibfield  {title} {\bibinfo {title} {Thrifty shadow estimation:
  Reusing quantum circuits and bounding tails},\ }\href
  {https://doi.org/10.1103/PhysRevLett.131.240602} {\bibfield  {journal}
  {\bibinfo  {journal} {Phys. Rev. Lett.}\ }\textbf {\bibinfo {volume} {131}},\
  \bibinfo {pages} {240602} (\bibinfo {year} {2023})}\BibitemShut {NoStop}%
\bibitem [{\citenamefont {Hu}\ \emph {et~al.}(2023)\citenamefont {Hu},
  \citenamefont {Choi},\ and\ \citenamefont {You}}]{HuCY23}%
  \BibitemOpen
  \bibfield  {author} {\bibinfo {author} {\bibfnamefont {H.-Y.}\ \bibnamefont
  {Hu}}, \bibinfo {author} {\bibfnamefont {S.}~\bibnamefont {Choi}},\ and\
  \bibinfo {author} {\bibfnamefont {Y.-Z.}\ \bibnamefont {You}},\ }\bibfield
  {title} {\bibinfo {title} {Classical shadow tomography with locally scrambled
  quantum dynamics},\ }\href {https://doi.org/10.1103/PhysRevResearch.5.023027}
  {\bibfield  {journal} {\bibinfo  {journal} {Phys. Rev. Res.}\ }\textbf
  {\bibinfo {volume} {5}},\ \bibinfo {pages} {023027} (\bibinfo {year}
  {2023})}\BibitemShut {NoStop}%
\bibitem [{\citenamefont {Ippoliti}\ \emph {et~al.}(2023)\citenamefont
  {Ippoliti}, \citenamefont {Li}, \citenamefont {Rakovszky},\ and\
  \citenamefont {Khemani}}]{IppoLRK23}%
  \BibitemOpen
  \bibfield  {author} {\bibinfo {author} {\bibfnamefont {M.}~\bibnamefont
  {Ippoliti}}, \bibinfo {author} {\bibfnamefont {Y.}~\bibnamefont {Li}},
  \bibinfo {author} {\bibfnamefont {T.}~\bibnamefont {Rakovszky}},\ and\
  \bibinfo {author} {\bibfnamefont {V.}~\bibnamefont {Khemani}},\ }\bibfield
  {title} {\bibinfo {title} {Operator relaxation and the optimal depth of
  classical shadows},\ }\href@noop {} {\bibfield  {journal} {\bibinfo
  {journal} {Phys. Rev. Lett.}\ }\textbf {\bibinfo {volume} {130}},\ \bibinfo
  {pages} {230403} (\bibinfo {year} {2023})}\BibitemShut {NoStop}%
\bibitem [{\citenamefont {Bertoni}\ \emph {et~al.}(2024)\citenamefont
  {Bertoni}, \citenamefont {Haferkamp}, \citenamefont {Hinsche}, \citenamefont
  {Ioannou}, \citenamefont {Eisert},\ and\ \citenamefont
  {Pashayan}}]{BertHHI24}%
  \BibitemOpen
  \bibfield  {author} {\bibinfo {author} {\bibfnamefont {C.}~\bibnamefont
  {Bertoni}}, \bibinfo {author} {\bibfnamefont {J.}~\bibnamefont {Haferkamp}},
  \bibinfo {author} {\bibfnamefont {M.}~\bibnamefont {Hinsche}}, \bibinfo
  {author} {\bibfnamefont {M.}~\bibnamefont {Ioannou}}, \bibinfo {author}
  {\bibfnamefont {J.}~\bibnamefont {Eisert}},\ and\ \bibinfo {author}
  {\bibfnamefont {H.}~\bibnamefont {Pashayan}},\ }\bibfield  {title} {\bibinfo
  {title} {Shallow shadows: Expectation estimation using low-depth random
  {Clifford} circuits},\ }\href
  {https://doi.org/10.1103/PhysRevLett.133.020602} {\bibfield  {journal}
  {\bibinfo  {journal} {Phys. Rev. Lett.}\ }\textbf {\bibinfo {volume} {133}},\
  \bibinfo {pages} {020602} (\bibinfo {year} {2024})}\BibitemShut {NoStop}%
\bibitem [{\citenamefont {Levy}\ \emph {et~al.}(2024)\citenamefont {Levy},
  \citenamefont {Luo},\ and\ \citenamefont {Clark}}]{LevyLC24}%
  \BibitemOpen
  \bibfield  {author} {\bibinfo {author} {\bibfnamefont {R.}~\bibnamefont
  {Levy}}, \bibinfo {author} {\bibfnamefont {D.}~\bibnamefont {Luo}},\ and\
  \bibinfo {author} {\bibfnamefont {B.~K.}\ \bibnamefont {Clark}},\ }\bibfield
  {title} {\bibinfo {title} {Classical shadows for quantum process tomography
  on near-term quantum computers},\ }\href@noop {} {\bibfield  {journal}
  {\bibinfo  {journal} {Phys. Rev. Res.}\ }\textbf {\bibinfo {volume} {6}},\
  \bibinfo {pages} {013029} (\bibinfo {year} {2024})}\BibitemShut {NoStop}%
\bibitem [{\citenamefont {Schuster}\ \emph {et~al.}(2024)\citenamefont
  {Schuster}, \citenamefont {Haferkamp},\ and\ \citenamefont
  {Huang}}]{SchuHH24}%
  \BibitemOpen
  \bibfield  {author} {\bibinfo {author} {\bibfnamefont {T.}~\bibnamefont
  {Schuster}}, \bibinfo {author} {\bibfnamefont {J.}~\bibnamefont
  {Haferkamp}},\ and\ \bibinfo {author} {\bibfnamefont {H.-Y.}\ \bibnamefont
  {Huang}},\ }\bibfield  {title} {\bibinfo {title} {Random unitaries in
  extremely low depth},\ }\href@noop {} {\bibfield  {journal} {\bibinfo
  {journal} {arXiv:2407.07754}\ } (\bibinfo {year} {2024})}\BibitemShut
  {NoStop}%
\bibitem [{\citenamefont {Van~den Nest}\ \emph {et~al.}(2004)\citenamefont
  {Van~den Nest}, \citenamefont {Dehaene},\ and\ \citenamefont
  {De~Moor}}]{NestDM04}%
  \BibitemOpen
  \bibfield  {author} {\bibinfo {author} {\bibfnamefont {M.}~\bibnamefont
  {Van~den Nest}}, \bibinfo {author} {\bibfnamefont {J.}~\bibnamefont
  {Dehaene}},\ and\ \bibinfo {author} {\bibfnamefont {B.}~\bibnamefont
  {De~Moor}},\ }\bibfield  {title} {\bibinfo {title} {Graphical description of
  the action of local {Clifford} transformations on graph states},\ }\href
  {https://doi.org/10.1103/PhysRevA.69.022316} {\bibfield  {journal} {\bibinfo
  {journal} {Phys. Rev. A}\ }\textbf {\bibinfo {volume} {69}},\ \bibinfo
  {pages} {022316} (\bibinfo {year} {2004})}\BibitemShut {NoStop}%
\bibitem [{\citenamefont {Hein}\ \emph {et~al.}(2004)\citenamefont {Hein},
  \citenamefont {Eisert},\ and\ \citenamefont {Briegel}}]{HeinEB04}%
  \BibitemOpen
  \bibfield  {author} {\bibinfo {author} {\bibfnamefont {M.}~\bibnamefont
  {Hein}}, \bibinfo {author} {\bibfnamefont {J.}~\bibnamefont {Eisert}},\ and\
  \bibinfo {author} {\bibfnamefont {H.~J.}\ \bibnamefont {Briegel}},\
  }\bibfield  {title} {\bibinfo {title} {Multiparty entanglement in graph
  states},\ }\href@noop {} {\bibfield  {journal} {\bibinfo  {journal} {Phys.
  Rev. A}\ }\textbf {\bibinfo {volume} {69}},\ \bibinfo {pages} {062311}
  (\bibinfo {year} {2004})}\BibitemShut {NoStop}%
\bibitem [{\citenamefont {Nezami}\ and\ \citenamefont
  {Walter}(2020)}]{NezaW20}%
  \BibitemOpen
  \bibfield  {author} {\bibinfo {author} {\bibfnamefont {S.}~\bibnamefont
  {Nezami}}\ and\ \bibinfo {author} {\bibfnamefont {M.}~\bibnamefont
  {Walter}},\ }\bibfield  {title} {\bibinfo {title} {Multipartite entanglement
  in stabilizer tensor networks},\ }\href
  {https://doi.org/10.1103/PhysRevLett.125.241602} {\bibfield  {journal}
  {\bibinfo  {journal} {Phys. Rev. Lett.}\ }\textbf {\bibinfo {volume} {125}},\
  \bibinfo {pages} {241602} (\bibinfo {year} {2020})}\BibitemShut {NoStop}%
\bibitem [{\citenamefont {Cormick}\ \emph {et~al.}(2006)\citenamefont
  {Cormick}, \citenamefont {Galv\~ao}, \citenamefont {Gottesman}, \citenamefont
  {Paz},\ and\ \citenamefont {Pittenger}}]{CormGGP06}%
  \BibitemOpen
  \bibfield  {author} {\bibinfo {author} {\bibfnamefont {C.}~\bibnamefont
  {Cormick}}, \bibinfo {author} {\bibfnamefont {E.~F.}\ \bibnamefont
  {Galv\~ao}}, \bibinfo {author} {\bibfnamefont {D.}~\bibnamefont {Gottesman}},
  \bibinfo {author} {\bibfnamefont {J.~P.}\ \bibnamefont {Paz}},\ and\ \bibinfo
  {author} {\bibfnamefont {A.~O.}\ \bibnamefont {Pittenger}},\ }\bibfield
  {title} {\bibinfo {title} {Classicality in discrete {Wigner} functions},\
  }\href@noop {} {\bibfield  {journal} {\bibinfo  {journal} {Phys. Rev. A}\
  }\textbf {\bibinfo {volume} {73}},\ \bibinfo {pages} {012301} (\bibinfo
  {year} {2006})}\BibitemShut {NoStop}%
\bibitem [{\citenamefont {Gross}(2006)}]{Gros06}%
  \BibitemOpen
  \bibfield  {author} {\bibinfo {author} {\bibfnamefont {D.}~\bibnamefont
  {Gross}},\ }\bibfield  {title} {\bibinfo {title} {Hudson's theorem for
  finite-dimensional quantum systems},\ }\href@noop {} {\bibfield  {journal}
  {\bibinfo  {journal} {J. Math. Phys.}\ }\textbf {\bibinfo {volume} {47}},\
  \bibinfo {eid} {122107} (\bibinfo {year} {2006})}\BibitemShut {NoStop}%
\bibitem [{\citenamefont {Zhu}(2016)}]{Zhu16P}%
  \BibitemOpen
  \bibfield  {author} {\bibinfo {author} {\bibfnamefont {H.}~\bibnamefont
  {Zhu}},\ }\bibfield  {title} {\bibinfo {title} {Permutation symmetry
  determines the discrete {Wigner} function},\ }\href@noop {} {\bibfield
  {journal} {\bibinfo  {journal} {Phys. Rev. Lett.}\ }\textbf {\bibinfo
  {volume} {116}},\ \bibinfo {pages} {040501} (\bibinfo {year}
  {2016})}\BibitemShut {NoStop}%
\bibitem [{\citenamefont {Gross}\ \emph {et~al.}(2021)\citenamefont {Gross},
  \citenamefont {Nezami},\ and\ \citenamefont {Walter}}]{GrosNW21}%
  \BibitemOpen
  \bibfield  {author} {\bibinfo {author} {\bibfnamefont {D.}~\bibnamefont
  {Gross}}, \bibinfo {author} {\bibfnamefont {S.}~\bibnamefont {Nezami}},\ and\
  \bibinfo {author} {\bibfnamefont {M.}~\bibnamefont {Walter}},\ }\bibfield
  {title} {\bibinfo {title} {{Schur–Weyl} duality for the {Clifford} group
  with applications: Property testing, a robust {Hudson} theorem, and de
  {Finetti} representations},\ }\href@noop {} {\bibfield  {journal} {\bibinfo
  {journal} {Commun. Math. Phys.}\ }\textbf {\bibinfo {volume} {385}},\
  \bibinfo {pages} {1325} (\bibinfo {year} {2021})}\BibitemShut {NoStop}%
\bibitem [{\citenamefont {Raussendorf}\ \emph {et~al.}(2023)\citenamefont
  {Raussendorf}, \citenamefont {Okay}, \citenamefont {Zurel},\ and\
  \citenamefont {Feldmann}}]{RausOZF23}%
  \BibitemOpen
  \bibfield  {author} {\bibinfo {author} {\bibfnamefont {R.}~\bibnamefont
  {Raussendorf}}, \bibinfo {author} {\bibfnamefont {C.}~\bibnamefont {Okay}},
  \bibinfo {author} {\bibfnamefont {M.}~\bibnamefont {Zurel}},\ and\ \bibinfo
  {author} {\bibfnamefont {P.}~\bibnamefont {Feldmann}},\ }\bibfield  {title}
  {\bibinfo {title} {The role of cohomology in quantum computation with magic
  states},\ }\href {https://doi.org/10.22331/q-2023-04-13-979} {\bibfield
  {journal} {\bibinfo  {journal} {{Quantum}}\ }\textbf {\bibinfo {volume}
  {7}},\ \bibinfo {pages} {979} (\bibinfo {year} {2023})}\BibitemShut {NoStop}%
\bibitem [{\citenamefont {Durt}\ \emph {et~al.}(2010)\citenamefont {Durt},
  \citenamefont {Englert}, \citenamefont {Bengtsson},\ and\ \citenamefont
  {{\.{Z}}yczkowski}}]{DurtEBZ10}%
  \BibitemOpen
  \bibfield  {author} {\bibinfo {author} {\bibfnamefont {T.}~\bibnamefont
  {Durt}}, \bibinfo {author} {\bibfnamefont {B.-G.}\ \bibnamefont {Englert}},
  \bibinfo {author} {\bibfnamefont {I.}~\bibnamefont {Bengtsson}},\ and\
  \bibinfo {author} {\bibfnamefont {K.}~\bibnamefont {{\.{Z}}yczkowski}},\
  }\bibfield  {title} {\bibinfo {title} {On mutually unbiased bases},\
  }\href@noop {} {\bibfield  {journal} {\bibinfo  {journal} {Int. J. Quantum
  Inf.}\ }\textbf {\bibinfo {volume} {08}},\ \bibinfo {pages} {535} (\bibinfo
  {year} {2010})}\BibitemShut {NoStop}%
\bibitem [{\citenamefont {Appleby}(2009)}]{Appl09P}%
  \BibitemOpen
  \bibfield  {author} {\bibinfo {author} {\bibfnamefont {D.~M.}\ \bibnamefont
  {Appleby}},\ }\href@noop {} {\bibinfo {title} {{Properties of the extended
  Clifford group with applications to SIC-POVMs and MUBs}}} (\bibinfo {year}
  {2009}),\ \Eprint {https://arxiv.org/abs/0909.5233} {arXiv:0909.5233}
  \BibitemShut {NoStop}%
\bibitem [{\citenamefont {Zhu}(2015)}]{Zhu15M}%
  \BibitemOpen
  \bibfield  {author} {\bibinfo {author} {\bibfnamefont {H.}~\bibnamefont
  {Zhu}},\ }\bibfield  {title} {\bibinfo {title} {Mutually unbiased bases as
  minimal {Clifford} covariant 2-designs},\ }\href@noop {} {\bibfield
  {journal} {\bibinfo  {journal} {Phys. Rev. A}\ }\textbf {\bibinfo {volume}
  {91}},\ \bibinfo {pages} {060301(R)} (\bibinfo {year} {2015})}\BibitemShut
  {NoStop}%
\bibitem [{\citenamefont {Appleby}(2005)}]{Appl05}%
  \BibitemOpen
  \bibfield  {author} {\bibinfo {author} {\bibfnamefont {D.~M.}\ \bibnamefont
  {Appleby}},\ }\bibfield  {title} {\bibinfo {title} {Symmetric informationally
  complete-positive operator valued measures and the extended {C}lifford
  group},\ }\href@noop {} {\bibfield  {journal} {\bibinfo  {journal} {J. Math.
  Phys.}\ }\textbf {\bibinfo {volume} {46}},\ \bibinfo {pages} {052107}
  (\bibinfo {year} {2005})}\BibitemShut {NoStop}%
\bibitem [{\citenamefont {Scott}\ and\ \citenamefont {Grassl}(2010)}]{ScotG10}%
  \BibitemOpen
  \bibfield  {author} {\bibinfo {author} {\bibfnamefont {A.~J.}\ \bibnamefont
  {Scott}}\ and\ \bibinfo {author} {\bibfnamefont {M.}~\bibnamefont {Grassl}},\
  }\bibfield  {title} {\bibinfo {title} {Symmetric informationally complete
  positive-operator-valued measures: A new computer study},\ }\href@noop {}
  {\bibfield  {journal} {\bibinfo  {journal} {J. Math. Phys.}\ }\textbf
  {\bibinfo {volume} {51}},\ \bibinfo {pages} {042203} (\bibinfo {year}
  {2010})}\BibitemShut {NoStop}%
\bibitem [{\citenamefont {Zhu}(2010)}]{Zhu10}%
  \BibitemOpen
  \bibfield  {author} {\bibinfo {author} {\bibfnamefont {H.}~\bibnamefont
  {Zhu}},\ }\bibfield  {title} {\bibinfo {title} {{SIC} {POVM}s and {C}lifford
  groups in prime dimensions},\ }\href@noop {} {\bibfield  {journal} {\bibinfo
  {journal} {J. Phys. A: Math. Theor.}\ }\textbf {\bibinfo {volume} {43}},\
  \bibinfo {pages} {305305} (\bibinfo {year} {2010})}\BibitemShut {NoStop}%
\bibitem [{\citenamefont {Fuchs}\ \emph {et~al.}(2017)\citenamefont {Fuchs},
  \citenamefont {Hoang},\ and\ \citenamefont {Stacey}}]{FuchHS17}%
  \BibitemOpen
  \bibfield  {author} {\bibinfo {author} {\bibfnamefont {C.~A.}\ \bibnamefont
  {Fuchs}}, \bibinfo {author} {\bibfnamefont {M.~C.}\ \bibnamefont {Hoang}},\
  and\ \bibinfo {author} {\bibfnamefont {B.~C.}\ \bibnamefont {Stacey}},\
  }\bibfield  {title} {\bibinfo {title} {The {SIC} question: History and state
  of play},\ }\href@noop {} {\bibfield  {journal} {\bibinfo  {journal}
  {Axioms}\ }\textbf {\bibinfo {volume} {6}},\ \bibinfo {pages} {21} (\bibinfo
  {year} {2017})}\BibitemShut {NoStop}%
\bibitem [{\citenamefont {Lunt}\ \emph {et~al.}(2021)\citenamefont {Lunt},
  \citenamefont {Szyniszewski},\ and\ \citenamefont {Pal}}]{LuntSP21}%
  \BibitemOpen
  \bibfield  {author} {\bibinfo {author} {\bibfnamefont {O.}~\bibnamefont
  {Lunt}}, \bibinfo {author} {\bibfnamefont {M.}~\bibnamefont {Szyniszewski}},\
  and\ \bibinfo {author} {\bibfnamefont {A.}~\bibnamefont {Pal}},\ }\bibfield
  {title} {\bibinfo {title} {Measurement-induced criticality and entanglement
  clusters: A study of one-dimensional and two-dimensional {Clifford}
  circuits},\ }\href {https://doi.org/10.1103/PhysRevB.104.155111} {\bibfield
  {journal} {\bibinfo  {journal} {Phys. Rev. B}\ }\textbf {\bibinfo {volume}
  {104}},\ \bibinfo {pages} {155111} (\bibinfo {year} {2021})}\BibitemShut
  {NoStop}%
\bibitem [{\citenamefont {Sierant}\ and\ \citenamefont
  {Turkeshi}(2023)}]{SierT23}%
  \BibitemOpen
  \bibfield  {author} {\bibinfo {author} {\bibfnamefont {P.}~\bibnamefont
  {Sierant}}\ and\ \bibinfo {author} {\bibfnamefont {X.}~\bibnamefont
  {Turkeshi}},\ }\bibfield  {title} {\bibinfo {title} {Controlling entanglement
  at absorbing state phase transitions in random circuits},\ }\href
  {https://doi.org/10.1103/PhysRevLett.130.120402} {\bibfield  {journal}
  {\bibinfo  {journal} {Phys. Rev. Lett.}\ }\textbf {\bibinfo {volume} {130}},\
  \bibinfo {pages} {120402} (\bibinfo {year} {2023})}\BibitemShut {NoStop}%
\bibitem [{\citenamefont {Piroli}\ \emph {et~al.}(2023)\citenamefont {Piroli},
  \citenamefont {Li}, \citenamefont {Vasseur},\ and\ \citenamefont
  {Nahum}}]{PiroLVN23}%
  \BibitemOpen
  \bibfield  {author} {\bibinfo {author} {\bibfnamefont {L.}~\bibnamefont
  {Piroli}}, \bibinfo {author} {\bibfnamefont {Y.}~\bibnamefont {Li}}, \bibinfo
  {author} {\bibfnamefont {R.}~\bibnamefont {Vasseur}},\ and\ \bibinfo {author}
  {\bibfnamefont {A.}~\bibnamefont {Nahum}},\ }\bibfield  {title} {\bibinfo
  {title} {Triviality of quantum trajectories close to a directed percolation
  transition},\ }\href {https://doi.org/10.1103/PhysRevB.107.224303} {\bibfield
   {journal} {\bibinfo  {journal} {Phys. Rev. B}\ }\textbf {\bibinfo {volume}
  {107}},\ \bibinfo {pages} {224303} (\bibinfo {year} {2023})}\BibitemShut
  {NoStop}%
\bibitem [{\citenamefont {Webb}(2016)}]{Webb16}%
  \BibitemOpen
  \bibfield  {author} {\bibinfo {author} {\bibfnamefont {Z.}~\bibnamefont
  {Webb}},\ }\bibfield  {title} {\bibinfo {title} {The {C}lifford group forms a
  unitary 3-design},\ }\href@noop {} {\bibfield  {journal} {\bibinfo  {journal}
  {Quantum Inf. Comput.}\ }\textbf {\bibinfo {volume} {16}},\ \bibinfo {pages}
  {1379} (\bibinfo {year} {2016})}\BibitemShut {NoStop}%
\bibitem [{\citenamefont {Zhu}(2017)}]{Zhu17MC}%
  \BibitemOpen
  \bibfield  {author} {\bibinfo {author} {\bibfnamefont {H.}~\bibnamefont
  {Zhu}},\ }\bibfield  {title} {\bibinfo {title} {Multiqubit {Clifford} groups
  are unitary 3-designs},\ }\href@noop {} {\bibfield  {journal} {\bibinfo
  {journal} {Phys. Rev. A}\ }\textbf {\bibinfo {volume} {96}},\ \bibinfo
  {pages} {062336} (\bibinfo {year} {2017})}\BibitemShut {NoStop}%
\bibitem [{\citenamefont {Guralnick}\ and\ \citenamefont
  {Tiep}(2005)}]{GuraT05}%
  \BibitemOpen
  \bibfield  {author} {\bibinfo {author} {\bibfnamefont {R.~M.}\ \bibnamefont
  {Guralnick}}\ and\ \bibinfo {author} {\bibfnamefont {P.~H.}\ \bibnamefont
  {Tiep}},\ }\bibfield  {title} {\bibinfo {title} {Decompositions of small
  tensor powers and {L}arsen's conjecture},\ }\href@noop {} {\bibfield
  {journal} {\bibinfo  {journal} {Representation Theory}\ }\textbf {\bibinfo
  {volume} {9}},\ \bibinfo {pages} {138} (\bibinfo {year} {2005})}\BibitemShut
  {NoStop}%
\bibitem [{\citenamefont {Kueng}\ and\ \citenamefont {Gross}(2015)}]{KuenG13}%
  \BibitemOpen
  \bibfield  {author} {\bibinfo {author} {\bibfnamefont {R.}~\bibnamefont
  {Kueng}}\ and\ \bibinfo {author} {\bibfnamefont {D.}~\bibnamefont {Gross}},\
  }\href {http://arxiv.org/abs/1510.02767} {\bibinfo {title} {{Qubit stabilizer
  states are complex projective 3-designs}}} (\bibinfo {year} {2015}),\
  \bibinfo {note} {poster at QIP 2013},\ \Eprint
  {https://arxiv.org/abs/1510.02767} {arXiv:1510.02767} \BibitemShut {NoStop}%
\bibitem [{\citenamefont {Haferkamp}\ \emph {et~al.}(2023)\citenamefont
  {Haferkamp}, \citenamefont {Montealegre-Mora}, \citenamefont {Heinrich},
  \citenamefont {Eisert}, \citenamefont {Gross},\ and\ \citenamefont
  {Roth}}]{HafeMHE23}%
  \BibitemOpen
  \bibfield  {author} {\bibinfo {author} {\bibfnamefont {J.}~\bibnamefont
  {Haferkamp}}, \bibinfo {author} {\bibfnamefont {F.}~\bibnamefont
  {Montealegre-Mora}}, \bibinfo {author} {\bibfnamefont {M.}~\bibnamefont
  {Heinrich}}, \bibinfo {author} {\bibfnamefont {J.}~\bibnamefont {Eisert}},
  \bibinfo {author} {\bibfnamefont {D.}~\bibnamefont {Gross}},\ and\ \bibinfo
  {author} {\bibfnamefont {I.}~\bibnamefont {Roth}},\ }\bibfield  {title}
  {\bibinfo {title} {Efficient unitary designs with a system-size independent
  number of non-{Clifford} gates},\ }\href@noop {} {\bibfield  {journal}
  {\bibinfo  {journal} {Commun. Math. Phys.}\ }\textbf {\bibinfo {volume}
  {397}},\ \bibinfo {pages} {995–1041} (\bibinfo {year} {2023})}\BibitemShut
  {NoStop}%
\bibitem [{\citenamefont {Mao}\ \emph {et~al.}(2024)\citenamefont {Mao},
  \citenamefont {Yi},\ and\ \citenamefont {Zhu}}]{MaoYZ24}%
  \BibitemOpen
  \bibfield  {author} {\bibinfo {author} {\bibfnamefont {C.}~\bibnamefont
  {Mao}}, \bibinfo {author} {\bibfnamefont {C.}~\bibnamefont {Yi}},\ and\
  \bibinfo {author} {\bibfnamefont {H.}~\bibnamefont {Zhu}},\ }\href@noop {}
  {\bibinfo {title} {The magic in qudit shadow estimation based on the
  {Clifford} group}} (\bibinfo {year} {2024}),\ \bibinfo {note} {to be
  posted}\BibitemShut {NoStop}%
\bibitem [{\citenamefont {Heath-Brown}\ and\ \citenamefont
  {Patterson}(1979)}]{HeatP79}%
  \BibitemOpen
  \bibfield  {author} {\bibinfo {author} {\bibfnamefont {D.~R.}\ \bibnamefont
  {Heath-Brown}}\ and\ \bibinfo {author} {\bibfnamefont {S.~J.}\ \bibnamefont
  {Patterson}},\ }\bibfield  {title} {\bibinfo {title} {The distribution of
  {Kummer} sums at prime arguments},\ }\href
  {https://doi.org/doi:10.1515/crll.1979.310.111} {\bibfield  {journal}
  {\bibinfo  {journal} {Journal f\"ur die reine und angewandte Mathematik}\
  }\textbf {\bibinfo {volume} {1979}},\ \bibinfo {pages} {111} (\bibinfo {year}
  {1979})}\BibitemShut {NoStop}%
\bibitem [{\citenamefont {Dunn}\ and\ \citenamefont
  {Radziwi\l\l}(2021)}]{DunnR21}%
  \BibitemOpen
  \bibfield  {author} {\bibinfo {author} {\bibfnamefont {A.}~\bibnamefont
  {Dunn}}\ and\ \bibinfo {author} {\bibfnamefont {M.}~\bibnamefont
  {Radziwi\l\l}},\ }\href@noop {} {\bibinfo {title} {Bias in cubic {Gauss}
  sums: {Patterson}'s conjecture}} (\bibinfo {year} {2021}),\ \Eprint
  {https://arxiv.org/abs/arXiv:2109.07463} {arXiv:2109.07463} \BibitemShut
  {NoStop}%
\bibitem [{\citenamefont {Diaconis}\ and\ \citenamefont
  {Shahshahani}(1994)}]{DiacS94}%
  \BibitemOpen
  \bibfield  {author} {\bibinfo {author} {\bibfnamefont {P.}~\bibnamefont
  {Diaconis}}\ and\ \bibinfo {author} {\bibfnamefont {M.}~\bibnamefont
  {Shahshahani}},\ }\bibfield  {title} {\bibinfo {title} {On the eigenvalues of
  random matrices},\ }\href@noop {} {\bibfield  {journal} {\bibinfo  {journal}
  {J. Appl. Probab.}\ }\textbf {\bibinfo {volume} {31}},\ \bibinfo {pages} {49}
  (\bibinfo {year} {1994})}\BibitemShut {NoStop}%
\bibitem [{\citenamefont {Rains}(1998)}]{Rain98}%
  \BibitemOpen
  \bibfield  {author} {\bibinfo {author} {\bibfnamefont {E.~M.}\ \bibnamefont
  {Rains}},\ }\bibfield  {title} {\bibinfo {title} {Increasing subsequences and
  the classical groups},\ }\href@noop {} {\bibfield  {journal} {\bibinfo
  {journal} {Electron. J. Combin.}\ }\textbf {\bibinfo {volume} {5}},\ \bibinfo
  {pages} {R12} (\bibinfo {year} {1998})}\BibitemShut {NoStop}%
\bibitem [{\citenamefont {Calderbank}\ and\ \citenamefont
  {Shor}(1996)}]{CaldS96}%
  \BibitemOpen
  \bibfield  {author} {\bibinfo {author} {\bibfnamefont {A.~R.}\ \bibnamefont
  {Calderbank}}\ and\ \bibinfo {author} {\bibfnamefont {P.~W.}\ \bibnamefont
  {Shor}},\ }\bibfield  {title} {\bibinfo {title} {Good quantum
  error-correcting codes exist},\ }\href
  {https://doi.org/10.1103/PhysRevA.54.1098} {\bibfield  {journal} {\bibinfo
  {journal} {Phys. Rev. A}\ }\textbf {\bibinfo {volume} {54}},\ \bibinfo
  {pages} {1098} (\bibinfo {year} {1996})}\BibitemShut {NoStop}%
\bibitem [{\citenamefont {Steane}(1996)}]{Stea96}%
  \BibitemOpen
  \bibfield  {author} {\bibinfo {author} {\bibfnamefont {A.}~\bibnamefont
  {Steane}},\ }\bibfield  {title} {\bibinfo {title} {Multiple-particle
  interference and quantum error correction},\ }\href
  {https://doi.org/10.1098/rspa.1996.0136} {\bibfield  {journal} {\bibinfo
  {journal} {Proc. R. Soc. Lond. A}\ }\textbf {\bibinfo {volume} {452}},\
  \bibinfo {pages} {2551} (\bibinfo {year} {1996})}\BibitemShut {NoStop}%
\bibitem [{\citenamefont {Howard}\ and\ \citenamefont {Vala}(2012)}]{HowaV12}%
  \BibitemOpen
  \bibfield  {author} {\bibinfo {author} {\bibfnamefont {M.}~\bibnamefont
  {Howard}}\ and\ \bibinfo {author} {\bibfnamefont {J.}~\bibnamefont {Vala}},\
  }\bibfield  {title} {\bibinfo {title} {Qudit versions of the qubit
  $\ensuremath{\pi}/8$ gate},\ }\href
  {https://doi.org/10.1103/PhysRevA.86.022316} {\bibfield  {journal} {\bibinfo
  {journal} {Phys. Rev. A}\ }\textbf {\bibinfo {volume} {86}},\ \bibinfo
  {pages} {022316} (\bibinfo {year} {2012})}\BibitemShut {NoStop}%
\bibitem [{\citenamefont {Lidl}\ and\ \citenamefont
  {Niederreiter}(1997)}]{LidlN97book}%
  \BibitemOpen
  \bibfield  {author} {\bibinfo {author} {\bibfnamefont {R.}~\bibnamefont
  {Lidl}}\ and\ \bibinfo {author} {\bibfnamefont {H.}~\bibnamefont
  {Niederreiter}},\ }\href {https://books.google.com.hk/books?id=xqMqxQTFUkMC}
  {\emph {\bibinfo {title} {Finite Fields}}},\ \bibinfo {series} {Encyclopedia
  of Mathematics and Its Applications}, Vol.~\bibinfo {volume} {20}\ (\bibinfo
  {publisher} {Cambridge University Press},\ \bibinfo {year}
  {1997})\BibitemShut {NoStop}%
\bibitem [{\citenamefont {Berndt}\ \emph {et~al.}(1998)\citenamefont {Berndt},
  \citenamefont {Evans},\ and\ \citenamefont {Williams}}]{BernEW98book}%
  \BibitemOpen
  \bibfield  {author} {\bibinfo {author} {\bibfnamefont {B.~C.}\ \bibnamefont
  {Berndt}}, \bibinfo {author} {\bibfnamefont {R.~J.}\ \bibnamefont {Evans}},\
  and\ \bibinfo {author} {\bibfnamefont {K.~S.}\ \bibnamefont {Williams}},\
  }\href@noop {} {\emph {\bibinfo {title} {Gauss and Jacobi Sums}}}\ (\bibinfo
  {publisher} {Wiley},\ \bibinfo {year} {1998})\BibitemShut {NoStop}%
\bibitem [{\citenamefont {Cameron}(2000)}]{Came00book}%
  \BibitemOpen
  \bibfield  {author} {\bibinfo {author} {\bibfnamefont {P.~J.}\ \bibnamefont
  {Cameron}},\ }\href@noop {} {\bibinfo {title} {Notes on classical groups}}
  (\bibinfo {year} {2000}),\ \bibinfo {note} {available at
  \url{http://www.maths.qmul.ac.uk/~pjc/class_gps/cg.pdf}}\BibitemShut
  {NoStop}%
\bibitem [{\citenamefont {Horn}\ and\ \citenamefont
  {Johnson}(2013)}]{HornJ13book}%
  \BibitemOpen
  \bibfield  {author} {\bibinfo {author} {\bibfnamefont {R.}~\bibnamefont
  {Horn}}\ and\ \bibinfo {author} {\bibfnamefont {C.}~\bibnamefont {Johnson}},\
  }\href@noop {} {\emph {\bibinfo {title} {Matrix Analysis}}},\ \bibinfo
  {edition} {2nd}\ ed.\ (\bibinfo  {publisher} {Cambridge University Press},\
  \bibinfo {address} {Cambridge, UK},\ \bibinfo {year} {2013})\BibitemShut
  {NoStop}%
\end{thebibliography}%

\end{document}